\documentclass[11pt]{book}

\usepackage{latexsym}
\usepackage{amsmath}
\usepackage{amssymb}
\usepackage[all]{xy}

\newtheorem{satz}{Theorem}
\newtheorem{proof}{Proof}
\newtheorem{lemma}{Lemma}
\newtheorem{proofL}{Proof of Lemma}

\begin{document}

\begin{titlepage}
    \vspace*{10ex}
    \huge
    \begin{center}
     \bf{Renormalization and Effective Actions \\ for General Relativity} 
    \end{center}
    \vspace{12ex}
    \Large
    \begin{center}
     Dissertation \\ zur Erlangung des Doktorgrades \\ des Fachbereichs Physik \\ der Universit\"at Hamburg            \end{center}
    \vspace{12ex}
    \begin{center}
     vorgelegt von \\ Falk Neugebohrn \\ aus Filderstadt
    \end{center}
    \vspace{2ex}
    \begin{center}
     Hamburg 2007
    \end{center}
\end{titlepage}

\thispagestyle{empty}
\begin{tabular}[bl!]{ll}
\vspace{16cm}\\
Gutachter der Dissertation:             & Prof.~Dr. Gerhard Mack\\
                                        & Prof.~Dr. Klaus Fredenhagen\\
Gutachter der Disputation:              & Prof.~Dr. Gerhard Mack\\
                                        & Prof.~Dr. Jan Louis\\
Datum der Disputation:                  & 3. April 2007\\
Vorsitzender des Pr\"ufungsausschusses:  & Prof.~Dr. Jochen Bartels\\
Vorsitzender des Promotionsausschusses: & Prof.~Dr. G\"unter Huber\\
Departmentleiter:                       & Prof.~Dr. Robert Klanner\\
Dekan der Fakult\"at f\"ur Mathematik,  &  \\
 Informatik und Naturwissenschaften:    &  Prof.~Dr. Arno Fr\"uhwald
\end{tabular}

\clearpage
\pagenumbering{roman}

~
\vspace{2cm} 

\begin{center}\bf Abstract \end{center}

Quantum gravity is analyzed from the viewpoint of the renormalization group. The analysis is based on methods introduced by J. Polchinski concerning the perturbative renormalization with flow equations. In the first part of this work, the program of renormalization with flow equations is reviewed and then extended to effective field theories that have a finite UV cutoff. This is done for a scalar field theory by imposing additional renormalization conditions for some of the nonrenormalizable couplings. It turns out that one so obtains a statement on the predictivity of the effective theory at scales far below the UV cutoff. In particular, nonrenormalizable theories can be treated without problems in the proposed framework. In the second part, the standard covariant BRS quantization program for Euclidean Einstein gravity is applied. A momentum cutoff regularization is imposed and the resulting violation of the Slavnov-Taylor identities is discussed. Deriving Polchinski's renormalization group equation for Euclidean quantum gravity, the predictivity of effective quantum gravity at scales far below the Planck scale is investigated with flow equations. A fine-tuning procedure for restoring the violated Slavnov-Taylor identities is proposed and it is argued that in the effective quantum gravity context, the restoration will only be accomplished with finite accuracy. Finally, the no-cutoff limit of Euclidean quantum gravity is analyzed from the viewpoint of the Polchinski method. It is speculated whether a limit with nonvanishing gravitational constant might exist where the latter would ultimatively be determined by the cosmological constant and the masses of the elementary particles.

\clearpage

~
\vspace{2cm}

\begin{center}\bf Zusammenfassung\end{center}

Die Quantentheorie der Gravitation wird untersucht vom Blickwinkel der Renor-mierungsgruppe. Die Analyse basiert auf Methoden von J. Polchinski bez\"uglich der pertubativen Renormierung mit Flussgleichungen. Im ersten Teil der Arbeit wird das Programm der Renormierung mit Flussgleichungen vorgestellt und erweitert auf den Fall effektiver Feldtheorien mit endlichem UV Cutoff. Dieses wird durchgef\"uhrt f\"ur eine skalare Feldtheorie durch Einf\"uhrung zus\"atzlicher Renormierungsbedingungen f\"ur einige der nichtrenormierbaren Kopplungskonstanten. Es stellt sich heraus, dass man auf solche Weise eine Aussage \"uber die Prediktivit\"at der Theorie auf Skalen weit unterhalb des UV Cutoffs bekommt. Insbesondere k\"onnen nichtrenormierbare Theorien in dem entwickelten Rahmen problemlos behandelt werden. Im zweiten Teil wird die \"ubliche kovariante BRS Quantisierung der Gravitation durchgefü\"uhrt. Eine Cutoff-Regularisierung wird eingesetzt und die daraus resultierende Verletzung der Slavnov-Taylor Identit\"aten diskutiert. Nach der Herleitung der Polchinski-Renormierungsgruppengleichung f\"ur Euklidische Quantengravitation folgt eine Analyse der Prediktivit\"at effektiver Quantengravitation auf Skalen weit unterhalb der Planckskala mit Hilfe der Flussgleichungen. Eine ''fine-tuning''-Prozedur zur Wiederherstellung der Slavnov-Taylor Identit\"aten wird vorgestellt und es wird dargelegt, dass im Kontext der effektiven Quantengravitation die Wiederherstelllung nur mit endlicher Genauigkeit gelingen kann. Zuletzt folgt eine Analyse des No-Cutoff Limes der Quantengravitation im Rahmen der Polchinski-Methode. Es wird spekuliert, dass ein solcher Limes f\"ur nichtverschwindendene Werte der Gravitationskonstante existieren k\"onnte, wobei letztere dann bestimmt w\"are durch die kosmologische Konstante und die Massen der Elementarteilchen.
\clearpage

\bibliographystyle{plain}
\tableofcontents

\chapter{Introduction}
\pagenumbering{arabic}
One of the most intriguing unsolved problems in contemporary physics is the unification of Einstein's theory of gravitation, general relativity, with quantum mechanics. It is well-known that the traditional approach to such a quantum theory of gravity, i.e. the covariant quantization program which treats general relativity as a flat-space field theory, leads to severe conceptual difficulties. In particluar, it turns out that the theory is perturbatively nonrenormalizable \cite{Hooft} \cite{Goro}. This means an infinite number of free parameters has to be fixed.  For this reason, Einstein gravity was long considered a failure as a quantum field theory, and different strategies have been pursued. Among those are string theory and loop quantum gravity. However, until today none of these approaches has been fully successful.

In the recent years, there have been some renewed efforts concerning the quantization of general relativity along the lines of ''conventional'' quantum field theory. 

On the one hand, there have been attemps to treat quantum gravity as an effective field theory \cite{Dono1} in the framework of perturbation theory. An effective theory, unlike a truly fundamental one, cannot be valid up to arbitrary scales. For the case of quantum gravity, the limiting scale is the Planck energy ($\sim 1.2 \cdot 10^{19}$ GeV ) as indicated by the size of the gravitational constant. In the effective field theory approach, the nonrenormalizability of general relativity is not an issue, and it has been shown that quantum corrections to Newton's potential as well as the Schwarzschild and Kerr metrics can be derived \cite{BB2} \cite{BB3}. However, one has to content oneself with predictions of finite accuracy, and one does not get new insights on the physics above the Planck scale.

On the other hand, the nonperturbative properties of (Euclidean) quantum gravity have been studied using the renormalization group flow equations \cite{MR1}. Evidence for the existence of a  non-Gaussian fixed point in the space of couplings has been found which would render the theory nonperturbatively renormalizable along the lines of Weinberg's ''aymptotic safety'' scenario \cite{Wein2}. In particular, one may hope that only a finite number of free parameters has to be fixed, corresponding to a finite number of dimensions of the UV critical surface \cite{MR2}. Applications to cosmology, black hole physics and the structure of spacetime have been discussed \cite{MR3} \cite{MR4} \cite{MR5}.  However, in order to perform the analysis, the space of actions has to be truncated. Thus, one cannot completely rule out the possibility that the non-Gaussian fixed point is only an artefact of the truncation. 

One could ask if there is any relation between the effective field theory approach (which implicitly relies on the Gaussian fixed point) and the nonperturbative renormalization group analysis. In particular, the latter (being the more general one) should have a certain domain where the former is valid.

Asking so one is led to a paper by Polchinski \cite{Pol} concering the perturbative renormalization of scalar $\phi^4$ theory by means of the renormalization group flow equations. In Polchinski's proof of renormalizability, no Feynman diagrams are needed and the cumbersome analysis of overlapping divergences employing Zimmerman's forest formula is avoided. Instead, the task is accomplished by bounding inductively the solutions of the renormalization group flow equations, which are a system of first order differential equations. It is crucial for the analysis that the couplings are small, i.e. that one is in the vicinity of the Gaussian fixed point. The method has turned out to be particularly simple and also transparent from a conceptual point of view.

It is therefore not surprising that Polchinski's work has stimulated various contributions over the following years.
 Among those are a simplified and mathemetically rigorous version of Polchinski's original proof employing physical renormalization conditions \cite{KKS}, extensions to composite operator renormalization \cite{KK3} \cite{KK4} and Symanzik's improved actions \cite{Wiec}, and a proof of the perturbative renormalizability of QED via flow equations \cite{KK5}. Finally, also the perturbative renormalization of spontaneously broken Yang-Mills theory has been accomplished \cite{KM} in the framework of the flow equations.

The reader notices that an analysis of quantum gravity along the lines of the Polchinski method is still missing. This is what the present work is about.

Since Polchinski's concepts implicitly rely on the Gaussian fixed point (as does perturbation theory), it seems likely that for such an analysis one will again be confronted with the perturbative nonrenormalizability of general relativity. However, it will be shown that it is possible to extend the method of renormalization via flow equations to the case of a (perturbatively) nonrenormalizable theory if one treats it as an effective field theory that has a finite cutoff. In fact, it will turn out that in doing so, one obtains a statement on the amount of predictivity the effective field theory has at scales far below the cutoff. One aim of the present work is therefore to investigate the predictivity of effective quantum gravity at energies that are small as compared to the Planck scale.

Applying a renormalization group analysis to quantum gravity, one has to surmount another serious obstacle. The cutoff regularization that has to be imposed in oder to obtain the flow equations inevitably violates the local gauge invariance of any gauge theory. For the case of general relativity, this means the symmetry under general coordinate transformations will be broken. Since on the quantum level gauge invariance is expressed in terms of BRS invariance and Slavnov-Taylor identities for Green's functions, one will end up with violated Slavnov-Taylor identities. It has been shown \cite{KM} that the analogous problem appearing in the renormalization of Yang-Mills theory with flow equations can be cured by a so-called \textit{fine-tuning} procedure which ultimatively aims at the restoration of the Slavnov-Taylor identities in the no-cutoff limit. In the present work, a similar procedure for effective quantum gravity will be proposed. The important difference to the Yang-Mills case lies in the fact that since a finite cutoff has to be retained, the restoration of the Slavnov-Taylor identities  will only be accomplished with \textit{finite accuracy}.

Although the results that are obtained from our analysis are already known from the ''conventional'' treatment of Einstein gravity as an effective field theory  \cite{Dono1} using dimensional regularization etc., in our opinion the investigation with flow equations has the benefit to be particularly systematic and transparent while employing the ''modern'' language of the renormalization group. In fact, the methods that are developed in this work for investigating effective field theories with flow equations should be applicable to \textit{any} effective field theory, not just gravity. Comparing to a full nonperturbative renormalization group treatment $\grave{a}$ la \cite{MR1}, our results will only be valid near the Gaussian fixed point. However, we do not need to truncate the space of actions, and it will be shown that for generic initial conditions the effective Lagrangian of quantum gravity is attracted towards a finite dimensional submanifold in the space of possible Lagrangians at scales far below the Planck scale. 

Finally, it has turned out that the Polchinski analysis of quantum gravity allows for an interesting speculation concerning the no-cutoff limit of the theory. In the framework of the flow equations, a nonrenormalizable theory is defined as a theory with field and symmetry content such that \textit{no} renormalizable interactions, except for kinetic and mass terms, are permitted. In the no-cutoff limit, such theories become free when an analysis $\grave{a}$ la Polchinski is applied. We find that quantum gravity without a cosmological constant is nonrenormalizable in the described sense, as it has already been conjectured by S. Weinberg in his book on quantum field theory \cite{Wein1}. However, for \textit{nonzero} values of the cosmological constant the situation is different because some renormalizable interactions are introduced. This gives rise to the speculation whether quantum gravity with cosmological constant has a no-cutoff limit with \textit{nonvanishing} gravitational constant. Moreover, since in the no-cutoff limit the nonrenormalizable running couplings are determined by the renormalizable ones, the value of the gravitational constant should then be determined by the cosmological constant. It is furthermore pointed out that a related situation may be produced by coupling massive fields to gravity, leading to the speculation of a gravitational constant that is determined by the cosmological constant and the masses of the elementary particles. The latter might indicate a deeper relation between the Higgs mechanism and the gravitational force.

\medskip
The structure of the thesis goes as follows. In chapter ($\ref{OvM}$), we will give an introduction into the key concepts of renormalization with flow equations. Furthermore, the strategy for extending the method to effective field theories is outlined, ultimatively leading to estimates concerning the predictivity of effective theories at scales far below the cutoff. The overview in chapter  ($\ref{OvM}$) should serve as some kind of map through the inductive proofs carried out in perturbation theory in the next two chapters.

The perturbative renormalization of scalar field theory via flow equations is reviewed in full detail in chapter ($\ref{RenFlow}$). We proceed along the lines of \cite{KKS} who presented an improved and considerably shortened version of Polchinski's original proof \cite{Pol}. As compared to their version, we include the following generalizations. Instead of $\phi^4$, we will consider $\phi^3 + \phi^4$ theory\footnote{Thus, we do not restrict ourselves to couplings assigned to operators with even numbers of fields.}. Moreover, nonvanishing bare values of the nonrenormalizable couplings are allowed for from the very beginning, and an alternative proof of the uniqueness of the no-cutoff limit is employed. These generalizations will turn out essential when the concepts are applied to effective quantum gravity in chapter ($\ref{PredGrav}$).

In chapter ($\ref{EffFlow}$), the extension of the Polchinski method to effective field theories is developed rigorously for a scalar field theory following the strategy described in chapter ($\ref{OvM}$). To do so, the vertex functions of an effective potential being the solution of the Polchinski renormalization group equation are expanded in generalized\footnote{This means perturbation theory in the renormalized renormalizable and some bare nonrenormalizbale couplings.} perturbation theory. New bounds for the vertex functions are established in an attempt to unify the results of \cite{Wiec} and \cite{KKS}, and additional renormalization conditions for some of the nonrenormalizable couplings are imposed at an arbitrary renormalization scale. These will be referred to as ''improvement conditions''. We prove that there exist small initial conditions for the nonrenormalizable couplings at the bare scale such that appropriately chosen improvement conditions can be met. The main advantage of our approach as compared to \cite{Wiec} is that the case of vanishing renormalizable couplings does not pose any problems and thus also nonrenormalizable theories can be treated. Finally, it will be established that the improvement conditions lead to an enhanced predictivity of the effective field theory at scales far below the cutoff.  

In chapter ($\ref{PolGR}$) we turn to quantum gravity. The standard covariant BRS quantization procedure for Euclidean Einstein gravity is reviewed, as it can be found for instance in \cite{Stelle} and \cite{CLM}. Our dynamical variable is a perturbation of the metric density $\sqrt{g} \ g^{\mu \nu}$ around flat space. Unlike the authors mentioned, we allow for nonvanishing values of the cosmological constant\footnote{The problems that arise due to the cosmological term when gravity is treated as a flat-space QFT are discussed, as well as possible resolutions. }.  
A momentum cutoff regularization for the generating functional of quantum gravity is introduced and the resulting violation of the gauge invariance and hence the Slavnov-Taylor-identities (STI) is discussed. We review a fine-tuning procedure that has been shown \cite{KM} to cure the analogous problem occuring in the perturbative renormalization of Yang-Mills theory via flow equations, and propose first implications of a similar procedure for quantum gravity. Finally, the Polchinski renormalization group equation for Euclidean quantum gravity is derived.

The concepts developed in chapters ($\ref{RenFlow}$) and ($\ref{EffFlow}$) are applied to Euclidean quantum gravity in chapter ($\ref{PredGrav}$). A bare action containing all field invariants that are permitted by general coordinate invariance is introduced, and the relation to higher derivative gravity is discussed. It is argued that in the effective field theory approach the known unitarity problems \cite{Stelle} will not appear. As a first step of the analysis, we disregard the violation of the Slavnov-Taylor identities and establish bounds for vertex functions of the gravity effective potential in generalized perturbation theory. It is shown that by introducing appropriate notations, we may proceed in close analogy to the case of the scalar field theory considered in chapters ($\ref{RenFlow}$) and ($\ref{EffFlow}$).  A set of (for the time being) arbitrary renormalization and improvement conditions is imposed. By inverting the renormalization group trajectory, it is argued that the improvement conditions force the cutoff of effective quantum gravity to be the Planck scale. Finally, we establish that the family of theories described by the arbitrary renormalization and improvement conditions is predictive at scales far below the Planck scale with finite accuracy. 

We then proceed to the restoration of the STI. Introducing bare regularized BRS variations, the violated Slavnov-Taylor identities (vSTI) are worked out.  Bounds for vertex functions carrying the  nonlinear BRS variations as operator insertions are established, and we note that a crucial difference to the Yang-Mills case lies in the fact that the gravity BRS fields contain \textit{nonrenormalizable parts}. By imposing renormalization and improvement conditions for the BRS variations, it is proven that the dependence of these vertex functions on the bare initial conditions is suppressed at scales far below the Planck scale\footnote{This is similar to the statements concerning the predictivity of the effective theory.}. It turns out that the violation of the STI can be described in terms of vertex functions carrying a space-time integrated operator insertion having canonical dimension $5$. It is therefore argued that the STI can be restored to \textit{finite accuracy} if {one particular} set of arbitrary renormalization and improvement conditions for the couplings and BRS variations can be determined such that the relevant and leading irrelevant parts of the vertex functions describing the violation of the STI are driven small at scales far below the Planck scale. Here, ''small'' means the order of accuracy to which the theory is predictive. 

In the last section of chapter ($\ref{PredGrav}$), we consider the no-cutoff limit $\Lambda_0 \rightarrow \infty$ of quantum gravity from the viewpoint of the analysis with flow equations. The vertex functions of the gravity effective potential are expanded solely in the renormalizable couplings, and their boundedness and convergence is established in the limit $\Lambda_0 \rightarrow \infty$. Applying the same program to the vertex functions carrying the nonlinear BRS variations as operator insertions, we observe that the nonrenormalizable parts of the gravity BRS fields will go away in the no-cutoff limit if smallness of the bare BRS couplings is imposed. It is shown that if the latter constraint is dropped, convergence of the BRS vertex functions may still be proven. Proceeding with the restoration of the STI, we argue that for zero renormalized cosmological constant $\Lambda_K=0$ the theory will become free as $\Lambda_0 \rightarrow \infty$, and that the latter statement is compatible with gauge invariance. It is speculated whether a \textit{nonzero} cosmological constant $\Lambda_K \ne 0$ might lead to a \textit{nonvanishing} value of the gravitational constant in the no-cutoff limit, and it is pointed out that the gravitational coupling should then become determined by the cosmological constant. Finally, we observe that a similar effect might be obtained by coupling massive fields to gravity, leading to speculations if the gravitational constant is given in terms of the cosmological constant and the masses of the elementary particles as $\Lambda_0 \rightarrow \infty$.

In chapter ($\ref{Outlook}$) we conclude with a discussion and an outlook. Some relations of our results to the ''conventional'' treatment \cite{Dono1} of quantum gravity as an effective field theory as well to the nonperturbative analysis employing the renormalization group flow equations \cite{MR1} are given.


\newcommand{\ta}{\otimes_{\mathcal{A}}}
\newcommand{\fund}{(\frac{1}{2},0)}
\newcommand{\coco}{(0,\frac{1}{2})}
\newcommand{\A}{\mathcal{A}}

\chapter{Overview of the Method} \label{OvM}
In this chapter we will give an introduction into renormalization with flow equations and into effective field theories from the viewpoint of the renormalization group. Both topics will be investigated in full detail in perturbation theory in chapters ($\ref{RenFlow}$) and ($\ref{EffFlow}$). This overview serves as some kind of map through the perturbative woods.

\begin{section}[The Wilson/Polchinski RGE]{The Wilson/Polchinski Renormalization Group Equations}  \label{RGInt}

Let $W(J)$ be the generating functional of an Euclidean Quantum Field Theory (QFT) involving a scalar field $\phi$ and the action $S(\phi)$,
\begin{equation}
W(J)= \int \mathcal{D} \phi e^{  S(\phi) + \langle J, \phi  \rangle }.  \label{W}
\end{equation}
$J$ denotes an external source, and the scalar product $\langle, \rangle$ refers to some position or momentum space integration. The functional ($\ref{W}$) contains all information of the QFT. Suppose that we are only interested in the physics below some scale $\Lambda$. We split up the fields according to their momentum degrees of freedom:
\begin{eqnarray}
\phi &=& \phi_L + \phi_H \nonumber \\
&\phi_H: &  k^2 > \Lambda^2 \nonumber \\
&\phi_L:& k^2 < \Lambda^2 .
\end{eqnarray}
In addition, we restrict the source term\footnote{This is not an essential ingredient. See Appendix ($\ref{LZ}$) for details.} to momenta below $\Lambda$, 
\begin{equation}
J=0 \ \ \text{for} \ \ k^2 > \Lambda^2.
\end{equation}
If we integrate out $\phi_H$, we are left with a functional integral with an upper frequency cutoff $\Lambda$ and an {effective action} $S_e(\phi_L)$:
\begin{eqnarray}
W(J) &=& \int \mathcal{D} \phi_L \mathcal{D} \phi_H  e^{ S(\phi_L 
+ \phi_H ) + \langle J, \phi_L \rangle } \nonumber \\ &=& \int \mathcal{D} \phi_L e^{S_e(\phi_L) + \langle J, \phi_L \rangle }  \label{intout}
\end{eqnarray}
where
\begin{equation}
e^{ S_e(\phi_L)} := \int \mathcal{D} \phi_H  e^{ S(\phi_L + \phi_H )}.
\end{equation}
Generally, $S_e$ will contain all possible interactions of the fields $\phi_L$ and their derivatives as a compensation for the removal of the Fourier modes $\phi_H$.

The change of $S_e$ while integrating out field modes is described by a renormalization group differential equation (RGE). This is the {Wilson equation} for the effective action:
\begin{equation}
- \Lambda \frac{d}{d \Lambda} S_e(\Lambda) = \mathcal{F}(S_e(\Lambda)) . \label{WRG}
\end{equation}
In the following, we will consider Polchinski's version \cite{Pol} of the RGE ($\ref{WRG}$).  Therefore we write\footnote{In chapter ($\ref{RenFlow}$), we will employ a slightly refined version of ($\ref{Spol}$), see eq. ($\ref{SpolM}$). }
\begin{eqnarray}
S_e(\phi,  \Lambda) = -\frac{1}{2} \langle \phi, \Delta^{-1}_\Lambda \phi \rangle + L(\phi, \Lambda)  \label{Spol}
\end{eqnarray}
where $\Delta_\Lambda$ is the propagator of the QFT multiplied with some cutoff function $K(k^2/ \Lambda^2)$, and $L(\phi, \Lambda)$ is a (not necessarily local) interaction term. The cutoff function is taken to be
\begin{eqnarray}
K(z) = \left\{ \begin{array}{c c c}   1 & ,&  0 \le z \le 1 \\ \mbox{smooth} & ,& 1 < z < 4   \\  0 & ,& 4 \le z . \end{array} \right.  \label{cutf}
\end{eqnarray}
Hence, $K(k^2/ \Lambda^2)$ suppresses the propagation of field modes with momenta $k^2>\Lambda$. If we demand invariance of $W(J)$ under a change of the cutoff $\Lambda$ as is implied by eq. ($\ref{intout}$),
\begin{eqnarray}
 \Lambda \frac{d}{d \Lambda} W(J) \stackrel{!}{=} 0,
\end{eqnarray}
it turns out that the effective potential $L(\phi, \Lambda)$ must satisfy \textit{Polchinski's equation}. In coordinate space it reads
\begin{eqnarray}
- \Lambda \frac{d}{d \Lambda} L  = \frac{1}{2} \int_x \int_y  \left( \Lambda \frac{d}{d \Lambda} \Delta_\Lambda \right) \left(\frac{\delta L}{\delta \phi(x)} \frac{\delta L}{\delta \phi(y)}   + \frac{\delta^2 L}{\delta \phi(x) \delta\phi(y)} \right) . \nonumber \\ \label{pol}
\end{eqnarray}
Eq. ($\ref{pol}$) has a simple graphical interpretation: as modes are removed from the propagator, compensating terms must be added in the effective potential  $L(\phi, \Lambda)$.

Let $\mathcal{O}_i(x, \phi)$ be local composite field operators and let $D_{\mathcal{O}_i}$ denote their canonical dimensions\footnote{See section ($\ref{GENR}$) for the determination of the canonical dimension of a field.}. We define initial conditions (a ''bare potential'') at some scale $\Lambda_0$ which from now on will be referred to as the UV cutoff of the QFT:
\begin{equation}
L^0(\phi, \Lambda_0)= \sum_{i} \rho^0_i \int_x \mathcal{O}_i(x, \phi).    \label{barepot}
\end{equation}
The coefficients $ \rho^0_i$ are called ''bare couplings''. They have canonical dimensions
\begin{equation}
D_{\rho^0_i} = d - D_{\mathcal{O}_i}
\end{equation}
where $d$ is the number of space-time dimensions. Solving the Polchinski RGE ($\ref{pol}$) employing the initial conditions ($\ref{barepot}$) now yields a trajectory
\begin{equation}
[\Lambda, \Lambda_0] \rightarrow L(\phi, \Lambda, \Lambda_0, \rho_i^0)  \label{traj}
\end{equation}
which, in turn, leads to a generating functional depending on the initial conditions\footnote{Sometimes we will just write sloppily $W(J, \Lambda_0)$ instead of $W(J, \Lambda_0, \rho_i^0)$. } at $\Lambda_0$:
\begin{equation}
W(J, \Lambda_0, \rho_i^0)= \int \mathcal{D} \phi e^{ -\frac{1}{2} \langle \phi, \Delta^{-1}_\Lambda \phi \rangle + L(\phi, \Lambda, \Lambda_0, \rho_i^0) + \langle J, \phi \rangle }.  \label{WP}
\end{equation}
It is always possible to perform a (position-space) derivative expansion of the effective potential $L(\phi, \Lambda)$ into local composite field operators and a nonlocal remainder term  $R^{(s)}$,
\begin{eqnarray}
L(\phi, \Lambda) &=& \sum_{D_{\rho_i} \ge -s} \rho_i(\Lambda) \int_x \mathcal{O}_i(x, \phi) +  R^{(s)}(\phi) . \label{effpot}
\end{eqnarray} 
$\rho_i(\Lambda)$ are {running coupling constants} with associated canonical dimensions $D_{\rho_i} = d - D_{\mathcal{O}_i}
$, and $s \in \mathbb{N}$ is some index. See Appedix ($\ref{MDV}$) for more details on the expansion ($\ref{effpot}$).

In the following, we will often need dimensionless running coupling constants which are defined by
\begin{equation}
\lambda_i(\Lambda) := \Lambda^{-D_{\rho_i}} \rho_i(\Lambda). \label{dlcoup}
\end{equation}

Furthermore, it will turn out necessary to consider small deviations 
\begin{equation}
\delta L (\Lambda):= L(\Lambda) - \overline{L}(\Lambda)
\end{equation}
from a solution $ \overline{L}(\Lambda)$ of the RGE ($\ref{pol}$). They obey a {linearized Polchinski RGE}:
\begin{eqnarray}
- \Lambda \frac{d}{d \Lambda} \delta L &=& \frac{1}{2} \int_x \int_y  \left( \Lambda \frac{d}{d \Lambda} \Delta_\Lambda \right)  \left(2 \frac{\delta \overline{L}}{\delta \phi(x)} \frac{\delta }{\delta \phi(y)}   + \frac{\delta^2}{\delta \phi(x) \delta\phi(y)} \right) \delta L \nonumber \\ &:=& M (\delta L) . \label{poll}
\end{eqnarray}

\end{section}

\begin{section}{Renormalization via Flow Equations}  \label{RenFlowOver}
Let us begin by recalling what renormalizability of a Quantum Field Theory means. Therefore, we have to distinguish between renormalizable (or \textit{relevant}\footnote{We do not distinguish between marginal and relevant couplings.}) couplings $\rho_a$ and nonrenormalizable (or \textit{irrelevant}) couplings $\rho_n$ by their associated canonical dimensions $D_{\rho_i}$:
\begin{eqnarray}
\rho_a :& & D_{\rho_a} \ge 0 \nonumber \\
\rho_n :& &D_{\rho_n} < 0.
\end{eqnarray}
The significance of this distinction will become clear throughout this section\footnote{The meaning of relevant, marginal and irrelevant directions in the space of actions near a fixed point is explained in Appendix ($\ref{LRG}$)}. Having employed an UV momentum cutoff regularization as proposed in the last section, the renormalizability of a QFT can be summarized as follows.

\bigskip
\noindent
\textbf{Renormalizability}: \textit{While taking the UV cutoff  $\Lambda_0$ to infinity holding the renormalizable couplings fixed at some renormalization scale $\Lambda_R < \Lambda_0$, all other quantities of the theory must converge to limits as inverse powers of the UV cutoff}.

\bigskip
Note that this in particular means that the running nonrenormalizable couplings become determined by the renormalizable ones in the no-cutoff limit. We will come back to this important point in more detail. If the field and symmetry content of a theory is such that it does not permit \textit{any} renormalizable interactions, that is no couplings except for kinetic and mass terms, it is called \textit{nonrenormalizable} \cite{Wein1}. As we will discuss at the end of section ($\ref{EffFlowOver}$), the running nonrenormalizable couplings then vanish in the no-cutoff limit.

In the following, we will go through the essential steps of the renormalization programme for the scalar field theory introduced in the last section employing the method of flow equations proposed by Polchinski \cite{Pol}. We begin by choosing the initial conditions ($\ref{barepot}$), i.e. the bare couplings $\rho_i^0$, to lay on an {initial surface} in the space of couplings whose coordinates are the bare renormalizable couplings: 
\begin{equation}
\rho_n^0= \rho_n^0 (\rho_a^0). \label{inis}
\end{equation}
Hence, the dimension of the initial surface amounts to the number of renormalizable couplings.  A particular simple, but not necessary choice for the surface ($\ref{inis}$) would be to take $\rho_n^0= 0$. 

The solution of eq. ($\ref{pol}$) now becomes
\begin{equation}
L=L(\phi, \Lambda, \Lambda_0, \rho_a^0).
\end{equation}
As mentioned above, renormalization requires to fix the renormalizable couplings $\rho_a$ at some renormalization scale $\Lambda_R$. This is done by specifying {renormalization conditions} for them:
\begin{equation}
\rho_a(\Lambda_R, \Lambda_0, \rho_b^0) = \rho^R_a .  \label{rcond}
\end{equation}
The renormalization conditions imply that
\begin{eqnarray}
\Lambda_0 \frac{d}{d \Lambda_0} \rho_a(\Lambda_R, \Lambda_0, \rho_b^0) &=& \Lambda_0 \left( \frac{\partial \rho_a }{\partial \Lambda_0} +  \frac{\partial \rho_a}{\partial \rho_b^0}  \frac{\partial \rho_b^0}{\partial \Lambda_0}         \right)   \nonumber \\ &{=}& 0. \label{fix}
\end{eqnarray}
This leads to an {implicit definition} of the bare ''coordinates'' $\rho_a^0$ as functions of the renormalizable couplings, the UV cutoff and the renormalization scale:
\begin{equation}
\rho^0_a= \rho^0_a(\Lambda_R, \Lambda_0, \rho^R_a). \label{impli}
\end{equation}
Eq. ($\ref{impli}$) means that we do not have to know about the initial values $\rho_a^0$ at the bare scale $\Lambda_0$ in order to solve the Polchinski RGE ($\ref{pol}$), but instead may use the renormalization conditions $\rho^R_a$ at some arbitrary renormalization scale $\Lambda_R$ as an input\footnote{Together with the definition of the initial surface ($\ref{inis}$). However, it will turn out that within certain limits, its exact form has no influence on the (renormalized) theory in the limit $\Lambda_0 \rightarrow \infty$.}. This is a good thing, because typically the bare scale is not accessible to any measurements. 

If $\partial \rho_a / {\partial \rho_b^0}$ is invertible, eq. ($\ref{fix}$) yields
\begin{eqnarray}
\frac{\partial \rho_b^0}{\partial \Lambda_0} = - \left( \frac{\partial \rho_a}{\partial \rho_b^0} \right)^{-1} \frac{\partial \rho_a }{\partial \Lambda_0} .
\end{eqnarray}
Therefore, at $\Lambda=\Lambda_R$  the quantity
\begin{eqnarray}
V(\Lambda) := \Lambda_0 \left( \frac{\partial L }{\partial \Lambda_0} -   \frac{\partial L }{\partial \rho_b^0}  \left( \frac{\partial \rho_a}{\partial \rho_b^0} \right)^{-1} \frac{\partial \rho_a }{\partial \Lambda_0} \right)  \label{v} 
\end{eqnarray}
is the total derivative of $L(\Lambda_R)$ with respect to $\Lambda_0$ holding the $\rho_a(\Lambda_R)$ fixed:
\begin{eqnarray}
V(\Lambda_R)&=& \left. \Lambda_0 \frac{d }{d \Lambda_0} L \right|_{\rho_a(\Lambda_R)=\text{fixed}} . \label{vt}
\end{eqnarray}
Note that the last term on the RHS of ($\ref{v}$) can be  interpreted such that the part of the original $\Lambda_0$-dependence of $L$ which is due to the renormalizable couplings is subtracted of. Thus $V(\Lambda_R)$ points towards the {irrelevant directions} in the space of couplings. This becomes even more obvious if we think of $L$ as given by a derivative expansion ($\ref{effpot}$) with $s=\infty$. Then $V(\Lambda)$ can be expressed in terms of running couplings $\rho_i(\Lambda)$:
\begin{eqnarray}
V_i(\Lambda) = \Lambda_0 \left( \frac{\partial \rho_i }{\partial \Lambda_0} -   \frac{\partial \rho_i }{\partial \rho_b^0}  \left( \frac{\partial \rho_a}{\partial \rho_b^0} \right)^{-1} \frac{\partial \rho_a }{\partial \Lambda_0} \right)  . \label{v_i} 
\end{eqnarray}
Clearly, we have $V_a(\Lambda)=0$, which in turn means that all relevant components of ($\ref{v}$) vanish. 

As is explained in Appendix ($\ref{LRG}$), the nonrenormalizable directions in the space of couplings correspond to negative eigenvalues of the linearized RG transformation in the zero-coupling limit. Hence, deviations in these directions get damped away in the infrared as powers of $\Lambda/\Lambda_0$. Since there is nothing discontiuous about the RG transformation as the couplings are changed, the latter will still be true for sufficiently small couplings\footnote{That is in the vicinity of the Gaussian fixpoint.}. Thus, we expect that the quantity $V(\Lambda_R)$ should be driven small as powers of $\Lambda_R/\Lambda_0$ for small enough  couplings.

To prove this, we have to find a RGE for $V(\Lambda)$. Since $\Lambda_0 \frac{\partial }{\partial\Lambda_0} L$ and $\frac{\partial }{\partial \rho^0_a} L$ both satisfy the linearized Polchinski RGE ($\ref{poll}$), we obtain
\begin{eqnarray}
- \Lambda \frac{d}{d \Lambda}  V &=& M (V) -  \frac{\partial L}{\partial \rho_b^0}  \left( \frac{\partial \rho_a}{\partial \rho_b^0} \right)^{-1}  M_{a} (V) \nonumber \\ &:=& N(V). \label{RG_v}
\end{eqnarray}
In the next chapter, we will investigate in detail an equation similar to ($\ref{RG_v}$) in perturbation theory. In order to get an idea of the results that can be expected of such an analysis, we may approximate the operator $N$ by the canonical dimension $D_{\rho_l}$ of the \textit{least} irrelevant coupling $\rho_l$ of the QFT:
\begin{equation}
N(V) \approx  D_{\rho_l} \ V .   \label{Neasy}
\end{equation}
This is justified for sufficiently small couplings by the reasoning above. By virtue of ($\ref{Neasy}$), we may now easily integrate eq. ($\ref{RG_v}$). As integration limits we choose $\Lambda_R$ and $\Lambda_0$ and arrive at
\begin{equation}
V(\Lambda_R) \approx V(\Lambda_0) \left( \frac{\Lambda_R}{\Lambda_0} \right)^{-D_{\rho_l}} . \label{vis}
\end{equation}
For some proper definition of a norm $|| \ ||$, the local composite operators of the effective potential ($\ref{effpot}$) can be estimated as  $|| \int_x \mathcal{O}_i(x, \phi) || \sim \Lambda^{-D_{\rho_i}}$. We assume that
\begin{equation}
||V(\Lambda_0)|| \le 1 \label{smallv0}
\end{equation}
corresponding to sufficiently small initial values of the nonrenormalizable couplings\footnote{Remember that only $V_n(\Lambda) \ne 0$.}:
\begin{equation}
\rho^0_n \le \Lambda_0^{D_{\rho_n}}, \ \ \ D_{\rho_n}< 0.   \label{sini}
\end{equation}
Note that in terms of the dimensionless couplings defined in eq. ($\ref{dlcoup}$), eq. ($\ref{sini}$) means $\lambda_n(\Lambda_0) \le 1$. With eq. ($\ref{vt}$) we thus may write
\begin{equation}
 || \left. \Lambda_0 \frac{d }{d \Lambda_0} L \right|_{\rho_a=\text{fixed}}|| \le  \left( \frac{\Lambda_R}{\Lambda_0} \right)^{ -D_{\rho_l}}
\end{equation}
which again can be integrated with respect to $\Lambda_0$. The result is
\begin{eqnarray}
|| L(\Lambda_R, \Lambda'_0) -  L(\Lambda_R, \Lambda''_0) || &\le&  \left( \frac{\Lambda_R}{\Lambda'_0} \right)^{-D_{\rho_l}}-\left( \frac{\Lambda_R}{\Lambda''_0} \right)^{-D_{\rho_l}} .  \label{conv}
\end{eqnarray}
From Cauchy's criterion now follows the existence of a no-cutoff-limit limit\footnote{We adopt the notation $L^{cont}$ from \cite{Wiec}.}  $L^{cont}(\Lambda_R) = L(\Lambda_R, \infty)$  which implies \textit{renormalizability} as we have defined it at the beginning of this section. For having fixed the renormalizable couplings via the renormalization conditions ($\ref{rcond}$), all other quantities of the theory (contained in $L$) converge to limits as inverse powers of the UV cutoff $\Lambda_0$. This result can also be expressed on the level of the effective action:  
\begin{equation}
|| S_e(\Lambda_R,\Lambda_0, \rho_a^0(\Lambda_R, \Lambda_0, \rho^R_a))- S_e^{cont}( \Lambda_R, \rho^R_a) || \le \left( \frac{\Lambda_R}{\Lambda_0} \right)^{-D_{\rho_l}} \label{convS}
\end{equation}
where $S_e^{cont}( \Lambda_R, \rho^R_a):= \lim_{\Lambda_0 \rightarrow \infty} S_e(\Lambda_R,\Lambda_0, \rho_a^0(\Lambda_R, \Lambda_0, \rho^R_a))$. Note that once the effective action is known, we may calculate the generating functional $W(J)$ of the QFT via ($\ref{intout}$) without needing to worry about divergences, because $S_e^{cont}( \Lambda_R, \rho^R_a)$ has an effective cutoff $\Lambda_R$.

We would like to stress that it is \textit{not} being said that the running nonrenormalizable couplings $\rho_n(\Lambda, \Lambda_0)$ go to zero when $\Lambda_0 \rightarrow \infty$, but that their values $\rho_n^{cont}(\Lambda):= \lim_{\Lambda_0 \rightarrow \infty} \rho_n(\Lambda, \Lambda_0)$ are {determined} by the renormalizable couplings $\rho_a(\Lambda)$ in this limit. This can be understood intuitively by noting that while integrating out field modes, the renormalizable couplings generate new contributions to the nonrenormalizable ones.

Furthermore, it is crucial that the initial values of the nonrenormalizable couplings are small in the sense of eq. ($\ref{sini}$).  Keeping $\rho^0_n$ of inverse powers of some smaller scale $\Lambda_D < \Lambda_0$ would destroy the convergence of $S_e$ to its no-cutoff limit $S_e^{cont}$ since $V(\Lambda_0)$ would be blown up .

The results stated above, in particular the no-cutoff limits $L^{cont}$ and $S_e^{cont}$, are {independent} of the exact choice of the initial surface in eq. ($\ref{inis}$) as long as the initial values of the nonrenormalizable couplings are taken to be sufficiently small. This can be seen as follows. Consider a given initial surface $\rho_n^0= \rho_n^0 (\rho_a^0)$. An easy way to change its ''shape'' is employing a parametrization
\begin{eqnarray}
\rho_n^0 \rightarrow \tilde{\rho}_n^0 := t \rho_n^0 , \ \ \ t \in [0,1].  \label{shape}
\end{eqnarray}
The case $t=0$ corresponds to $\rho_n^0=0$, whereas $t=1$ yields again the original surface $\rho_n^0 (\rho_a^0)$. Moreover, we also allow for the ''coordinates'' $\rho_a^0$, i.e. the bare renormalizable couplings,  to depend on the parameter $t$:
\begin{eqnarray}
\rho_a^0= \rho_a^0(t). 
\end{eqnarray}
The total dependence of the potential $L$ on $t$, $\Lambda$ and $\Lambda_0$ now reads
\begin{eqnarray}
L &=& L \left(\Lambda, \Lambda_0,\rho_a^0(t),  \tilde{\rho}_n^0(t,\rho_a^0(t) )\right) \nonumber \\ & \equiv &  L \left(\Lambda, \Lambda_0,\rho_a^0(t), t \right). \label{tdep}
\end{eqnarray}
In analogy to the definition of $V(\Lambda)$ we define the quantity
\begin{eqnarray}
W(\Lambda) :=   \frac{\partial L }{\partial t}  -   \frac{\partial L }{\partial \rho_b^0}  \left( \frac{\partial \rho_a}{\partial \rho_b^0} \right)^{-1} \frac{\partial \rho_a }{\partial t} .  \label{w} 
\end{eqnarray}
At $\Lambda= \Lambda_R$,  $W$ is the total derivative of $L$ with respect to the parameter $t$ holding the renormalizable couplings $\rho_a(\Lambda_R)$ fixed:
\begin{eqnarray}
W(\Lambda_R)&=& \left. \frac{d}{d t} L \right|_{\rho_a(\Lambda_R)=\text{fixed}} . \label{wt}
\end{eqnarray}
Thus, $W(\Lambda_R)$ describes the impact of a change of the shape of the initial surface $\grave{a}$ la ($\ref{shape}$) on $L$ while at the same time the coordinates $\rho_a^0$ are adjusted\footnote{This means that we move to another point $(\rho_a^0(t), \tilde{\rho}_n^0(t, {\rho}_a^0 (t))) $ on the initial surface.} such that their renormalized values $\rho_a^R$ remain fixed at the renormalization scale $\Lambda_R$.

Since $W(\Lambda)$ has been defined in analogy to $V(\Lambda)$, it obeys the RGE ($\ref{RG_v}$): 
\begin{eqnarray}
- \Lambda \frac{d}{d \Lambda}  W &=& N(W). \label{RG_w}
\end{eqnarray}
We may proceed as in the analysis of $V(\Lambda)$ and arrive at an equation corresponding to eq. ($\ref{vis}$):
\begin{equation}
W(\Lambda_R)  \approx W(\Lambda_0) \left( \frac{\Lambda_R}{\Lambda_0} \right)^{-D_{\rho_l}}  \label{wis}
\end{equation}
where $D_{\rho_l}$ still denotes the canonical dimension of the least irrelevant coupling of the QFT.  As an initial condition for W at the UV cutoff $\Lambda_0$, we again assume that
\begin{eqnarray}
|| W(\Lambda_0) || \le 1.
\end{eqnarray}
This is justified for small initial values of the nonrenormalizable couplings $\grave{a}$ la eq. ($\ref{sini}$). With eq. ($\ref{wt}$) we arrive at
\begin{equation}
|| \left.  \frac{d}{d t} L \right|_{\rho_a(\Lambda_R)=\text{fixed}} || \le  \left( \frac{\Lambda_R}{\Lambda_0} \right)^{ -D_{\rho_l}} .
\end{equation}
Integrating over $t$ with integration limits $0$ and $1$ yields the result
\begin{eqnarray}
|| L(\Lambda_R, \Lambda_0, \rho_a^0(\Lambda_R, \Lambda_0,\rho_a^R, 1 ),1)- L(\Lambda_R,\Lambda_0,\rho_a^0(\Lambda_R, \Lambda_0,\rho_a^R, 0), 0) || \le \left( \frac{\Lambda_R}{\Lambda_0} \right)^{-D_{\rho_l}} . \nonumber \\ \label{ind}
\end{eqnarray}
Since eq. ($\ref{ind}$) is valid for \textit{any} initial surface $ \rho_n^0( \rho_a^0)= \tilde{\rho}_n^0(1,\rho_a^0(1) )$ as long as it satisfies eq. ($\ref{sini}$), we conclude with the triangle inequality that for two different initial surfaces $ \rho_{n}^{0, A}( \rho_a^{0, A})$ and $ \rho_{n}^{0,B}( \rho_a^{0, B})$ which are in accordance with eq. ($\ref{sini}$)
\begin{eqnarray}
|| L(\Lambda_R, \Lambda_0, \rho_a^{0, A}(\Lambda_R, \Lambda_0,\rho_a^R))- L(\Lambda_R,\Lambda_0,\rho_a^{0, B}(\Lambda_R, \Lambda_0,\rho_a^R)) || \le \left( \frac{\Lambda_R}{\Lambda_0} \right)^{-D_{\rho_l}}. \nonumber \\ \label{indf}
\end{eqnarray}
Eq. ($\ref{indf}$) shows that the no-cutoff limit $L^{cont}$ is {independent} of the choice of the initial surface as long as the initial values for the nonrenormalizable couplings are small in the sense of eq. ($\ref{sini}$). 

Note that for finite $\Lambda_0$, we can also interpret eq. ($\ref{indf}$) in the following way. At the scale $\Lambda_R$, the ignorance about the exact values of the bare couplings $ \rho_{{n}}^{0}$ amounts to an indetermination of the potential $L(\Lambda_R)$ of the order of $\left( {\Lambda_R}/{\Lambda_0} \right)^{-D_{\rho_l}}$. This ''effective field theory'' interpretation will become important in the next section.

\end{section}

\begin{section}[Effective Field Theories with Flow Equations]{Effective Field Theories from the Viewpoint of the Renormalization Group} \label{EffFlowOver}

In addition to the renormalizable couplings $\rho_a$, one could think of fixing one or more of the nonrenormalizable couplings $\rho_n$ at the renormalization scale $\Lambda_R$. This defines additional renormalization conditions which we will refer to\footnote{Following \cite{Wiec}.} as \textit{improvement conditions}. We will give the motivations for pursuing such a strategy shortly. Before, let us note that the nonrenormalizable couplings cannot be set to arbitrary values since they are already determined by the $\rho_a$ in the limit $\Lambda_0 \rightarrow \infty$:
\begin{equation}
\lim_{\Lambda_0 \rightarrow \infty}  \rho_n(\Lambda_R,\Lambda_0, \rho_a^0(\Lambda_R, \Lambda_0, \rho^R_a)) = \rho_n^{cont}( \Lambda_R, \rho^R_a).
\end{equation}
This has been shown in the last section, and in fact it is the reason why the $\rho_n$ are called \textit{nonrenormalizable}. To be more precise, for small initial values $\rho_n^0(\rho^0_a)$ in the sense of eq. ($\ref{sini}$),
\begin{equation*}
\rho^0_n \le \Lambda_0^{D_{\rho_n}}, \ \ \ D_{\rho_n}< 0,
\end{equation*}
a running dimensionless nonrenormalizable coupling $\lambda_n$ satisfies at $\Lambda=\Lambda_R$
\begin{equation}
|| \lambda_n(\Lambda_R, \Lambda_0, \rho_a^0(\Lambda_R, \Lambda_0, \rho^R_a))- \lambda_n^{cont}(\Lambda_R, \rho^R_a) || \le \left( \frac{\Lambda_R}{\Lambda_0} \right)^{-D_{\rho_l}}  \label{convNRC}
\end{equation}
where  $ \lambda_n^{cont}(\Lambda_R, \rho^R_a)= \lim_{\Lambda_0 \rightarrow \infty} \lambda_n(\Lambda_R, \Lambda_0, \rho_a^0(\Lambda_R, \Lambda_0, \rho^R_a))$ and $D_{\rho_l}$ is again the canonical dimension of the least irrelevant coupling of the QFT. Eq. ($\ref{convNRC}$)  follows from eqns. ($\ref{conv}$), ($\ref{convS}$) and a derivative expansion ($\ref{effpot}$). It seems reasonable and will be proven in perturbation theory in chapter ($\ref{EffFlow}$) that the inversion is also true: 

Let $\lambda_n^{NR}$ be some real numbers which satisfy
\begin{equation}
|| \lambda_n^{NR}- \lambda_n^{cont}(\Lambda_R, \rho^R_a) || \le \left( \frac{\Lambda_R}{\Lambda_0} \right)^{-D_{\rho_l}}.  \label{cond1}
\end{equation}
Then there exist initial values $\rho^0_n(\rho^0_a) $ which are small $\grave{a}$ la eq. ($\ref{sini}$) such that the improvement conditions
\begin{equation}
\lambda_n(\Lambda_R, \Lambda_0, \rho_a^0(\Lambda_R, \Lambda_0, \rho^R_a)) =\lambda_n^{NR}    \label{cond2}
\end{equation}
can be met.

By fixing some of the couplings  $\rho_n$ to values $\rho_n^{NR}$ at the renormalization scale $\Lambda_R$, the following two different aims may be pursued.
\begin{itemize} 
\item The {rate of convergence} of the effective action $S_e(\Lambda, \Lambda_0)$ to its no-cutoff limit $S_e^{cont}(\Lambda)$ as $\Lambda_0 \rightarrow \infty$ can be {improved} by keeping one or more of the nonrenormalizable couplings $\rho_n(\Lambda_R, \Lambda_0)$ fixed at their ''{continuum values}'' $\rho_n^{cont}(\Lambda_R)$. This has been shown by C. Wieczerkowski \cite{Wiec}.

\item One or more of the nonrenormalizable couplings $\rho_n$ can be fixed at the renormalization scale  $\Lambda_R$ to values $\rho_n^{NR}$ that {differ} from the values $\rho_n^{cont}(\Lambda_R)$ they would take in the limit $\Lambda_0 \rightarrow \infty $ if no improvement conditions were specified for them\footnote{The notation $\rho_n^{cont}(\Lambda_R)$ always refers to the case without improvement conditions for the $\rho_n(\Lambda_R)$.}:
\begin{equation}
\rho_n^{NR} \ne \rho_n^{cont}(\Lambda_R). 
\end{equation}
A possible motivation for this may be experimental input, since in some QFTs it might turn out that the values of some nonrenormalizable couplings  $\rho_n(\Lambda_R)$ measured by experiments do not coincide with their respective $\rho_n^{cont}(\Lambda_R)$. Remember that the latter can in principle be calculated\footnote{In perturbation theory in the renormalizable couplings $\rho_a(\Lambda_R)$. } if the $\rho_a(\Lambda_R)$ are known. In fact, this scenario seems to apply to the case of quantum Einstein gravity - see chapter ($\ref{PredGrav}$). As follows from eq. ($\ref{convNRC}$), we then {cannot} send  the UV cutoff $\Lambda_0$ to infinity if the initial values $\rho_n^0$ for all nonrenormalizable couplings are supposed to be small. We are therefore dealing with an \textit{effective field theory} which, unlike a fundamental theory, cannot be valid up to arbitrary scales. However, we will show that even if we are forced to keep the bare scale $\Lambda_0$ finite, we are still left with a theory that is {predictive} at scales $\Lambda << \Lambda_0$ with {finite accuracy}.

\end{itemize}
We will now briefly review the first point, and then outline the strategy for the second one.

Let again $D_{\rho_l}$ denote the canonical dimension of the least irrelevant coupling $\rho_l$ of the QFT, and let $s$ be some integer number which is henceforth called ''improvement index''. Our aim is to introduce improvement conditions for those nonrenormalizable couplings $\rho_n$ that have canonical dimension $0 > D_{\rho_n} > D_{\rho_l}-s$, and to study the impact of the improvement conditions on the convergence properties ($\ref{convS}$) of the effective action. To distinguish the couplings for which renormalization and improvement conditions have been specified from the others, we refine our notation to the effect that couplings with canonical dimensions $D_{\rho_i} > (D_{\rho_l}-s)$ are denoted by $\rho_{\tilde{a}}$, and those with $D_{\rho_i} \le (D_{\rho_l}-s)$ by $\rho_{\tilde{n}}$. The initial surface in the space of couplings is now taken to be
\begin{equation}
\rho^0_{\tilde{n}}= \rho^0_{\tilde{n}}(\rho^0_{\tilde{a}}) . \label{inis2}
\end{equation}
Its dimension amounts to the numer of renormalizable couplings plus the number of nonrenormalizable couplings which have canonical dimensions $ 0 > D_{\rho_{\tilde{a}}} > (D_{\rho_l}-s)$.
Consequently, the solution of the Polchinski RGE ($\ref{pol}$) becomes
\begin{equation}
L=L(\phi, \Lambda, \Lambda_0, \rho_{\tilde{a}}^0).
\end{equation}
The {renormalization and improvement conditions} are defined as follows:
\begin{eqnarray}
\rho_{\tilde{a}}(\Lambda_R, \Lambda_0, \rho_{\tilde{a}}^0) = \left\{ \begin{array}{l} \rho^R_{\tilde{a}}, \ \ \ \ \ D_{\rho_{\tilde{a}}}  \ge 0 \\ \rho^{cont}_{\tilde{a}}(\Lambda_R, \rho_{\tilde{a}}^R), \ \ \   0 > D_{\rho_{\tilde{a}}} > (D_{\rho_l}-s) . \end{array} \right.   \nonumber \\
\end{eqnarray}
In analogy to eq. ($\ref{impli}$), they lead to an implicit definition of the bare couplings
\begin{equation}
\rho_{\tilde{a}}^0=\rho_{\tilde{a}}^0(\Lambda_R,\Lambda_0,\rho_{\tilde{a}}^R, \rho^{cont}_{\tilde{a}}(\Lambda_R, \rho_{\tilde{a}}^R) ). 
\end{equation}
A repetition of the analysis in the last section using an extended version of the quantity $V(\Lambda)$ of eq.($\ref{v}$) ,
\begin{eqnarray}
V(\Lambda) = \Lambda_0 \left( \frac{\partial L }{\partial \Lambda_0} -   \frac{\partial L }{\partial \rho_{\tilde{b}}^0}  \left( \frac{\partial \rho_{\tilde{a}}}{\partial \rho_{\tilde{b}}^0} \right)^{-1} \frac{\partial \rho_{\tilde{a}} }{\partial \Lambda_0} \right),  \label{vext} 
\end{eqnarray}
then yields {$s$-improved convergence} of the effective action:
\begin{equation}
|| S_e(\Lambda_R,\Lambda_0, \rho_{\tilde{a}}^0(\Lambda_R, \Lambda_0, \rho^R_{\tilde{a}}, \rho^{cont}_{\tilde{a}}(\Lambda_R, \rho_{\tilde{a}}^R) ))- S_e^{cont}(\Lambda_R,\rho^R_{\tilde{a}}) || \sim \left( \frac{\Lambda_R}{\Lambda_0} \right)^{-(D_{\rho_l} -s) }  .
\end{equation}
We will now focus on the case where nonrenormalizable couplings are fixed to values that \textit{differ} from the ones they would take in the limit $\Lambda_0 \rightarrow \infty$ if no improvement conditions were specified for them, $\rho_n^{NR} \ne \rho_n^{cont}(\Lambda_R)$. We define renormalization and improvement conditions
\begin{eqnarray}
\rho_{\tilde{a}}(\Lambda_R, \Lambda_0, \rho_{\tilde{a}}^0) = \left\{ \begin{array}{l} \rho^R_{\tilde{a}}, \ \ \  \ D_{\rho_{\tilde{a}}}  \ge 0 \\ \rho^{NR}_{\tilde{a}}, \ \  \  0 > D_{\rho_{\tilde{a}}} > (D_{\rho_l}-s)  \end{array} \right. \nonumber .
\end{eqnarray}
The $\rho^{NR}_{\tilde{a}}$ are taken such that their dimensionless counterparts $\lambda^{NR}_{\tilde{a}}(\Lambda):= \Lambda^{-D_{\rho_{\tilde{a}}}} \rho^{NR}_{\tilde{a}}$ satisfy
\begin{equation}
|| \lambda^{NR}_{\tilde{a}}(\Lambda_R) - \lambda^{cont}_{\tilde{a}}(\Lambda_R, \rho_{\tilde{a}}^R) || \le 
\left( \frac{\Lambda_R}{\Lambda_{D}} \right)^{-D_{\rho_l}}  \label{cond1a}
\end{equation}
where $\Lambda_D > \Lambda_R$ is some scale. We demand that the initial values $\rho_{\tilde{a}}^0, \ \ 0 > D_{\rho_{\tilde{a}}} > (D_{\rho_l}-s)$, remain small\footnote{One might ask the question whether abandoning the requirement of small initial values for the nonrenormalizable couplings $\rho_{\tilde{a}}^0, \ \ 0 > D_{\rho_{\tilde{a}}} > (D_{\rho_l}-s)$, while keeping it for the $\rho_{\tilde{n}}^0, \ \ \Delta_{\tilde{n}} \le (D_{\rho_l}-s)$ , would allow to fix nonrenormalizable couplings $\rho_{\tilde{a}}(\Lambda_R)$ at values $\rho_{\tilde{a}}^{NR} \ne \rho_{\tilde{a}}^{cont}(\Lambda_R)$, where the $\rho_{\tilde{a}}^{NR}$ satisfy eq. ($\ref{cond1a}$), \textit{even for an UV cutoff} $\Lambda_0 \ge \Lambda_D$. At least if one works in perturbation theory in the nonrenormalizable couplings the answer is that this is \textit{not} possible because then, counterterms for \textit{all} couplings in order to cancel the arising divergences are needed. The counterterms are the initial values $\rho_{\tilde{a}}^0$ and $\rho_{\tilde{n}}^0$. Thus, if any divergence occurs that needs large $\rho_{\tilde{a}}^0$ to be cancelled, there will be others that require large $\rho_{\tilde{n}}^0$.} as defined in eq. ($\ref{sini}$). Thus, the UV cutoff of the theory is now restricted to  
\begin{eqnarray}
\Lambda_0 \le \Lambda_D \label{restr}
\end{eqnarray}
as follows from the discussion at the beginning of this section, leading to eqns. ($\ref{cond1}$) and ($\ref{cond2}$). This fact can also be interpreted in the way that a measurement of the values of the nonrenormalizable couplings at the scale $\Lambda_R$  \textit{defines} the UV cutoff scale of the theory to be $\Lambda_0= \Lambda_D$. In the following, we will employ this view.

At scales $\Lambda << \Lambda_D$, the theory remains predictive to finite accuracy. This can be seen as follows. As we have done in eq. ($\ref{shape}$), we use a parametrization to change the shape of the initial surface ($\ref{inis2}$):
\begin{eqnarray}
\rho_{\tilde{n}}^0 \rightarrow \tilde{\rho}_{\tilde{n}}^0 := t \rho_{\tilde{n}}^0 , \ \ \ t \in [0,1].  \label{shape2}
\end{eqnarray}
Allowing also the coordinates $\rho_{\tilde{a}}^0$ to depend on the parameter $t$, the total $t$-dependence of the potential $L$ becomes, in analogy to eq. ($\ref{tdep}$),
\begin{eqnarray}
L &=&   L \left(\Lambda, \Lambda_D,\rho_{\tilde{a}}^0(t), t \right).
\end{eqnarray}
We extend the definition of the quantity $W(\Lambda)$ of eq. ($\ref{w}$) to 
\begin{eqnarray}
W(\Lambda) =   \frac{\partial L }{\partial t}  -   \frac{\partial L }{\partial \rho_{\tilde{b}}^0}  \left( \frac{\partial \rho_{\tilde{a}}}{\partial \rho_{\tilde{b}}^0} \right)^{-1} \frac{\partial \rho_{\tilde{a}} }{\partial t}   \label{wext} 
\end{eqnarray}
and repeat the corresponding steps that follow in section ($\ref{RenFlowOver}$). In doing so, we employ the assumption that the initial values $\rho_{\tilde{n}}^0$ are small. The result is
\begin{eqnarray}
|| L(\Lambda_R, \Lambda_D, \rho_{\tilde{a}}^0(\Lambda_R, \Lambda_D,\rho_{\tilde{a}}^R, 1 ),1)- L(\Lambda_R,\Lambda_D,\rho_{\tilde{a}}^0(\Lambda_R, \Lambda_D,\rho_{\tilde{a}}^R, 0), 0) || \le \left( \frac{\Lambda_R}{\Lambda_D} \right)^{-(D_{\rho_l}-s)} \nonumber . \\ \label{ind2}
\end{eqnarray}
Again, we conclude that for two different sets of initial values $ \rho_{\tilde{n}}^{0, A}( \rho_{\tilde{a}}^{0, A})$ and $ \rho_{\tilde{n}}^{0,B}( \rho_{\tilde{a}}^{0, B})$ which are in accordance with eq. ($\ref{sini}$) we have
\begin{eqnarray}
|| L(\Lambda_R, \Lambda_D, \rho_{\tilde{a}}^{0, A}(\Lambda_R, \Lambda_D,\rho_{\tilde{a}}^R))- L(\Lambda_R,\Lambda_D,\rho_{\tilde{a}}^{0, B}(\Lambda_R, \Lambda_D,\rho_{\tilde{a}}^R)) || \le \left( \frac{\Lambda_R}{\Lambda_D} \right)^{-(D_{\rho_l}-s)}. \nonumber \\ \label{indf2}
\end{eqnarray}
Eq.  ($\ref{indf2}$) shows that at the scale $\Lambda= \Lambda_R$ the ignorance about the exact values of the bare couplings $ \rho_{\tilde{n}}^{0}$ amounts to an indetermination of the potential $L(\Lambda_R)$ of the order of $\left( {\Lambda_R}/{\Lambda_D} \right)^{-(D_{\rho_l}-s)}$. If we know about the  potential $L$, we know about the effective action $S_e$ which in turn allows us to determine the generating functional $W(J)$ of the QFT, see eq. ($\ref{intout}$).  Hence, the knowledge of $L(\Lambda_R)$ to an accuracy of $\left( {\Lambda_R}/{\Lambda_D} \right)^{-(D_{\rho_l}-s)}$ leads to a QFT that is predictive with an accuracy of $\left( {\Lambda_R}/{\Lambda_D} \right)^{-(D_{\rho_l}-s)}$.

Comparing eq. ($\ref{indf2}$) to eq. ($\ref{indf}$) where renormalization conditions only for the renormalizable couplings have been specified, we furthermore see that the introduction of the improvement conditions has led to an enhanced predictivity of the effective field theory at scales $\Lambda << \Lambda_0$.

\bigskip
Let us summarize the results of this section.
\begin{itemize}
\item The deviation of the values to which the nonrenormalizable couplings are fixed, $\rho_{\tilde{a}}(\Lambda_R, \Lambda_0, \rho_{\tilde{a}}^0)=\rho^{NR}_{\tilde{a}}, \ \ 0 > D_{\rho_{\tilde{a}}} > (D_{\rho_l}-s)$, from the values $\rho^{cont}_{\tilde{a}}$ they would take in the limit $\Lambda_0 \rightarrow \infty$ if no improvement conditions were specified for them, defines an UV cutoff scale $\Lambda_D$.

\item Although we have to keep the UV cutoff $\Lambda_0$ below the scale $\Lambda_D$, for $\Lambda_0 = \Lambda_D$ the {theory remains predictive to an accuracy of $(\Lambda_R/\Lambda_D)^{(-(D_{\rho_l}-s))}$} at $\Lambda = \Lambda_R$.

\item As we fix more nonrenormalizable couplings at the renormalization scale, the improvement index $s$ increases and we achieve more predictivity.

\end{itemize}
This is the {paradigm of effective field theories} from the viewpoint of the renormalization group. In particular, it may be applied to the case of a nonrenormalizable theory which has been defined at the beginning of section ($\ref{RenFlowOver}$) as a theory that does not allow for any renormalizable couplings, except for kinetic and mass terms. The Fermi theory of weak interactions is a well-known example for such a theory, as well as Einstein gravity without a cosmological constant- see chapter ($\ref{PredGrav}$). In a nonrenormalizable theory we have 
\begin{equation}
\rho_n^{cont}(\Lambda)=0
\end{equation}
because there are no renormalizable interactions which generate new contributions to the nonrenormalizable ones while integrating out field modes. Thus, the latter die out according to their canonical dimensions in the limit $\Lambda_0 \rightarrow \infty$. However, one important point of this section was to show that ''a nonrenormalizable theory is just as good as a renormalizable theory for computations, provided one is satisfied with finite accuracy.''\cite{Mano}  

\end{section}

\chapter{Renormalization via Flow Equations} \label{RenFlow}

The perturbative renormalization of scalar field theory with flow equations is reviewed. We proceed along the lines of G. Keller, C. Kopper, M. Salmhofer \cite{KKS} who presented an improved and considerably shortened version of Polchinski's original proof \cite{Pol}. Compared to \cite{KKS}, our version contains the following generalizations.  We do not restrict ourselves to couplings assigned to operators with even numbers of fields.\footnote{We do not require the theory to be invariant under $\phi \rightarrow - \phi$.} This amounts to considering $\phi^3 + \phi^4$ theory instead of solely $\phi^4$. Moreover, we allow for nonvanishing values of the nonrenormalizable couplings at the UV cutoff scale from the very beginning and employ an alternative proof of the uniqueness of the no-cutoff limit. The latter will turn out crucial for investigating the predictivity of an effective field theory in the chapter ($\ref{EffFlow}$). One can regard the present chapter as a rigorous version of section  ($\ref{RenFlowOver}$).

\begin{section}{Renormalization of scalar $\phi^3 + \phi^4$ field theory}

\begin{subsection}{RG inequalities for vertex functions} \label{RGIs}

The central aim of this chapter will be to prove the boundedness and convergence of solutions $L(\phi, \Lambda, \Lambda_0)$ of the Polchinski equation ($\ref{pol}$) in the no-cutoff limit $\Lambda_0 \rightarrow \infty$ while renormalization conditions for the renormalizable couplings are imposed at some renormalization scale $\Lambda_R$. This program will be carried out for a scalar field theory in perturbation theory in the (renormalized) renormalizable couplings. As we have pointed out in the last chapter, the potential $L(\Lambda)$ corresponds to an effective action $S_e(\Lambda)$ which in turn leads via eq. ($\ref{WP}$) to the determination of a generating functional $W(J)$ of a QFT. Remember that since  $S_e(\Lambda)$ contains an effective cutoff $\Lambda$, we do not have to worry about possible divergences any more once finite bounds for $L(\phi, \Lambda, \Lambda_0)$ have been established for the limit $\Lambda_0 \rightarrow \infty$. 

In the following we will always work in momentum space. We begin by rewriting the effective action ($\ref{Spol}$):
\begin{eqnarray} 
S_e(\phi, \Lambda)= \int \frac{ d^4k}{(2 \pi)^4} \left(  -\frac{1}{2} \phi(k) (k^2+ B) K^{-1}(k^2/\Lambda^2) \phi(-k) + A \frac{(2 \pi)^4}{2} \delta(k) \phi(k) \right) + L(\phi, \Lambda). \nonumber \\ \label{SpolM}
\end{eqnarray}
$K$ is the cutoff function defined in eq. ($\ref{cutf}$). Note that in ($\ref{SpolM}$), we have included terms linear and bilinear in the fields associated with couplings $A$ and $B$ into the ''free'' part of the action. It will turn out that  the definition of the effective potential $L$ following from ($\ref{SpolM}$) is more suitable for the upcoming proofs of boundedness and convergence, as compared to the one associated with eq. ($\ref{Spol}$) of chapter ($\ref{OvM}$). 

The redefinition of $L$ leads to an additional term in the Polchinski renormalization group equation, as compared to the original version ($\ref{pol}$). By requiring invariance of the generating functional ($\ref{WP}$) under a change of the scale $\Lambda$ we obtain (in momentum space)
\begin{eqnarray}
- \Lambda \frac{d}{d \Lambda} L  = \frac{1}{2} \int d^4k   (2 \pi)^4  \Lambda \frac{d}{d \Lambda} \Delta_\Lambda  \left(\frac{\delta L}{\delta \phi(k)} \frac{\delta L}{\delta \phi(-k)}   + \frac{\delta^2 L}{\delta \phi(k) \delta\phi(-k)} + A \delta(k) \frac{\delta L}{\delta \phi(k)}    \right) \nonumber \\ \label{polm}
\end{eqnarray}
where the regularized propagator is given by 
\begin{equation}
\Delta_\Lambda=\frac{K(k^2/\Lambda^2)}{k^2+B}.
\end{equation}
At this point, we will not prove the modified RGE ($\ref{polm}$), but instead refer the reader to the proof of Theorem ($\ref{PolQG}$) where an analogous RGE for Euclidean quantum gravity is established.

Furthermore we would like to point out that we do not require the renormalization constant $B$ of the ''mass  term'' in ($\ref{SpolM}$) to be positive. This poses no problems to the analysis of the effective potential $L(\phi, \Lambda)$ as long as we keep $\Lambda^2 > |B|$, as can be seen from the RGE ($\ref{polm}$). Due to the properties ($\ref{cutf}$) of the cutoff function, $\Lambda \frac{d}{d \Lambda} \Delta_\Lambda  $ has compact support $\Lambda< k < 4 \Lambda$. Thus, only momenta $k>\Lambda$ contribute to the integral in ($\ref{polm}$) and therefore to a solution $L(\phi,\Lambda, \Lambda_0)$ of the RGE\footnote{One can see $\Lambda$ and $\Lambda_0$ as an IR and UV momentum cutoffs for $L(\phi,\Lambda, \Lambda_0)$ respectively.}.

We will now focus on the analysis of the effective potential $L(\phi, \Lambda)$. Therefore, we  expand $L(\phi, \Lambda)$ in powers of the fields $\phi$:
\begin{eqnarray}
L(\phi, \Lambda) = \sum_{n=1}^{\infty} \frac{1}{n!} \int \frac{d^4 k_1 ... d^4 k_{n}}{(2 \pi )^{4n-4}} L_{n} (k_1,...,k_{n}, \Lambda) \delta^4 \big( \sum_i k_i \big) \phi(k_1)... \phi(k_{n}) . \label{Lexp}
\end{eqnarray}
The expansion coefficients $L_{n} (k_1,...,k_{n}, \Lambda)$ are henceforth called \textit{vertex functions}\footnote{They are {not} identical to the conventional vertex functions of the 1 particle irreducible (PI) effective action but instead are related to the connected amputated Green's functions. See Appendix ($\ref{LZ}$) for details.}. 

As it is explained in Appedix ($\ref{D}$), the (momentum-space) fields $\phi(k)$ have canonical dimension $D_{\phi(k)}=-3$. Thus from eq. ($\ref{Lexp}$) it follows that the canonical dimension of the vertex functions is given by
\begin{eqnarray}
D_{L_n} = 4-n.   \label{DL}
\end{eqnarray}
The Polchinski RGE ($\ref{polm}$) can be reformulated in terms of the $L_{n} (k_1,...,k_{n}, \Lambda)$  and then be generalized to a RGE for momentum derivatives of vertex functions. The resulting equation is
\begin{eqnarray}
\frac{d}{d \Lambda} \partial^p L_{n}(k_1,...k_n, \Lambda) &=&  - \sum_{p_1,p_2,p_3: \ \sum p_i =p} \sum_{l=1}^{n} \Big(  \partial^{p_1} \frac{d}{d \Lambda} \Delta_\Lambda (K, \Lambda) \ \partial^{p_2} L_{l}(k_1,...,k_{l-1}, K, \Lambda )  \nonumber \\ && \partial^{p_3} L_{n+2-l }(k_l, ..., k_n, -K, \Lambda) \Big) + \frac{1}{2} \binom{n}{l-1}  \ \text{permutations} \nonumber \\ && - \frac{1}{2} \int \frac{d^4 k}{(2 \pi )^4}  \partial^p L_{n+2}(k_1, ..., k_n, k, -k, \Lambda) \frac{d}{d \Lambda}  \Delta_\Lambda (k, \Lambda) \nonumber \\ && - \frac{1}{2} A \ \partial^p L_{n+1}(k_1, ..., k_n, 0, \Lambda) \frac{d}{d \Lambda}  \Delta_\Lambda (0, \Lambda) . \label{RGE0}
\end{eqnarray}
Here, we have defined $K:= \sum k_i$ and employed the notation
\begin{equation}
\partial^p := \partial^{\mu_1}_{i_1,j_1}...\partial^{\mu_p}_{i_p,j_p}
\end{equation}
where
\begin{equation}
\partial^\mu_{i,j}:= \frac{\partial}{\partial k_i^\mu} -\frac{\partial}{\partial k_j^\mu}
\end{equation}
takes care of the fact that $L_{n}(k_1,...,k_n, \Lambda)$ makes only sense for $\sum_i k_i=0$.\footnote{Write $L_{n}(k_1,...k_{n-1},K_n, \Lambda)$ with $K_n:= -\sum_{i=1}^{n-1} k_i$. Then $\frac{d}{dk_i} L_n = \frac{\partial}{\partial k_i} L_n + \frac{\partial}{\partial K_n} L_n \frac{\partial K_n}{\partial k_i}= \frac{\partial}{\partial k_i} L_n - \frac{\partial}{\partial K_n} L_n $. }. 

Using Taylor's theorem we expand the vertex functions $L_{n} (k_1,...,k_{n}, \Lambda)$ around $k_i=0$. The expansions up to $\mathcal{O}(k^0)$ and to  $\mathcal{O}(k^2)$ read\footnote{Terms with odd powers of momenta vanish because of the invariance of $L_n$ under the orthogonal group, see Appendix ($\ref{MDV}$) .}
\begin{eqnarray}
L_{n} (k_1,...,k_{n}, \Lambda) &=& L_{n} (0,...,0, \Lambda) + \sum_{i_1,i_2=1}^{n-1} k_{i_1}^{\mu_1} k_{i_2}^{\mu_2} \int_0^1 d \tau (1-\tau) \partial^{\mu_1}_{i_1,n} \partial^{\mu_2}_{i_2,n} L_{n} (\tilde{k}_1,...,\tilde{k}_{n}, \Lambda)|_{\tilde{k_i}=\tau k_i} \nonumber \\ \label{TE0} 
\end{eqnarray}
\begin{eqnarray}
L_{n} (k_1,...,k_{n}, \Lambda) &=& L_{n} (0,...,0, \Lambda) +  \frac{1}{2} \sum_{i_1,i_2=1}^{n-1}  k_{i_1}^{\mu_1} k_{i_2}^{\mu_2} \ \partial^{\mu_1}_{i_1,n} \partial^{\mu_2}_{i_2,n} L_{n} (\tilde{k}_1,...,\tilde{k}_{n}, \Lambda)|_{\tilde{k_i}=0} \nonumber \\ && + \frac{1}{3!} \sum_{i_1,...,i_4=1}^{n-1} k_{i_1}^{\mu_1} ... k_{i_4}^{\mu_4} \int_0^1 d \tau (1-\tau)^3 \partial^{\mu_1}_{i_1,n} ... \partial^{\mu_4}_{i_4,n} L_{n} (\tilde{k}_1,...,\tilde{k}_{n}, \Lambda)|_{\tilde{k_i}=\tau k_i} . \nonumber \\ \label{TE1}
\end{eqnarray}
The last terms in  ($\ref{TE0}$), ($\ref{TE1}$) are remainder terms, respectively. The expansion coefficients in ($\ref{TE1}$) give rise to the definition of running coupling constants\footnote{See Appendix ($\ref{MDV}$) for details on the definition of $\rho_3(\Lambda)$.} $\rho_a(\Lambda)$:
\begin{eqnarray} 
\rho_1(\Lambda) &:=& L_1(0,\Lambda) \label{rcc1} \\
\rho_2(\Lambda) &:=& L_2(0,0,\Lambda) \\
\rho_3(\Lambda) \ \delta^{\mu \nu} &:=& \partial^\mu_{1,2} \partial^\nu_{1,2} L_2(k_1,k_2,\Lambda)|_{k_1=k_2=0} \\
\rho_4(\Lambda) &:=& L_3(0,0,0,\Lambda) \label{rcc4} \\
\rho_5(\Lambda) &:=& L_4(0,0,0,0,\Lambda). \label{rcc5}
\end{eqnarray}
These are equivalent to the coupling constants introduced in the position-space derivative expansion ($\ref{effpot}$), as is shown in Appendix  ($\ref{MDV}$). The corresponding dimensionless couplings $\lambda_a(\Lambda)$ are defined as in  eq.  ($\ref{dlcoup}$). All couplings $\rho_a$ are renormalizable, i.e. their mass dimensions $D_{\rho_a}$ satisfy
\begin{equation}
D_{\rho_a} \ge 0
\end{equation}
as follows from eq. ($\ref{DL}$).
In an analogous way we may define nonrenormalizable coupling constants $\rho_n$ with $D_{\rho_n} < 0$. This will be done in chapter ($\ref{EffFlow}$). We now impose renormalization conditions at some renormalization scale $\Lambda_R< \Lambda$:
\begin{eqnarray} 
\rho_1(\Lambda_R) &=& 0 \label{rc1} \\
\rho_2(\Lambda_R) &=& 0 \\
\rho_3(\Lambda_R) &=& 0 \\
\rho_4(\Lambda_R) &=& \rho_4^R \label{c3} \\
\rho_5(\Lambda_R) &=& \rho_5^R  . \label{c4} 
\end{eqnarray}
Note that the $\rho_i(\Lambda_R), \ i=1...3,$ are associated with the constants $A$ and $B$ appearing in the ''free'' part of the effective action  ($\ref{SpolM}$) and the field strength renormalization. For the dimensionless couplings, we introduce the notation\footnote{If we write $\lambda_a^R$, we always mean $\lambda_a^R(\Lambda)$, and not $\lambda_a^R(\Lambda_R)$. }
\begin{eqnarray}
\lambda_a^R(\Lambda) &:=& \Lambda^{-D_{\rho_a}}  \rho_a^R .  \label{dCR}
\end{eqnarray}
The vertex functions $L_{n} (k_1,...,k_{n}, \Lambda)$ are being evaluated in perturbation theory in the renormalized renormalizable couplings $\rho_4^R$ and $\rho_5^R$:
\begin{eqnarray}
L_{n} (k_1,...,k_{n}, \Lambda) = \sum_{r_1, r_2=0}^{\infty} (\rho_4^R)^{r_1} (\rho_5^R)^{r_2} L_{n}^{(r_1,r_2)} (k_1,...,k_{n}, \Lambda).  \label{Pert}
\end{eqnarray}
Note that since the expansion parameter $\rho_4^R$ is dimensionful, the canonical dimensions of the expansion coefficients $L_{n}^{(r_1,r_2)}$ become dependent of the order in perturbation theory in $\rho_4^R$. 

We will work mostly with dimensionless vertex functions. These follow from eq. ($\ref{DL}$) as\footnote{The dimensionless couplings $\lambda_a(\Lambda)$ defined in ($\ref{dlcoup}$) can also be introduced by replacing $L_n$ in the definitions ($\ref{rcc1}$)-($\ref{rcc5}$) by $A_n$ and by multiplying the result with powers of $\Lambda$ that correspond to the number of derivatives present. }
\begin{equation}
 A_{n} (k_1,...,k_{n}, \Lambda) :=  \Lambda^{n-4} L_{n} (k_1,...,k_{n}, \Lambda)  . \label{dV0}
\end{equation}
Accordingly, we may expand
\begin{equation}
A_{n} (k_1,...,k_{n}, \Lambda) = \sum_{r_1, r_2=0}^{\infty} (\lambda_4^R)^{r_1} (\lambda_5^R)^{r_2} A_{n}^{(r_1,r_2)} (k_1,...,k_{n}, \Lambda) \label{Apert}
\end{equation}
where
\begin{equation}
 A_{n}^{(r_1,r_2)} (k_1,...,k_{n}, \Lambda) :=  \Lambda^{n+ r_1 -4} L_{n}^{(r_1,r_2)} (k_1,...,k_{n}, \Lambda) . \label{dV}
\end{equation}
The perturbative expansion ($\ref{Apert}$) makes only sense if the dimensionless couplings are small. Therefore we impose as an additional constraint to the renormalization conditions ($\ref{rc1}$)-($\ref{c4}$)
\begin{equation}
\lambda_a^R(\Lambda) \le 1.  \label{smallR}
\end{equation}
Note that because of the definition ($\ref{dCR}$), this in particular means $\rho_4^R \le \Lambda_R$ for $\Lambda \ge \Lambda_R$ and thus
\begin{equation}
\lambda_4^R(\Lambda) \le \frac{\Lambda_R}{\Lambda} \ . \label{smallRExp}
\end{equation}
Moreover, we remember that we have included terms linear and bilinear in the fields asscociated with ''couplings'' $A$ and $B$ into the ''free'' part of the action in eq. ($\ref{SpolM}$). From the RGE ($\ref{RGE0}$) then follows that for increasing order $(r_1, r_2)$ in perturbation theory in $\rho_4^R$, $\rho_5^R$ of the vertex functions $ A_{n}^{(r_1,r_2)} (\Lambda) $ there will be graphs involving increasing powers of $A$'s. Thus, for $\Lambda^{-3} A \ge 1$ the vertex functions will become arbitrarily large at high orders $(r_1, r_2)$. In order to avoid such behaviour and remembering the discussion at the beginning of this section concerning the values of $B$, we impose as conditions for $B$, $A$ 
\begin{eqnarray}
|\Lambda^{-2} B | &\le& 1 \label{smallm} \\
|\Lambda^{-3} A | &\le&  \label{smallg} 1.
\end{eqnarray}
To $0 th$ order in perturbation theory in $\lambda_4^R$,  $\lambda_5^R$ the vertex functions vanish:
\begin{eqnarray}
 A_{n}^{(0,0)} (k_1,...,k_{n}, \Lambda) = 0. \label{zero}
\end{eqnarray}
Finally, we define an overall order in perturbation theory via
\begin{equation}
r:= r_1 + r_2. 
\end{equation}
Due to the compact support of $dK/ d \Lambda$ it is easy to see that there are constants $C$ and $D_n$ such that
\begin{eqnarray}
\int \frac{d^4k}{(2 \pi)^4}  \lVert \Lambda^3 \frac{d}{d \Lambda} \Delta_\Lambda  \rVert  &<& C \Lambda^4 \label{q1} \\ \lVert \frac{\partial^n}{\partial k^n} \Lambda^3 \frac{d}{d \Lambda} \Delta_\Lambda \rVert &<& D_n \Lambda^{-n} \label{q2}
\end{eqnarray}
where the notation $\lVert \ \rVert$ is defined by
\begin{equation}
\lVert f(k_1, ..., k_{n}) \rVert = \max_{k_i^2 \le 4 \Lambda} |f(k_1,...,k_{n} ,\Lambda)   |     \label{norm}
\end{equation}
for some function $f$ of one or more momenta. We now rewrite the  RGE ($\ref{RGE0}$) in terms of the dimensionless vertex functions  ($\ref{dV0}$) and express the resulting equation in perturbation theory in $\lambda_4^R$ and $\lambda_5^R$. Applying the bounds ($\ref{q1}$) and ($\ref{q2}$) as well as the condition ($\ref{smallg}$) for the renormalization constant $A$, we arrive at the key RG inequality 
\begin{eqnarray}
\lVert \frac{d}{d \Lambda} \Lambda^{4-n-r_1} \partial^p A^{(r_1, r_2)}_{n}(\Lambda)  \rVert &\le& c_{n,p} \ \Lambda^{3-n-r_1} \Bigg( \lVert \partial^p A^{(r_1, r_2)}_{n+2}(\Lambda) \rVert + \lVert \partial^p A^{(r_1, r_2)}_{n+1}(\Lambda) \rVert  \nonumber \\ && +  \sum_{...} \Lambda^{-p_1} \lVert \partial^{p_2} A_{l}^{(s_1, s_2)}(\Lambda) \rVert \lVert \partial^{p_3} A^{(r_1-s_1, r_2-s_2)}_{n+2-l}(\Lambda) \rVert \Bigg)  \nonumber \\ \label{RGI1}
\end{eqnarray}
where we introduced the abbreviation
\begin{equation}
 \sum_{...} \ := \sum_{p_1,p_2,p_3: \ \sum p_i =p} \ \sum_{l=1}^{n} \ \sum_{\substack{s_1, s_2=0 \\ 1 \le s \le r-1 }}^{r_1, r_2}  .
\end{equation} 
The inequality ($\ref{RGI1}$) is the starting point to deduce four more inequalities which turn out to be crucial to show boundedness and convergence of the vertex functions in the limit $\Lambda_0 \rightarrow \infty$:
\begin{enumerate}
\item We integrate eq. ($\ref{RGI1}$) \textit{down} from the scale $\Lambda_0$ to $\Lambda$. Using the triangle inequality, we get
\begin{eqnarray}
\hspace{-0.5cm} \Vert \Lambda^{4-n-r_1} \partial^p A^{(r_1, r_2)}_{n}(\Lambda) \rVert   &\le& \lVert \Lambda_0^{4-n-r_1} \partial^p A^{(r_1, r_2)}_{n}(\Lambda_0) \rVert \nonumber \\ && + \ c_{n,p}  \int_{\Lambda}^{\Lambda_0} ds \ s^{3-n-r_1} \Bigg( \lVert \partial^p A^{(r_1, r_2)}_{n+2}(s) \rVert +  \lVert \partial^p A^{(r_1, r_2)}_{n+1}(s) \rVert  \nonumber \\ && + \sum_{...} s^{-p_1} \lVert \partial^{p_2} A_{l}^{(s_1, s_2)}(s) \rVert \lVert \partial^{p_3} A^{(r_1-s_1, r_2-s_2)}_{n+2-l}(s) \rVert \Bigg)  . \nonumber \\ \label{RGI2}
\end{eqnarray} 
Eq. ($\ref{RGI2}$) is now being differentiated with respect to $\Lambda_0$. With regard to this we note that
\begin{eqnarray}
g(x,x_0) := \int_{x}^{x_0} h(s, x_0) ds = \int_{x}^{\infty} h(s, x_0) \Theta(x_0-s) ds 
\end{eqnarray}
for any function $h(s, x_0)$ and the step function $\Theta(x_0-s)$. Therefore
\begin{eqnarray}
\frac{d}{dx_0} g(x,x_0) &=& \int_{x}^{\infty} \left(\frac{d}{dx_0}  h(s, x_0) \right) \Theta(x_0-s) ds  + \int_{x}^{\infty} h(s, x_0) \delta(x_0-s) ds  \nonumber \\ &=& h(x_0, x_0) + \int_{x}^{x_0} \left(\frac{d}{dx_0}  h(s, x_0) \right) ds.
\end{eqnarray}
Consequently, we arrive at
\begin{eqnarray}
\hspace{-0.5cm} \lVert \frac{d}{d \Lambda_0} \Lambda^{4-n-r_1} \partial^p A^{(r_1, r_2)}_{n}(\Lambda) \rVert   &\le&   \lVert \frac{d}{d \Lambda_0} \Lambda_0^{4-n-r_1} \partial^p  A^{(r_1, r_2)}_{n}(\Lambda_0) \rVert \nonumber \\ &&  +  c_{n,p}  \  \Lambda_0^{3-n-r_1} \Bigg( \lVert \partial^p A^{(r_1, r_2)}_{n+2}(\Lambda_0) \rVert + \lVert \partial^p A^{(r_1, r_2)}_{n+1}(\Lambda_0) \rVert \nonumber \\  && +  \sum_{...} \Lambda_0^{-p_1} \lVert \partial^{p_2} A_{l}^{(s_1, s_2)}(\Lambda_0) \rVert \lVert \partial^{p_3} A^{(r_1-s_1, r_2-s_2)}_{n+2-l}(\Lambda_0) \rVert \Bigg)  \ \nonumber \\  &&+ c_{n,p} \int_{\Lambda}^{\Lambda_0} ds \ s^{3-n-r_1} \Bigg( \lVert  \frac{d}{d \Lambda_0} \partial^p A^{(r_1, r_2)}_{n+2}(s) \rVert \nonumber \\ && \qquad \qquad \qquad \qquad \qquad \qquad + \lVert  \frac{d}{d \Lambda_0} \partial^p A^{(r_1, r_2)}_{n+1}(s) \rVert  \nonumber \\  && + 2 \sum_{...} s^{-p_1} \lVert  \frac{d}{d \Lambda_0} \partial^{p_2} A_{l}^{(s_1, s_2)}(s) \rVert \lVert \partial^{p_3} A^{(r_1-s_1, r_2-s_2)}_{n+2-l}(s) \rVert \Bigg) .   \nonumber \\ \label{RGI3}
\end{eqnarray} 
\item We integrate eq. ($\ref{RGI1}$) \textit{up} from the scale $\Lambda_R$ to $\Lambda$ and make use of the triangle inequality:
\begin{eqnarray}
\hspace{-0.5cm} \lVert \Lambda^{4-n-r_1} \partial^p A^{(r_1, r_2)}_{n}(\Lambda)\rVert  &\le& \lVert \Lambda_R^{4-n-r_1} \partial^p A^{(r_1, r_2)}_{n}(\Lambda_R) \rVert  \nonumber \\ && + \ c_{n,p} \int_{\Lambda_R}^{\Lambda} ds \ s^{3-n-r_1} \Bigg( \lVert \partial^p A^{(r_1, r_2)}_{n+2}(s) \rVert + \lVert \partial^p A^{(r_1, r_2)}_{n+1}(s) \rVert  \nonumber \\ && + \sum_{...} s^{-p_1} \lVert \partial^{p_2} A_{l}^{(s_1, s_2)}(s) \rVert \lVert \partial^{p_3} A^{(r_1-s_1, r_2-s_2)}_{n+2-l}(s) \rVert \Bigg) \nonumber \\ \label{RGI4}
\end{eqnarray} 
Differentiating with respect to $\Lambda_0$ now yields
\begin{eqnarray}
\hspace{-0.5cm} \lVert \frac{d}{d \Lambda_0} \Lambda^{4-n-r_1} \partial^p A^{(r_1, r_2)}_{n}(\Lambda) \rVert &\le& \lVert \frac{d}{d \Lambda_0} \Lambda_R^{4-n-r_1} \partial^p A^{(r_1, r_2)}_{n}(\Lambda_R) \rVert \nonumber \\ && +  c_{n,p} \int_{\Lambda_R}^{\Lambda} ds \ s^{3-n-r_1} \Bigg( \lVert \frac{d}{d \Lambda_0} \partial^p A^{(r_1, r_2)}_{n+2}(s) \rVert \nonumber \\ && \qquad \qquad \qquad \qquad \qquad \qquad  + \lVert \frac{d}{d \Lambda_0} \partial^p A^{(r_1, r_2)}_{n+1}(s) \rVert \nonumber \\ && + 2 \sum_{...} s^{-p_1} \lVert \frac{d}{d \Lambda_0} \partial^{p_2} A_{l}^{(s_1, s_2)}(s) \rVert \lVert \partial^{p_3} A^{(r_1-s_1, r_2-s_2)}_{n+2-l}(s) \rVert \Bigg) . \nonumber \\ \label{RGI5}
\end{eqnarray} 

\end{enumerate}

\end{subsection}

\begin{subsection}{Boundedness and convergence of the vertex functions} \label{BaC}
We begin by proving in perturbation theory in $\lambda_4^R$ and $\lambda_5^R $ the boundedness of the norms $\lVert  \partial^p A^{(r_1, r_2)}_{n}(\Lambda) \rVert$ of the vertex functions as the UV cutoff $\Lambda_0 \rightarrow \infty$. Our analysis is based on the work of Keller, Kopper and Salmhofer \cite{KKS} who presented improved bounds for the vertex functions as compared to the original version of Polchinski \cite{Pol}. We generalize their treatment to the case of $\phi^3 + \phi^4$ theory.

\begin{satz}[Boundedness I] \label{BoundTh} Given the renormalization conditions ($\ref{rc1}$)-($\ref{c4}$) and assuming initial conditions
\begin{eqnarray}
||\partial^p A_{n}^{(r_1, r_2)}(p_1,...,p_{n}, \Lambda_0)|| \le \Lambda_0^{-p}  \left( \frac{\Lambda_0}{\Lambda_R} \right)^{r_1} Pln \left( \frac{\Lambda_0}{\Lambda_R} \right) \label{ini}
\end{eqnarray} 
for $n+p \ge 5$, to order $r_1, r_2$ in perturbation theory in $\lambda_4^R$ and $\lambda_5^R $
\begin{eqnarray}
||\partial^p A_{n}^{(r_1, r_2)}(p_1,...,p_{n}, \Lambda)|| \le \Lambda^{-p} \left( \frac{\Lambda}{\Lambda_R} \right)^{r_1} \left(  Pln\left( \frac{\Lambda}{\Lambda_R} \right) +  \frac{\Lambda}{\Lambda_0}  Pln \left( \frac{\Lambda_0}{\Lambda_R} \right) \right) \ \label{Bound}
\end{eqnarray}
where $Pln(z)$ denotes some polynomial in $ln(z)$ whose coefficients are taken to be nonnegative and $\Lambda_R \le \Lambda \le \Lambda_0$.
\end{satz}

\begin{proof} Since this proof will serve as a blueprint for the proofs of the upcoming Theorems ($\ref{ConvTh}$)-($\ref{PreTh}$), it will be performed in full detail. The proof is done via induction in both the overall order in perturbation theory, $r=r_1+r_2$, and the number of external legs, $n$, of the vertex functions $A_n^{(r_1, r_2)}$. This induction scheme is motivated by the fact that the terms appearing in the key renormalization group inequality (RGI) ($\ref{RGI1}$) satisfy
\begin{eqnarray}
 r_{LHS} > r_{RHS}  \ \ \vee \ \ r_{LHS}= r_{RHS} \ \wedge \ n_{LHS}< n_{RHS}.
\end{eqnarray}
Furthermore we note that
\begin{eqnarray}
A_n^{(r_1, r_2)}(\Lambda)=0 \ \ \text{for} \ \ n > 2r+2
\end{eqnarray}
and remember eq.($\ref{zero}$)
\begin{eqnarray*}
A_n^{(0,0)}(\Lambda)   =  0 .
\end{eqnarray*}
Thus, we are led to the following induction scheme:
\begin{itemize}
\item \textbf{Induction start}: ($\ref{Bound}$) holds true for
\begin{eqnarray}
\{(r,n): r=0  \ \wedge \ n \in  \mathbb{N} \} \vee \{ (r,n): r \ge 1 \ \wedge \ n > 2r+2  \}
\end{eqnarray}
\item \textbf{Induction hypothesis}: ($\ref{Bound}$) holds true for
\begin{eqnarray}
\{(r,n): r< r_0 \ \wedge \ n \in  \mathbb{N} \} \vee \{ (r,n): r = r_0 \ \wedge \ n > n_0  \}
\end{eqnarray}
\item \textbf{Induction step}: ($\ref{Bound}$) holds true for
\begin{eqnarray}
\{ (r,n): r = r_0 \ \wedge \ n = n_0  \ \forall n_0 \in \mathbb{N} \} .
\end{eqnarray}
\end{itemize}  
We will prove the induction step using the RGIs ($\ref{RGI2}$) and ($\ref{RGI4}$). To do so, we will need the following
\begin{lemma} \label{l}
Let $P^{(n)}$ denote some polynomial of degree $n$ with positive coefficients, and let $a \le b, \ a,b \in \mathbb{R}^{+}$ and $m \in \mathbb{Z}$. Then
\begin{eqnarray}
\int_a^b dx \ P^{(n)}(\ln x) x^{m} & \le & \left\{ \begin{array}{c c c}   \left[ - x^{m+1} P^{(n)}(\ln x) \right]_a^b & ,& m \le -2 \\ \left[  P^{(n+1)}(\ln x) \right]_a^b  &,  & m=-1 \\  \left[  x^{m+1} P^{(n)}(\ln x) \right]_a^b &, & m \ge 0 \ . \end{array} \right.
\end{eqnarray}
\end{lemma}

\begin{proofL} Using integration by parts one shows that
\begin{eqnarray} 
\int dx \ (\ln x)^n x^{m} = \left\{ \begin{array}{c c c}   \frac{(\ln x)^n }{1+m} x^{m+1} - \frac{n}{1+m} \int dx \ (\ln x)^{n-1} x^{m}  & ,& m \ne -1 \\ \frac{1}{1+n} (\ln x)^{n+1} .  &,  & m=-1 \ . \end{array} \right.  
\end{eqnarray} 
Thus Lemma ($\ref{l}$) follows immediately for $m=-1$ and recursively for $m \le -2$. For $m \ge 0$ the recursion yields a polynomial in $\ln x$ that may have negative coefficients. Let $\tilde{P}^{(n)}$ denote some polynomial with coefficients of arbitrary sign, and $P^{(n)}$ the corresponding polynomial with all coefficients made positive. Then
\begin{equation}
 \left[ x^{m+1} \tilde{P}^{(n)}(\ln x) \right]_a^b  \le   \left[ x^{m+1} P^{(n)}(\ln x) \right]_a^b 
\end{equation}
for $a \le b$, $m \ge -1$. Thus Lemma ($\ref{l}$) follows for $m \ge 0$.
\begin{flushright}
$\Box$
\end{flushright}
\end{proofL}
We now may proceed with the proof of Theorem ($\ref{BoundTh}$) .
\begin{enumerate}

\item \textbf{The case} $p+n \ge 5$:

We start with the RGI ($\ref{RGI2}$). On the RHS, we plug in eq. ($\ref{ini}$) as the initial condition at $\Lambda_0$ and eq. ($\ref{Bound}$) as the induction hypothesis. The result is
\begin{eqnarray}
\Vert \Lambda^{4-n- r_1} \partial^p A^{(r_1, r_2)}_{n}(\Lambda) \rVert   &\le&  \Lambda_R^{-r_1} \Lambda_0^{4-n -p}  Pln \left( \frac{\Lambda_0}{\Lambda_R} \right) \nonumber \\  & &  +  \underbrace{ \Lambda_R^{-r_1} \int_\Lambda^{\Lambda_0} ds \  s^{3-n-p} \left(  Pln \left( \frac{s}{\Lambda_R} \right) +  \frac{s}{\Lambda_0}   Pln \left( \frac{\Lambda_0}{\Lambda_R} \right) \right)}_{:=I_1(\Lambda, \Lambda_0)}. \nonumber \\ \label{i1}
\end{eqnarray}
With Lemma ($\ref{l}$) it follows that
\begin{eqnarray}
\hspace{-0.7em} I_1(\Lambda, \Lambda_0) & \le & \Lambda_R^{-r_1} \left[ -s^{4-n-p}   Pln\left( \frac{s}{\Lambda_R} \right) \right.  +  \left. \left\{ \begin{array}{c c c}   Pln \left( \frac{\Lambda_0}{\Lambda_R} \right) \Lambda_0^{-1}\ln(s) & ,& p+n =5 \\ -Pln \left( \frac{\Lambda_0}{\Lambda_R} \right)  \Lambda_0^{-1 }s^{5-n-p} &, &  p+n \ge 6  \end{array} \right. \right]_\Lambda^{\Lambda_0} . \nonumber \\
\end{eqnarray}
Plugging this into ($\ref{i1}$) and multiplying the whole equation with $\Lambda^{n+r_1-4}$ yields the bound ($\ref{Bound}$) for $n+p \ge 5$.

\pagebreak
\item \textbf{The case} $n + p \le 4$

We start with the RGI ($\ref{RGI4}$) and plug in  eq. ($\ref{Bound}$) as the induction hypothesis:
\begin{eqnarray}
\lVert \Lambda^{4-n-r_1} \partial^p A^{(r_1, r_2)}_{n}(\Lambda)\rVert  &\le& \lVert \Lambda_R^{4-n-r_1} \partial^p A^{(r_1, r_2)}_{n}(\Lambda_R) \rVert \nonumber \\ & &  + \Lambda_R^{-r_1} \int_{\Lambda_R}^{\Lambda} ds \  s^{3-n-p} \left( Pln\left( \frac{s}{\Lambda_R} \right) +  \frac{s}{\Lambda_0}   Pln \left( \frac{\Lambda_0}{\Lambda_R} \right) \right) . \nonumber \\ 
\end{eqnarray} 
The renormalization conditions ($\ref{rc1}$)-($\ref{c4}$) serve as initial conditions at $\Lambda_R$. We begin with $n=4$, $p=0$ and vanishing external momenta $k_i=0$: 
\begin{eqnarray}
\hspace{-1em} \lVert \Lambda^{-r_1} A^{(r_1, r_2)}_{4}(0,...,0, \Lambda)\rVert  &\le& \delta^{r_1 0} \delta^{r_2 1} + \underbrace{ \Lambda_R^{-r_1} \int_{\Lambda_R}^{\Lambda} ds \  s^{-1} \left(  Pln\left( \frac{s}{\Lambda_R} \right) +  \frac{s}{\Lambda_0}   Pln \left( \frac{\Lambda_0}{\Lambda_R} \right) \right)}_{:=I_2(\Lambda, \Lambda_0)}. \nonumber \\ \label{i2}
\end{eqnarray} 
With Lemma ($\ref{l}$) it follows that
\begin{eqnarray}
I_2(\Lambda, \Lambda_0) \le \Lambda_R^{-r_1} \left[ \left(  Pln \left( \frac{s}{\Lambda_R} \right) + \frac{s}{\Lambda_0}   Pln \left( \frac{\Lambda_0}{\Lambda_R} \right) \right)\right]_{\Lambda_R}^{\Lambda}.
\end{eqnarray}
We thus arrive at the bound ($\ref{Bound}$) for $A^{(r_1, r_2)}_{4}(0,...,0, \Lambda)$. Taylor's theorem ($\ref{TE0}$) allows us to reconstruct $A^{(r_1, r_2)}_{4}(k_1,...,k_4, \Lambda)$ where the remainder term is already bounded since it corresponds to the case $n+p \ge 5$.

We now proceed to $n=3$, $p=0$ and vanishing external momenta $k_i=0$: 
\begin{eqnarray}
\lVert \Lambda^{1 -r_1} A^{(r_1, r_2)}_{3}(0,0,0, \Lambda)\rVert  &\le& \delta^{r_1 1} \delta^{r_2 0} + \Lambda_R^{-r_1} \int_{\Lambda_R}^{\Lambda} ds \big(...\big) .
\end{eqnarray} 
With Lemma ($\ref{l}$) we solve the integral and establish the bound ($\ref{Bound}$) for $A^{(r_1, r_2)}_{3}(0,0,0, \Lambda)$. Again, the generalization to $A^{(r_1, r_2)}_{3}(k_1,k_2,k_3, \Lambda)$ follows with Taylor's theorem ($\ref{TE0}$).

In the same manner, we treat the case $n=2$, $p \le 2$, $k_i=0$:
\begin{eqnarray}
\lVert \Lambda^{2-r_1} A^{(r_1, r_2)}_{2}(0,0, \Lambda) \rVert & \le &  \Lambda_R^{-r_1} \int_{\Lambda_R}^{\Lambda} ds \ s \big(...\big) \\
\lVert \Lambda^{2-r_1} \partial^2 A^{(r_1, r_2)}_{2}(k_1,k_2, \Lambda)_{k_1=k_2=0} \rVert & \le & \Lambda_R^{-r_1} \int_{\Lambda_R}^{\Lambda} ds \ s^{-1} \big(...\big) .
\end{eqnarray}

\pagebreak
Thus ($\ref{Bound}$) follows with Lemma ($\ref{l}$)  for $A^{(r_1, r_2)}_{2}(0,0, \Lambda)$ and $\partial^2 A^{(r_1, r_2)}_{2}(k_1,k_2, \Lambda)_{k_i=0} $. These are the first two coefficients of a Taylor expansion ($\ref{TE1}$) of $A^{(r_1, r_2)}_{2}(k_1,k_2, \Lambda)$ around $k_i=0$. Since  the remainder term corresponds to the case $p+n \ge 5$, the bound ($\ref{Bound}$) is established for  $A^{(r_1, r_2)}_{2}(k_1,k_2, \Lambda)$. 

Finally, we treat the case $n=1$, $p=0$. Here, it suffices to prove ($\ref{Bound}$) for $A^{(r_1, r_2)}_{1}(0, \Lambda)$ since $A_1$ is only defined for $k=0$. 
\end{enumerate}
We thus have established the induction step for all $n_0 \in \mathbb{N}$.
\begin{flushright}
$\Box$
\end{flushright}
\end{proof}
The bounds ($\ref{Bound}$) for the vertex fuctions $A_{n}^{(r_1, r_2)}(\Lambda)$ would still allow an oscillatory dependence on the UV cutoff $\Lambda_0$. Therefore we will also prove convergence of the vertex functions as $\Lambda_0 \rightarrow \infty$. In Polchinki's original version \cite{Pol}, this step amounts to analyzing the quantity $V(\Lambda)$ that has been defined in eq. ($\ref{v}$) in the last chapter. However, as Keller, Kopper and Salmhofer have shown \cite{KKS}, a major shortcut in the proof can be achieved by pursuing a slightly different path. Instead of defining $V(\Lambda)$ as the total derivative of the effective potential $L(\Lambda)$ with respect to $\Lambda_0$ holding the renormalizable couplings fixed and \textit{then} analyzing associated vertex functions $V_n(\Lambda)$ by integrating a renormalization group inequality that can be deduced out of the RGE ($\ref{RG_v}$), they start by integrating the RGIs for the vertex functions $L_n(\Lambda)$ of the effective potential $L(\Lambda)$ and \textit{then} differentiate the result with respect to $\Lambda_0$. The fixing of the renormalizable couplings at the renormalization scale is incorporated through the initial conditions for the $L_n(\Lambda)$ at $\Lambda=\Lambda_R$.

We will follow the latter path and extend the analysis \cite{KKS} to the case of $\phi^3+ \phi^4$ theory.

\begin{satz}[Convergence] \label{ConvTh} Let there be renormalization conditions ($\ref{rc1}$)-($\ref{c4}$). Assume that to order $r_1, r_2$ in perturbation theory in $\lambda_4^R$ and $\lambda_5^R $ 
\begin{eqnarray}
||\partial^p A_{n}^{(r_1, r_2)}(p_1,...,p_{n}, \Lambda)|| \le \Lambda^{-p}  \left( \frac{\Lambda}{\Lambda_R} \right)^{r_1} Pln \left( \frac{\Lambda_0}{\Lambda_R} \right), \label{BoundW}
\end{eqnarray}
and that for $n+p \ge 5$
\begin{eqnarray}
|| \Lambda_0 \frac{d}{d \Lambda_0} \partial^p A_{n}^{(r_1, r_2)}(p_1,...,p_{n}, \Lambda_0)|| \le \Lambda_0^{-p}  \left( \frac{\Lambda_0}{\Lambda_R} \right)^{r_1} Pln \left( \frac{\Lambda_0}{\Lambda_R} \right) . \label{ini2}
\end{eqnarray}
Then
\begin{eqnarray}
||\Lambda_0 \frac{d}{d \Lambda_0} \partial^p A_{n}^{(r_1, r_2)}(p_1,...,p_{n}, \Lambda)|| \le \Lambda^{-p}  \left( \frac{\Lambda}{\Lambda_R} \right)^{r_1} \frac{\Lambda}{\Lambda_0} Pln \left( \frac{\Lambda_0}{\Lambda_R} \right)  \label{ConvE}
\end{eqnarray}
where $\Lambda_R \le \Lambda \le \Lambda_0$.

\end{satz}

\begin{proof} Note that eq. ($\ref{BoundW}$) is a weaker version of the bound ($\ref{Bound}$) and thus has already been established in Theorem ($\ref{BoundTh}$). We employ the same induction scheme as in the proof of Theorem ($\ref{BoundTh}$). In order to prove the induction step, we will use the RGIs ($\ref{RGI3}$) and ($\ref{RGI5}$).
\begin{enumerate}
\item \textbf{The case} $p+n \ge 5$:

We begin with the RGI ($\ref{RGI3}$). At the RHS, we plug in ($\ref{BoundW}$), ($\ref{ini2}$) as the initial condition at $\Lambda_0$ and ($\ref{ConvE}$) as the induction hypothesis. The result is
\begin{eqnarray}
|| \frac{d}{d \Lambda_0}  \Lambda^{4-n-r_1}  \partial^p A_{n}^{(r_1, r_2)}(p_1,...,p_{n}, \Lambda)|| & \le &  \Lambda_R^{-r_1} \Lambda_0^{3-n-p} Pln \left( \frac{\Lambda_0}{\Lambda_R} \right) \nonumber \\ &&  + \ \Lambda_R^{-r_1}  Pln \left( \frac{\Lambda_0}{\Lambda_R} \right) \Lambda_0^{-2} \int_\Lambda^{\Lambda_0} s^{4-n-p} ds \nonumber \\ &\le&   \Lambda_R^{-r_1} \Lambda_0^{-2} \Lambda^{5-n-p} Pln \left( \frac{\Lambda_0}{\Lambda_R} \right),
\end{eqnarray}
which is equivalent to ($\ref{ConvE}$).

\item \textbf{The case} $p+n \le 4$:

We start with the RGI ($\ref{RGI5}$) and plug in ($\ref{BoundW}$) and ($\ref{ConvE}$) as the induction hypothesis:
\begin{eqnarray}
\lVert \frac{d}{d \Lambda_0} \Lambda^{4-n-r_1} \partial^p A^{(r_1, r_2)}_{n}(\Lambda) \rVert &\le& \lVert \frac{d}{d \Lambda_0} \Lambda_R^{4-n-r_1} \partial^p A^{(r_1, r_2)}_{n}(\Lambda_R) \rVert \nonumber \\ && \ + \  \Lambda_R^{-r_1} Pln \left( \frac{\Lambda_0}{\Lambda_R} \right) \Lambda_0^{-2} \int_{\Lambda_R}^{\Lambda} ds \ s^{4-n-p} . \nonumber \\
\end{eqnarray}
Again, the renormalization conditions ($\ref{rc1}$)-($\ref{c4}$) for the renormalizable couplings $\rho_a$ serve as initial conditions at $\Lambda_R$. Since these conditions mean that the couplings $\rho_a$ are kept \textbf{fixed} at the renormalization scale, their total derivatives with respect to $\Lambda_0$ vanish at $\Lambda=\Lambda_R$:
\begin{eqnarray}
\frac{d}{d \Lambda_0} \rho_a (\Lambda_R) =0.
\end{eqnarray}
We therefore have
\begin{eqnarray}
\lVert \frac{d}{d \Lambda_0}  \rho^{(r_1, r_2)}_a (\Lambda) \rVert &\le&  \Lambda_R^{-r_1} Pln \left( \frac{\Lambda_0}{\Lambda_R} \right) \Lambda_0^{-2} \int_{\Lambda_R}^{\Lambda} ds \ s^{4-n-p} 
\end{eqnarray}
where $n, p$ have to be taken according to the definitions ($\ref{rcc1}$)-($\ref{rcc5}$) of the couplings $\rho_a$. Thus, eq. ($\ref{ConvE}$) follows easily for all $\frac{d}{d \Lambda_0} \rho_a(\Lambda)$. In complete analogy to the strategy set forth in the proof of Theorem ($\ref{BoundTh}$), we use Taylor expansions ($\ref{TE0}$) and ($\ref{TE1}$) to expand $A_n^{(r_1, r_2)}(k_1,...,k_n,\Lambda)$ around $k_i=0$ and thus to reconstruct the full $\frac{d}{d \Lambda_0} A_n^{(r_1, r_2)}(k_1,...,k_n,\Lambda)$. Since all remaining terms in the expansions correspond to the case $n+p \ge 5$, the claim ($\ref{ConvE}$) is established for $n+p \le 4$.
\end{enumerate}
\begin{flushright}
$\Box$
\end{flushright}
\end{proof}
Given Theorem ($\ref{ConvTh}$) and using Lemma ($\ref{l}$), we may easily integrate eq. ($\ref{ConvE}$) with respect to $\Lambda_0$ in order to explicitly show the convergence of the vertex functions $A_{n}^{(r_1, r_2)}(\Lambda, \Lambda_0)$ to a no-cutoff limit  
\begin{equation}
A_{n}^{cont \ (r_1, r_2)}(\Lambda) := \lim_{\Lambda_0 \rightarrow \infty} A_{n}^{(r_1, r_2)}(\Lambda, \Lambda_0). \label{contlimA} 
\end{equation}
The result is the analogon to eq. ($\ref{ConvE}$):
\begin{eqnarray}
|| A_{n}^{(r_1, r_2)}(p_1,...,p_{n}, \Lambda, \Lambda_0)- A_{n}^{cont \ (r_1, r_2)}(p_1,...,p_{n},\Lambda)  || \le \left( \frac{\Lambda}{\Lambda_R} \right)^{r_1} \frac{\Lambda}{\Lambda_0} Pln \left( \frac{\Lambda_0}{\Lambda_R} \right) . \nonumber \\   \label{conv2}
\end{eqnarray}

\end{subsection}

\begin{subsection}{Uniqueness of the no-cutoff limit}
In order to establish the boundedness and convergence of the dimensionless vertex functions $A^{(r_1, r_2)}_n(\Lambda)$ in Theorem  ($\ref{BoundTh}$) and ($\ref{ConvTh}$), it has been necessary to specify initial conditions at the UV cutoff scale $\Lambda_0$. These are the momentum derivatives of the vertex functions $\partial^p A^{(r_1, r_2)}_n(\Lambda_0)$ for $n+p \ge 5$. We will now explicitly show in perturbation theory in $\lambda_4^R$ and $\lambda_5^R $ that the no-cutoff limits $A_{n}^{cont \ (r_1, r_2)}(\Lambda)$ do not depend on these initial conditions as long as they are sufficiently small in the sense of eq. ($\ref{ini}$).

We follow the strategy set forth in section ($\ref{RenFlowOver}$). For a given set of initial conditions $\grave{a}$ la ($\ref{ini}$),
\begin{eqnarray}
||\partial^p A_{n}^{(r_1, r_2)}(p_1,...,p_{n}, \Lambda_0)|| \le \Lambda_0^{-p}  \left( \frac{\Lambda_0}{\Lambda_R} \right)^{r_1} Pln \left( \frac{\Lambda_0}{\Lambda_R} \right) , \ \ \ \ n+p \ge 5 \nonumber,
\end{eqnarray}
we construct another one via the parametrization
\begin{eqnarray}
\partial^p A_{n}^{(r_1, r_2)}(\Lambda_0) \rightarrow \partial^p \tilde{A}_{n}^{(r_1, r_2)}(\Lambda_0) := t \ \partial^p A_{n}^{(r_1, r_2)}( \Lambda_0), \ \ \ t \in [0,1], \ \ \ n+p \ge 5 . \nonumber \\  \label{shapeA}
\end{eqnarray}
Obviously, $t=0$ corresponds to the case $\partial^p A_{n}^{(r_1, r_2)}(\Lambda_0)=0,  \ \ n+p \ge 5$, whereas $t=1$ leaves the original set of initial conditions unchanged. Thus, the ''running'' vertex functions become dependent on the parameter $t$:
\begin{eqnarray}
\partial^p A_{n}^{(r_1, r_2)}= \partial^p A_{n}^{(r_1, r_2)}(p_1,...,p_{n}, \Lambda, \Lambda_0, t).
\end{eqnarray}
We could now proceed as in section ($\ref{RenFlowOver}$) and analyze the quantity $W(\Lambda)$ defined in eq. ($\ref{w}$), that is the total derivative of the vertex functions with respect to the parameter $t$ holding the renormalizable couplings fixed at $\Lambda=\Lambda_R$. Instead, we will choose a strategy very similar to the proof of convergence of the vertex functions in Theorem  ($\ref{ConvTh}$): we start by integrating the RGIs for the vertex functions, thereby employing the renormalization conditions, and \textit{then} differentiate the resulting inequalities with respect to the parameter $t$. The first step has already been performed in eqns. ($\ref{RGI2}$) and ($\ref{RGI4}$), and we proceed to the second one.
\begin{enumerate}
\item Differentiation of ($\ref{RGI2}$) with respect to the parameter $t$ yields
\begin{eqnarray}
\hspace{-0.5cm} \lVert \frac{d}{d t} \Lambda^{4-n-r_1} \partial^p A^{(r_1, r_2)}_{n}(\Lambda) \rVert   &\le&   \lVert \frac{d}{d t} \Lambda_0^{4-n-r_1} \partial^p  A^{(r_1, r_2)}_{n}(\Lambda_0) \rVert \nonumber \\ &&  + c_{n,p} \int_{\Lambda}^{\Lambda_0} ds \ s^{3-n-r_1} \Bigg( \lVert  \frac{d}{d t} \partial^p A^{(r_1, r_2)}_{n+2}(s) \rVert \nonumber \\ && \qquad \qquad \qquad \qquad \qquad \qquad  + \lVert  \frac{d}{d t} \partial^p A^{(r_1, r_2)}_{n+1}(s) \rVert  \nonumber \\  && + 2 \sum_{...} s^{-p_1} \lVert  \frac{d}{d t} \partial^{p_2} A_{l}^{(s_1, s_2)}(s) \rVert \lVert \partial^{p_3} A^{(r_1-s_1, r_2-s_2)}_{n+2-l}(s) \rVert \Bigg) .   \nonumber \\ \label{RGI3U}
\end{eqnarray} 
Note that in ($\ref{RGI3U}$),  the terms of ($\ref{RGI3}$) that stem from the differentiation of the upper bound of the integral in ($\ref{RGI2}$) are missing.
\item Differentiation of ($\ref{RGI4}$) with respect to the parameter $t$ yields
\begin{eqnarray}
\hspace{-0.5cm} \lVert \frac{d}{d t} \Lambda^{4-n-r_1} \partial^p A^{(r_1, r_2)}_{n}(\Lambda) \rVert &\le& \lVert \frac{d}{d t} \Lambda_R^{4-n-r_1} \partial^p A^{(r_1, r_2)}_{n}(\Lambda_R) \rVert \nonumber \\ && +  c_{n,p} \int_{\Lambda_R}^{\Lambda} ds \ s^{3-n-r_1} \Bigg( \lVert \frac{d}{d t} \partial^p A^{(r_1, r_2)}_{n+2}(s) \rVert \nonumber \\ && \qquad \qquad \qquad \qquad \qquad \qquad  + \lVert \frac{d}{d t} \partial^p A^{(r_1, r_2)}_{n+1}(s) \rVert  \nonumber \\ && + 2 \sum_{...} s^{-p_1} \lVert \frac{d}{d t} \partial^{p_2} A_{l}^{(s_1, s_2)}(s) \rVert \lVert \partial^{p_3} A^{(r_1-s_1, r_2-s_2)}_{n+2-l}(s) \rVert \Bigg) . \nonumber \\ \label{RGI5U}
\end{eqnarray} 
\end{enumerate}

\begin{satz}[Uniqueness of the no-cutoff limit] \label{UniTh} Let there be renormalization conditions ($\ref{rc1}$)-($\ref{c4}$). Assume that to order $r_1, r_2$ in perturbation theory in $\lambda_4^R$ and $\lambda_5^R $ 
\begin{eqnarray}
||\partial^p A_{n}^{(r_1, r_2)}(p_1,...,p_{n}, \Lambda)|| \le \Lambda^{-p}  \left( \frac{\Lambda}{\Lambda_R} \right)^{r_1} Pln \left( \frac{\Lambda_0}{\Lambda_R} \right), \label{BoundWU}
\end{eqnarray}
and that for $n+p \ge 5$
\begin{eqnarray}
|| \frac{d}{d t} \partial^p A_{n}^{(r_1, r_2)}(p_1,...,p_{n}, \Lambda_0)|| \le \Lambda_0^{-p}  \left( \frac{\Lambda_0}{\Lambda_R} \right)^{r_1} Pln \left( \frac{\Lambda_0}{\Lambda_R} \right) . \label{ini2U}
\end{eqnarray}
Then
\begin{eqnarray}
|| \frac{d}{d t} \partial^p A_{n}^{(r_1, r_2)}(p_1,...,p_{n}, \Lambda)|| \le \Lambda^{-p}  \left( \frac{\Lambda}{\Lambda_R} \right)^{r_1} \frac{\Lambda}{\Lambda_0} Pln \left( \frac{\Lambda_0}{\Lambda_R} \right)  \label{Uni}
\end{eqnarray}
where $\Lambda_R \le \Lambda \le \Lambda_0$.
\end{satz}
\begin{proof}
The proof goes in complete analogy to the proof of Theorem ($\ref{ConvTh}$) if the RGIs ($\ref{RGI3}$) and ($\ref{RGI5}$) are replaced by ($\ref{RGI3U}$) and ($\ref{RGI5U}$).
\begin{flushright}
$\Box$
\end{flushright}
\end{proof}
We now may integrate the inequality ($\ref{Uni}$) over $t$ with integration limits $0$ and $1$:
\begin{eqnarray}
|| \partial^p A_{n}^{(r_1, r_2)}(p_1,...,p_{n}, \Lambda,1) - \partial^p A_{n}^{(r_1, r_2)}(p_1,...,p_{n}, \Lambda,0) || \le \Lambda^{-p}  \left( \frac{\Lambda}{\Lambda_R} \right)^{r_1} \frac{\Lambda}{\Lambda_0} Pln \left( \frac{\Lambda_0}{\Lambda_R} \right) . \nonumber \\ \label{ind2}
\end{eqnarray}
Since eq. ($\ref{ind2}$) is valid for \textit{any} set of initial conditions $\partial^p A_{n}^{(r_1, r_2)}(\Lambda_0), \ n+p \ge5$, as long as they satisfy eq. ($\ref{ini}$), we conclude with the triangle inequality that for two different sets $\partial^p A_{n}^{A (r_1, r_2)}(\Lambda_0)$ and $\partial^p A_{n}^{B(r_1, r_2)}(\Lambda_0)$ which are in accordance with eq. ($\ref{ini}$) the associated ''running'' vertex functions satisfy
\begin{eqnarray}
|| \partial^p A_{n}^{A(r_1, r_2)}(p_1,...,p_{n}, \Lambda) - \partial^p A_{n}^{B(r_1, r_2)}(p_1,...,p_{n}, \Lambda) || \le \Lambda^{-p}  \left( \frac{\Lambda}{\Lambda_R} \right)^{r_1} \frac{\Lambda}{\Lambda_0} Pln \left( \frac{\Lambda_0}{\Lambda_R} \right) . \nonumber \\  \label{UniR}
\end{eqnarray}
This result is the analogon to eq. ($\ref{indf}$) and shows that the no-cutoff limit of the vertex functions $A_{n}^{cont \ (r_1, r_2)}(\Lambda)$ is {independent} of the choice of the initial surface as long as the initial values $\partial^p A_{n}^{(r_1, r_2)}(\Lambda_0), \ n+p \ge5$, remain sufficiently small.

For finite $\Lambda_0$, eq. ($\ref{UniR}$) means that at the scale $\Lambda$ the ignorance about the exact values of the bare vertex functions  $\partial^p A_{n}^{(r_1, r_2)}(\Lambda_0), \ n+p \ge5$, amounts to an indetermination of the ''running'' vertex functions $\partial^p A_{n}^{(r_1, r_2)}(p_1,...,p_{n}, \Lambda)$ of the order of $ \frac{\Lambda}{\Lambda_0} Pln \left( \frac{\Lambda_0}{\Lambda_R} \right)$.

\end{subsection}

\end{section}

\chapter[Effective Field Theories with Flow Equations]{Effective Field Theories from the Viewpoint of the Renormalization Group} \label{EffFlow}

We investigate the predictivity of an effective field theory that has a finite UV cutoff scale $\Lambda_0$ by means of the renormalization group flow equations. Therefore, additional nonrenormalizable coupling constants are introduced and the vertex functions of the effective potential are expanded into perturbation series in the renormalized renormalizable and some of the bare nonrenormalizable couplings. This is referred to as ''generalized perturbation theory''. New bounds for the vertex functions are established in generalized perturbation theory in an attempt to unify the results of \cite{Wiec} and \cite{KKS}, and improvement conditions for the nonrenormalizable couplings are imposed at the renormalization scale $\Lambda_R$. 
We prove that there exist small initial conditions for the nonrenormalizable couplings at the UV cutoff scale $\Lambda_0$ such that appropriately chosen improvement conditions can be met. The proof is done in perturbation theory in the renormalized renormalizable couplings and in the deviations of the renormalized nonrenormalizable couplings from the values they would take at the renormalization scale $\Lambda_R$ for vanishing initial conditions at the UV cutoff scale $\Lambda_0$. The main advantage of our approach as compared to \cite{Wiec} is that the case of vanishing renormalizable couplings does not pose any problems and thus also nonrenormalizable theories can be treated. Finally, it is proven in generalized perturbation theory that the improvement conditions lead to an enhanced predictivity of the effective field theory at scales $\Lambda << \Lambda_0$ for a finite UV cutoff $\Lambda_0$. The present chapter can be seen as a rigorous version of the concepts introduced in section ($\ref{EffFlowOver}$).

\begin{section}[Initial and improvement conditions]{Initial and improvement conditions for the nonrenormalizable couplings}

\begin{subsection}{Generalized perturbation theory and RG inequalities for vertex functions}  \label{GPT}
In this chapter, the main objective is to analyze the behaviour of the effective potential $L(\phi, \Lambda, \Lambda_0)$ as additional renormalization conditions for some of the nonrenormalizable couplings are introduced at the renormalization scale $\Lambda_R$. These come in addition to the renormalization conditions for the renormalizable couplings, and are again referred to as ''\textit{improvement conditions}''. The analysis will be carried out in perturbation theory in the renormalizable couplings and in some of the nonrenormalizable couplings.

In order to define nonrenormalizable coupling constants, we consider a Taylor expansion of the vertex fuctions $L_{n} (k_1,...,k_{n}, \Lambda)$ introduced in eq. ($\ref{Lexp}$) around $k_i=0$ up to  $\mathcal{O}(k^4)$: 
\begin{eqnarray}
L_{n} (k_1,...,k_{n}, \Lambda) &=& L_{n} (0,...,0, \Lambda) +  \frac{1}{2} \sum_{i_1,i_2=1}^{n-1}  k_{i_1}^{\mu_1} k_{i_2}^{\mu_2} \ \partial^{\mu_1}_{i_1,n} \partial^{\mu_2}_{i_2,n} L_{n} (\tilde{k}_1,...,\tilde{k}_{n}, \Lambda)|_{\tilde{k_i}=0} \nonumber \\ && +  \frac{1}{4!} \sum_{i_1...i_4=1}^{n-1}  k_{i_1}^{\mu_1} ... k_{i_4}^{\mu_4} \ \partial^{\mu_1}_{i_1,n} ... \partial^{\mu_4}_{i_4,n} L_{n} (\tilde{k}_1,...,\tilde{k}_{n}, \Lambda)|_{\tilde{k_i}=0}  \nonumber \\ && + \frac{1}{5!} \sum_{i_1,...,i_6=1}^{n-1} k_{i_1}^{\mu_1} ... k_{i_6}^{\mu_6} \int_0^1 d \tau (1-\tau)^5 \partial^{\mu_1}_{i_1,n} ... \partial^{\mu_6}_{i_6,n} L_{n} (\tilde{k}_1,...,\tilde{k}_{n}, \Lambda)|_{\tilde{k_i}=\tau k_i} . \nonumber \\ \label{TE2}
\end{eqnarray}
The expansion coefficients in ($\ref{TE2}$) may be used to define the additional running coupling constants\footnote{See Appendix ($\ref{MDV}$) for details on the definitions of the couplings.}, which we will denote by $\rho_{\tilde{a}}(\Lambda)$ in accordance with the notation employed in section ($\ref{EffFlowOver}$):
\begin{eqnarray} 
\rho_{6}(\Lambda) \ \delta^{\mu \nu} &:=& {\partial^\mu_{i,3}} \partial^\nu_{i,3} L_3(k_1,k_2,k_3,\Lambda)|_{k_i=0},  \ \ \ i=1,2 \label{rcc6} \\
 \rho_{7}(\Lambda)  \ \delta^{\mu \nu} &:=& \partial_{i,3}^{\mu} {\partial_{j,3}^{\nu}} L_3(k_1,k_2,k_3,\Lambda)|_{k_i=0} , \ \ \  i \ne j = 1,2 \label{rcc7} \\
\rho_8(\Lambda) &:=& L_5(0,...,0, \Lambda) \label{rcc8} \\
\rho_{9}(\Lambda) \ I^{\mu \nu \rho \sigma} &:=& \partial^\mu_{1,2} \partial^\nu_{1,2} \partial^\rho_{1,2} \partial^\sigma_{1,2} L_2(k_1,k_2,\Lambda)|_{k_i=0} \\
\rho_{10}(\Lambda) \ \delta^{\mu \nu} &:=& \partial^\mu_{i,4} \partial^\nu_{i,4} L_4(k_1,k_2,k_3,k_4,\Lambda)|_{k_i=0} ,  \ \ \ i=1...3 \\
 \rho_{11}(\Lambda)  \ \delta^{\mu \nu} &:=& \partial_{i,4}^{\mu} \partial_{j,4}^{\nu}  L_4(k_1,k_2,k_3,k_4,\Lambda)|_{k_i=0} , \ \ \  i \ne j = 1...3 \\ 
\rho_{12}(\Lambda) &:=& L_6(0,...,0,\Lambda)  \label{rcc12}
\end{eqnarray}
where $I^{\mu \nu \rho \sigma}:= \delta^{\mu \nu} \delta^{\rho  \sigma} + \delta^{\mu \rho} \delta^{\nu  \sigma} + \delta^{\mu \sigma} \delta^{\nu  \rho}$. The couplings ($\ref{rcc6}$)-($\ref{rcc12}$) are indeed nonrenormalizable since their mass dimensions $D_{\rho_{\tilde{a}}}$ satisfy
\begin{eqnarray}
D_{\rho_{\tilde{a}}} = \left\{ \begin{array}{c c c} -1  & ,& {\tilde{a}}=6,7,8 \\ -2  &,  &  {\tilde{a}}=9,...,12     \ , \end{array} \right. 
\end{eqnarray}
see eq. ($\ref{DL}$).
In the following, we will restrict our considerations to the couplings with $D_{\rho_{\tilde{a}}} \ge -1$, that is from now on ${\tilde{a}}=1...8$ unless stated otherwise.

For the bare nonrenormalizable couplings $\rho_{\tilde{a}}(\Lambda_0)$, we define initial conditions
\begin{eqnarray}
\rho_{\tilde{a}}(\Lambda_0) := \rho_{\tilde{a}}^0, \ \ \ \tilde{a} =6...8. \label{iniCNR}
\end{eqnarray}
In analogy to the notation introduced in eq. ($\ref{dCR}$), we employ
\begin{eqnarray}
\lambda_{\tilde{a}}^0(\Lambda) := \Lambda^{-D_{\rho_{\tilde{a}}}} \rho_{\tilde{a}}^0 . \label{dCNR}
\end{eqnarray}
The vertex functions $L_{n} (k_1,...,k_{n}, \Lambda)$ are expanded in perturbation theory in the renormalized renormalizable couplings $\rho_4^R$ and $\rho_5^R$ introduced in eqns. ($\ref{c3}$)-($\ref{c4}$) and in the bare nonrenormalizable couplings  $\rho_{\tilde{a}}^0, \  \tilde{a} =6...8$, of eqns. ($\ref{iniCNR}$):
\begin{eqnarray}
L_{n} (k_1,...,k_{n}, \Lambda) = \sum_{r_1, ...,r_{9}=0}^{\infty} (\rho_4^R)^{r_1} (\rho_5^R)^{r_2} (\rho_{6}^0)^{r_3} (\rho_{7}^0)^{r_4}  (\rho_{8}^0)^{r_5}  L_{n}^{(r_1,...,r_5)} (k_1,...,k_{n}, \Lambda). \nonumber \\  \label{GenPert}
\end{eqnarray}
We refer to the expansion ($\ref{GenPert}$) as ''generalized perturbation theory'', in contrast to the expansion ($\ref{Pert}$) that is done only in the renormalizable couplings $\rho_4^R$ and $\rho_5^R$. The reason for using bare nonrenormalizable couplings as additional expansion parameters is that in the integrations performed in the inductive proofs of upcoming Theorems  ($\ref{BoundThII}$) and ($\ref{PreTh}$), the initial values for the nonrenormalizable couplings have to be specified at the bare scale $\Lambda_0$.

In the following, we will work with the dimensionless vertex functions $A_n(\Lambda)$ defined in eq. ($\ref{dV0}$). With
\begin{eqnarray}
A_{n}^{(r_1, ..., r_5)} (k_1,...,k_{n}, \Lambda) := \Lambda^{n-4 +r_1 - (r_3 +r_4 +r_5) )}  L^{(r_1, ..., r_5)}_{n} (k_1,...,k_{n}, \Lambda) \label{Apert2_Pert}
\end{eqnarray}
and the definitions ($\ref{dCR}$), ($\ref{dCNR}$) we may write
\begin{eqnarray}
A_{n} (k_1,...,k_{n}, \Lambda) = \sum_{r_1, ...,r_{5}=0}^{\infty} (\lambda_4^R)^{r_1} (\lambda_5^R)^{r_2} (\lambda_{6}^0)^{r_3} (\lambda_{7}^0)^{r_4} (\lambda_{8}^0)^{r_5}  A_{n}^{(r_1,...,r_5)} (k_1,...,k_{n}, \Lambda). \nonumber \\ \label{Apert2}
\end{eqnarray}
The perturbative expansion ($\ref{Apert2}$) is sensible only for small dimensionless couplings. Therefore we again impose eq. ($\ref{smallR}$) as an additional constraint to the renormalization conditions ($\ref{rc1}$)-($\ref{c4}$) and equally demand that the initial conditions ($\ref{iniCNR}$) are such that
\begin{eqnarray}
\lambda_{\tilde{a}}^0(\Lambda) \le 1, \ \ \ \ {\tilde{a}}=6...8.  \label{smallNR}
\end{eqnarray}
Note that because of the definition ($\ref{dCNR}$), this implies $\rho_{\tilde{a}}^0 \le \Lambda_0^{-1}, \ {\tilde{a}}=6...8$, for $\Lambda \le \Lambda_0$ and thus
\begin{equation}
\lambda_{\tilde{a}}^0(\Lambda) \le \frac{\Lambda}{\Lambda_0}, \ \ \ \ {\tilde{a}}=6...8.  \label{smallNRExp}
\end{equation}
To $0 th$ order in perturbation theory in all coupling constants the vertex functions vanish:
\begin{eqnarray}
 A_{n}^{(0,...,0)} (k_1,...,k_{n}, \Lambda) = 0. \label{zero2}
\end{eqnarray}
Finally, we define overall orders in perturbation theory via
\begin{eqnarray}
r &:=& r_1 +... + r_5 \\
r_{NR} &:=& r_3+ r_4 + r_5 .
\end{eqnarray}
$r_{NR}$ gives the overall order in perturbation theory in the bare nonrenormalizable couplings $\lambda_{6}^0$,  $\lambda_{7}^0$ and $\lambda_{8}^0$.

We rewrite the  RGE ($\ref{RGE0}$) in terms of the dimensionless vertex functions  ($\ref{dV0}$) and express the resulting equation in perturbation theory in the couplings $\lambda_4^R$, $\lambda_5^R$, $\lambda_{6}^0$,  $\lambda_{7}^0$ and $\lambda_{8}^0$. Applying the bounds ($\ref{q1}$) and ($\ref{q2}$) as well as the condition ($\ref{smallg}$) for the renormalization constant $A$ we arrive at 
\begin{eqnarray} 
&& \hspace{-0.7cm} \lVert \frac{d}{d \Lambda} \Lambda^{4-n-r_1 + r_{NR} } \partial^p A^{(r_1,..., r_5)}_{n}(\Lambda)  \rVert \nonumber \\   && \qquad \qquad \qquad \le   c_{n,p} \ \Lambda^{3-n-r_1 + r_{NR} } \Bigg( \lVert \partial^p A^{(r_1,..., r_5)}_{n+2}(\Lambda) \rVert + \lVert \partial^p A^{(r_1,..., r_5)}_{n+1}(\Lambda) \rVert \nonumber \\ && \qquad \qquad \qquad  \qquad \qquad +  \sum_{...} \Lambda^{-p_1} \lVert \partial^{p_2} A_{l}^{(s_1,..., s_5)}(\Lambda) \rVert \lVert \partial^{p_3} A^{(r_1-s_1,..., r_5-s_5)}_{n+2-l}(\Lambda) \rVert \Bigg) \nonumber \\ \label{RGI1_NR}
\end{eqnarray} 
where
\begin{eqnarray}
 \sum_{...} \ := \sum_{p_1,p_2,p_3: \ \sum p_i =p} \ \sum_{l=1}^{n} \ \sum_{\substack{s_1,..., s_5=0 \\ 1 \le s \le r-1 }}^{r_1,..., r_5}  .  \label{sumabb}
\end{eqnarray} 
The RGI ($\ref{RGI1_NR}$) is the analogon to the RGI ($\ref{RGI1}$) of chapter ($\ref{RenFlow}$). Integrating  ($\ref{RGI1_NR}$) with respect to $\Lambda$ and differentiating the resulting RGIs with respect to $\Lambda_0$, we can easliy deduce RGIs that correspond to the RGIs ($\ref{RGI2}$), ($\ref{RGI3}$), ($\ref{RGI4}$) and ($\ref{RGI5}$) of chapter ($\ref{RenFlow}$). The same holds true for the RGIs ($\ref{RGI3U}$) and ($\ref{RGI5U}$). We leave this step to the reader.
 
\end{subsection}

\begin{subsection}{Boundedness of vertex functions in generalized perturbation theory}
We establish new bounds for the norms $\lVert  \partial^p A^{(r_1, ..., r_5)}_{n}(\Lambda) \rVert$ of the vertex functions in generalized perturbation theory. These bounds can be viewn as the unification of the improved bounds of Keller, Kopper and Salmhofer \cite{KKS} and the results of Wieczerkowski \cite{Wiec} concerning boundedness of vertex functions in generalized perturbation theory. The bounds established in Theorem ($\ref{BoundTh}$) are included as the special case of $0th$ order in perturbation theory in the nonrenormalizable couplings.

\begin{satz}[Boundedness II] \label{BoundThII}
Given the renormalization conditions ($\ref{rc1}$)-($\ref{c4}$) and the initial conditions ($\ref{iniCNR}$), and assuming that
\begin{eqnarray}
||\partial^p A_{n}^{(r_1,..., r_5)}(p_1,...,p_{n}, \Lambda_0)|| \le \Lambda_0^{-p}  \left( \frac{\Lambda_0}{\Lambda_R} \right)^{r_1} Pln \left( \frac{\Lambda_0}{\Lambda_R} \right) \label{iniNR}
\end{eqnarray} 
for $n+p \ge 6$, to order $r_1,..., r_5$ in perturbation theory in $\lambda_4^R$, $\lambda_5^R$, $\lambda_{6}^0$,  $\lambda_{7}^0$ and $\lambda_{8}^0$
\begin{eqnarray}
&& ||\partial^p A_{n}^{(r_1,..., r_5)}(p_1,...,p_{n}, \Lambda)|| \nonumber \\ &&  \qquad \qquad  \qquad  \le \Lambda^{-p} \left( \frac{\Lambda}{\Lambda_R} \right)^{r_1} \left( \frac{\Lambda_0}{\Lambda} \right)^{r_{NR}}   \left( \delta_{r_{NR}, 0} \ Pln\left( \frac{\Lambda}{\Lambda_R} \right) +  \frac{\Lambda}{\Lambda_0}  Pln \left( \frac{\Lambda_0}{\Lambda_R} \right) \right) \nonumber \\ \label{BoundNR}
\end{eqnarray}
where $Pln(z)$ denotes some polynomial in $ln(z)$ whose coefficients are taken to be nonnegative and $\Lambda_R \le \Lambda \le \Lambda_0$.

\end{satz}

\begin{proof} At first we note that to $0th$ order in perturbation theory in the bare nonrenormalizable couplings  $\lambda_{6}^0$,  $\lambda_{7}^0$ and $\lambda_{8}^0$, that is for $r_{NR} = 0$, Theorem ($\ref{BoundThII}$) reduces to Theorem ($\ref{BoundTh}$) that has already been proven. We therefore proceed to the case $r_{NR} > 0$ and follow the induction scheme employed in the proof of Theorem ($\ref{BoundTh}$), that is the proof is done via induction in both the overall order in perturbation theory, $r=r_1+...+r_5$, and the number of external legs, $n$, of the vertex functions $A_n^{(r_1, ..., r_5)}$. The induction start is now given by 
\begin{eqnarray}
\{(r,n): r=0  \ \wedge \ n \in  \mathbb{N} \} \vee \{ (r,n): r \ge 1 \ \wedge \ n > 3r+2  \} 
\end{eqnarray}
because of the perturbation theory in the bare couplings.
\begin{enumerate}
\item \textbf{The case} $p+n \ge 6$: On the RHS of the analogon to the RGI ($\ref{RGI2}$), we plug in eq. ($\ref{iniNR}$) as the initial condition at $\Lambda_0$ and eq. ($\ref{BoundNR}$) as the induction hypothesis. The integrals can be solved easily and as a result, ($\ref{BoundNR}$) is established.

\item \textbf{The case} $p+n = 5$: Plugging eqns. ($\ref{iniCNR}$) as initial conditions at $\Lambda_0$ and eq. ($\ref{BoundNR}$) as the induction hypothesis into the analogon to the RGI ($\ref{RGI2}$), the bound ($\ref{BoundNR}$) can be proven for all expansion coefficients ($\ref{rcc6}$)-($\ref{rcc8}$). We will demonstrate this for the example $n=3$, $p=2$:
\begin{eqnarray}
\hspace{-0.5cm} ||\Lambda^{1-r_1 + r_{NR} } \partial^2 A^{(r_1, ..., r_5)}_3(k_1,k_2,k_3,\Lambda)|_{k_i=0}|| &\le& \delta^{r 1} \delta^{r_6 1} + \delta^{r 1} \delta^{r_7 1} \nonumber \\ && + \ \Lambda_R^{-r_1} \Lambda_0^{r_{NR}}   \int_{\Lambda}^{\Lambda_0} ds \ s^{-2}  \frac{s}{\Lambda_0}   Pln \left( \frac{\Lambda_0}{\Lambda_R}  \right) . \nonumber \\
\end{eqnarray} 
Solving the intergral and multiplying the whole inequality with $\Lambda^{r_1 -1 - r_{NR} }$ yields ($\ref{BoundNR}$).

\item \textbf{The case} $p+n \le 4$: The bound ($\ref{BoundNR}$) has to be proven for the expansion coefficients ($\ref{rcc1}$) -($\ref{rcc5}$). We start with the analogon to the RGI ($\ref{RGI4}$) and plug in eq. ($\ref{BoundNR}$) as the induction hypothesis. Because we are treating the case $r_{NR} > 0$, all initial conditions at $\Lambda_R$ vanish.  Solving the integrals establishes the bound ($\ref{BoundNR}$).

\end{enumerate}
\begin{flushright}
$\Box$
\end{flushright}
\end{proof}

\end{subsection}

\begin{subsection}{Inversion of the RG trajectory}  \label{InvSec}

At the renormalization scale $\Lambda_R$ we impose improvement conditions for the nonrenormalizable couplings $\rho_{\tilde{a}}, \ \tilde{a}=6...12$. The canonical dimensions of the couplings for which improvement conditions are defined determines an improvement index $s$ that has been introduced in section ($\ref{EffFlowOver}$):
\begin{eqnarray}
 s=1: && \rho_{\tilde{a}}(\Lambda_R) = \rho^{NR}_{\tilde{a}} , \ \  {\tilde{a}}=6,...,8    \label{impro1} \\  
 s=2: && \rho_{\tilde{a}}(\Lambda_R) = \rho^{NR}_{\tilde{a}}   , \ \  {\tilde{a}}=6,...,12       \label{impro2} \\
 \vdots \ \ \ \ \ &  & \ \ \ \ \ \ \ \ \ \ \ \vdots   \nonumber 
\end{eqnarray}
We restrict ourselves to the case $s=1$. In analogy to eqns. ($\ref{dCR}$), ($\ref{dCNR}$) we introduce dimensionless versions of the couplings ($\ref{impro1}$):
\begin{eqnarray}
\lambda^{NR}_{\tilde{a}}(\Lambda) := \Lambda^{- D_{\rho_{\tilde{a}}} } \rho^{NR}_{\tilde{a}}. \label{improdl}
\end{eqnarray}
For small initial values of the nonrenormalizable couplings as implied by eq. ($\ref{smallNR}$), 
\begin{eqnarray}
\lambda_{\tilde{a}}^0(\Lambda_0) \le 1, \ \ \ {\tilde{a}}=6,...,8, \label{smalliniNR}
\end{eqnarray}
the improvement conditions ($\ref{impro1}$) cannot be chosen freely. It follows from Theorem ($\ref{ConvTh}$) and eq. ($\ref{conv2}$) that for initial conditions $\grave{a}$ la eq.  ($\ref{smalliniNR}$) we have
\begin{eqnarray}
|| \lambda_{\tilde{a}}^{(r_1, r_2)}(\Lambda_R, \Lambda_0)- \lambda_{\tilde{a}}^{cont \ (r_1, r_2)}(\Lambda_R)  ||  & \le & \left( \frac{\Lambda}{\Lambda_R} \right)^{r_1} \frac{\Lambda_R}{\Lambda_0} \ Pln \left( \frac{\Lambda_0}{\Lambda_R} \right)  , \ \ \ \ \tilde{a}=6...8, \nonumber \\ \label{conv2p}
\end{eqnarray}
where $\lambda_{\tilde{a}}^{cont \ (r_1, r_2)}(\Lambda_R):= \lambda_{\tilde{a}}^{(r_1, r_2)}(\Lambda_R, \infty)$ and the index $(r_1,r_2)$ refers to perturbation theory in the renormalizable couplings $\lambda_4^R$ and $\lambda_5^R$. Since the latter are small according to eqns. ($\ref{smallR}$), ($\ref{smallRExp}$) we conclude that for small initial values ($\ref{smalliniNR}$) the running nonrenormalizable couplings must satisfy at $\Lambda=\Lambda_R$
\begin{eqnarray}
|| \lambda_{\tilde{a}}(\Lambda_R, \Lambda_0)- \lambda_{\tilde{a}}^{cont}(\Lambda_R)  || &\le&  \frac{\Lambda_R}{\Lambda_0} Pln \left( \frac{\Lambda_0}{\Lambda_R} \right)  , \ \ \ \ \tilde{a}=6...8, \label{conv2b}
\end{eqnarray}
in agreement to the discussion at the beginning of section ($\ref{EffFlowOver}$).

In the following, it will be proven in perturbation theory that the inversion is also true: for given improvement conditions ($\ref{improdl}$) that satisfy
\begin{eqnarray}
|| \lambda_{\tilde{a}}^{NR}(\Lambda_R)- \lambda_{\tilde{a}}^{cont}(\Lambda_R)  || &\le&  \frac{\Lambda_R}{\Lambda_D} Pln \left( \frac{\Lambda_D}{\Lambda_R} \right)  , \ \ \ \Lambda_D > \Lambda_R,  \ \ \ \tilde{a}=6...8,  \ \ \label{con1}
\end{eqnarray}
there exist small initial conditions $\grave{a}$ la eq.  ($\ref{smalliniNR}$) such that the improvement conditions can be met:
\begin{eqnarray}
\lambda_{\tilde{a}}(\Lambda_R, \Lambda_0) = \lambda_{\tilde{a}}^{NR}(\Lambda_R)   \label{con2}
\end{eqnarray}
for an UV cutoff scale $\Lambda_0$ that obeys $\Lambda_0 \le \Lambda_D$. 

\bigskip

We begin by expanding the running dimensionless coupling constants $\lambda_{\tilde{a}}(\Lambda ), \ \tilde{a}=6...8,$ in perturbation theory in the renormalized renormalizable couplings $\lambda_4^R$ and $\lambda_5^R$ and in the bare nonrenormalizable couplings  $\lambda_6^0$, $\lambda_7^0$ and $\lambda_8^0$:
\begin{eqnarray}
\lambda_{\tilde{a}} (\Lambda) = \sum_{r_1, ...,r_{5}=0}^{\infty} \lambda_{\tilde{a}}^{(r_1,...,r_5)} (\Lambda)  (\lambda_4^R)^{r_1} (\lambda_5^R)^{r_2} (\lambda_{6}^0)^{r_3} \lambda_{7}^0)^{r_4}  (\lambda_{8}^0)^{r_5}  .  \label{Lambpert}
\end{eqnarray}
Eq. ($\ref{Lambpert}$) follows essentially from ($\ref{Apert2}$) and the definitions ($\ref{rcc6}$)-($\ref{rcc8}$). Next, we introduce auxiliary variables $\overline{\lambda}_{\tilde{a}}(\Lambda)$ which are defined as the values the running dimensionless couplings take for vanishing initial conditions ($\ref{iniCNR}$). This can also be viewed as the $0th$ order in perturbation theory in the nonrenormalizable couplings:
\begin{eqnarray}
\overline{\lambda}_{\tilde{a}}(\Lambda):= \sum_{\substack{r_1, ...,r_{5}=0 \\ r_{NR}=0}}^{\infty} \lambda_{\tilde{a}}^{(r_1,...,r_5)} (\Lambda)  (\lambda_4^R)^{r_1} (\lambda_5^R)^{r_2} (\lambda_{6}^0)^{r_3} (\lambda_{7}^0)^{r_4}  (\lambda_{8}^0)^{r_5}  \label{zRcoup}
\end{eqnarray}
where $r_{NR}$ is again the overall order in perturbation theory in the nonrenormalizable couplings. Defining the deviations
\begin{eqnarray}
\Delta \lambda_{\tilde{a}}(\Lambda) := \lambda_{\tilde{a}} (\Lambda) -\overline{\lambda}_{\tilde{a}}(\Lambda)
\end{eqnarray}
we may write
\begin{eqnarray}
\Delta \lambda_{\tilde{a}}(\Lambda) = \sum_{\substack{r_1, ...,r_{5}=0 \\ r_{NR} \ge 1}}^{\infty} \lambda_{\tilde{a}}^{(r_1,...,r_5)} (\Lambda)  (\lambda_4^R)^{r_1} (\lambda_5^R)^{r_2} (\lambda_{6}^0)^{r_3} (\lambda_{7}^0)^{r_4}  (\lambda_{8}^0)^{r_5}  . \label{devi}
\end{eqnarray}
Furthermore, we note that
\begin{eqnarray}
\lambda_{\tilde{a}}^{(0,0,1,0,0)} &=& \delta_{\tilde{a}, 6} \\
\lambda_{\tilde{a}}^{(0,0,0,1,0)} &=& \delta_{\tilde{a}, 7} \\
\lambda_{\tilde{a}}^{(0,0,0,0,1)} &=& \delta_{\tilde{a}, 8} .
\end{eqnarray}
Thus, for $\tilde{a}=6...8$
\begin{eqnarray}
\lambda_{\tilde{a}}^0= \sum_{\substack{r_1, ...,r_{5}=0 \\ r_{NR} \ge 1 \\ r=1 }}^{\infty} \lambda_{\tilde{a}}^{(r_1,...,r_5)} (\Lambda)  (\lambda_4^R)^{r_1} (\lambda_5^R)^{r_2} (\lambda_{6}^0)^{r_3} (\lambda_{7}^0)^{r_4}  (\lambda_{8}^0)^{r_5}
\end{eqnarray}
and we finally arrive at
\begin{eqnarray}
\lambda_{\tilde{a}}^0(\Lambda) = \Delta \lambda_{\tilde{a}}(\Lambda) - \sum_{\substack{r_1, ...,r_{5}=0 \\ r_{NR} \ge 1 \\ r \ge 2 }}^{\infty} \lambda_{\tilde{a}}^{(r_1,...,r_5)} (\Lambda)  (\lambda_4^R)^{r_1} (\lambda_5^R)^{r_2} (\lambda_{6}^0)^{r_3} (\lambda_{7}^0)^{r_4}  (\lambda_{8}^0)^{r_5} . \label{inv}
\end{eqnarray}
Eq. ($\ref{devi}$) gives us the deviations $\Delta \lambda_{\tilde{a}} (\Lambda)$ in perturbation theory in the renormalized renormalizable couplings $\lambda_4^R$ and $\lambda_5^R$ and the bare nonrenormalizable couplings  $\lambda_6^0$, $\lambda_7^0$ and $\lambda_8^0$, whereas eq. ($\ref{inv}$) serves as the starting point for the inversion. This will become clear in the following.

We expand the couplings $\lambda_6^0$, $\lambda_7^0$ and $\lambda_8^0$ in perturbation theory in $\lambda_4^R$ and $\lambda_5^R$ and in the deviations 
\begin{eqnarray}
\Delta \lambda_{\tilde{a}}^R(\Lambda) := ({\Lambda_R}/{\Lambda})^{D_{\rho_{\tilde{a}}}} \Delta \lambda_{\tilde{a}}(\Lambda_R)  , \ \ \ {\tilde{a}}=6,...,8.   \label{devR}
\end{eqnarray}
The result is
\begin{eqnarray}
\lambda_{\tilde{a}}^0(\Lambda) =  \sum  \lambda_{\tilde{a}}^{0 \ (l_1, ..., l_5)}(\Lambda) (\lambda_4^R)^{l_1} (\lambda_5^R)^{l_2} (\Delta \lambda_6^R)^{l_3} (\Delta \lambda_{7}^R)^{l_4}  (\Delta \lambda_{8}^R)^{l_5}.   \label{L_0pert}
\end{eqnarray}
For $\tilde{a}=6...8$ we have
\begin{eqnarray}
\lambda_{\tilde{a}}^{0 \ (0, 0, 0, 0 , 0)} &=& \lambda_{\tilde{a}}^{0 \ (1, 0, 0, 0 ..., 0)} = \lambda_{\tilde{a}}^{0 \ (0, 1, 0, 0 ..., 0)} =0  \\
\lambda_{\tilde{a}}^{0 \ (0, 0, 1, 0 , 0)} &=& \delta_{\tilde{a}, 6}  \\
\lambda_{\tilde{a}}^{0 \ (0, 0, 0, 1 , 0)} &=& \delta_{\tilde{a}, 7}  \\
\lambda_{\tilde{a}}^{0 \ (0, 0, 0, 0 , 1)} &=& \delta_{\tilde{a}, 8} .
\end{eqnarray}
In order to keep the complexity of the upcoming equations under control, some notations have to be introduced:
\begin{eqnarray} 
\mathbf{A}_{i_1} &:=& (A_{i_1}^1, ..., A_{i_1}^5 ), \ \ \ i_1=1...r_3 \\
\mathbf{B}_{i_2} &:=& (B_{i_2}^1, ..., B_{i_2}^5 ),  \ \ \ i_2=1...r_4 \\
\mathbf{C}_{i_3} &:=& (C_{i_3}^1, ..., C_{i_3}^5 ), \ \ \ i_3=1...r_5,
\end{eqnarray} 
the ''norms''
\begin{eqnarray}
|\mathbf{A}_{i_1}| &:=& A_{i_1}^1+ ...+ A_{i_1}^5, \ \ \ i_1=1...r_3 \\
& \vdots &   \nonumber
\end{eqnarray}
and
\begin{eqnarray}
\mathbf{T} &:=& \mathbf{A}_1+...+ \mathbf{A}_{r_3} +  \mathbf{B}_1+...+ \mathbf{C}_{r_5}.
\end{eqnarray}
We now plug the perturbative expansions ($\ref{L_0pert}$) into eq. ($\ref{inv}$) and arrive at a system of equations for the coefficients $\lambda_{\tilde{a}}^{0 \ (l_1, ..., l_5) }(\Lambda)$ of the expansion ($\ref{L_0pert}$):
\begin{eqnarray}
&& \hspace{-0.7cm} \lambda_{\tilde{a}}^{0 \ (l_1, ..., l_5) }(\Lambda) \nonumber \\  && =  \Delta \lambda^{(l_1, ..., l_5)}_{\tilde{a}}(\Lambda)  - \sum_{\substack{s_1, s_2=0 \\ s_1+s_2 \le l_1+l_2-1}}^{l_1, l_2} \sum_{\substack{r_3, r_4 ,r_{5}=0 \\ r_{NR} \ge 1} }^{\infty} \ \sum^\infty_{ \begin{subarray}{l}  \mathbf{A}_1... \mathbf{A}_{r_3}  \\ \mathbf{B}_1... \mathbf{B}_{r_4}  \\ \mathbf{C}_1... \mathbf{C}_{r_5}=0  \end{subarray}} \Big( \delta_{T^1, s_1} \delta_{T^2, s_2} \delta_{T^3, l_3} \delta_{T^4, l_4}  \delta_{T^5, l_5}   \nonumber \\ &&  \ \quad  \qquad \qquad \qquad  \qquad  \qquad \left.  \lambda_{\tilde{a}}^{(l_1-s_1,l_2-s_2, r_3,r_4,r_5)} (\Lambda)   \lambda_{6}^{0 \ \mathbf{A}_1}... \lambda_{6}^{0 \ \mathbf{A}_{r_3}} \lambda_{7}^{0 \ \mathbf{B}_{1}} ... \lambda_{8}^{0 \ \mathbf{C}_{r_5}}   \right) \nonumber \\&& \quad \ - \sum_{\substack{r_3, r_4 ,r_{5}=0 \\ r_{NR} \ge 2} }^{\infty} \ \sum^\infty_{ \begin{subarray}{l}  \mathbf{A}_1... \mathbf{A}_{r_3}  \\ \mathbf{B}_1... \mathbf{B}_{r_4}  \\ \mathbf{C}_1... \mathbf{C}_{r_5}=0  \end{subarray}} \left( \delta_{T^1, l_1}  ... \delta_{T^5, l_5}   \lambda_{\tilde{a}}^{(0,0, r_3,r_4,r_5)} (\Lambda) \lambda_{6}^{0 \ \mathbf{A}_1}... \lambda_{6}^{0 \ \mathbf{A}_{r_3}} \lambda_{7}^{0 \ \mathbf{B}_{1}} ... \lambda_{8}^{0 \ \mathbf{C}_{r_5}}    \right) \nonumber \\ \label{MainInv}
\end{eqnarray}
where $\tilde{a}=6...8$. Note that with $l:=l_1+ ...+ l_5$, 
\begin{eqnarray}
\Delta \lambda^{(l_1, ..., l_5)}_{6}(\Lambda_R) &=& \delta_{l_3, 1 } \delta_{l, 1} \label{weissnicht1} \\
\Delta \lambda^{(l_1, ..., l_5)}_{7}(\Lambda_R) &=& \delta_{l_4, 1 } \delta_{l, 1} \\
\Delta \lambda^{(l_1, ..., l_5)}_{8}(\Lambda_R) &=& \delta_{l_5, 1 } \delta_{l, 1} .  \label{weissnicht2}
\end{eqnarray}
Eq. ($\ref{MainInv}$) is our key equation for the inversion of the RG trajectory, as is shown by the following theorem.

\begin{satz}[Inversion of the RG trajectory]  \label{invTh}
For $\tilde{a}=6...8$ and $l_{\Delta}:= l_3+ l_4+ l_5$ let 
\begin{eqnarray}
|| \Delta \lambda^{(l_1, ..., l_5)}_{\tilde{a}}(\Lambda) || &\le&  \left( \frac{\Lambda}{\Lambda_R} \right)^{l_1} \left( \frac{\Lambda_0}{\Lambda} \right)^{l_\Delta}  \frac{\Lambda}{\Lambda_0}  Pln \left( \frac{\Lambda_0}{\Lambda_R} \right) .
\end{eqnarray}
Then to order $l_1,...,l_5$ in perturbation theory in $\lambda_4^R$, $\lambda_5^R$ and the deviations $\Delta \lambda_6^R$, $\Delta \lambda_7^R$ and $\Delta \lambda_8^R$
\begin{eqnarray}
|| \lambda^{0 \ (l_1, ..., l_5)}_{\tilde{a}}(\Lambda) || &\le&  \left( \frac{\Lambda}{\Lambda_R} \right)^{l_1} \left( \frac{\Lambda_0}{\Lambda} \right)^{l_\Delta} \frac{\Lambda}{\Lambda_0}  Pln \left( \frac{\Lambda_0}{\Lambda_R} \right)    \label{smalliniP}
\end{eqnarray}
where $\Lambda_R \le \Lambda \le \Lambda_0$.
\end{satz}
\begin{proof} The proof is done via induction in the overall order in perturbation theory $l=l_1+...+l_5$ of the couplings $\lambda_4^R$, $\lambda_5^R$ and the deviations $\Delta \lambda_6^R$, $\Delta \lambda_7^R$ and $\Delta \lambda_8^R$. We note that
\begin{eqnarray}
\lambda^{0 \ (l_1, ..., l_5)}_{\tilde{a}}(\Lambda) = \Delta \lambda^{(l_1, ..., l_5)}_{\tilde{a}}(\Lambda) \ \ \ \text{for} \  \ \ l=1,
\end{eqnarray}
and thus the induction start is established. In order to prove the induction step, we make use of eq. ($\ref{MainInv}$). The overall orders in perturbation theory $|\mathbf{A}_{i_1}| $, $|\mathbf{B}_{i_2}| $,  $|\mathbf{C}_{i_3}| $ of all coefficients $\lambda_{6}^{0 \ \mathbf{A}_{i_1}}$, $\lambda_{7}^{0 \ \mathbf{B}_{i_2}}$, $\lambda_{8}^{0 \ \mathbf{C}_{i_3}}$  appearing on rhe RHS of ($\ref{MainInv}$) satisfy
\begin{eqnarray}
|\mathbf{A}_{i_1}|, |\mathbf{B}_{i_2}|, |\mathbf{C}_{i_3}|  &<& l.
\end{eqnarray}
This follows from the constraints on the summation bounds appearing on the RHS of ($\ref{MainInv}$).  We thus plug eq. ($\ref{smalliniP}$) as the induction hypothesis into the RHS of eq. ($\ref{MainInv}$).  The bounds for $\lambda_{\tilde{a}}^{(l_1-s_1,l_2-s_2, r_3,r_4,r_5)} (\Lambda)$ and $\lambda_{\tilde{a}}^{(0,0, r_3,r_4,r_5)} (\Lambda)$ follow from Theorem ($\ref{BoundThII}$) and the induction step can be established.
\begin{flushright}
$\Box$
\end{flushright}
\end{proof}
The couplings $\lambda_4^R$, $\lambda_5^R$ are small as defined in eqns. ($\ref{smallR}$), ($\ref{smallRExp}$). If we demand also smallness of the deviations ($\ref{devR}$),
\begin{eqnarray}
\Delta \lambda_{\tilde{a}}^R(\Lambda) \le 1, \ \ \ {\tilde{a}}=6,...,8, \label{smallinidev}
\end{eqnarray}
Theorem ($\ref{invTh}$) implies that for
\begin{eqnarray}
|| \Delta \lambda_{\tilde{a}}(\Lambda) || &\le&  \frac{\Lambda}{\Lambda_0}  Pln \left( \frac{\Lambda_0}{\Lambda_R} \right)
\end{eqnarray}
we have 
\begin{eqnarray}
||\lambda^0_{\tilde{a}}(\Lambda) || &\le&  \frac{\Lambda}{\Lambda_0}  Pln \left( \frac{\Lambda_0}{\Lambda_R} \right). \label{smallindeed}
\end{eqnarray}
Let us remember that $\overline{\lambda}_{\tilde{a}}(\Lambda)$ refers to the case of vanishing initial conditions ($\ref{iniCNR}$), as follows from the definition ($\ref{zRcoup}$). Thus, Theorem ($\ref{ConvTh}$) and eq. ($\ref{conv2b}$) imply
\begin{eqnarray}
|| \overline{\lambda}_{\tilde{a}}(\Lambda_R, \Lambda_0)- \lambda_{\tilde{a}}^{cont}(\Lambda_R)  || \le  \frac{\Lambda_R}{\Lambda_0} Pln \left( \frac{\Lambda_0}{\Lambda_R} \right)  , \ \ \ \ \tilde{a}=6...8. 
\end{eqnarray}
Since $\Delta \lambda_{\tilde{a}}(\Lambda) = \lambda_{\tilde{a}} (\Lambda) -\overline{\lambda}_{\tilde{a}}(\Lambda)$,  we conclude with the triangle inequality that if
\begin{eqnarray}
|| {\lambda}_{\tilde{a}}(\Lambda_R, \Lambda_0)- \lambda_{\tilde{a}}^{cont}(\Lambda_R)  || \le  \frac{\Lambda_R}{\Lambda_0} Pln \left( \frac{\Lambda_0}{\Lambda_R} \right)  , \ \ \ \ \tilde{a}=6...8,   \label{smalldevcont}
\end{eqnarray}
the bare couplings are small $\grave{a}$ la eq. ($\ref{smallindeed}$). The conjecture discussed at the beginning of this section in eqns. ($\ref{con1}$) and ($\ref{con2}$) has therefore been proven: there exist small initial conditions $\lambda^0_{\tilde{a}}(\Lambda), \ \tilde{a}=6...8$, such that the improvement conditions $\lambda_{\tilde{a}}(\Lambda_R, \Lambda_0) = \lambda_{\tilde{a}}^{NR}(\Lambda_R)$ for the nonrenormalizable couplings can be met as long as they satisfy eq. ($\ref{con1}$) .

\bigskip
In eq. ($\ref{L_0pert}$), we have expanded the bare couplings $\lambda_6^0$, $\lambda_7^0$ and $\lambda_8^0$ in perturbation theory in $\lambda_4^R$ and $\lambda_5^R$ and in the deviations $\Delta \lambda_{\tilde{a}}^R, \ \tilde{a}=6...8$. This approach is suitable for analyzing the predictivity of an effective field theory for a given UV cutoff scale $\Lambda_0$, as will be shown in the next section. In particular, the special case where the renormalized renormalizable couplings $\lambda_4^R$ and $\lambda_5^R$ are chosen to be zero can be treated without problems. 

However, if the aim is to fix nonrenrormalizable couplings at their continuum values $\lambda_{\tilde{a}}^{cont}(\Lambda_R)$ in order to achieve improved convergence in the limt $\Lambda_0 \rightarrow \infty$, this approach is problematic because of the $\Lambda_0$-dependence of the deviations $\Delta \lambda_{\tilde{a}}^R(\Lambda, \Lambda_0), \ \tilde{a}=6...8$. In this case, it seems reasonable to follow the strategy outlined by Wieczerkowski \cite{Wiec} and to replace the expansion ($\ref{L_0pert}$) by one that is done in the renormalizable couplings only\footnote{This essentially amounts to consider the equations beginning with eq. ($\ref{L_0pert}$) in the case of $0th$ order in perturbation theory in the deviations $\Delta \lambda_{\tilde{a}}^R$, $\tilde{a}=6...8$, and to withdraw eqns. ($\ref{weissnicht1}$)-($\ref{weissnicht2}$). }:
\begin{eqnarray}
\lambda_{\tilde{a}}^0(\Lambda) =  \sum  \lambda_{\tilde{a}}^{0 \ (l_1, l_2)}(\Lambda) (\lambda_4^R)^{l_1} (\lambda_5^R)^{l_2} .   \label{L_0pertR}
\end{eqnarray}
A repetition of the analysis outlined in this section then yields an equation\footnote{In \cite{Wiec}, the part of this equation that corresponds to the last term of the RHS of ($\ref{MainInv}$) seems to be missing. } similar to eq. ($\ref{MainInv}$) and ultimately the analogon of Theorem ($\ref{invTh}$) in perturbation theory in the renormalizable couplings $\lambda_4^R$ and $\lambda_5^R$. 

\end{subsection}

\end{section}

\begin{section}[Predictivity of effective field theories]{Predictivity of effective field theories and improved convergence}  \label{Presec}

In order to investigate the predictivity of an effective field theory, we follow the strategy set forth in section ($\ref{EffFlowOver}$). For a given set of initial conditions $\grave{a}$ la ($\ref{iniNR}$),
\begin{eqnarray}
||\partial^p A_{n}^{(r_1,..., r_5)}(p_1,...,p_{n}, \Lambda_0)|| \le \Lambda_0^{-p}  \left( \frac{\Lambda_0}{\Lambda_R} \right)^{r_1} Pln \left( \frac{\Lambda_0}{\Lambda_R} \right) , \ \ \ \ n+p \ge 6 \nonumber,
\end{eqnarray}
we construct another one via the parametrization
\begin{eqnarray}
\partial^p A_{n}^{(r_1, ..., r_5)}(\Lambda_0) \rightarrow \partial^p \tilde{A}_{n}^{(r_1, ..., r_5)}(\Lambda_0) := t \ \partial^p A_{n}^{(r_1, ..., r_5)}( \Lambda_0), \ \ \ t \in [0,1], \ \ \ n+p \ge 6 . \nonumber \\  \label{shapeANR}
\end{eqnarray}
Obviously, $t=0$ corresponds to the case $\partial^p A_{n}^{(r_1,..., r_5)}(\Lambda_0)=0,  \ \ n+p \ge 6$, whereas $t=1$ leaves the original set of initial conditions unchanged. The ''running'' vertex functions become dependent on the parameter $t$:
\begin{eqnarray}
\partial^p A_{n}^{(r_1, ..., r_5)}= \partial^p A_{n}^{(r_1, ..., r_5)}(p_1,...,p_{n}, \Lambda, \Lambda_0, t).
\end{eqnarray}

\begin{satz}[Predictivity of Effective Field Theories] \label{PreTh} Let there be renormalization conditions ($\ref{rc1}$)-($\ref{c4}$) and  improvement conditions ($\ref{impro1}$), meaning the improvement index is chosen to be $s=1$. Assume that to order $r_1, ..., r_5$ in perturbation theory in $\lambda_4^R$, $\lambda_5^R$, $\lambda_{6}^0$,  $\lambda_{7}^0$ and $\lambda_{8}^0$
\begin{eqnarray}
||\partial^p A_{n}^{(r_1,..., r_5)}(p_1,...,p_{n}, \Lambda)|| \le \Lambda^{-p}  \left( \frac{\Lambda}{\Lambda_R} \right)^{r_1}  \left( \frac{\Lambda_0}{\Lambda} \right)^{r_{NR}} Pln \left( \frac{\Lambda_0}{\Lambda_R} \right), \label{BoundWP}
\end{eqnarray}
and that for $n+p \ge 6$
\begin{eqnarray}
|| \frac{d}{d t} \partial^p A_{n}^{(r_1,..., r_5)}(p_1,...,p_{n}, \Lambda_0)|| \le \Lambda_0^{-p}  \left( \frac{\Lambda_0}{\Lambda_R} \right)^{r_1} Pln \left( \frac{\Lambda_0}{\Lambda_R} \right) . \label{ini2P}
\end{eqnarray}
Then 
\begin{eqnarray}
|| \frac{d}{d t} \partial^p A_{n}^{(r_1,..., r_5)}(p_1,...,p_{n}, \Lambda)|| \le \Lambda^{-p}  \left( \frac{\Lambda}{\Lambda_R} \right)^{r_1}  \left( \frac{\Lambda_0}{\Lambda} \right)^{r_{NR}} \left(\frac{\Lambda}{\Lambda_0}\right)^{2} Pln \left( \frac{\Lambda_0}{\Lambda_R} \right) \nonumber \\ \label{Pre}
\end{eqnarray}
where $\Lambda_R \le \Lambda \le \Lambda_0$.
\end{satz}

\begin{proof} We note that ($\ref{BoundWP}$) is a weaker bound than ($\ref{BoundNR}$) and thus has already been established in Theorem ($\ref{BoundThII}$). We follow the by now well-known induction scheme outlined in the proof of Theorem ($\ref{BoundTh}$).
\begin{enumerate}

\item \textbf{The case} $p+n \ge 6$: On the RHS of the analogon to the RGI ($\ref{RGI3U}$), we plug in eq. ($\ref{ini2P}$) as the initial condition at $\Lambda_0$, eq. ($\ref{Pre}$) as the induction hypothesis and eq. ($\ref{BoundWP}$). The integrals can be solved and as a result, ($\ref{Pre}$) is established.

\item \textbf{The case} $p+n \le 5$: The inequality ($\ref{Pre}$) has to be proven for the derivatives of the expansion coefficients ($\ref{rcc1}$) -($\ref{rcc5}$) and ($\ref{rcc6}$)-($\ref{rcc8}$) with respect to the parameter $t$. We start with the analogon to the RGI ($\ref{RGI5U}$) and plug in eq. ($\ref{Pre}$) as the induction hypothesis and eq. ($\ref{BoundWP}$). All initial conditions at $\Lambda_R$ vanish because the renormalization conditions ($\ref{rc1}$)-($\ref{c4}$) and improvement conditions ($\ref{impro1}$) mean that the couplings $\rho_{\tilde{a}}, \ \tilde{a}=1...8,$ are independent of $t$ at $\Lambda=\Lambda_R$.  Solving the integrals establishes ($\ref{Pre}$).
\end{enumerate}
\begin{flushright}
$\Box$
\end{flushright}
\end{proof}
We can integrate the inequality ($\ref{Pre}$) over $t$ with integration limits $0$ and $1$:
\begin{eqnarray}
|| \partial^p A_{n}^{(r_1,..., r_5)}(p_1,...,p_{n}, \Lambda,1) - \partial^p A_{n}^{(r_1, ..., r_5)}(p_1,...,p_{n}, \Lambda,0) ||  \qquad \qquad \qquad \nonumber \\ \le \Lambda^{-p}  \left( \frac{\Lambda}{\Lambda_R} \right)^{r_1} \left( \frac{\Lambda_0}{\Lambda} \right)^{r_{NR}}  \left(\frac{\Lambda}{\Lambda_0}\right)^{2} Pln \left( \frac{\Lambda_0}{\Lambda_R} \right) .  \label{ind3}
\end{eqnarray}
Since eq. ($\ref{ind3}$) is valid for \textit{any} set of initial conditions $\partial^p A_{n}^{(r_1,..., r_5)}(\Lambda_0), \ n+p \ge6$, as long as they satisfy eq. ($\ref{iniNR}$), we conclude with the triangle inequality that for two different sets $\partial^p A_{n}^{A (r_1,..., r_5)}(\Lambda_0)$ and $\partial^p A_{n}^{B(r_1,..., r_5)}(\Lambda_0)$ which are in accordance with eq. ($\ref{iniNR}$) the associated ''running'' vertex functions satisfy
\begin{eqnarray}
|| \partial^p A_{n}^{A(r_1,..., r_5)}(p_1,...,p_{n}, \Lambda) - \partial^p A_{n}^{B(r_1,..., r_5)}(p_1,...,p_{n}, \Lambda) || \qquad \qquad \qquad \nonumber \\ \le \Lambda^{-p}  \left( \frac{\Lambda}{\Lambda_R} \right)^{r_1} \left( \frac{\Lambda_0}{\Lambda} \right)^{r_{NR}}  \left(\frac{\Lambda}{\Lambda_0}\right)^{2} Pln \left( \frac{\Lambda_0}{\Lambda_R} \right) . \label{PreF}
\end{eqnarray}
As has been stressed before, the perturbative approach is sensible only for small dimensionless expansion parameters $\grave{a}$ la eqns. ($\ref{smallR}$) for  $\lambda_4^R$, $\lambda_5^R$ and ($\ref{smallNR}$)  for $\lambda_{6}^0$,  $\lambda_{7}^0$, $\lambda_{8}^0$. Theorem ($\ref{invTh}$) and the thereof derived eqns. ($\ref{smalldevcont}$),  ($\ref{smallindeed}$)  ensure that for improvement conditions which satisfy ($\ref{con1})$, that is
\begin{eqnarray*}
|| \lambda_{\tilde{a}}^{NR}(\Lambda_R)- \lambda_{\tilde{a}}^{cont}(\Lambda_R)  || &\le&  \frac{\Lambda_R}{\Lambda_D} Pln \left( \frac{\Lambda_D}{\Lambda_R} \right)  , \ \ \ \Lambda_D > \Lambda_R,  \ \ \ \tilde{a}=6...8,   
\end{eqnarray*}
the couplings $\lambda_{6}^0$,  $\lambda_{7}^0$, $\lambda_{8}^0$ are small in the sense of eqns. ($\ref{smallNR}$) as long as the UV cutoff stays below the scale $\Lambda_D$, $\Lambda_0 \le \Lambda_D$.

Coming back to the discussion at the end of section ($\ref{EffFlowOver}$), we now have proven that the ignorance about the exact initial values $ \partial^p A_{n}^{(r_1,..., r_5)}(\Lambda_0), \ \ n+p \ge 6$, amounts to an indetermination of the ''running'' vertex functions $\partial^p A_{n}^{(r_1,..., r_5)}(\Lambda)$ of the order of $\left( \frac{\Lambda_0}{\Lambda} \right)^{r_{NR}}  \left(\frac{\Lambda}{\Lambda_0}\right)^{2} Pln \left( \frac{\Lambda_0}{\Lambda_R} \right) $. Note that for small couplings $\lambda_{6}^0$,  $\lambda_{7}^0$, $\lambda_{8}^0$ in the sense of eqns. ($\ref{smallNR}$), ($\ref{smallNRExp}$) this amounts to an indetermination of $\partial^p A_{n}(\Lambda)$ of the order of $\left(\frac{\Lambda}{\Lambda_0}\right)^{2} Pln \left( \frac{\Lambda_0}{\Lambda_R} \right) $.   Compared to Theorem ($\ref{UniTh}$) and eq. ($\ref{UniR}$) where only renormalization conditions for the renormalizable couplings have been specified, the improvement conditions ($\ref{impro1}$) for $s=1$ have thus led to enhanced predictivity. 

\bigskip
We will now briefly comment on the case where the nonrenrormalizable couplings $ \rho_{\tilde{a}}(\Lambda_R) , \ \tilde{a}=6..8,$ are fixed to their no-cutoff values $\rho_{\tilde{a}}^{cont}(\Lambda_R)$ in order to achieve improved convergence in the limt $\Lambda_0 \rightarrow \infty$. Following \cite{Wiec}, we then assume that the initial conditions for the nonrenormalizable couplings are defined in perturbation theory in the renormalizable couplings, see eq. ($\ref{L_0pertR}$). Consequently, also the vertex functions are given in perturbation theory in the renormalizable couplings $\lambda_4^R$ and $\lambda_5^R$. Employing the RGIs ($\ref{RGI3}$) and ($\ref{RGI5}$) and following the path set forth in Theorems  ($\ref{PreTh}$) and  ($\ref{ConvTh}$),  improved convergence as compared to eq. ($\ref{conv2}$)  of the vertex functions to their continuum limits can be established:
\begin{eqnarray}
|| A_{n}^{(r_1, r_2)}(p_1,...,p_{n}, \Lambda, \Lambda_0)- A_{n}^{cont \ (r_1, r_2)}(p_1,...,p_{n},\Lambda)  || \le \left( \frac{\Lambda}{\Lambda_R} \right)^{r_1} \left(\frac{\Lambda}{\Lambda_0}\right)^2 Pln \left( \frac{\Lambda_0}{\Lambda_R} \right) . \nonumber \\   \label{conv2impro}
\end{eqnarray}
\end{section}


\begin{chapter}[Polchinski's Equation for Quantum Gravity]{Polchinski's Equation for Euclidean Quantum Gravity}  \label{PolGR}

The standard covariant BRS quantization procedure for Euclidean Einstein gravity is reviewed, as it can be found for instance in \cite{Stelle}, \cite{CLM}. Unlike these authors, we include a cosmological term. Our dynamical variable is a perturbation of the metric density $\sqrt{g} \ g^{\mu \nu}$ around flat space, and in contrast to \cite{Hooft}, \cite{Dono1}, we do not work in the background field method. We state the Slavnov-Taylor-Identities (STI) for quantum Einstein gravity and deduce graviton and ghost propagators. The problems that arise due to the cosmological constant when gravity is treated as a flat-space QFT are discussed, as well as possible resolutions. 
A momentum cutoff regularization for the generating functional is employed and the resulting violation of the gauge invariance and hence the Slavnov-Taylor-Identities is discussed. We review a fine-tuning procedure that has been shown \cite{KM} to cure the analogous problem occuring in the perturbative renormalization of Yang-Mills theory via flow equations. The ultimate aim of this procedure is the restoration of the STI in the no-cutoff limit. First implications of an analogous procedure for quantum gravity are proposed.
Finally, we establish the Polchinski renormalization group equation for Euclidean quantum gravity and discuss the properties that a solution of this RGE must have in order to represent a valid candidate for a quantum theory of gravitation. 

\begin{section}{Covariant quantization of Euclidean gravity}

\begin{subsection}{BRS quantization of Einstein gravity with a cosmological constant}   \label{BRSQ}

The Einstein-Hilbert {action} for Euclidean gravity with a cosmological constant is defined as\footnote{See Appendix ($\ref{Conv}$) for our conventions.}
\begin{equation}
S_{EH}= \frac{1}{{\lambda^2}} \int d^4 x \sqrt{g} \left( -{4\Lambda_K} + {2} R \right) \label{EH}
\end{equation}
where $\lambda^2 := 32 \pi G$ is proportional to Newtons constant $G$ and $\Lambda_K$ is the cosmological constant. The canonical dimensions of $\lambda$ and $\Lambda_K$ are 
\begin{eqnarray}
[\lambda ] &=& \Lambda^{-1} \label{cdN} \\
 {[} \Lambda_K ]  &=& \Lambda^{2} \label{cdK}
\end{eqnarray}
where $\Lambda$ is some scale of mass. We would like to remind the reader of the definitions of the geometrical quantities in the action ($\ref{EH}$):
\begin{eqnarray}
g &=& \det g_{\mu \nu} \\
{R}^{\rho \sigma}_{\ \ \mu \nu} &=& \partial_\mu {\Gamma}^{\rho \sigma}_{\ \ \nu} - \partial_{\nu} \Gamma^{\rho \sigma}_{\ \ \mu} + \Gamma^{\rho}_{\  \alpha \mu} {\Gamma}^{ \alpha \sigma}_{\ \ \nu} - \Gamma^{\rho}_{\  \alpha \nu} {\Gamma}^{ \alpha \sigma}_{\ \ \mu} \label{DCur}  \\
\Gamma_{\sigma \mu \nu} &=& \frac{1}{2} (\partial_\mu g_{\sigma \nu} + \partial_\nu g_{\sigma \mu}- \partial_\sigma g_{\mu \nu}) \label{DChr} \\
R_{\mu \nu} &=& R^{\alpha}_{\ \mu \alpha \nu} \\
R &=& R^\mu_{\ \mu}  . \label{LCC}
\end{eqnarray}
Obviously, the {metric} $g_{\mu \nu}$ is the only dynamical variable the Lagrangian depends on. It turns out convenient to define the contravariant metric density (weight $+1$)
\begin{eqnarray}
\tilde{g}^{\mu \nu}:= \sqrt{g} \ g^{\mu \nu} \label{gdens}
\end{eqnarray}
where $g^{\mu \nu}$ is the inverse metric, $g^{\mu \rho} g_{\rho \nu}= \delta^\mu_{\ \nu}$, and to rewrite the action ($\ref{EH}$) in terms of $\tilde{g}^{\mu \nu}$. With 
\begin{eqnarray}
\det \tilde{g}^{\mu \nu} &=& g^2 \det g^{\mu \nu} = g^2 g^{-1} \nonumber \\ &=& g
\end{eqnarray}
and $\tilde{g}^{\mu \rho} \tilde{g}_{\rho \nu}= \delta^\mu_{\ \nu}$ we find \cite{CLM}
\begin{eqnarray}
S_{EH} &=& \frac{1}{{\lambda^2}} \int d^4 x \left( -4 \Lambda_K \sqrt{\det \tilde{g}^{\mu \nu}} - \frac{1}{2}  \tilde{g}^{\mu \nu} \tilde{g}_{\alpha \beta}\tilde{g}_{\gamma \delta} \partial_\mu \tilde{g}^{\beta \gamma}  \partial_\nu \tilde{g}^{\alpha \delta} \right. \nonumber \\ &&  \qquad  \qquad \qquad \left. + \frac{1}{4} \tilde{g}^{\mu \nu} \tilde{g}_{\alpha \beta} \tilde{g}_{\gamma \delta} \partial_\mu \tilde{g}^{\alpha \beta} \partial_\nu \tilde{g}^{\gamma \delta} + \tilde{g}_{\alpha \beta} \partial_\mu \tilde{g}^{\alpha \nu} \partial_\nu \tilde{g}^{\mu \beta} \right).
\end{eqnarray}
The tensor density $\tilde{g}^{\mu \nu}$ is split up into a flat Euclidean background $\delta^{\mu \nu}$ and a perturbation part $h^{\mu \nu}$:
\begin{eqnarray}
\tilde{g}^{\mu \nu}= \delta^{\mu \nu} + \lambda h^{\mu \nu} . \label{h}
\end{eqnarray}
The quantity $h^{\mu \nu}$ will be our dynamical field variable. With $\frac{1}{x}=\frac{1}{x_0} - \frac{1}{x_0^2} (x - x_0) + \frac{1}{x_0^3} (x - x_0)^2+ ...$ we find
\begin{eqnarray}
\tilde{g}_{\mu \nu}= \delta_{\mu \nu} - \lambda h_{\mu \nu} + \lambda^2 h_{\mu}^{\ \alpha} h_{\alpha \nu} - \lambda^3 h_{\mu}^{\ \alpha} h_{\alpha }^{\ \beta} h_{\beta \nu} +  ...
\end{eqnarray}
where $h_{\mu \nu}= \delta_{\mu \alpha} \delta_{\nu \beta} h^{\alpha \beta}$. In addition, because of $\det x= \exp( tr \ln x)$ and $\ln x = \sum_{n=1}^\infty \frac{1}{n} (-1)^{n+1} (x-1)^n $ it follows that
\begin{eqnarray}
\sqrt{\det \tilde{g}^{\mu \nu}} = 1 + \frac{1}{2} \lambda h - \frac{1}{4} \lambda^2 h^\mu_{\ \rho} h^\rho_{\ \mu} + \lambda^2 \frac{1}{8} h^2 + ... .  \label{KKexp}
\end{eqnarray}
where the notation $h^\mu_{\ \mu}=h$ has been employed\footnote{Sometimes we will use the notation $h$ also as an abbreviation for $h_{\mu \nu}$.}. Thus, we can expand the Einstein-Hilbert action ($\ref{EH}$) in powers of $h_{\mu \nu}$. Since we are expanding around flat space, the scalar curvature $R$ vanishes to $0th$ order in $h_{\mu \nu}$. Furthermore, we would like to stress that for $\Lambda_K \ne 0$ the flat background is \textit{not} a solution of the field equations and we are left with a term linear in $h_{\mu \nu}$ in the action:
\begin{eqnarray}
S_{EH}= \int d^4 x \left( - \frac{4 \Lambda_K}{\lambda^2} + \mathcal{L}^{(1)} + \mathcal{L}^{(2)} + \mathcal{L}^{(3)} + ...  \right) . \label{gaexp}
\end{eqnarray}
Here, $\mathcal{L}^{(1)}$, $\mathcal{L}^{(2)}$,...  denote the Lagrangians linear, bilinear, etc. in $h$. The constant term $- \frac{4\Lambda_K}{\lambda^2}$ will merely give rise to a multiplicative factor in the generating functional of quantum gravity and we will drop it from now on. From Ref \cite{CLM} and eq. ($\ref{KKexp}$) follows\footnote{The contributions linear in $h$ that stem from the expansion of $R$ vanish because they are total derivatives. }
\begin{eqnarray}
\mathcal{L}^{(1)} &=& - 2 \frac{\Lambda_K}{\lambda}  h \label{l1} \\
\mathcal{L}^{(2)} &=&  \Lambda_K \big( h^\mu_{\ \rho} h^\rho_{\ \mu} -  \frac{1}{2} h^2 \big) - \frac{1}{2} \partial_\mu h_{\nu \rho } \partial^\mu h^{\nu \rho} + \frac{1}{4} \partial_\mu h \partial^\mu h + \partial_\mu h^{\mu \rho} \partial^\nu h_{\nu \rho} .  \nonumber \\ \label{l2}
\end{eqnarray}
As shown in Appendix ($\ref{cgdens}$), the tensor density ($\ref{gdens}$) transforms under infinitesimal general coordinate transformations  as 
\begin{eqnarray}
\tilde{g}^{\mu \nu} & \rightarrow & \tilde{g}^{\mu \nu}{'} = \tilde{g}^{\mu \nu} + \mathcal{L}_X \tilde{g}^{\mu \nu}
\end{eqnarray}
where
\begin{eqnarray}
\mathcal{L}_X \tilde{g}^{\mu \nu} = X^\rho \partial_\rho \tilde{g}^{\mu \nu} + \tilde{g}^{\mu \nu} \partial_\rho X^\rho - \tilde{g}^{\rho \nu} \partial_\rho  X^\mu - \tilde{g}^{\mu \rho} \partial_\rho X^\nu.    \label{lg0}
\end{eqnarray}
$\mathcal{L}_X$ denotes the Lie derivative with respect to a vector field $X=X^\rho \partial_\rho  $. Regarding our gravitational field variable $h_{\mu \nu}$, we  define ''gauge transformations''
\begin{eqnarray}
h^{\mu \nu}{'} &=& h^{\mu \nu} + \lambda^{-1}  \mathcal{L}_X \tilde{g}^{\mu \nu} \label{g1} \\
\delta^{\mu \nu}{'} &=& \delta^{\mu \nu} . \label{g2}
\end{eqnarray}
Here, the freedom to ''shift'' the changes induced by the coordinate transformation on $\tilde{g}^{\mu \nu}$ between the parts $\delta^{\mu \nu}$ and $\lambda h^{\mu \nu}$ has been exploited such that only the $h_{\mu \nu}$ field is transformed and the flat Euclidean background $\delta^{\mu \nu}$ remains fixed. With eqns. ($\ref{h}$) and ($\ref{lg0}$) it follows that 
\begin{equation}
\mathcal{L}_X \tilde{g}^{\mu \nu} = \delta^{\mu \nu} \partial_\rho X^\rho - \delta^{\rho \nu} \partial_\rho  X^\mu - \delta^{\mu \rho} \partial_\rho X^\nu + \lambda \big( X^\rho \partial_\rho h^{\mu \nu} + h^{\mu \nu} \partial_\rho X^\rho - h^{\rho \nu} \partial_\rho  X^\mu - h^{\mu \rho} \partial_\rho X^\nu \big) \label{lg}
\end{equation}
in agreement with Ref \cite{Stelle}. 

In order to quantize gravity, the gauge has to be fixed. As a gauge condition we employ the harmonic gauge:
\begin{eqnarray}
F^\mu_{\ \rho \sigma}(h^{\rho \sigma})  &=& 0 \\
\text{with} \ \ \ F^\mu_{\ \rho \sigma}  &=&  \delta^\mu_{\ \rho} \partial_\sigma   . \label{GFH}
\end{eqnarray}
In the following, we will often use the abbreviation $F^\mu=F^\mu_{\ \rho \sigma}(h^{\rho \sigma}) $. The Faddeev-Popov procedure \cite{Peskin} then adds a gauge fixing term
\begin{equation}
S_{GF}(h) = - \frac{1}{2 \xi} \int d^4 x  F^\mu F_\mu \label{GFL}
\end{equation}
as well as a ghost term\footnote{Note that since $\frac{\delta (F^\mu)'}{\delta X^\nu} =\frac{\delta (F^\mu)'}{\delta h'_{\rho \sigma}} \frac{\delta h'_{\rho \sigma}}{\delta X^\nu}$ and $\frac{\delta h'_{\rho \sigma}}{\delta X^\nu} C^\nu = \mathcal{L}_C g_{\rho \sigma}$ we have
$ \int d^4 x  \ \overline{C}_\mu \left(  \frac{\delta (F^{\mu})'}{\delta X^\nu} \right) C^\nu =  \int d^4 x  \ \overline{C}_\mu \left(  \frac{\delta (F^{\mu})'}{\delta h_{\rho \sigma }} \right)  \mathcal{L}_C g_{\rho \sigma} \label{ghR}. $ An expression of the latter kind was used by Ref \cite{MR1}.}
\begin{equation}
S_{GH}(h, C, \overline{C}) =  - \int d^4 x \ \overline{C}_\mu \left( \lambda \frac{\delta (F^{\mu})'}{\delta X^\nu} \right) C^\nu \label{gh}
\end{equation}
to the action where $F(^\mu){'}=F^\mu_{\ \rho \sigma}(h^{\rho \sigma}{'}) $. With eqns. ($\ref{lg}$) and ($\ref{GFH}$) we find
\begin{equation}
\lambda \frac{\delta (F^\mu)'}{\delta X^\nu} = - \delta^\mu_{\ \nu} \partial^\rho \partial_\rho  + \lambda \big( \partial_\rho \partial_\nu h^{\mu \rho} + \partial_\rho h^{\mu \rho} \partial_\nu - \delta^{\mu}_{\ \nu} \partial_\sigma h^{ \sigma \rho} \partial_\rho  - \delta^\mu_{\ \nu}   h^{\rho \sigma} \partial_\rho \partial_\sigma  \big)  \label{fp}
\end{equation}
in agreement with Ref \cite{CLM}. 

Let $\epsilon$ be an anticommuting constant parameter. One can check that the total action
\begin{equation}
S_{tot}= S_{EH}(h) + S_{GF}(h )+ S_{GH}(h, C, \overline{C})   \label{Stot}
\end{equation}
is invariant under the BRS transformations
\begin{eqnarray}
\delta_\epsilon  h^{\mu \nu} &=&  \mathcal{L}_C \tilde{g}^{\mu \nu} \epsilon \label{BRST1} \\
\delta_\epsilon C^\mu &=& \lambda C^\nu \partial_\nu C^\mu \epsilon \label{BRST2}\\
\delta_\epsilon \overline{C}_\mu &=&    \xi^{-1} F_\mu \epsilon. \label{BRST3}
\end{eqnarray}
These are in complete analogy to the corresponding BRS transformations of Yang Mills theory \cite{Peskin}. We may define an extended total action $\tilde{S}_{tot}$ by introducing sources $\beta_{\mu \nu} $ and $\tau_\mu$ that couple to the nonlinear BRS variations ($\ref{BRST1}$) and ($\ref{BRST2}$):
\begin{eqnarray}
\tilde{S}_{tot}:= S_{tot} + \int d^4 x \left( \beta_{\mu \nu} \mathcal{L}_C \tilde{g}^{\mu \nu}  + \tau_\mu  \lambda C^\nu \partial_\nu C^\mu  \right) . \label{StotE}
\end{eqnarray}
One can show that ($\ref{BRST1}$) and ($\ref{BRST2}$)  are nilpotent,
\begin{eqnarray}
\delta_\epsilon  \big( \mathcal{L}_C \tilde{g}^{\mu \nu} \epsilon  \big) &=& 0  \\
\delta_\epsilon  \big( C^\nu \partial_\nu C^\mu \epsilon \big) &=& 0 .
\end{eqnarray}
Thus, also $\tilde{S}_{tot}$ is invariant under the BRS transformations ($\ref{BRST1}$)-($\ref{BRST3}$) and the BRS invariance can be stated in the form\footnote{Since $f(x + \delta x)=f(x)+ \frac{\partial f}{\partial x} \delta x$ the invariance $f(x + \delta x)=f(x)$ means $\frac{\partial f}{\partial x} \delta x=0$.}
\begin{eqnarray}
\delta_\epsilon \tilde{S}_{tot} = \frac{\delta \tilde{S}_{tot}}{\delta \beta_{\mu \nu}} \frac{\delta \tilde{S}_{tot}}{\delta h^{\mu \nu}}  + \frac{\delta \tilde{S}_{tot}}{\delta \tau_{\mu}} \frac{\delta \tilde{S}_{tot}}{\delta C^{\mu }} + \xi^{-1} F_\mu \frac{\delta \tilde{S}_{tot}}{\delta \overline{C}_{\mu }} =0 .   \label{BRSTE}
\end{eqnarray}
Finally, we add a source term
\begin{equation}
S_{J}= \int d^4 x \left( t_{\mu \nu} h^{\mu \nu} + \overline{\sigma}_\mu C^\mu + \sigma^\mu \overline{C}_\mu  \right)
\end{equation}
to the action.

We are now ready to formally write down two unregularized generating functionals for Euclidean quantum Einstein gravity:
\begin{eqnarray}
W(J) = \int \mathcal{D} h_{\mu \nu} \mathcal{D} C^\mu \mathcal{D} \overline{C}_\mu \ e^{ \left. S_{tot} + S_{J} \right.}  \label{GFQG}
\end{eqnarray}
and an extended version
\begin{eqnarray}
\tilde{W}(J, \beta_{\mu \nu} ,  \tau_\mu) = \int \mathcal{D} h_{\mu \nu} \mathcal{D} C^\mu \mathcal{D} \overline{C}_\mu \ e^{ \left. \tilde{S}_{tot} + S_{J} \right.}  \label{GFQGE}
\end{eqnarray}
where we employed the notation
\begin{equation}
J = \{ t_{\mu \nu},\sigma_\mu , \overline{\sigma}^\mu \}.
\end{equation}
Note that $\tilde{W}(J, \beta_{\mu \nu} ,  \tau_\mu) $ involves the composite field operators ($\ref{BRST1}$) and ($\ref{BRST2}$). In addition, we introduce a generating functional $Z(J)$  for connected correlation functions,
\begin{equation}
W(J)=e^{- Z(J)}, \label{conn}
\end{equation}
and similarly an extended version $\tilde{Z}(J, \beta_{\mu \nu} , \tau_\mu)$. 

The (extended) generating functional $\tilde{W}(J, \beta_{\mu \nu} ,  \tau_\mu) $ as well as $\tilde{S}_{tot}$ are invariant under the BRS transformations ($\ref{BRST1}$)- ($\ref{BRST3}$), whereas the source term $S_J$ is \textit{not}. This leads to a condition for $\tilde{W}(J, \beta_{\mu \nu} ,  \tau_\mu) $, the Slavnov-Taylor-identities (STI). Introducing the BRS operator
\begin{eqnarray}
\mathcal{D}:= \int d^4x \left( t_{\mu \nu} \frac{\delta }{\delta \beta_{\mu \nu}}  + \overline{\sigma}_\mu \frac{\delta }{\delta \tau_{\mu}} + \xi^{-1} \sigma^\mu F_{\mu \rho \sigma} \Big( \frac{\delta }{\delta t_{\rho \sigma}} \Big) \right)  \label{BRSO}
\end{eqnarray}
they can be compactly summarized as
\begin{eqnarray}
\mathcal{D}  \tilde{W}(J, \beta_{\mu \nu} ,  \tau_\mu) = 0.  \label{STW}
\end{eqnarray}
The STI ($\ref{STW}$) take the same form when they are written down in terms of the functional $\tilde{Z}(J, \beta_{\mu \nu} , \tau_\mu)$,
\begin{eqnarray}
\mathcal{D}  \tilde{Z}(J, \beta_{\mu \nu} ,  \tau_\mu) = 0.  \label{STZ}
\end{eqnarray}
For completeness sake, let us also define classical fields
\begin{eqnarray}
\overline{h}_{ \mu \nu}=\frac{\delta \tilde{Z}}{\delta t^{\mu \nu}} \\ \xi^\mu=  \frac{\delta \tilde{Z}}{\delta \overline{\sigma}_\mu} \\ \overline{\xi}_\mu=  \frac{\delta \tilde{Z}}{\delta {\sigma}^\mu} 
\end{eqnarray}
and denote
\begin{equation}
\overline{\Phi}= \{\overline{h}_{ \mu \nu}, \xi^{\mu}, \overline{\xi}_\mu  \}.
\end{equation}
The Legendre transform $\Gamma(\overline{\Phi},\beta_{\mu \nu} , \tau_\mu)$ of $\tilde{Z}(J, \beta_{\mu \nu} , \tau_\mu)$,
\begin{eqnarray}
\Gamma(\overline{\Phi},\beta_{\mu \nu} , \tau_\mu)=\int d^4x \left(t_{\mu \nu} \overline{h}^{\mu \nu} + \overline{\sigma}_\mu \xi^{\mu} +  \sigma^\mu \overline{\xi}_\mu  \right) - \tilde{Z}(J, \beta_{\mu \nu} , \tau_\mu) ,
\end{eqnarray}
generates 1PI correlation functions. The Slavnov-Taylor-identities can be formulated in terms of $\Gamma(\overline{\Phi},\beta_{\mu \nu} , \tau_\mu)$:
\begin{eqnarray}
\int d^4x \left( \frac{\delta \Gamma}{\delta \beta^{\mu \nu}} \frac{\delta \Gamma}{\delta \overline{h}_{\mu \nu}}   + \frac{\delta \Gamma}{\delta \tau_{\mu}} \frac{\delta \Gamma}{\delta \xi^{\mu}}+ \xi^{-1} F_{\mu \rho \sigma} \big( {\overline{h}}^{ \rho \sigma} \big)  \frac{\delta \Gamma}{\delta \overline{\xi}_\mu } \right) &=& 0  .   \label{STI}
\end{eqnarray}
Again,  ($\ref{STI}$) are in complete analogy to the Slavnov-Taylor identities of Yang-Mills theory.

\end{subsection}

\begin{subsection}{The graviton and ghost propagators}  \label{propsec}
In order to find the graviton propagator in harmonic gauge, we have to include the gauge fixing terms ($\ref{GFL}$) into the Lagrangian bilinear in the quantum fields, $\mathcal{L}^{(2)}$. We wish to obtain a free-field Langrangian of the form
\begin{equation}
\mathcal{L}^{f} = - \frac{1}{2} h^{\mu \nu} \Delta^{-1}_{\mu \nu \rho \sigma} h^{\rho \sigma}   \label{Lf}
\end{equation}
where $\Delta_{\mu \nu \rho \sigma}$ is the graviton propagator. From eqns. ($\ref{l2}$), ($\ref{KKexp}$) and ($\ref{GFL}$) follows that
\begin{eqnarray}
\mathcal{L}^{(2)}+ \mathcal{L}_{GF} &=& -{\Lambda_K} \left(  \frac{1}{2} h^2  - h^\mu_{\ \rho} h^\rho_{\ \mu}  \right) -  \frac{1}{2 \xi} \partial_\mu h^{\mu \rho} \partial^\nu h_{\nu \rho}\nonumber \\ &&  - \frac{1}{2} \partial_\mu h_{\nu \rho } \partial^\mu h^{\nu \rho} + \frac{1}{4} \partial_\mu h \partial^\mu h + \partial_\mu h^{\mu \rho} \partial^\nu h_{\nu \rho} .
\end{eqnarray}
If we choose the parameter $\xi=\frac{1}{2}$ and use partial integration the above expression simplifies to
\begin{eqnarray}
\mathcal{L}^{(2)}+ \mathcal{L}_{GF} &=& {\Lambda_K} \left( h_{\mu \rho} h^{\mu \rho} - \frac{1}{2} h^2    \right) + \frac{1}{2} \left( h_{\nu \rho } \partial_\mu  \partial^\mu h^{\nu \rho} - \frac{1}{2}  h \partial_\mu \partial^\mu h  \right) . \nonumber \\
\end{eqnarray}
By comparison with eq.  ($\ref{Lf}$) we conclude that
\begin{eqnarray}
\Delta^{-1}_{\mu \nu \rho \sigma} &=&  -\frac{1}{2} \left( \delta_{\mu \rho} \delta_{\nu \sigma}+ \delta_{\mu \sigma} \delta_{\nu \rho}    - \delta_{\mu \nu} \delta_{\rho \sigma}       \right) (\partial_\alpha \partial^\alpha + 2\Lambda_K)\nonumber \\ & \equiv & -\mathcal{D}_{\mu  \nu \rho \sigma}  (\partial_\alpha \partial^\alpha + 2\Lambda_K ). \label{invprop}
\end{eqnarray}
To obtain the graviton propagator, $\Delta^{-1}_{\mu \nu \rho \sigma} $ has to be inverted. The operator $\mathcal{D}_{\mu  \nu \rho \sigma}$ is symmetric in $\mu, \nu$ and $\rho, \sigma$ and under the combined operation  $\mu \leftrightarrow \rho, \ \nu \leftrightarrow \sigma$. By direct calculation one finds that
\begin{eqnarray}
\mathcal{D}^{\alpha \beta \mu  \nu } \mathcal{D}_{\alpha \beta  \rho \sigma } = \frac{1}{2} \left(\delta^\mu_\rho \delta^\nu_\sigma + \delta^\mu_\sigma \delta^\nu_\rho    \right)   \label{one}
\end{eqnarray}
where again $\mathcal{D}^{\alpha \beta \mu  \nu }= \frac{1}{2} \left( \delta^{\alpha \mu} \delta^{\beta \nu}+ \delta^{\alpha \nu} \delta^{\beta \mu}    - \delta^{\alpha \beta} \delta^{\mu \nu} \right)$. Thus we conclude that the propagator for gravitons on a flat Euclidean background is
\begin{eqnarray}
\Delta_{\mu \nu \rho \sigma} (x-y) &=&  \int \frac{d^4k}{(2 \pi)^4}  e^{ik(x-y)} \Delta_{\mu \nu \rho \sigma} (k^2)   \label{gravpropPS}
\end{eqnarray}
where we have employed the Fourier transform
\begin{eqnarray}
\Delta_{\mu \nu \rho \sigma} (k^2) = \frac{1}{2} \frac{\delta_{\mu \rho} \delta_{\nu \sigma}+ \delta_{\mu \sigma} \delta_{\nu \rho}   -   \delta_{\mu \nu} \delta_{\rho \sigma}}{k^2 - 2 \Lambda_K}  . \label{gravprop}
\end{eqnarray}
The expression ($\ref{gravprop}$) deserves some comments. It is known from the literature \cite{zee} that in flat Euclidean space, the propagator for a massive spin 2 particle with five physical degrees of freedom is given by
\begin{equation}
\Delta_{\mu \nu \rho \sigma}^{Spin 2} (k^2)= \frac{1}{2} \frac{G_{\mu \rho} G_{\nu \sigma} + G_{\mu \sigma} G_{\nu \rho}  -   \frac{2}{3} G_{\mu \nu} G_{\rho \sigma}}{k^2+ m^2} \label{s2mprop}
\end{equation}
where $ G_{\mu \nu} = \delta_{\mu \nu} - (k_\mu k_\nu )/m^2$. As is explained in Appendix ($\ref{Conv}$), to $0th$ order in $\lambda$ physical sources $t^{\mu \nu}$ satisfy 
\begin{equation}
k_\mu t^{\mu \nu}=0.
\end{equation}
Thus, if we consider the one-particle exchange amplitude between two physical sources\footnote{For real sources $t^{\mu \nu}(k)$ we just have $t^{\mu \nu}{'}(k)=t^{\mu \nu}(-k)$.} $t^{\mu \nu}$ and $t^{\mu \nu}{'}$ for the massive spin 2 state described by ($\ref{s2mprop}$) we arrive at
\begin{eqnarray}
 \frac{ t_{\mu \nu} t^{\mu \nu}{'}-  \frac{1}{3} t t' }{k^2 + m^2}
\end{eqnarray}
 with $t= t^\mu_{\ \mu}$. On the other hand, the one-particle exchange amplitude derived with the graviton propagator ($\ref{gravprop}$) reads
\begin{eqnarray}
 \frac{ t_{\mu \nu} t^{\mu \nu}{'}-  \frac{1}{3} t t' }{k^2 - 2 \Lambda_K} -  \frac{1}{6} \frac{t t' }{k^2 - 2\Lambda_K} . \label{ex}
\end{eqnarray}
By comparison it follows immediately that for $\Lambda_K < 0$, the first term in ($\ref{ex}$) corresponds to the exchange of a massive spin 2 particle. However, the second term has the ''wrong sign'' and must be associated with a repulsive interaction due to a massive spin 0 ghost. This result is in agreement with \cite{Veltman}. For $\Lambda_K > 0$ things look actually worse: in this case one would naively conclude that ($\ref{ex}$) describes the exchange amplitude of a massive spin 2 particle and a massive spin 0 ghost with ''masses'' $i \sqrt{\Lambda_K}$.   

However, for $\Lambda_K=0$ there is no problem because then ($\ref{gravprop}$) is just the well-known propagator for a massless spin 2 field with two degrees of freedom. Note that there remains a difference between the graviton propagator ($\ref{gravprop}$) for $\Lambda_K=0$ and the massive spin 2 particle propagator ($\ref{s2mprop}$) in the limit $m \rightarrow 0$ even if they are coupled to physical sources $k_\mu t^{\mu \nu}=0$. This is known as the van Dam-Veltman-Zakharov (vDVZ) discontinuity \cite{Damvelt} \cite{Zak} and is ultimatively due to the extra degrees of freedom of the massive spin 2 state as compared to the massless case.

For $\Lambda_K \ne 0$ the appearance of the ghost and the imaginary ''masses''  $i \sqrt{\Lambda_K}$ respectively disqualify the \textit{linearized} theory described by the action
\begin{eqnarray}
S_{EH}^L=\int d^4 x \left( \mathcal{L}^{(1)} + \mathcal{L}^{(2)} +  \mathcal{L}_{GF}  \right)  \label{EHL}
\end{eqnarray}
 as a consistent quantum theory of gravity because negative probabilities, gravitons travelling faster than the speed of light and other problematic behaviour will occur. However, this can be seen as an artefact of our expansion around flat space that is \textit{not} a solution of the field equations for $\Lambda_K \ne 0$ and thus does not represent the ground state of the theory. In case of $\Lambda_K \ne 0$, the solutions of the field equations that correspond to the flat space solution for $\Lambda_K=0$ (i.e. that have the maximal number of symmetries, 10 Killing vectors) are either de Sitter ($\Lambda_K > 0$) or anti de Sitter ($\Lambda_K < 0$). See Appendix ($\ref{Conv}$) for some more details. Indeed it can be shown \cite{Novello} that in an expansion around (anti) de Sitter space, the spectrum of the theory with nonvanishing cosmological constant does not suffer from ghosts or other pathologies and the gravitational field involves only two physical degrees of freedom.

In this work, we will not attempt to enter a deeper discussion involving quantum field theory in curved spacetimes. However, it will turn out in the next chapter that the case $\Lambda_K \ne 0$ has some interesting implications for the renormalization group flow of gravity. Therefore, we propose the following pragmatic point of view. It can be argued \cite{Gaba} that a nonlinear completion\footnote{It is known \cite{Vain} that an analysis of this kind can cure the the vDVZ discontinuity problem.} of the action ($\ref{EHL}$) such that ultimatively, the full Einstein Hilbert action with cosmological constant ($\ref{EH}$) is restored,
\begin{eqnarray}
S_{EH}= S_{EH}^L + V(h),
\end{eqnarray}
 should eliminate the ghost, the longitudinal polarisations of a massive graviton and the wrong sign ''mass'' terms respectively. This is due to the nonlinear interactions in $V(h)$ which are supposed to allow for a background rearrangement such that the signature of the kinetic term of the ghost changes and a shift to the correct ground state of the theory occurs. Then one effectively ends up with a QFT in curved spacetime (de Sitter or anti de Sitter).  
     
Thus, we will \textit{formally} treat quantum gravity with (small) nonvanishing $\Lambda_K$ as a QFT in flat Euclidean space, keeping in mind that a more accurate treatment would amount to doing it in (anti) de Sitter space. Note that since we will work with an (effective) IR momentum space cutoff $\Lambda$, the occurrence of the wrong sign ''mass'' term in the propagator ($\ref{gravprop}$) for $\Lambda_K> 0$ will not pose any problems from a technical point of view as long as we keep $\Lambda^2> \Lambda_K$.

\medskip
Finally, it remains to derive the propagator $\Delta_{GH \mu \nu}$ for the ghost fields $\overline{C}^\nu$, $C_\mu$. This one can be extracted out of the ghost term ($\ref{gh}$). Employing eq. ($\ref{fp}$) we find
\begin{equation}
S_{GH}(h, C, \overline{C}) = \int d^4 x  \overline{C}_\mu \delta^\mu_{\ \nu} \partial^\rho \partial_\rho {C}^\nu + \mathcal{O}(\lambda)  \label{propghPS0}
\end{equation}
and thus
\begin{equation}
\Delta_{GH \mu \nu}(x-y)=  \int \frac{d^4 k}{(2 \pi)^4}  e^{ik(x-y)} \Delta_{GH \mu \nu}(k^2)  \label{propghPS}
\end{equation}
where
\begin{equation}
\Delta_{GH \mu \nu}(k^2)= \frac{\delta_{\mu \nu}}{k^2} .  \label{propgh}
\end{equation}

\end{subsection}

\end{section}

\begin{section}{Cutoff regularization and Polchinski's equation}

\begin{subsection}{UV cutoff regularization of gravity and violation of the gauge invariance}  \label{CutRegG}
We begin by decomposing the total action ($\ref{Stot}$) of Einstein gravity into the graviton and ghost kinetic terms, a term linear in the field steming from the cosmological constant and an interaction Lagrangian $\mathcal{L}_{int}^{EH} \big(h,C,\overline{C}\big)$:
\begin{eqnarray}
{S}_{tot} &=& \int_x - 2 \frac{\Lambda_K}{\lambda}  h   - \frac{1}{2} \langle h^{\mu \nu}, \Delta^{-1}_{\mu \nu \rho \sigma} h^{\rho \sigma} \rangle  -  \langle \overline{C}^\mu, \Delta_{GH \mu \nu}^{-1} {C}^\nu    \rangle  + \int_x  \mathcal{L}_{int}^{EH} \big(h,C,\overline{C}\big) . \nonumber \\  \label{Stotde}
\end{eqnarray}
In addition, we introduce an extended interaction Lagrangian $\tilde{\mathcal{L}}_{int}^{EH}(h,C,\overline{C}, \beta, \tau)$ involving the composite field operators $ \mathcal{L}_C \tilde{g}^{\mu \nu}$ and $ C^\nu \partial_\nu C^\mu  $ by 
\begin{eqnarray}
\tilde{\mathcal{L}}_{int}^{EH}(h,C,\overline{C}, \beta, \tau) := \mathcal{L}_{int}^{EH}(h,C,\overline{C}) +  \beta_{\mu \nu} \mathcal{L}_C \tilde{g}^{\mu \nu}  + \lambda \tau_\mu C^\nu \partial_\nu C^\mu  .
\end{eqnarray}
Note that from the definition ($\ref{StotE}$) of the extended total action follows that
\begin{eqnarray}
\tilde{S}_{tot} &=& \int_x - 2 \frac{\Lambda_K}{\lambda}  h   - \frac{1}{2} \langle h^{\mu \nu}, \Delta^{-1}_{\mu \nu \rho \sigma} h^{\rho \sigma} \rangle  -  \langle \overline{C}^\mu, \Delta_{GH \mu \nu}^{-1} {C}^\nu    \rangle  + \int_x  \tilde{\mathcal{L}}_{int}^{EH} \big(h,C,\overline{C}, \beta, \tau\big) . \nonumber \\  \label{StotEde}
\end{eqnarray}
In order to be able to write down a renormalization group equation for Euclidean quantum gravity, we have to employ a cutoff regularization. Let $\Lambda$ be some scale. We introduce a cutoff function\footnote{One can think of $K(-\partial^2/\Lambda^2)$ as $K(k^2/\Lambda^2)$ where $k^2$ are the eigenvalues of the operator $-\partial^2$. } $K(-\partial^2/\Lambda^2)$ which has the properties ($\ref{cutf}$):
\begin{eqnarray}
K(z) = \left\{ \begin{array}{c c c}   1 & ,&  0 \le z \le 1 \\ \mbox{smooth} & ,& 1 < z < 4   \\  0 & ,& 4 \le z . \end{array} \right. \nonumber
\end{eqnarray}
The UV regularization is done by multiplying the graviton propagator ($\ref{gravpropPS}$)  with the cutoff function,
\begin{eqnarray}
\Delta_{\mu \nu \rho \sigma} (x-y) &\rightarrow& \Delta^{\Lambda}_{\mu \nu \rho \sigma}(x-y) := K(-\partial^2/\Lambda^2) \Delta_{\mu \nu \rho \sigma}(x-y) . \label{propgravreg}
\end{eqnarray}
Similarly, we regularize the propagator ($\ref{propghPS}$) for the ghosts:
\begin{eqnarray}
\Delta_{\mu \nu GH} (x-y) &\rightarrow& \Delta^{\Lambda}_{\mu \nu GH}(x-y) := K(-\partial^2/\Lambda^2) \Delta_{\mu \nu GH} (x-y)   . \label{propghreg}
\end{eqnarray}
At this point, some important comments are in order. A momentum cutoff regularization inevitably violates the local gauge invariance of any gauge theory. To see this, consider a homogeneous gauge transformation of some field $\phi(x)$,
\begin{eqnarray}
\phi(x) \rightarrow \Omega(x) \phi(x).
\end{eqnarray}
In momentum space the gauge transformed field is given by a convolution of the field with the gauge transformation,
\begin{eqnarray}
\int d^4k \  \Omega(p-k) \phi(k),
\end{eqnarray}
and consequently any division of momenta is lost. On the quantum level, the gauge symmetry leads to BRS invariance of the total action and ultimatively to the Slavnov-Taylor-Identities for the generating functional. Hence, if gauge invariance is destroyed by the regulator, so will be the BRS invariance and we will end up with violated Slavnov-Taylor-identities (vSTI).

This fact poses a serious obstacle to any kind of analysis of a gauge theory involving RG flow equations. Moreover, since the STI are generated by nonlinear BRS transformations of the fields, additional complications arise because composite field operators will have to be renormalized, too. 

In the case of the perturbative renormalization of Yang-Mills (YM) theory via flow equations, it has been shown \cite{KM} that these problemes can be surmounted by the following procedure which has as its ultimate aim the restoration of the STI (and therefore the gauge invariance) of the theory in the limit $\Lambda_0 \rightarrow \infty$, i.e. when the UV cutoff is taken away:
\begin{enumerate}
\item As a first step, one disregards the violation of the STI and establishes a finite UV behaviour of the vertex functions of an effective potential $L$ for the theory at hand. Here, $L$ is defined as in eq. ($\ref{Spol}$). Since local gauge invariance is destroyed in the regularized theory, the RG flow will generally produce contributions to all those operators that are not forbidden by the unbroken global symmetries, e.g. $O(4)$ invariance. Thus, one has to add as counterterms to $L$ not only terms corresponding to the ''classical'' interaction part of the theory, but also all local renormalizable operators allowed by the unbroken global symmetries. In general, the number of such terms is much higher than the actual number of free parameters in the theory\footnote{In the case of spontaneously broken $SU(2)$ Yang Mills theory, there are 37 such counterterms needed, whereas only 9 physical renormalization conditions exist \cite{Mull}.}. 

Thus for an arbitrary set of renormalization conditions, a \textit{family} of finite theories is established. The procedure up to now is essentially the same as the renormalization of a scalar field theory described in chapter ($\ref{RenFlow}$).

\item The second step consists of finding one particular choice for the set of renormalization conditions such that in the limit $\Lambda_0 \rightarrow \infty$ the theory satisfies the STI, i.e. gauge invariance is restored. This so-called \textit{fine tuning procedure} is rather involved and, roughly speaking, goes as follows \cite{KM}. 

It turns out that the violation of the STI can be expressed through vertex functions carrying a space-time integrated operator insertion which depends on the UV cutoff $\Lambda_0$. These vertex functions will vanish in the limit $\Lambda_0 \rightarrow \infty$ if the set of renormalization conditions can be chosen such that the inserted vertex functions have vanishing renormalization conditions for all of their relevant local parts. However, it turns out that there can be more relevant parts of these insertions\footnote{Again for the example of spontaneously broken $SU(2)$ Yang Mills theory, the insertion has canonical dimension 5 and there are 53 relevant parts of it \cite{Mull}.} than renormalization conditions describing the family of finite theories mentioned above. In this case it has also to be shown that there are linear interdependences between the relevant parts of the insertions. This can be achieved by relating the BRS variation of the bare action (which can be worked out explicitly) with the vSTI for the vertex fuctions by means of the flow equations. 

Clearly, the fine-tuning procedure involves the renormalization of composite field operators, namley the BRS variations of the fields and the space-time integrated operator insertion describing the vSTI. In the context of renormalization via flow equations, this problem has been studied in \cite{KK3}, \cite{KK4}. 

To summarize this second step once more from a slightly different point of view, one could say that the STI \textit{determine} all arbitrarily chosen renormalization conditions of step one but the ''real'' free parameters of the gauge theory. 
\end{enumerate}
Now let us come back to our analysis of quantum gravity as an effective field theory by means of the RG flow equations. Since the cutoff regularization ($\ref{propgravreg}$) violates the invariance of the Einstein-Hilbert action ($\ref{gaexp}$) under the local ''gauge'' transformations ($\ref{g1}$), the BRS invariance ($\ref{BRSTE}$) of the (extended) total gravity action will be broken and as a result, we will observe a violation of the Slavnov-Taylor-Identities  ($\ref{STW}$) and  ($\ref{STI}$) respectively.

Thus, in order to arrive at a theory that yields sensible physical predictions, one would expect that some kind of fine-tuning procedure as described above for the case of Yang-Mills theory is needed also for quantum gravity in order to restore gauge invariance and hence the STI in the end. However, the situation for gravity is somewhat different because contrary to the Yang-Mills case, the theory is perturbatively nonrenormalizable (at least without cosmological constant). Thus, it will in general not make sense to discuss the restoration of the STI in the limit $\Lambda_0 \rightarrow \infty$, i.e. when the UV cutoff is taken away. We will come back to this point in the next chapter and propose a fine-tuning procedure adapted to the case of quantum gravity.

For the time being, we will restrain ourselves to the implications of step one of our discussion on the previous page for quantum gravity. This means that the cutoff regularization ($\ref{propgravreg}$) and the resulting violation of the gauge/BRS symmetry ($\ref{BRSTE}$) enforces the introduction of counterterms for \textit{all} local renormalizable operators still allowed by the unbroken global symmetries. In our case, this is just the Euclidean $O(4)$ invariance. Hence, a bare total gravity action ${S}_{tot}({\Lambda_0})$ will be defined in the following way:
\begin{eqnarray}
{S}_{tot}(\Lambda_0) &=& \int_x A  \frac{(2 \pi)^4}{2}  h   - \frac{1}{2} \langle h^{\mu \nu}, \Delta^{\Lambda_0 \ -1}_{\mu \nu \rho \sigma} h^{\rho \sigma} \rangle  - \langle \overline{C}^\mu, \Delta_{GH \mu \nu}^{\Lambda_0 \ -1} {C}^\nu    \rangle  + {L} \big(h,C,\overline{C}, \Lambda_0 \big) . \nonumber \\  \label{StotdeReg}
\end{eqnarray}
Note that we included an arbitrary renormalization constant\footnote{The factor ${(2 \pi)^4}/{2}$ is chosen in order to obtain a nice RGE in the next section. } $A$ for the term linear in the field $h$. Furthermore, it is understood that the cosmological ''mass term''  $-2 \Lambda_K$ appearing in the graviton propagator ($\ref{gravprop}$) is replaced by another arbitrary renormalization constant $B_1$, and that the ghost propagator becomes equipped with a ''mass'' squared $B_2$.  See eqns. ($\ref{gravpropG}$) and ($\ref{ghpropG}$) for the new momentum space propagators. Finally, we introduced the bare interaction term
\begin{eqnarray}
{L} \big(h,C,\overline{C}, \Lambda_0 \big)  := \int_x  \mathcal{L}_{int}^{EH} \big(h,C,\overline{C}\big) + L_{C.T.}\big(h,C,\overline{C},\Lambda_0     \big).  \label{LGrav}
\end{eqnarray}
In eq. ($\ref{LGrav}$), $\mathcal{L}_{int}^{EH}$ means the ''classical'' interaction Lagrangian of quantum Einstein gravity introduced in eq. ($\ref{Stotde}$), whereas $L_{C.T.}\big(h,C,\overline{C},\Lambda_0  \big)$ contains all local counterterms necessary to cancel the upcoming divergences. Note that since $\mathcal{L}_{int}^{EH}$ involves (infinitely many) nonrenormalizable operators as will be explained in section ($\ref{GENR}$), in the standard perturbative  treatment we would also need counterterms corresponding to higher order field invariants (i.e. higher powers of the  curvature, such as $R^2$) that are not present in $\mathcal{L}_{int}^{EH}$ \cite{Hooft} \cite{Goro}. This would be the case even if a symmetry respecting regulator would be used, such as dimensional regularization. However, if the theory is treated as an effective field theory with flow equations along the lines developed in chapter ($\ref{EffFlow}$), only a finite number of counterterms will be needed depending on an improvement index $s$. This will become clear in the next chapter.

By virtue of the definition ($\ref{StotdeReg}$), we may now easily write down an UV regularized generating functional for quantum gravity:
\begin{eqnarray}
W(J; {\Lambda_0}) = \int \mathcal{D} h_{\mu \nu} \mathcal{D} C^\mu \mathcal{D} \overline{C}_\mu \ e^{ \left. S_{tot}({\Lambda_0}) + S_{J} \right.} . \label{GFQGReg}
\end{eqnarray}
Defining also an extended bare total gravity action $\tilde{S}_{tot}({\Lambda_0})$ via
\begin{eqnarray}
 \tilde{S}_{tot}(\Lambda_0) &=& \int_x A  \frac{(2 \pi)^4}{2}  h   -  \frac{1}{2} \langle h^{\mu \nu}, \Delta^{\Lambda_0 \ -1}_{\mu \nu \rho \sigma} h^{\rho \sigma} \rangle  -  \langle \overline{C}^\mu, \Delta_{GH \mu \nu}^{\Lambda_0 \ -1} {C}^\nu    \rangle  + \tilde{L} \big(h,C,\overline{C} , \beta, \tau, \Lambda_0 \big)  \nonumber \\  \label{StotdeRegE}
\end{eqnarray}
where
\begin{eqnarray}
\tilde{L} \big(h,C,\overline{C}, \beta, \tau, \Lambda_0 \big)  := \int_x \tilde{\mathcal{L}}_{int}^{EH}(h,C,\overline{C}, \beta, \tau) + \tilde{L}_{C.T.}\big(h,C,\overline{C}, \beta, \tau,    \Lambda_0  \big),  \label{LGravE}
\end{eqnarray}
we arrive at the extended regularized functional
\begin{eqnarray}
W(J,\beta_{\mu \nu},\tau_\mu; {\Lambda_0} ) = \int \mathcal{D} h_{\mu \nu} \mathcal{D} C^\mu \mathcal{D} \overline{C}_\mu \ e^{ \left. \tilde{S}^{\Lambda_0}_{tot} + S_{J} \right.} . \label{GFQGEReg}
\end{eqnarray}
The latter can be used as a starting point for deducing violated Slavnov-Taylor-Identities. Note that counterterms for the nonlinear BRS variations have been included in ($\ref{LGravE}$).

\end{subsection}

\begin{subsection}{Polchinski's equation for Euclidean quantum gravity}  \label{PolGravSec}
In this section, we will work in momentum space. In the spirit of the effective action ($\ref{SpolM}$) and our definition ($\ref{StotdeReg}$) of a bare gravity action,  we introduce an effective total action $S_{tot}({\Lambda})$ for quantum gravity referring to some scale $\Lambda$. To do so, we employ regularized momentum space graviton and ghost propagators,
\begin{eqnarray}
\Delta^{\Lambda}_{\mu \nu \rho \sigma}(k^2) &=& \frac{1}{2} K(k^2/\Lambda^2) \ \frac{\delta_{\mu \rho} \delta_{\nu \sigma}+ \delta_{\mu \sigma} \delta_{\nu \rho}    -\delta_{\mu \nu} \delta_{\rho \sigma}}{k^2 + B_1}  \label{gravpropG} \\
\Delta^{\Lambda}_{GH \mu \nu} (k^2) &=& - K(k^2/\Lambda^2)  \ \frac{\delta_{\mu \nu}}{k^2+ B_2}.   \label{ghpropG} 
\end{eqnarray}
Note that two arbitrary renormalization constants $B_1$ and $B_2$ (''mass'' squares) have been included in  ($\ref{gravpropG}$)  and ($\ref{ghpropG}$). Our effective action for quantum gravity takes the following form: 
\begin{eqnarray} 
S_{tot}(\Lambda) &=&  \int \frac{d^4k}{(2 \pi)^4} \bigg(  A  \frac{(2 \pi)^4}{2} \ \delta(k) h(k)    -  \frac{1}{2} C_1 h^{\mu \nu}(k) \left(\Delta^{\Lambda}_{\mu \nu \rho \sigma}(k^2)    \right)^{-1} h^{\rho \sigma}(-k)  \nonumber  \\ && \ \ \ \ \ \ \ \  - C_2 \ \overline{C}^\mu(k) \left(\Delta^{\Lambda}_{GH \mu \nu}(k^2)\right)^{-1} {C}^\nu(-k) \bigg) +  L \Big(h, C, \overline{C}; \Lambda \Big)  \label{SGraveff}
\end{eqnarray}
where $L(h, C, \overline{C}; \Lambda)$ is a not necessarily local interaction term depending on the fields $h_{\mu \nu}, C^\mu,  \overline{C}_\mu$ and the scale $\Lambda$. We will refer to it as ''effective potential''.

The action ($\ref{SGraveff}$) gives rise to an Euclidean quantum field theory described by the generating functional
\begin{eqnarray}
W(J) = \int \mathcal{D} h_{\mu \nu} \mathcal{D} C^\mu \mathcal{D} \overline{C}_\mu \ e^{ S_{tot}(\Lambda) + S_{J} } .  \label{WgravR}
\end{eqnarray}
The dependence of $L(h, C, \overline{C}; \Lambda)$ on the scale $\Lambda$ is given by the Polchinksi RGE for quantum gravity, which will be deduced in the following theorem.
\begin{satz}[Polchinski's equation for Euclidean quantum gravity]  \label{PolQG}
Let $\Lambda$ be some scale, and assume\footnote{This is not an essential ingredient. See Appendix ($\ref{LZ}$) for details.} that
\begin{equation}
J(k)=0 \ \ \text{for} \ \ k^2>\Lambda^2.  \label{noJ}
\end{equation}
Under a change of $\Lambda$ the generating functional $W(J)$ defined in eq. ($\ref{WgravR}$) remains unchanged,
\begin{eqnarray}
\Lambda \frac{d}{d \Lambda} W(J) =0,
\end{eqnarray}
if the effective potential $L(h, C, \overline{C}; \Lambda)$ satisfies
\begin{eqnarray}
-\Lambda \frac{d}{d \Lambda} L &=& \int d^4k \frac{(2 \pi)^4}{2 } \Lambda \frac{d}{d \Lambda} \Delta^{\Lambda}_{\mu \nu \rho \sigma} 
\left( \frac{\delta L}{\delta h_{\mu \nu}} \frac{\delta L}{\delta h_{\rho \sigma}} + \frac{\delta^2 L}{\delta h_{\mu \nu} \delta h_{\rho \sigma}  }  + A  \delta (k) \delta^{\mu \nu }  \frac{\delta L}{\delta h_{\rho \sigma}}    \right) \nonumber \\ && \ \ \ + \ \int d^4k {(2 \pi)^4} \Lambda \frac{d}{d \Lambda} \Delta^{\Lambda}_{GH \mu \nu} 
\left( \frac{\delta L}{\delta \overline{C}_{\mu}} \frac{\delta L}{\delta {C}_{\nu}} + \frac{\delta^2 L}{\delta \overline{C}_{\mu} \delta {C}_{\nu} }    \right). \label{polgrav}
\end{eqnarray}

\end{satz}

\begin{proof}
Differentiation of $W(J)$ with respect to $\Lambda$ yields
\begin{eqnarray}
\Lambda \frac{d}{d \Lambda} W &=& \int \mathcal{D} h_{\mu \nu} \mathcal{D} C^\mu \mathcal{D} \overline{C}_\mu  \left[ \int \frac{d^4k}{(2 \pi)^4} \left( -\frac{1}{{2}}h^{\mu \nu} \Lambda \frac{d}{d \Lambda} \left( \Delta^{\Lambda}_{\mu \nu \rho \sigma} \right)^{-1} h^{\rho \sigma} \right. \right.  \nonumber \\ & &  \ \ \ \left. - \ \overline{C}^\mu \Lambda \frac{d}{d \Lambda} \left( \Delta^{\Lambda}_{GH \mu \nu} \right)^{-1} {C}^\nu  \right) +  \left. \Lambda \frac{d}{d \Lambda} L(h, C, \overline{C}; \Lambda) \right] e^{ S_{tot}(\Lambda) + S_{J}} .  \nonumber \\  \label{DiffW}
\end{eqnarray}
Due to eq. ($\ref{noJ}$), $J(k)$ has no overlap with $\frac{d K}{d \Lambda}$ for a cutoff function which has the properties ($\ref{cutf}$). Moreover, we have\footnote{See also eq. ($\ref{one}$).}
\begin{eqnarray}
\big(\Delta_{\mu \nu \alpha \beta}^\Lambda(k^2)\big)^{-1} \Delta^{\alpha \beta \rho \sigma}(k^2) &=& \frac{1}{2} \left(\delta_\mu^\rho \delta_\nu^\sigma + \delta_\mu^\sigma \delta_\nu^\rho    \right) K^{-1} \\
\big(\Delta_{GH \mu \alpha}^\Lambda(k^2)\big)^{-1} \Delta^{\alpha \nu}_{GH}(k^2) &=& \delta_\mu^\nu K^{-1}.
\end{eqnarray}
Using $ K^{-2} \Lambda \frac{d K}{d \Lambda} =- \Lambda \frac{d }{d \Lambda} K^{-1}$ and neglecting field-independent terms which only change $W(J; \Lambda)$ by an overall factor we therefore find that
\begin{eqnarray}
&& \hspace{-0.7cm} \int d^4k \ \Lambda \frac{d K}{d \Lambda} \int \mathcal{D} h_{\mu \nu} \frac{\delta }{\delta  h^{\mu \nu}} \left[ \left( h^{\mu \nu} K^{-1} + \frac{(2 \pi)^4}{2}  \Delta^{\mu \nu \rho \sigma} \frac{\delta }{\delta  h^{\rho \sigma}}   \right) e^{S_{tot}(\Lambda)} \right]  \nonumber \\ && \quad = \  \int \mathcal{D} h_{\mu \nu}  \int {d^4k} \ \frac{1}{2} \left[ \frac{1}{(2\pi)^4 }  h^{\mu \nu} \Lambda \frac{d }{d \Lambda} \left( \Delta^{\Lambda}_{\mu \nu \rho \sigma} \right)^{-1} h^{\rho \sigma} \right. \nonumber \\ && \qquad \ \ + \  (2\pi)^4 \left. \Lambda \frac{d}{d \Lambda} \Delta^{\Lambda}_{\mu \nu \rho \sigma} 
\left( \frac{\delta L}{\delta h_{\mu \nu}} \frac{\delta L}{\delta h_{\rho \sigma}} + \frac{\delta^2 L}{\delta h_{\mu \nu} \delta h_{\rho \sigma} } + A   \delta (k) \delta^{\mu \nu }  \frac{\delta L}{\delta h_{\rho \sigma}}   \right)  \right] e^{S_{tot}(\Lambda)}. \nonumber \\
\end{eqnarray}
Employing the identity $\{ \frac{\delta }{\delta C_{\mu}}, C_\mu \}= \{ \frac{\delta }{\delta \overline{C}_{\mu}}, \overline{C}_\mu \}=1$ and remembering that $\frac{\delta }{\delta C^{\mu}} \frac{\delta }{\delta \overline{C}^{\nu}} = -\frac{\delta }{\delta \overline{C}^{\nu}} \frac{\delta }{\delta C^{\mu}}$  we similarly have
\begin{eqnarray}
& & \hspace{-0.7cm} \int d^4k  \Lambda \frac{d K}{d \Lambda} \int \mathcal{D} C^\mu \mathcal{D} \overline{C}_\mu \left[ \frac{\delta }{\delta \overline{C}^{\mu}} \left( \overline{C}^{\mu} K^{-1} e^{S_{tot}(\Lambda)} \right) \right. \nonumber \\ && \left. \ \qquad \qquad \qquad \qquad \qquad \qquad + \ \frac{\delta }{\delta C^{\mu}} \left( \left( C^{\mu} K^{-1} + (2 \pi)^4 \Delta_{GH }^{\mu \nu}  \frac{\delta }{\delta \overline{C}^{\nu}} \right) e^{S_{tot}(\Lambda)} \right) \right]  \nonumber \\& & = \ - \int \mathcal{D} C^\mu \mathcal{D} \overline{C}_\mu \int d^4k \left[ \frac{1}{(2 \pi)^4} \overline{C}^{\mu} \Lambda \frac{d}{d \Lambda} \left(\Delta^{\Lambda}_{GH \mu \nu} \right)^{-1} {C}^\nu  \right. \nonumber \\ && \qquad \qquad \qquad \qquad \qquad +  \left. (2\pi)^4 \Lambda \frac{d}{d \Lambda} \Delta^{\Lambda}_{GH \mu \nu} 
\left( \frac{\delta L}{\delta \overline{C}_{\mu}} \frac{\delta L}{\delta {C}_{\nu}} + \frac{\delta^2 L}{\delta \overline{C}_{\mu} \delta {C}_{\nu} }    \right) \right] e^{S_{tot}(\Lambda)}.  \nonumber \\
\end{eqnarray}
Hence, if we choose in eq. ($\ref{DiffW}$) the quantity  $-\Lambda \frac{d}{d \Lambda} L$ as proposed by the Polchinski RGE ($\ref{polgrav}$) we obtain
\begin{eqnarray}
\hspace{-1em} \Lambda \frac{d}{d \Lambda} W &=& \int d^4k  \ \Lambda \frac{d K}{d \Lambda} \int \mathcal{D} h_{\mu \nu}  \mathcal{D} C^\mu \mathcal{D} \overline{C}_\mu \left[ \frac{\delta }{\delta  h^{\mu \nu}} \left( \left( h^{\mu \nu} K^{-1} + \frac{(2 \pi)^4}{2}  \Delta^{\mu \nu \rho \sigma} \frac{\delta }{\delta  h^{\rho \sigma}}   \right) e^{S_{tot}(\Lambda)} \right) \right. \nonumber \\ && + \left. \frac{\delta }{\delta \overline{C}^{\mu}} \left( \overline{C}^{\mu} K^{-1} e^{S_{tot}(\Lambda)} \right) + \frac{\delta }{\delta C^{\mu}} \left( \left( C^{\mu} K^{-1} + (2 \pi)^4 \Delta_{GH }^{\mu \nu}  \frac{\delta }{\delta \overline{C}^{\nu}} \right) e^{S_{tot}(\Lambda)} \right)  \right]   \nonumber \\ &=&0 .
\end{eqnarray}
\begin{flushright}
$\Box$
\end{flushright}
\end{proof}
We would like to point out that the derivation of the Polchinski equation does \textit{not} depend on translational invariance of the effective potential \cite{Gerhard}.

The Polchinski RGE ($\ref{polgrav}$) for Euclidean quantum gravity can also be written down in position space:
\begin{eqnarray}
&& \hspace{-0.7cm}  -\Lambda \frac{d}{d \Lambda} L \nonumber \\ && = \frac{1}{2 } \int_{xy} \Lambda \frac{d}{d \Lambda} \Delta^{\Lambda}_{\mu \nu \rho \sigma} 
\left( \frac{\delta L}{\delta h_{\mu \nu}(x)} \frac{\delta L}{\delta h_{\rho \sigma}(y)} + \frac{\delta^2 L}{\delta h_{\mu \nu}(x) \delta h_{\rho \sigma}(y) }  + A \delta(x-y) \delta^{\mu \nu }  \frac{\delta L}{\delta h_{\rho \sigma}(x)}    \right) \nonumber \\ && \qquad \qquad \ + \ \int_{xy} \Lambda \frac{d}{d \Lambda} \Delta^{\Lambda}_{GH \mu \nu} 
\left( \frac{\delta L}{\delta \overline{C}_{\mu}(x)} \frac{\delta L}{\delta {C}_{\nu}(y)} + \frac{\delta^2 L}{\delta \overline{C}_{\mu}(x) \delta {C}_{\nu}(y) }    \right) . \label{polgravps}
\end{eqnarray}
Please refer to \cite{Mull} for a general derivation of position-space RGEs.

In analogy to the definition of the effective action in eq. ($\ref{SGraveff}$), we may introduce an extended effective action $\tilde{S}_{tot}(\Lambda)$ by employing an extended effective potential\footnote{It is understood that at some bare scale $\Lambda_0$, the extended effective potential $\tilde{L}$ is given by eq. ($\ref{LGravE}$).} 
\begin{equation}
\tilde{L} \big(h,C,\overline{C}, \beta, \tau, \Lambda \big). 
\end{equation}
The latter involves again the couplings $\beta_{\mu \nu},\tau_\mu$ to the BRS composite fields ($\ref{BRST1}$) and ($\ref{BRST2}$) and leads to an extended generating functional $W(J,\beta_{\mu \nu},\tau_\mu)$. One can check\footnote{Remember that the derivation of the Polchinski equation does {not} depend on translational invariance of the effective potential.} that the dependence of $\tilde{L} \big(h,C,\overline{C}, \beta, \tau, \Lambda \big)$ on the scale is again given by the Polchinski RGE ($\ref{polgrav}$) if we perform the substitution
\begin{eqnarray}
L(h, C, \overline{C}; \Lambda) \rightarrow \tilde{L} \big(h,C,\overline{C}, \beta, \tau, \Lambda \big).
\end{eqnarray}
If we would find a way to solve the RGE ($\ref{polgrav}$) for the (extended) effective potential, employing the bare potential ($\ref{LGravE}$) as initial condition at $\Lambda_0$, we would obtain a trajectory 
\begin{equation}
[\Lambda, \Lambda_0] \rightarrow \tilde{L} \big(h,C,\overline{C}, \beta, \tau, \Lambda, \Lambda_0 \big).
\end{equation}
The solution $\tilde{L} \big(h,C,\overline{C}, \beta, \tau, \Lambda, \Lambda_0 \big)$ then leads via eq. ($\ref{WgravR}$) to a generating functional $W(J,\beta_{\mu \nu},\tau_\mu; \Lambda_0 )$ of some QFT\footnote{See Appendix ($\ref{LZ}$) for a cleaner relation between $\tilde{L} \big(h,C,\overline{C}, \beta, \tau, \Lambda, \Lambda_0 \big)$ and  $W(J,\beta_{\mu \nu},\tau_\mu; \Lambda_0 )$ which does not need source terms that satisy the condition ($\ref{noJ}$).}. If this QFT would satisfy the following conditions,
\begin{enumerate}
\item only a finite number of renormalization conditions at some renormalization scale $\Lambda_R$ has to be imposed in order to retain a finite solution $L$ in the limit $\Lambda_0 \rightarrow \infty$

\item the arbitrary renormalization conditions of the various couplings that arise due to the symmetry breaking cutoff regularization can be chosen in such a way that the generating functional $W(J,\beta_{\mu \nu},\tau_\mu; \Lambda_0 )$ satisfies the STI ($\ref{STW}$) of quantum gravity in the limit $\Lambda_0 \rightarrow \infty$

\item crucial physical requirements such as the unitarity of the quantum theory are satisfied

\item the observed coupling strenghts, i.e. of the cosmological constant $\Lambda_K$ and Newtons constant $G={\lambda^2}/{32 \pi}$, are reproduced correctly, 
\end{enumerate}
then we would presumably have found the quantum theory of gravitation. However, it is well-known that so far, all attempts to find such a solution (at least in perturbation theory) have failed in one or the other points stated above. We will discuss this in more detail in the next chapter, and we will investigate the problem from the viewpoint of renormalization via flow equations and the analogous treatment of effective field theories proposed in chapter ($\ref{EffFlow}$).

\end{subsection}

\end{section}

\end{chapter}

\begin{chapter}{Euclidean Quantum Gravity via Flow Equations}  \label{PredGrav}
A bare action for quantum gravity containing all field invariants that are permitted by general coordinate invariance is introduced. The relation to higher derivative gravity is discussed, and it is argued that in the effective field theory approach the known unitarity problems \cite{Stelle} will not appear. In the following, the methods that have been developed in chapter ($\ref{EffFlow}$) for investigating the predictivity of effective field theories are applied to effective quantum gravity. As a first step, we disregard the violation of the Slavnov-Taylor identities (STI) and establish bounds for the vertex functions of the gravity effective potential in generalized perturbation theory in the renormalized renormalizable and some of the bare nonrenormalizable couplings. It is shown that by introducing appropriate notations, we may proceed in close analogy to the case of the scalar field theory considered in chapters  ($\ref{RenFlow}$)  and  ($\ref{EffFlow}$).  A set of (for the time being) arbitrary renormalization and improvement conditions is imposed. By inverting the renormalization group trajectory, it is argued that the improvement conditions force the UV cutoff $\Lambda_0$ of effective quantum gravity to be the Planck scale $M_P$. Finally, we establish that the family of theories described by the arbitrary renormalization and improvement conditions is predictive at scales far below the Planck scale with finite accuracy. We then proceed to the restoration of the STI. Introducing bare regularized BRS variations, the violated Slavnov-Taylor identities (vSTI) for the extended effective potential are worked out at the value $\Lambda=0$ of the floating cutoff.  Bounds for the vertex functions carrying the nonlinear BRS variations as operator insertions are established, and we note that a crucial difference to the Yang-Mills case lies in the fact that the gravity BRS fields contain \textit{nonrenormalizable parts}. By imposing renormalization and improvement conditions for the BRS variations, it is proven that the dependence of these vertex functions on the bare initial conditions is suppressed at scales far below the Planck scale\footnote{This is similar to the statements concerning the predictivity of the effective theory.}. The violation of the STI can be described in terms of vertex functions carrying a space-time integrated operator insertion having canonical dimension $5$. It is therefore argued that the STI can be restored to \textit{finite accuracy} if {one particular} set of arbitrary renormalization and improvement conditions for the couplings and BRS variations can be determined such that the relevant and leading irrelevant parts of the vertex functions describing the violation of the STI are driven small at scales far below the Planck scale. Here, ''small'' means the order of accuracy to which the theory is predictive. In the last section of this chapter, we consider the no-cutoff limit $\Lambda_0 \rightarrow \infty$  of quantum gravity from the viewpoint of the analysis with flow equations $\grave{a}$ la Polchinski. The vertex functions of the gravity effective potential are expanded solely in the renormalizable couplings, and their boundedness and convergence is established in the limit $\Lambda_0 \rightarrow \infty$ while the STI are still violated. Applying the same program to the vertex functions carrying the nonlinear BRS variations as operator insertions, we observe that the nonrenormalizable parts of the gravity BRS fields will go away in the no-cutoff limit if smallness of the bare BRS couplings is imposed. It is, however, shown that if the latter constraint is dropped, convergence of the BRS vertex functions may still be proven. Proceeding with the restoration of the STI of quantum gravity, we argue that for zero renormalized cosmological constant $\Lambda_K=0$ the theory will become free as $\Lambda_0 \rightarrow \infty$, and that the latter statement is compatible with gauge invariance. It is speculated whether a \textit{nonzero} cosmological constant $\Lambda_K \ne 0$ might lead to a \textit{nonvanishing} value of the gravitational constant in the no-cutoff limit, and we point out that the gravitational coupling should then become determined by the cosmological constant. Finally, we observe that a similar effect might be obtained by coupling massive fields to gravity, leading to speculations if the gravitational constant is given in terms of the cosmological constant and the masses of the elementary particles as $\Lambda_0 \rightarrow \infty$. We ask whether this indicates a Higgs-gravity connection.

\begin{section}{The bare action for effective quantum gravity}
\begin{subsection}{Quantum Einstein gravity without a cosmological constant is perturbatively nonrenormalizable} \label{GENR}
As it is explained in Appendix ($\ref{D}$), the canoncial dimension of a field of some QFT follows from the large momentum behaviour of its dedicated propagator. Thus we can determine the canonical dimensions $D_h$, $D_C$ and $D_{\overline{C}}$ of the graviton and ghost fields $h_{\mu \nu}$, $C_\mu$ and $\overline{C}_\nu$ introduced in section ($\ref{BRSQ}$) by looking at their respective momentum space propagators ($\ref{gravprop}$) and ($\ref{propgh}$). We find that in $d=4$
\begin{equation}
D_h=D_C= D_{\overline{C}}= 1 . \label{mhd}
\end{equation}
Let us now come back to the expansion ($\ref{gaexp}$) of the Einstein-Hilbert action with a cosmological constant,
\begin{eqnarray}
S_{EH}= \int d^4 x \left( - \frac{4\Lambda_K}{\lambda^2} + \mathcal{L}^{(1)} + \mathcal{L}^{(2)} + \mathcal{L}^{(3)} + ...  \right) . \nonumber
\end{eqnarray}
In eqns. ($\ref{l1}$) and ($\ref{l2}$) we have already given the Lagrangians linear and bilinear in the gravitational field $h$, $\mathcal{L}^{(1)} $ and $\mathcal{L}^{(2)}$. We will now derive some terms of higher powers in $h$. From eq. ($\ref{KKexp}$) and Ref. \cite{CLM} we obtain
\begin{eqnarray}
\mathcal{L}^{(3)} &=& \lambda \Lambda_K  \left(-\frac{2}{3} h^\mu_{\ \nu} h^\nu_{\ \rho} h^\rho_{\ \mu} + \frac{1}{2} h h^\mu_{\ \rho} h^\rho_{\ \mu} - \frac{1}{12} h^3 \right)   - \lambda  h^{\mu \nu} \left( \frac{1}{2} \partial_\mu h_{\rho \sigma} \partial_\nu h^{\rho \sigma}  - \frac{1}{4} \partial_\mu h \partial_\nu h \right.  \nonumber \\ && \hspace{4cm} + \ \partial^\rho h_{\mu \sigma} \partial^\sigma h_{\rho \nu}  - \partial_\rho h_{\mu \sigma} \partial^\rho h^{\nu \sigma} +  \partial_\rho h_{\mu \nu} \partial^\rho h   \bigg) \label{l3} \\
\mathcal{L}^{(4)} &=& \lambda^2 \Lambda_K \left(-\frac{1}{96} h^4 + \frac{1}{8} h^2 h^\mu_{\ \rho} h^\rho_{\ \mu} - \frac{1}{3} h h^\mu_{\ \nu} h^\nu_{\ \rho} h^\rho_{\ \mu} \right. \nonumber \\&& \hspace{4cm} \left.- \ \frac{1}{8} (h^\mu_{\ \rho} h^\rho_{\ \mu})^2  +\frac{1}{2} h^\mu_{\ \nu} h^\nu_{\ \rho} h^\rho_{\ \sigma} h^{\sigma}_{\ \mu}   \right) + \ \mathcal{O}(\lambda^2) \label{l4} \\
 \mathcal{L}^{(5)} &=& \lambda^3 \Lambda_K \left( -\frac{2}{5} h^\mu_{\ \nu} h^\nu_{\ \rho} h^\rho_{\ \sigma} h^{\sigma}_{\ \lambda}  h^{\lambda}_{\ \mu} + \frac{1}{6} (h^\mu_{\ \rho} h^\rho_{\ \mu}) (h^\mu_{\ \nu} h^\nu_{\ \rho} h^\rho_{\ \mu}) + \frac{1}{4} h h^\mu_{\ \nu} h^\nu_{\ \rho} h^\rho_{\ \sigma} h^{\sigma}_{\ \mu}  \nonumber \right. \\ && \left. \ -\  \frac{1}{16} h (h^\mu_{\ \rho} h^\rho_{\ \mu})^2  - \frac{1}{12} h^2 h^\mu_{\ \nu} h^\nu_{\ \rho} h^\rho_{\ \mu} + \frac{1}{48} h^3 h^\mu_{\ \rho} h^\rho_{\ \mu} - \frac{1}{960} h^5 \right) + \ \mathcal{O}(\lambda^3) . \nonumber \\ \label{l5}
\end{eqnarray}
By virtue of eqns. ($\ref{l1}$), ($\ref{l2}$) and ($\ref{l3}$)- ($\ref{l5}$) we have now explicitly given all field operators\footnote{Modulo the operators involving the ghost fields.} of quantum Einstein gravity with a cosmological constant up to canonical dimension $5$. Of course, the expansion does not stop there. In fact eq. ($\ref{gaexp}$) provides us with an infinite sum of operators of ever increasing canonical dimension, involving higher and higher powers of the gravitational coupling $\lambda$. Hence, the latter must have negative canonical dimension, as was already pointed out in eq. ($\ref{cdN}$).

The conventional wisdom is therefore that quantum Einstein gravity is perturbatively nonrenormalizable by dimensional analysis. Indeed it has been shown by t'Hooft and Veltman \cite{Hooft}, employing the background field method and dimensional regularization, that the 1-loop divergence of pure gravity with $\Lambda_K=0$ is given by 
\begin{eqnarray}
\frac{1}{8 \pi \epsilon} \left(\frac{1}{120} \overline{R}^2 + \frac{7}{20} \overline{R}_{\mu \nu} \overline{R}^{\mu \nu} \right).  \label{div1}
\end{eqnarray}
Here, the parameter $\epsilon=4-d$ is the deviation of the space-time dimension $d$ of $4$, and the overlined curvatures refer to the background metric $\overline{g}_{\mu \nu}$\footnote{In the background field method (BFM), one employs a splitup $g_{\mu \nu}= \overline{g}_{\mu \nu} + h_{\mu \nu}$ of the metric where $\overline{g}_{\mu \nu}$ is the background field. The latter is often required to satisfy the classical field equations. See \cite{BDJ} for a good review of the BFM.}. Clearly, the divergence ($\ref{div1}$) is not proportional to the original Einstein action ($\ref{EH}$), which is the kind of behaviour one would expect from a nonrenormalizable theory.

However, for pure gravity with $\Lambda_K=0$ and a background field that satisfies the Einstein equation ($\ref{EEQ}$), we have $\overline{R}_{\mu \nu}=0$ and thus the divergence ($\ref{div1}$) disappears- pure gravity is one loop finite! This is no longer true in the presence of matter fields, and furthermore it has been shown \cite{Goro} that at two loops a divergence
\begin{eqnarray}
\frac{209 \lambda^2}{2880 (16 \pi^2)^2} \frac{1}{\epsilon} \overline{R}^{\mu \nu}_{\ \ \gamma \delta} \overline{R}^{\gamma \delta}_{\ \ \rho \sigma} \overline{R}^{\rho \sigma}_{\ \ \mu \nu}
\end{eqnarray}
of pure gravity appears which remains even after the Einstein equations have been used. Hence, the conclusion seems inescapable: quantum Einstein gravity is perturbatively nonrenormalizable.

Note, however, that the inclusion of the cosmological term does not only provide us with a term linear in $h$ and some kind of ''mass'' term in the action, with problematic consequences (see the discussion in section ($\ref{propsec}$)), but also with additional interactions, i.e. operators involving more than two fields. In particular, there appear two operators in $\mathcal{L}^{(3)} $ and $\mathcal{L}^{(4)} $ multiplied with the cosmological constant\footnote{The operators have ''mixed'' couplings $\lambda \Lambda_k$ and $\lambda^2 \Lambda_K$.} which have canonical dimensions $\le 4$ and which are therefore renormalizable. This would not have been the case had we only considered the expansion of the curvature scalar $R$. Then, apart from the kinetic term, only operators with canonical dimension $\ge 5$ appear in the action.
 
We will discuss systematically the various operators appearing in the expanded gravity action in section ($\ref{ToyM}$). The inclusion of the cosmological term will then lead to some interesting speculations in section ($\ref{GravNoCut}$).

\end{subsection}

\begin{subsection}{The general action for gravity} \label{GAFG}
The underlying principle of general relativity is the invariance of the theory under general coordinate transformations. Hence, an action of a theory describing gravitation should be gauge invariant in that sense. The Einstein-Hilbert action ($\ref{EH}$), of course, meets this requirement. However, the principle of gauge invariance does not define the theory completely, since infinitely many invariants can be constructed out of the metric. In addition to the operators $\sqrt{g}$ and $\sqrt{g}R$ appearing in ($\ref{EH}$), there are invariants $\sqrt{g} R^2$, $\sqrt{g} R_{\mu \nu} R^{\mu \nu}$, $\sqrt{g} R^3$ etc. Hence, the most general action for a theory of gravitation takes the following form:
\begin{eqnarray}
S_{grav}&=& \int d^4 x \sqrt{g}  \left( -4 \frac{\Lambda_K}{\lambda^2} + \frac{2}{\lambda^2} R + c_1 R_{\mu \nu} R^{\mu \nu} - c_2 R^2 + ... \right) \label{Sgrav}
\end{eqnarray}
where $c_1, c_2$ are additional (dimensionless) coupling constants and $...$ means higher powers of $R$, $R_{\mu \nu}$ and $R_{\mu \nu \rho \sigma}$. Note that we did not include a term $\sqrt{g} R_{\mu \nu \rho \sigma} R^{\mu \nu \rho \sigma}$ in the action ($\ref{Sgrav}$) because of the Gauss-Bonnet topological invariance in four dimensions \cite{Stelle}:
\begin{eqnarray}
\int d^4 x \sqrt{g}   \left( R_{\mu \nu \rho \sigma} R^{\mu \nu \rho \sigma} - 4 R_{\mu \nu} R^{\mu \nu} + R^2 \right) =0
\end{eqnarray}
for space-times topologically equivalent to flat space. 

The experimental bounds on the couplings $c_1$, $c_2$ are very poor. From \cite{Dono1} we obtain $c_1, c_2 \le 10^{74}$, whereas the coefficients of yet higher powers of $R$ have essentially no experimental constraints.  The basic reason for this is that at the energies accessible at our present-day experiments, the curvature is so small that higher powers of $R$ are even smaller. We will demonstrate this behaviour in an explicit calulation at the end of this section.

Thus, there is no real justification\footnote{One may, however, demand that the classical equations of motion for the metric tensor be second order partial differential equations \cite{MackART}. Then one is restricted to the Einstein-Hilbert term.} for preferring the Einstein-Hilbert action ($\ref{EH}$) over the general action ($\ref{Sgrav}$) in a theory of gravitation. It cannot be argued on the basis of symmetry nor experimental input. Unlike in other theories, renormalizability is not a criterion, either: as has been discussed in the last section, a quantum theory based on the Einstein-Hilbert term already involves infinitely many nonrenormalizable operators.

We therefore refine our definition of a bare gravity action given in eq. ($\ref{StotdeReg}$) as follows:
\begin{eqnarray}
{S}_{tot}(\Lambda_0) &=& \int_x A  \frac{(2 \pi)^4}{2}  h   - \frac{1}{2} \langle h^{\mu \nu}, \Delta^{\Lambda_0 \ -1}_{\mu \nu \rho \sigma} h^{\rho \sigma} \rangle  -  \langle \overline{C}^\mu, \Delta_{GH \mu \nu}^{\Lambda_0 \ -1} {C}^\nu    \rangle  + {L} \big(h,C,\overline{C}, \Lambda_0 \big) \nonumber \\ \label{StotdeRegG}
\end{eqnarray}
where the bare interaction term ${L} \big(h,C,\overline{C}, \Lambda_0 \big)$ is now understood to contain all interaction terms of the  classical  general action ($\ref{Sgrav}$) expanded in powers of the gravitational field $h$, as well as counterterms necessary to cancel the upcoming divergences:  
\begin{eqnarray}
{L} \big(h,C,\overline{C}, \Lambda_0 \big)  := \int_x  \mathcal{L}_{int}^{grav} \big(h,C,\overline{C}\big) + L_{C.T.}\big(h,C,\overline{C},\Lambda_0     \big).  \label{LGravG}
\end{eqnarray}
Hence, our bare gravity action ($\ref{StotdeRegG}$) contains an infinite number of coupling constants $\Lambda_K, \lambda, c_1, c_2, ...$. In the next section, we will show that if we employ an effective field theory approach with flow equations $\grave{a}$ la chapter ($\ref{EffFlow}$), only a finite number\footnote{Depending on the number of renormalization and improvement conditions that are imposed.} of counterterms has to be included in $L_{C.T.}\big(h,C,\overline{C},\Lambda_0  \big)$. Remember that due to the symmetry-breaking cutoff regularization, among these will be counterterms for all operators allowed by the unbroken global $O(4)$ invariance.

The inverse regularized propagators appearing in  ($\ref{StotdeRegG}$) are still (the position-space versions of) the graviton and ghost propagators ($\ref{gravpropG}$) and ($\ref{ghpropG}$) that have been extracted out of the Einstein-Hilbert term. This means that contributions bilinear in $h$ that stem from the higher order field invariants in ($\ref{Sgrav}$) have been included in the interaction term ($\ref{LGravG}$). 

From the literature \cite{Stelle} it is known that a quantum theory of gravitation that is based on the action
\begin{eqnarray}
S_{R2}&=& \int d^4 x \sqrt{g}  \left( \frac{2}{\lambda^2} R + c_1 R_{\mu \nu} R^{\mu \nu} - c_2 R^2 \right) \label{SR^2}
\end{eqnarray}
is renormalizable, but plagued by unitarity problems. We will refer to  ($\ref{SR^2}$) as $R^2$-gravity. Since our bare gravity action ($\ref{StotdeRegG}$) contains the operators of ($\ref{SR^2}$), we have to discuss the implications of $R^2$-gravity for our work. 

To do so, we consider the contribution bilinear in $h$ to the interaction term $\int_x \mathcal{L}_{int}^{grav} \big(h,C,\overline{C}\big)$ that is generated by the operator 
\begin{eqnarray}
\int_x \sqrt{g} \left( c_1 R_{\mu \nu} R^{\mu \nu} - c_2 R^2 \right).  \label{R^2}
\end{eqnarray}
From the definitions ($\ref{DCur}$) and ($\ref{DChr}$) of the curvature tensor and the Christoffel symbols follows that each operator $R_{\mu \nu \rho \sigma}$ and its contractions $R$, $R_{\mu \nu}$ contain two derivatives. Hence for an expansion  ($\ref{h}$) of the metric density, the term bilinear in the gravitational field $h$ emerging out of ($\ref{R^2}$) will be of the form 
\begin{eqnarray}
L^{quad}_{R^2}(h)= - \frac{c_i \lambda^2}{2} \int_x h \partial^4 h.  \label{R2toyI}
\end{eqnarray}
The exact form of ($\ref{R2toyI}$) is rather lenghty due to the proliferations of Lorentz indices and we will not give it here. A momentum space version can be found in \cite{Stelle}. In order to extract the relevant physics, let us instead consider a toy model of the bare action ($\ref{StotdeRegG}$) that contains only two operators:
\begin{eqnarray}
 \int_x \frac{1}{2} h \partial^2 h - \frac{c_i \lambda^2}{2} h \partial^4 h \label{R2toy}
\end{eqnarray}
where the second term in ($\ref{R2toy}$) is treated as an interaction. The momentum space propagator of the model is thus again
\begin{eqnarray}
\frac{1}{k^2},
\end{eqnarray}
whereas the vertex factor of the ''interaction'' ($\ref{R2toyI}$) follows from
\begin{eqnarray}
\int_{x_1 x_2} e^{i x_1 k_1} e^{i x_2 k_2}  \frac{\delta}{\delta h(x_1)} \frac{\delta }{\delta h(x_2)} L^{quad}_{R^2}(h)   &=& -\frac{c_i \lambda^2}{2} \delta(k_1+ k_2) (k_1^4+ k_2^4)
\end{eqnarray}
as
\begin{eqnarray}
\tau(k)&=& -c_i \lambda^2 k^4.
\end{eqnarray}
Using the geometric series, we may calculate the full or ''dressed'' propagator\footnote{This is of course what we would have obtained had we used the full action ($\ref{R2toy}$) to derive the propagator.}:
\begin{eqnarray}
G^{(2)}(k^2) &=& \frac{1}{k^2} - \frac{1}{k^2} c_i \lambda^2 k^4 \frac{1}{k^2} + \frac{1}{k^2} c_i \lambda^2  k^4 \frac{1}{k^2}  c_i \lambda^2  k^4 \frac{1}{k^2}   + ... \nonumber \\ &=& \frac{1}{ k^2 + c_i \lambda^2 k^4 } .  \label{R2prop}
\end{eqnarray}
From ($\ref{R2prop}$) two important conclusions can be drawn. First, we may apply an expansion into partial fractions
\begin{eqnarray}
 \frac{1}{ k^2 + c_i \lambda^2 k^4 } = \frac{1}{k^2} - \frac{1}{k^2+ 1/( c_i \lambda^2)} .  \label{PBZ}
\end{eqnarray}
The second term on the RHS of ($\ref{PBZ}$) corresponds to the propagator of a massive scalar, but with the wrong overall sign.  Hence, the full propagator ($\ref{R2prop}$) leads to a potential that includes an additional Yukawa term:
\begin{eqnarray}
V(r) \sim - \frac{1}{r} \left( 1  - {e^{-r/\sqrt{c_i \lambda^2}}} \right).  \label{R2pot}
\end{eqnarray}
From ($\ref{R2pot}$) it is now clear why the experimental bounds on the couplings $c_i$ are so crude. Since $\lambda \sim 10^{-35}m$ it requires very large values of $c_i$ to produce observable effects.

Second, if we look at the denominator of the propagator ($\ref{R2prop}$) we see that the $k^4$ term becomes dominant over $k^2$ for
\begin{eqnarray}
k^2 > c_i \lambda^2 . \label{k^4l}
\end{eqnarray}
A $1/k^4$ propagator makes the theory perturbatively renormalizable, as can be argued already on the level of dimensional analysis. From ($\ref{Dphix}$) follows that the canonical dimension of a field $\phi$ is $D_\phi=0$ for an UV behaviour $\sigma=4$ of the associated propagator, and thus all operators of the (expanded) $R^2$-action ($\ref{SR^2}$) have dimensions $\le 4$ and are therefore renormalizable. 

On the other hand, the $R^2$-propagator introduces bad behaviour such as ghostlike particles and hence violation of the unitarity \cite{Stelle}. However, it can be shown \cite{Simon1} that these kind of problems only arise in the high energy domain  ($\ref{k^4l}$), i.e. when the UV behaviour of the theory is governed by the $1/k^4$ term and $R^2$-gravity is treated as a fundamental theory. If we introduce a cutoff and restrict ourselves to momenta $k^2 < c_i \lambda^2$, the $R^2$ terms produce only small corrections to the Einstein-Hilbert theory and no problematic behaviour is introduced. This argument holds still true when even higher order field invariants are included in the action, as we have done in ($\ref{Sgrav}$).

We conclude this section by noting that we also refine the definition of the bare extended action $\tilde{S}_{tot}({\Lambda_0})$ given in eq. ($\ref{StotdeRegE}$) by replacing ($\ref{LGravE}$) with
\begin{eqnarray}
\tilde{L} \big(h,C,\overline{C}, \beta, \tau, \Lambda_0 \big)  := \int_x \tilde{\mathcal{L}}_{int}^{grav}(h,C,\overline{C}, \beta, \tau) + \tilde{L}_{C.T.}\big(h,C,\overline{C}, \beta, \tau,    \Lambda_0  \big)   \label{LGravEG}
\end{eqnarray}
where
\begin{eqnarray}
\tilde{\mathcal{L}}_{int}^{grav}(h,C,\overline{C}, \beta, \tau) := \mathcal{L}_{int}^{grav}(h,C,\overline{C}) +   \lambda^{-1} \beta_{\mu \nu} \mathcal{L}_C \tilde{g}^{\mu \nu}  + \tau_\mu C^\nu \partial_\nu C^\mu.    
\end{eqnarray}
Concerning the counterterms appearing in $\tilde{L}_{C.T.}\big(h,C,\overline{C}, \beta, \tau, \Lambda_0  \big)$, we refer the reader to the remarks following eq. ($\ref{LGravG}$). Note however that additional counterterms for the nonlinear BRS variations have been included in ($\ref{LGravEG}$).

\end{subsection}

\end{section}

\begin{section}[Predictivity of effective Euclidean quantum gravity]{Predictivity of effective Euclidean quantum gravity from the viewpoint of the renormalization group}

\begin{subsection}{Analysis without Slavnov-Taylor identities}  \label{ToyM}
As has been discussed in the last section, a quantum theory of gravitation whose UV behaviour is governed by the $1/k^2$ propagator of the Einstein-Hilbert term of the bare action ($\ref{StotdeRegG}$) involves infinitely many nonrenormalizable operators. Thus it is clear by the arguments given in section ($\ref{EffFlowOver}$) that we will not be able to fix the (nonrenormalizable) couplings\footnote{The cosmological constant $\Lambda_K$ is a renormalizable coupling. See section  ($\ref{GravNoCut}$).} $\lambda, c_1, c_2...$ to arbitrary values at some renormalization scale $\Lambda_R$ while at the same time sending the UV cutoff of the theory to infinity, $\Lambda_0 \rightarrow \infty$.

However, as has been outlined in section ($\ref{EffFlowOver}$) and proven in chapter ($\ref{EffFlow}$), it is possible to extract information out of a nonrenormalizable theory as long as one keeps a finite UV cutoff $\Lambda_0$ and contents oneself with predictions of finite accuracy. We will therefore show in the following how the program of chapter ($\ref{EffFlow}$) can be applied to the theory of gravitation described by the general action ($\ref{Sgrav}$).

To do so, let us once more examine the structure of  ($\ref{Sgrav}$),
\begin{eqnarray}
S_{grav}&=& \int d^4 x \sqrt{g}  \left( -4 \frac{\Lambda_K}{\lambda^2} + \frac{2}{\lambda^2} R + c_1 R_{\mu \nu} R^{\mu \nu} - c_2 R^2 + ... \right). \nonumber
\end{eqnarray}
As we have discussed in the last section, each curvature tensor $R$ comes with two derivatives. Hence, if we express the action in terms of the metric density $\tilde{g}^{\mu \nu}= \sqrt{g} \ g^{\mu \nu} $ and apply the splitup ($\ref{KKexp}$) 
\begin{eqnarray}
\tilde{g}^{\mu \nu}= \delta^{\mu \nu} + \lambda h^{\mu \nu} \nonumber,
\end{eqnarray}
the expansion of ($\ref{Sgrav}$) in powers of $h$ takes the following schematic structure:
\begin{eqnarray}
S_{grav}& \sim & \int d^4 x \Big( \frac{\Lambda_K}{\lambda}  h + h (\partial^2  + \Lambda_K) h + \lambda \Lambda_K h^3 + \lambda^2 \Lambda_K h^4 + \lambda h^2 \partial^2 h + \lambda^3 \Lambda_K h^5  \nonumber \\ &&   \hspace{1.5cm} + \ c_1 \lambda^2 h \partial^4 h + c_2 \lambda^2 h \partial^4 h  + \lambda^2 h^3 \partial^2 h +  \lambda^4 \Lambda_K h^6 ...   \Big) .  \label{Stoy}
\end{eqnarray}
In ($\ref{Stoy}$) we have dropped all Lorentz indices and ordered the field operators with respect to their increasing canonical dimensions (with the exception of the kinetic term, which appears together with the ''mass'' term). Note that all combinations\footnote{The operators $\int_x h^{n_1} \partial^{2m} h^{n_2}, \ n_1+ n_2=n+1,$ and $\int_x h^n \partial^{2m} h$ are not linearly independent, as can by shown by integration by parts.}
\begin{eqnarray}
\int_x h^n \partial^{2m} h, \ \ \ n,m \in \mathbb{N}
\end{eqnarray}
of powers of fields $h$ and derivatives $\partial^{2}$ possible appear in ($\ref{Stoy}$).

In order to perform the analysis of chapter ($\ref{EffFlow}$), a cutoff regularization had to be imposed. We have already done this for the action ($\ref{Sgrav}$) in the last section, leading to the definition of a bare gravity action in eq. ($\ref{StotdeRegG}$). As has been discussed in section ($\ref{CutRegG}$), the cutoff regularization violates the BRS invariance of the gauge fixed action, and hence in eq. ($\ref{LGravG}$) counterterms for all operators allowed by the unbroken global $O(4)$ invariance have been included in the bare potential ${L} \big(h,C,\overline{C}, \Lambda_0 \big)$.

In section ($\ref{CutRegG}$), we have reviewed a fine-tuning procedure developed for the perturbative renormalization of Yang-Mills theory with flow equations which aims at the restoration of the Slavnov-Taylor identities (and therefore the gauge/BRS symmetry) of the theory in the no-cutoff limit. The first step of this procedure consists of introducing counterterms for all operators allowed by the unbroken global symmetries, and of establishing the boundedness and convergence of the vertex functions in the limit $\Lambda_0 \rightarrow \infty$ for an arbitrary set of renormalization conditions. The second step then amounts to the determination of one particular choice of these renormalization conditions such that the STI are restored in the limit $\Lambda_0 \rightarrow \infty$.  

We will now propose an analogon to the first step of this procedure for the theory of gravitation described by the general action ($\ref{Sgrav}$). The main difference to Yang-Mills theory lies once more in the fact that gravity is nonrenormalizable: the expansion of ($\ref{Sgrav}$) leads to infinitely many nonrenormalizable operators, and hence renormalization conditions for some nonrenormalizable operators will have to be imposed, too. In section ($\ref{EffFlowOver}$) and chapter ($\ref{EffFlow}$), these have been referred to as improvement conditions, and it has been pointed out that their introduction in general prevents us from taking the no-cutoff limit $\Lambda_0 \rightarrow \infty$.  

The starting point of our analysis will be once more the definition of a momentum-space effective total action for quantum gravity:
\begin{eqnarray} 
S_{tot}(\Lambda) &=&  \int \frac{d^4k}{(2 \pi)^4} \bigg(  A  \frac{(2 \pi)^4}{2} \ \delta(k) h(k)    -  \frac{1}{2} h^{\mu \nu}(k) \left(\Delta^{\Lambda}_{\mu \nu \rho \sigma}(k^2)    \right)^{-1} h^{\rho \sigma}(-k)  \nonumber  \\ && \ \ \ \ \ \ \ \  - \ C^\mu(k) \left(\Delta^{\Lambda}_{GH \mu \nu}(k^2)\right)^{-1} \overline{C}^\nu(-k) \bigg) +  L \Big(h, C, \overline{C}; \Lambda \Big)  \label{SGraveffG}
\end{eqnarray}
where  $L(h, C, \overline{C}; \Lambda)$ is a not necessarily local interaction term. While ($\ref{SGraveffG}$) looks exactly the same as the effective gravity action ($\ref{SGraveff}$), note that it is understood that $S_{tot}(\Lambda_0) $ is now given by the bare gravity action ($\ref{StotdeRegG}$) introduced in the last section. 

In the following, the behaviour of the effective potential $L(h, C, \overline{C}; \Lambda)$ will be investigated as renormalization and improvement conditions are imposed. Therefore, we have to define vertex functions by expanding $L$ in powers of the fields $h_{\mu \nu}$, $C^\mu$ and  $\overline{C}_\mu$.  To keep things as simply as possible, let us introduce the following notations for the fields and indices:
\begin{eqnarray}
{\Phi}&:=& \left(h_{\mu \nu}, C^\mu, \overline{C}_\mu \right) \label{Gvec}   \\
\textbf{n} &:=& \left( n_h, n_C, n_{\overline{C}} \right) \ , \ \ \ n:= n_h + n_C + n_{\overline{C}}.
\end{eqnarray}
We will also sometimes refer to the ''components'' $h$, $C$ and  $\overline{C}$ of the vector $\Phi$ by the placeholder $\phi$. The vertex functions can now be defined as $n$-fold functional derivatives of the potential $L(\Phi, \Lambda)$: 
\begin{eqnarray}
\delta^4(k_1 + ... + k_n ) L_n(k_1, ..., k_n, \Lambda) = (2 \pi )^{4n} \delta^{(n)}_{\hat{\Phi}} L(\Phi, \Lambda) \big|_{\Phi=0}    \label{LexpG}
\end{eqnarray}
where
\begin{equation}
 \delta^{(n)}_{\hat{\Phi}} := \frac{\delta}{\delta \phi(k_1)} ... \frac{\delta}{\delta \phi(k_n)}.
\end{equation}
In our simplifying notation, we have suppressed the $O(4)$ tensor structure of the vertex functions, as well as the assignment of the momenta to the multiindex $\textbf{n}$. Note that the $L_n(k_1, ..., k_n, \Lambda)$ are symmetric (antisymmetric) upon permuting the variables belonging to the graviton fields $h_{\mu \nu}$ (the ghost fields $C^\mu$, $\overline{C}_\mu $) .

Formally, the vertex functions ($\ref{LexpG}$) look exactly the same as those introduced in eq. ($\ref{Lexp}$) for the scalar field theory. This is of course a benefit of our notation, and in fact it was the reason for employing it. Hence, if we define running coupling constants ${\rho}_i(\Lambda)$  as coefficients of Taylor expansions of the vertex functions $L_n(k_1, ..., k_n, \Lambda)$ around $k_i=0$, we may just adopt the definitions of eqns. ($\ref{rcc1}$)-($\ref{rcc5}$) and ($\ref{rcc6}$)-($\ref{rcc12}$). Note that the so introduced couplings have to be understood as $k_i$-tuples
\begin{eqnarray}
{\rho}_i(\Lambda)=\big( \rho^1_i(\Lambda), ..., \rho^{k_i}_i(\Lambda) \big)  \label{ktup}
\end{eqnarray}
in order to account for the different combinations of fields $h$, $C$, $\overline{C}$ and the  $O(4)$ tensor structure entering the vertex functions. However, all couplings $\rho^j_i(\Lambda)$ belonging to a $k_i$-tuple ${\rho}_i(\Lambda)$ share the following two properties:
\begin{enumerate}
\item They emerge from vertex functions that have the same number of external legs
\item All $\rho^j_i(\Lambda), \ j=1...k$, have the same canonical dimension $D_{\rho_i}$.
\end{enumerate}
The latter property follows from the first and the fact that all fields $h$, $C$, $\overline{C}$ have canonical dimension $D_\phi=1$, as has been pointed out in eq. ($\ref{mhd}$). Since from the point of view of our analysis the attributes above are what matters, we will not have to distinguish between couplings belonging to the same $k_i$-tuple ${\rho}_i(\Lambda)$.

In order to avoid possible confusion, we would like to stress that ultimatively, the running couplings ${\rho}_i(\Lambda)$ are \textit{not} expected to be all independent from each other. In fact, they are supposed to be given in terms of the ''physical'' running couplings $\Lambda_K(\Lambda),\lambda(\Lambda), c_1(\Lambda), c_2(\Lambda)...$. But this is only after some kind of restoration of the Slavnov-Taylor identities, which will be the subject of the next section. For the moment, the STI are violated by the cutoff regularization and we have to treat all couplings as independent. Note that an intuitive way of understanding the effect of the restoration of the STI is to relate the couplings $\rho_i$ to the coefficients of the expansion ($\ref{Stoy}$). 

In the following table, we have collected all running couplings ${\rho}_i(\Lambda)$ having canonical dimension $D_{\rho_i} \ge -2$. They are given by the definitions of eqns. ($\ref{rcc1}$)-($\ref{rcc5}$) and ($\ref{rcc6}$)-($\ref{rcc12}$). In addition, we give the corresponding position space composite field operators, together with their canonical dimensions $D_{\mathcal{O}_i}= 4-D_{\rho_i}$. They are related by a derivative expansion of the position-space effective potential $L(\Phi, \Lambda)$, as is discussed in Appendix ($\ref{MDV}$). Remember that in our notation each $\phi$ is a placeholder for $h$, $C$ or $\overline{C}$, where of course the fields have to be distributed such that symmetries as $O(4)$ invariance are respected\footnote{Note also that because of a symmetry under global phase transformations, only pairs $C_\mu \mathcal{O}^{\mu \nu}(h) \overline{C}_\nu$ will appear in the composite field operators $\mathcal{O}_i$. If we assign ''ghost numbers'' $+1$ and $-1$ to $C_\mu$ and  $\overline{C}_\nu$ respectively, this means that there will be only vertex functions associated with ghost number $0$.}. Finally, we state the associated coefficients of the expansion ($\ref{Stoy}$) in order to give a relation of the ${\rho}_i(\Lambda)$ to the ''physical'' couplings.

\begin{table}[here]
\hspace{-0.4cm}\begin{tabular}{|c|l|c|c|c|c|} \hline  
Coupling $\rho_i$  & Definition via momentum-space& $D_{\rho_i}$  &  Position-space  & $D_{\mathcal{O}_i}$ & ''Physical \\  & vertex functions $L_n$ &  & operator $\mathcal{O}_i$  &  & coupling'' \\ \hline \hline $\rho_1(\Lambda)$ & $L_1(0,\Lambda)$ & $3$ & $\phi$ &  $1$ & ${\Lambda_K}/{\lambda}$ \\ \hline $\rho_2(\Lambda)$ & $L_2(0,0,\Lambda)$ & $2$ & $\phi^2$ & $2$   & ${\Lambda_K}$ \\ \hline $\delta^{\mu \nu} \rho_3(\Lambda) $ & $\partial^\mu_{1,2} \partial^\nu_{1,2} L_2(k_1,k_2,\Lambda)|_{k_i=0} $ & $0$ & $\phi \partial^2 \phi $ & $4$ & $1$ \\ \hline $ \rho_4 (\Lambda)$ & $L_3(0,0,0,\Lambda) $ & $1$ & $\phi^3$ & $3$ & ${\lambda} {\Lambda_K}$ \\ \hline $ \rho_5(\Lambda)$ & $L_4(0,0,0,0,\Lambda)$ & $0$ & $\phi^4$ & $4$ & ${\lambda}^2 {\Lambda_K} $ \\ \hline $\delta^{\mu \nu} \rho_{6}(\Lambda) $ & ${\partial^\mu_{1,3}} \partial^\nu_{1,3} L_3(k_1,k_2,k_3,\Lambda)|_{k_i=0}$& $ -1$ & $\phi^2 \partial^2 \phi$ & $5$ & $\lambda$ \\ \hline $\delta^{\mu \nu} \rho_{7}(\Lambda) $ &  $ \partial_{1,3}^{\mu} {\partial_{2,3}^{\nu}} L_3(k_1,k_2,k_3,\Lambda)|_{k_i=0} $ & $-1$ & $\phi \partial \phi \partial \phi$ & $5$ & $\lambda$ \\ \hline $\rho_8(\Lambda)$ & $L_5(0,...,0, \Lambda)$ & $-1$ & $\phi^5$ & $5$ & ${\lambda}^3 {\Lambda_K}$ \\ \hline $  I^{\mu \nu \rho \sigma} \rho_9(\Lambda)$ & $\partial^\mu_{1,2} \partial^\nu_{1,2} \partial^\rho_{1,2} \partial^\sigma_{1,2} L_2(k_1,k_2,\Lambda)|_{k_i=0}$ & $-2$ & $\phi \partial^4 \phi $ & $6$ & $c_i \lambda^2  $ \\ \hline $ \delta^{\mu \nu} \rho_{10}(\Lambda) $ & $\partial^\mu_{1,4} \partial^\nu_{1,4} L_4(k_1,k_2,k_3,k_4,\Lambda)|_{k_i=0}$ & $-2$ & $\phi^3 \partial^2 \phi$ & $6$ & $\lambda^2$ \\ \hline  $  \delta^{\mu \nu} \rho_{11}(\Lambda) $ & $\partial_{1,4}^{\mu} \partial_{2,4}^{\nu}  L_4(k_1,k_2,k_3,k_4,\Lambda)|_{k_i=0} $ & $-2$ & $\phi^2 \partial \phi \partial \phi$ & $6$ & $\lambda^2$ \\ \hline $\rho_{12}(\Lambda) $ & $L_6(0,...,0,\Lambda)$ & $-2$ & $\phi^6$ & $6$ & $\lambda^4 \Lambda_K$ \\ \hline 
\end{tabular} 
\caption{Some couplings and field operators of effective quantum gravity}   \label{tab}
\end{table}
The quantity $I^{\mu \nu \rho \sigma}$ appearing in Table ($\ref{tab}$) has been defined in section ($\ref{GPT}$) as $I^{\mu \nu \rho \sigma}= \delta^{\mu \nu} \delta^{\rho  \sigma} + \delta^{\mu \rho} \delta^{\nu  \sigma} + \delta^{\mu \sigma} \delta^{\nu  \rho}$. Furthermore, it has been pointed out in Appendix ($\ref{MDV}$) that by integration by parts, the operators $\int_x \phi^2 \partial^2 \phi$ and $\int_x \phi \partial \phi \partial \phi$ are not linearly independent. Thus, one can merge the couplings $\rho_6$ and $\rho_7$ into just one coupling constant associated with the operator $\int_x \phi^2 \partial^2 \phi$, which turns out to be
\begin{eqnarray}
\rho_{6/7}=-\rho_{6}(\Lambda) + \frac{1}{2} \rho_{7}(\Lambda).  \label{r67}
\end{eqnarray}
Eq. ($\ref{r67}$) means componentwise addition of the $k_i$-tuples $\rho_{6}$, $\rho_{7}$. A similar relation holds true for the couplings $\rho_{10}$ and $\rho_{11}$.

Let us now come to the renormalization and improvement conditions that have to be imposed for the couplings $\rho_i$ in order to establish bounds for the vertex functions $L_n(\Lambda)$ of the effective potential $L(\Phi, \Lambda)$. As we have already stressed, at this point of our analysis we will have to impose a set of arbitrary  renormalization and improvement conditions. However, the question arises for which couplings improvement conditions should be introduced. Recall from section ($\ref{EffFlowOver}$) that the canonical dimensions of the couplings for which improvement conditions are specified determines an improvement index $s$, that in turn is related to the amount of predictivity the effective field theory will have in the end.

To answer this question, let us consider the experimental values of the ''physical'' coupling constants $\Lambda_K, \lambda, c_1, c_2...$ appearing in the action ($\ref{Sgrav}$). They correspond to the renormalized values of the couplings.  Recent astrophysical data \cite{Krauss} strongly suggests that the cosmological constant $\Lambda_K$ should be nonzero and positive, while extremely small. The observational bounds are \cite{Dono1}
\begin{eqnarray}
|\Lambda_K| \le 10^{-83} \text{GeV}^2.  \label{Lex}
\end{eqnarray}
The coupling $\lambda$ has been defined at the beginning of section ($\ref{BRSQ}$) as $\lambda^2 = 32 \pi G$ where $G$ is Newton's constant. Thus we have
\begin{eqnarray}
\lambda \sim (10 \cdot M_{P})^{-1}  \label{lex}
\end{eqnarray}
where $M_{P} \sim 1.2 \cdot 10^{19}$ GeV is the Planck scale. In the following, we will simply speak of $\lambda$ as given by the inverse Planck scale, ignoring the factor of ten appearing in eq. ($\ref{lex}$). 

Concerning the couplings $c_1$,  $c_2$, we have discussed in the last section that their experimental bounds are very poor,
\begin{eqnarray}
c_1, c_2 \le 10^{74},
\end{eqnarray}
whereas the coefficients of yet higher invariants have essentially no experimental constraints. 

Now if we look at the expansion ($\ref{Stoy}$)  and Table ($\ref{tab}$) respectively, we find that all field operators $\mathcal{O}_i$ that have canonical dimensions $D_{\mathcal{O}_i} \le 5$ (corresponding to running couplings $\rho_i(\Lambda)$ with $D_{\rho_i} \ge -1$) are associated with the ''known'' couplings $\Lambda_K$ and $\lambda$. The ''unknown'' couplings $c_i$ appear for the first time associated with the operator
\begin{eqnarray}
\phi \partial^4 \phi
\end{eqnarray}
which has canonical dimension $D_{\mathcal{O}_i}=6$, corresponding to the running coupling constant $\rho_9(\Lambda)$ with $D_{\rho_9}=-2$. 

We therefore conclude that renormalization conditions for the couplings $\rho_i, \ i=1...5$, and improvement conditions for the couplings $\rho_i, \ i=6...8,$ will have to be imposed, the latter having canonical dimensions $D_{\rho_i} = -1$. Doing so, we adopt the notations employed in chapters ($\ref{RenFlow}$) and ($\ref{EffFlow}$):
\begin{eqnarray} 
\rho_{\tilde{a}}(\Lambda_R) &=& 0, \ \ {\tilde{a}}=1,...,3 \label{rcG1} \\
\rho_{\tilde{a}}(\Lambda_R) &=& \rho_{\tilde{a}}^R, \ \ {\tilde{a}}=4,5  \label{rcG2}   \\
\rho_{\tilde{a}}(\Lambda_R) &=& \rho^{NR}_{\tilde{a}} , \ \  {\tilde{a}}=6,...,8   . \label{rcG3}  
\end{eqnarray}
Because of the lack of experimental input for the couplings $c_1, c_2,...$ it will not make sense to specify improvement conditions for couplings $\rho_i$ with $D_{\rho_i} \le -2$. Comparing to section ($\ref{EffFlowOver}$) and chapter ($\ref{EffFlow}$), we recognize that the situation corresponds to an improvement index $s=1$. 

Before we can apply the analysis of chapter ($\ref{EffFlow}$) to the vertex functions ($\ref{LexpG}$) of our effective gravity potential, we have to add two more ingredients. First, the Theorems ($\ref{BoundThII}$), ($\ref{invTh}$) and ($\ref{PreTh}$) established in chapter ($\ref{EffFlow}$) refer to a fixed bare scale $\Lambda_0$. Moreover, remember that a crucial point for proving boundedness and convergence of vertex functions (Theorems ($\ref{BoundTh}$), ($\ref{ConvTh}$) and ($\ref{BoundThII}$)) and for estimating the predictivity of an effective field theory at scales $\Lambda << \Lambda_0$ (Theorems ($\ref{PreTh}$) and ($\ref{UniTh}$)) has been the smallness of the nonrenormalizable couplings at the bare scale $\Lambda_0$. Employing again dimensionless couplings $\lambda_i(\Lambda) = \Lambda^{-D_{\rho_i}} \rho_i(\Lambda)$, this amounts to the constraint
\begin{eqnarray}
\lambda_n(\Lambda_0) \le 1   \label{smallbG}
\end{eqnarray}
for the bare values of the nonrenormalizable couplings\footnote{We have adopted again the notation introduced in section ($\ref{RenFlowOver}$) where renormalizable couplings have been denoted by $\lambda_a$ and nonrenormalizable ones by $\lambda_n$.}. Therefore, the following two questions arise:
\begin{enumerate}
\item What is the bare scale $\Lambda_0$ for effective quantum gravity? 
\item Are the bare nonrenormalizable couplings $\rho_n(\Lambda_0)=\rho_n^0$ sufficiently small? 
\end{enumerate}
We will now show how these questions can be adressed by virtue of an analogon to Theorem ($\ref{invTh}$) and the experimental bounds for the physical couplings $\Lambda_K, \lambda, c_1, c_2...$.  To do so, let us recall some definitions and notations and adapt them to the case of effective quantum gravity. 

The bare values for the couplings $\rho_{\tilde{a}}, \  \tilde{a} =6...8$, are again denoted by
\begin{eqnarray}
\rho_{\tilde{a}}(\Lambda_0)= \rho_{\tilde{a}}^0.  \label{iniCNRG}
\end{eqnarray}
We will use  dimensionless couplings
\begin{eqnarray}
\lambda_{\tilde{a}}^R(\Lambda) &=& \Lambda^{-D_{\rho_{\tilde{a}}}}  \rho_{\tilde{a}}^R \label{lgren}  \\
\lambda_{\tilde{a}}^0(\Lambda) &=& \Lambda^{-D_{\rho_{\tilde{a}}}} \rho_{\tilde{a}}^0,   \label{lgBare}
\end{eqnarray}
as well as dimensionless vertex functions $A_n$ which are related to the $L_n(k_1, ..., k_n, \Lambda)$ defined in eq. ($\ref{LexpG}$) by
\begin{eqnarray}
 A_{n} (k_1,...,k_{n}, \Lambda) =  \Lambda^{n-4} L_{n} (k_1,...,k_{n}, \Lambda)  . \label{AexpG}
\end{eqnarray}
The dimensionless vertex functions are expanded\footnote{Please refer to eq. ($\ref{Apert2_Pert}$) of section ($\ref{GPT}$) for a definition of the coefficients $A_{n}^{(r_1,...,r_5)} (k_1,...,k_{n}, \Lambda)$.} in perturbation theory in the renormalized renormalizable couplings $\lambda_4^R$ and $\lambda_5^R$ and the bare nonrenormalizable couplings  $\lambda_{\tilde{a}}^0, \  \tilde{a} =6...8$:
\begin{eqnarray}
A_{n} (k_1,...,k_{n}, \Lambda) = \sum_{r_1, ...,r_{5}=0}^{\infty} (\lambda_4^R)^{r_1} (\lambda_5^R)^{r_2} (\lambda_{6}^0)^{r_3} (\lambda_{7}^0)^{r_4} (\lambda_{8}^0)^{r_5}  A_{n}^{(r_1,...,r_5)} (k_1,...,k_{n}, \Lambda). \nonumber \\ \label{Apert3}
\end{eqnarray}
Note that since the couplings $\lambda_{\tilde{a}}$ are $k_{\tilde{a}}$-tuples as defined in eq. ($\ref{ktup}$), it is understood that the indices $r_i$ denoting the orders in perturbation theory are multiindices 
\begin{eqnarray}
r_i = \big( r^1_i, ..., r^{k_{\tilde{a}}}_i \big), \ \ \ \ |r_i|:=r^1_i+ ...+ r^{k_{\tilde{a}}}_i,   \label{tupord}
\end{eqnarray}
and that
\begin{eqnarray}
(\lambda_{\tilde{a}})^{r_i} = (\lambda_{\tilde{a}}^{1})^{r^1_i} \cdot... \cdot (\lambda_{\tilde{a}}^{k_{\tilde{a}}})^{r_i^{k_{\tilde{a}}}}.   \label{tupexp}
\end{eqnarray}
The perturbative expansion ($\ref{Apert3}$) is sensible only for small dimensionless couplings. Therefore we impose as additional constraints to the renormalization and initial conditions
\begin{eqnarray}
\lambda_{\tilde{a}}^R(\Lambda) &\le& 1,  \ \ \ \ {\tilde{a}}=4,5  \label{smallRG}  \\
\lambda_{\tilde{a}}^0(\Lambda) &\le& 1, \ \ \ \ {\tilde{a}}=6...8.  \label{smallNRG}
\end{eqnarray}
Because of the definitions ($\ref{lgren}$) and ($\ref{lgBare}$), this implies $\rho_{4}^R \le \Lambda_R$ for $\Lambda \ge \Lambda_R$ and $\rho_{\tilde{a}}^0 \le \Lambda_0^{-1}, \ {\tilde{a}}=6...8$, for $\Lambda \le \Lambda_0$. Thus
\begin{eqnarray}
\lambda_{4}^R(\Lambda) &\le&  {\Lambda_R}/{\Lambda},    \label{smallRGExp}  \\
\lambda_{\tilde{a}}^0(\Lambda) &\le& {\Lambda}/{\Lambda_0}, \ \ \ \ {\tilde{a}}=6...8.  \label{smallNRGExp}
\end{eqnarray}
At this point we already give an analogon to Theorem ($\ref{BoundTh}$) concerning the boundedness of vertex functions for the case of effective quantum gravity. To do so, we define an overall order $r_{NR} = |r_3| + |r_4| + |r_5|$ in perturbation theory in the bare nonrenormalizable couplings $\lambda_{\tilde{a}}^0, \  \tilde{a} =6...8$, and assume for the time being that there exists some bare scale $\Lambda_0$ where the nonrenormalizable couplings are small $\grave{a}$ la eq. ($\ref{smallbG}$). This will then be justified later.

\begin{satz}[Boundedness of Gravity Vertex Functions]   \label{BoundThG}
Given the renormalization conditions ($\ref{rcG1}$)-($\ref{rcG2}$) and the initial conditions ($\ref{iniCNRG}$), and assuming that
\begin{eqnarray}
||\partial^p A_{n}^{(r_1,..., r_5)}(p_1,...,p_{n}, \Lambda_0)|| \le \Lambda_0^{-p}  \left( \frac{\Lambda_0}{\Lambda_R} \right)^{|r_1|} Pln \left( \frac{\Lambda_0}{\Lambda_R} \right) \label{iniNRG}
\end{eqnarray} 
for $n+p \ge 6$, to order $r_1,..., r_5$ in perturbation theory in $\lambda_4^R$, $\lambda_5^R$, $\lambda_{6}^0$,  $\lambda_{7}^0$ and $\lambda_{8}^0$
\begin{eqnarray}
&& ||\partial^p A_{n}^{(r_1,..., r_5)}(p_1,...,p_{n}, \Lambda)|| \nonumber \\ &&  \qquad \qquad  \qquad  \le \Lambda^{-p} \left( \frac{\Lambda}{\Lambda_R} \right)^{|r_1|} \left( \frac{\Lambda_0}{\Lambda} \right)^{r_{NR}}   \left( \delta_{r_{NR}, 0} \ Pln\left( \frac{\Lambda}{\Lambda_R} \right) +  \frac{\Lambda}{\Lambda_0}  Pln \left( \frac{\Lambda_0}{\Lambda_R} \right) \right) \nonumber \\ \label{BoundNRG}
\end{eqnarray}
where $r_{NR} = |r_3| + |r_4| + |r_5|$ and $\Lambda_R \le \Lambda \le \Lambda_0$.

\end{satz}
The  condition ($\ref{iniNRG}$) amounts to the assumption of small inital values for couplings $\rho_{\tilde{n}}(\Lambda_0), \ \tilde{n} \ge 9$, as follows from their definitions in Table ($\ref{tab}$). The proof of Theorem ($\ref{BoundThG}$) goes  in analogy to the proof of Theorem ($\ref{BoundTh}$), and we will therefore skip it.

In order to proceed, we will again need the auxiliary variables $\overline{\lambda}_{\tilde{a}}(\Lambda)$ introduced in section  ($\ref{InvSec}$). They are defined as the values the running dimensionless couplings take for initial conditions $\lambda_{\tilde{a}}^0=0, \  \tilde{a} =6...8$, of the nonrenormalizable couplings:
\begin{eqnarray}
\overline{\lambda}_{\tilde{a}}(\Lambda) = \sum_{\substack{r_1, ...,r_{5}=0 \\ r_{NR}=0}}^{\infty} \lambda_{\tilde{a}}^{(r_1,...,r_5)} (\Lambda)  (\lambda_4^R)^{r_1} (\lambda_5^R)^{r_2} (\lambda_{6}^0)^{r_3} (\lambda_{7}^0)^{r_4}  (\lambda_{8}^0)^{r_5}     \label{zRcoupG}
\end{eqnarray}
where $r_{NR} = |r_3| + |r_4| + |r_5|$.  The $\overline{\lambda}_{\tilde{a}}(\Lambda)$ give rise to the definition of deviations $\Delta \lambda_{\tilde{a}}(\Lambda) = \lambda_{\tilde{a}} (\Lambda) -\overline{\lambda}_{\tilde{a}}(\Lambda)$, which we will need in the following form:
\begin{eqnarray}
\Delta \lambda_{\tilde{a}}^R(\Lambda) = ({\Lambda_R}/{\Lambda})^{D_{\rho_{\tilde{a}}}} \Delta \lambda_{\tilde{a}}(\Lambda_R)  , \ \ \ {\tilde{a}}=6,...,8.   \label{devRG}
\end{eqnarray}
We will consider expansions in the renormalized renormalizable couplings $\lambda_4^R$, $\lambda_5^R$ and the deviations $\Delta \lambda_{\tilde{a}}^R, \ {\tilde{a}}=6,...,8$, where the order of the expansion will be denoted\footnote{The order of an expansion in the couplings $\lambda_4^R$, $\lambda_5^R$, $\lambda_{6}^0$,  $\lambda_{7}^0$ and $\lambda_{8}^0$ has been denoted by $(r_1, ..., r_5)$.} by $(l_1, ..., l_5)$. As we have dicussed above, all couplings have to be understood as $k_{\tilde{a}}$-tuples, and this holds also true for the  deviations $\Delta \lambda_{\tilde{a}}^R$. The expansion has therefore to be understood in the sense of eq. ($\ref{tupexp}$), and we introduce
\begin{eqnarray}
|l_i|:=l^1_i+ ...+ l^{k_{\tilde{a}}}_i
\end{eqnarray}
in analogy to eq. ($\ref{tupord}$). We are now ready to formulate an analogon to Theorem ($\ref{invTh}$) for effective quantum gravity.

\begin{satz}[Inversion of the RG Trajectory of Effective Quantum Gravity]  \label{invThG}
For $\tilde{a}=6...8$ and $l_{\Delta}:= |l_3|+ |l_4|+ |l_5|$ let 
\begin{eqnarray}
|| \Delta \lambda^{(l_1, ..., l_5)}_{\tilde{a}}(\Lambda) || &\le&  \left( \frac{\Lambda}{\Lambda_R} \right)^{|l_1|} \left( \frac{\Lambda_0}{\Lambda} \right)^{l_\Delta}  \frac{\Lambda}{\Lambda_0}  Pln \left( \frac{\Lambda_0}{\Lambda_R} \right) .
\end{eqnarray}
Then to order $l_1,...,l_5$ in perturbation theory in $\lambda_4^R$, $\lambda_5^R$ and the deviations $\Delta \lambda_6^R$, $\Delta \lambda_7^R$ and $\Delta \lambda_8^R$
\begin{eqnarray}
|| \lambda^{0 \ (l_1, ..., l_5)}_{\tilde{a}}(\Lambda) || &\le&  \left( \frac{\Lambda}{\Lambda_R} \right)^{|l_1|} \left( \frac{\Lambda_0}{\Lambda} \right)^{l_\Delta} \frac{\Lambda}{\Lambda_0}  Pln \left( \frac{\Lambda_0}{\Lambda_R} \right)   
\end{eqnarray}
where $\Lambda_R \le \Lambda \le \Lambda_0$.
\end{satz}
In proving Theorem ($\ref{invThG}$), we have to apply\footnote{It is therefore understood that the bare values for the couplings $\rho_{\tilde{n}}(\Lambda_0), \ \tilde{n} \ge 9$, are small.} Theorem ($\ref{BoundThG}$). The proof goes in analogy to the derivation of Theorem ($\ref{invTh}$) discussed in section ($\ref{InvSec}$), and we will not give it here. 

For small couplings $\lambda_4^R(\Lambda), \lambda_5^R(\Lambda)$ in the sense of eqns. ($\ref{smallRG}$), ($\ref{smallRGExp}$) and small deviations\footnote{See also the discussion at the end of this section concerning the smallness of couplings.} $\Delta \lambda_{\tilde{a}}^R(\Lambda) \le 1, \ {\tilde{a}}=6,...,8$, Theorem ($\ref{invThG}$) implies that if the $\Delta \lambda_{\tilde{a}}(\Lambda)$ are small,
\begin{eqnarray}
|| \Delta \lambda_{\tilde{a}}(\Lambda) || &\le&  \frac{\Lambda}{\Lambda_0}  Pln \left( \frac{\Lambda_0}{\Lambda_R} \right),  \ \ \ \ \tilde{a} =6...8,  \label{ifsmallG}
\end{eqnarray}
then also the bare values of the nonrenormalizable couplings $\lambda_{\tilde{a}}^0$ are small:
\begin{eqnarray}
||\lambda^0_{\tilde{a}}(\Lambda) || &\le&  \frac{\Lambda}{\Lambda_0}  Pln \left( \frac{\Lambda_0}{\Lambda_R} \right) \ \ \ \ \tilde{a} =6...8. \label{smallindeedG}
\end{eqnarray} 
Let us now discuss the implications of eqns.  ($\ref{ifsmallG}$) and  ($\ref{smallindeedG}$) for gravity. To do so, we analyze the auxiliary variables $\overline{\lambda}_{\tilde{a}}(\Lambda)$ entering the deviations  $\Delta \lambda_{\tilde{a}}(\Lambda) = \lambda_{\tilde{a}} (\Lambda) -\overline{\lambda}_{\tilde{a}}(\Lambda)$. They have been defined in eq.  ($\ref{zRcoupG}$) as the values the running couplings $\lambda_{\tilde{a}}$  take to $0th$ order in perturbation theory in the bare nonrenormalizable couplings  $\lambda_6^0$, $\lambda_7^0$ and $\lambda_8^0$. Hence, they are given solely in terms of the perturbative contributions in the renormalizable couplings $\lambda_4^R$ and $\lambda_5^R$. If we look at Table ($\ref{tab}$), we see that the latter are both associated with the cosmological constant $\Lambda_K$:
\begin{eqnarray}
\rho_4 \sim \lambda \Lambda_K , \ \ \ \ \rho_5 \sim \lambda^2 \Lambda_K   \label{rel1}
\end{eqnarray} 
where $\lambda$ is the gravitational coupling. On the other hand, the relations of the couplings $\lambda_{\tilde{a}}(\Lambda_R), \ {\tilde{a}}=6,...,8$,  to their ''physical'' counterparts are given by
\begin{eqnarray}
\rho_6, \ \rho_7 \sim \lambda, \ \ \ \ \rho_8 \sim  \lambda^3 \Lambda_K .  \label{rel2}
\end{eqnarray} 
Note that the cosmological constant does \textit{not} enter the couplings $\lambda_6, \ \lambda_7 $, and recall that the relations ($\ref{rel1}$) and ($\ref{rel2}$) should hold after some fine-tunig procedure has been imposed in order to restore the Slavnov-Taylor identities. 

As we have pointed out in eq. ($\ref{Lex}$), experimental input bounds the cosmological constant to be extremely small. In particular, if we compare it to the size of the gravitational constant $G \sim \lambda^2 \sim M_P^{-2}$, we arrive at the famous ratio
\begin{eqnarray}
\lambda^2 \Lambda_K \sim 10^{-120}.  \label{120}
\end{eqnarray} 
It is therefore reasonable to assume that the size of the auxiliary variables $\overline{\lambda}_{\tilde{a}}(\Lambda)$ calculated in perturbation theory in $\lambda_4^R$ and $\lambda_5^R$ will be neglectible to those of the couplings ${\lambda}_{\tilde{a}}(\Lambda), \ {\tilde{a}}=6, 7$. The latter should be of the order $\Lambda/M_P$ as is suggested by eq. ($\ref{rel2}$).  Thus, the deviations $\Delta \lambda_{\tilde{a}}(\Lambda) = \lambda_{\tilde{a}} (\Lambda) -\overline{\lambda}_{\tilde{a}}(\Lambda)$ can be estimated as
\begin{eqnarray}
|| \Delta \lambda_{\tilde{a}}(\Lambda) || &\le&  \frac{\Lambda}{M_P}  Pln \left( \frac{M_P}{\Lambda_R} \right),  \ \ \ \ \tilde{a} =6...8.  \label{smallG}
\end{eqnarray}
Now if we demand that the bare values $\lambda_{\tilde{a}}^0, \  \tilde{a} =6...8$, can be taken small in order to arrive at renormalized values $\lambda_{\tilde{a}}(\Lambda) \sim \Lambda/M_P$, eqns. ($\ref{ifsmallG}$) and  ($\ref{smallindeedG}$) tell us that the UV cutoff scale of effective quantum gravity must not exceed the Planck scale. This motivates us to define the bare scale $\Lambda_0$ of quantum gravity to be
\begin{eqnarray}
\Lambda_0 \sim M_P.  \label{cutP}
\end{eqnarray}
If the Planck scale is considered to be the characteristic energy scale of the theory, it is natural to assume that all ''bare'' physical coupling constants $\Lambda_K(M_P)$, $\lambda(M_P)$, $c_1(M_P)$, $c_2(M_P)...$ entering the action ($\ref{StotdeRegG}$) should be of the order of this scale. This, in turn, suggests that the (dimensionless) vertex functions $A_n(\Lambda)$ of the gravity effective potential $L(\Phi, \Lambda)$ should satisfy
\begin{eqnarray}
|| A_{n}(p_1,...,p_{n}, M_P)||  \sim  1  \label{PlanckNat}
\end{eqnarray}
where the norm $|| \ ||$ is defined as in eq. ($\ref{norm}$). In the following, we will always assume that  ($\ref{PlanckNat}$) holds true, and therefore in particular take it for granted that the bare values of all nonrenormalizable couplings $\rho^0_n$ will be small in the sense of eq. ($\ref{smallbG}$).

From the viewpoint of a ''characteristic energy scale $M_P$'' and eq. ($\ref{PlanckNat}$), it is of course disturbing that the renormalized value of the cosmological constant is so small. It requires a very delicate adjustment of the bare potential $L(\Phi, \Lambda_0)$ and therefore a lot of ''fine-tuning'' in order to reproduce the experimental bounds ($\ref{Lex}$) and the factor of $10^{-120}$ appearing in eq. ($\ref{120}$). We have nothing new to say about this puzzle. Note however that from a technical point of view, there is of course no problem in imposing a tiny renormalized value for the cosmological constant, other than it seems unnatural.

Let us now investigate the predictivity of effective quantum gravity at scales far below the Planck scale, $\Lambda << M_P$, as the renormalization and improvement conditions ($\ref{rcG1}$), ($\ref{rcG2}$), ($\ref{rcG3}$) are imposed as well as the assumption ($\ref{PlanckNat}$). To do so, we proceed as in section ($\ref{Presec}$) and introduce bare vertex functions depending on a parameter $t$
\begin{eqnarray}
\partial^p \tilde{A}_{n}^{(r_1, ..., r_5)}(M_P) := t \ \partial^p A_{n}^{(r_1, ..., r_5)}( M_P), \ \ \ t \in [0,1], \ \ \ n+p \ge 6 . \nonumber \\  \label{shapeANRG}
\end{eqnarray}
By varying $t$, the impact of a change of the initial conditions ($\ref{shapeANRG}$) on the ''running'' vertex functions $A_n(\Lambda)$ can be studied. This is of course related to the effect of not knowing the exact values of the couplings $c_1(M_P)$, $c_2(M_P)...$. The $A_n(\Lambda)$ become dependent on the parameter $t$, 
\begin{eqnarray}
\partial^p A_{n}^{(r_1, ..., r_5)}= \partial^p A_{n}^{(r_1, ..., r_5)}(p_1,...,p_{n}, \Lambda, M_P, t),
\end{eqnarray}
and we may formulate an analogon to Theorem ($\ref{PreTh}$) for effective quantum gravity. This will be done again in perturbation theory in the renormalized renormalizable couplings $\lambda_4^R$ and $\lambda_5^R$ and the bare nonrenormalizable couplings  $\lambda_{\tilde{a}}^0, \  \tilde{a} =6...8$. Since all couplings have to be understood as $k_{\tilde{a}}$-tuples, recall the notation conventions of eqns. ($\ref{tupord}$) and ($\ref{tupexp}$).

\begin{satz}[Predictivity of Effective Quantum Gravity]  \label{PreThG}
Let there be renormalization conditions ($\ref{rcG1}$)-($\ref{rcG2}$) and  improvement conditions ($\ref{rcG3}$). Assume that to order $r_1, ..., r_5$ in perturbation theory in $\lambda_4^R$, $\lambda_5^R$, $\lambda_{6}^0$,  $\lambda_{7}^0$ and $\lambda_{8}^0$
\begin{eqnarray}
||\partial^p A_{n}^{(r_1,..., r_5)}(p_1,...,p_{n}, \Lambda)|| \le \Lambda^{-p}  \left( \frac{\Lambda}{\Lambda_R} \right)^{|r_1|}  \left( \frac{M_P}{\Lambda} \right)^{r_{NR}} Pln \left( \frac{M_P}{\Lambda_R} \right), \label{BoundWPG}
\end{eqnarray}
and that for $n+p \ge 6$
\begin{eqnarray}
|| \frac{d}{d t} \partial^p A_{n}^{(r_1,..., r_5)}(p_1,...,p_{n}, M_P)|| \le M_P^{-p}  \left( \frac{M_P}{\Lambda_R} \right)^{|r_1|} Pln \left( \frac{M_P}{\Lambda_R} \right) . \label{ini2PG}
\end{eqnarray}
Then 
\begin{eqnarray}
|| \frac{d}{d t} \partial^p A_{n}^{(r_1,..., r_5)}(p_1,...,p_{n}, \Lambda)|| \le \Lambda^{-p}  \left( \frac{\Lambda}{\Lambda_R} \right)^{|r_1|}  \left( \frac{M_P}{\Lambda} \right)^{r_{NR}} \left(\frac{\Lambda}{M_P}\right)^{2} Pln \left( \frac{M_P}{\Lambda_R} \right) \nonumber \\ \label{PreG}
\end{eqnarray}
where $r_{NR} = |r_3| + |r_4| + |r_5|$ and  $\Lambda_R \le \Lambda \le M_P$.
\end{satz}
For deducing Theorem  ($\ref{PreThG}$), we refer the reader to the proof Theorem ($\ref{PreTh}$). Note that the conditions ($\ref{BoundWPG}$) and ($\ref{ini2PG}$) are satisfied by virtue of Theorem ($\ref{BoundThG}$) and the assumption ($\ref{PlanckNat}$).  

Integrating eq. ($\ref{PreG}$) and employing the triangle inequality we conclude that for two different sets of bare vertex functions $\partial^p A_{n}^{A (r_1,..., r_5)}(M_P)$ and $\partial^p A_{n}^{B(r_1,..., r_5)}(M_P)$ which are small $\grave{a}$ la eq. ($\ref{PlanckNat}$) the associated ''running'' vertex functions satisfy
\begin{eqnarray}
|| \partial^p A_{n}^{A(r_1,..., r_5)}(p_1,...,p_{n}, \Lambda) - \partial^p A_{n}^{B(r_1,..., r_5)}(p_1,...,p_{n}, \Lambda) || \qquad \qquad \qquad \nonumber \\ \le \Lambda^{-p}  \left( \frac{\Lambda}{\Lambda_R} \right)^{r_1} \left( \frac{M_P}{\Lambda} \right)^{r_{NR}}  \left(\frac{\Lambda}{M_P}\right)^{2} Pln \left( \frac{M_P}{\Lambda_R} \right) . \label{PreFG}
\end{eqnarray}
It has therefore been established that the ignorance about the precise initial values $ \partial^p A_{n}^{(r_1,..., r_5)}(M_P), \ \ n+p \ge 6$, amounts to an indetermination of the ''running'' vertex functions $\partial^p A_{n}^{(r_1,..., r_5)}(\Lambda)$ of the order given by eq. ($\ref{PreFG}$).  Note that for small couplings $\lambda_{\tilde{a}}^R(\Lambda), \lambda_{\tilde{a}}^0(\Lambda)$ in the sense of eqns. ($\ref{smallRG}$)-($\ref{smallNRGExp}$) the latter implies 
\begin{eqnarray}
||A^A_{n}(\Lambda)-A^B_{n}(\Lambda)|| \le \left(\frac{\Lambda}{M_P}\right)^{2} Pln \left( \frac{M_P}{\Lambda_R} \right). \label{PreFGnp}
\end{eqnarray}
As follows from the definitions in Table ($\ref{tab}$), the initial values $ \partial^p A_{n}^{(r_1,..., r_5)}(M_P), \ \ n+p \ge 6$, correspond to bare couplings $\rho_{\tilde{n}}(\Lambda_0), \ \tilde{n} \ge 9$, which (after some fine-tuning procedure in order to restore the STI) are related to the ''physical''couplings $c_1, c_2, ...$ of the general gravity action ($\ref{Sgrav}$). Thus we conclude that our ignorance about these couplings results in an effective theory of quantum gravity that is predictive at scales $\Lambda < \Lambda_0$ to an accuray given by eq. ($\ref{PreFGnp}$). 

Let us finally point out one more time that the perturbation theory in the couplings $\lambda_4^R$, $\lambda_5^R$, $\lambda_{6}^0$,  $\lambda_{7}^0$ and $\lambda_{8}^0$ employed in Theorems ($\ref{BoundThG}$) and ($\ref{PreThG}$) makes only sense if the (dimensionless) couplings are small $\grave{a}$ la eqns. ($\ref{smallRG}$) and ($\ref{smallNRG}$).  The same holds true for the perturbation theory in the  couplings $\lambda_4^R$, $\lambda_5^R$ and the deviations $\Delta \lambda_{\tilde{a}}^R, \ {\tilde{a}}=6,...,8$ employed in Theorem ($\ref{invThG}$). If we look at Table ($\ref{tab}$), we see that the couplings $\lambda_{\tilde{a}}, \  {\tilde{a}}=2,...,8$, are associated with products of the cosmological constant $\Lambda_K$ and the gravitational constant $\lambda \sim M_P^{-1}$. Moreover, remember that in eq. ($\ref{smallG}$) we estimated the deviations as $\Delta \lambda_{\tilde{a}}^R \sim \Lambda/M_P$, and that our scale has to stay above the ''mass'' $\sqrt{\Lambda_K}$ as has been pointed out at the beginning of section ($\ref{RGIs}$). Remembering the definitions ($\ref{lgren}$), ($\ref{lgBare}$)  and ($\ref{devRG}$) it seems therefore reasonable to assume that (at least after some restoration of the STI) the couplings remain small on scales $\Lambda$ that are in the range
\begin{equation}
\sqrt{\Lambda_K} < \Lambda < M_P
\end{equation}
which means $10^{-42} GeV < \Lambda < 10^{19} GeV$. However, as has been dicussed in eq. ($\ref{smallg}$) of section ($\ref{RGIs}$), also the coupling $\rho_1 \sim \Lambda_K/\lambda$ has to remain small in order to prove Theorems ($\ref{BoundThG}$)- ($\ref{PreThG}$).  This would suggest a lower bound of $\Lambda > \Lambda_K M_P \sim 10^{-21} GeV$. Note that the latter is probably an artefact\footnote{See the discussion in section ($\ref{propsec}$).} of using a flat Euclidean background which does not fullfill the field equations for $\Lambda_K \ne 0$.

The renormalization of the extended effective potential $\tilde{L} \big(h,C,\overline{C}, \beta, \tau, \Lambda \big) $ involving couplings of the sources $\beta$, $\tau$ to the nonlinear BRS variations ($\ref{BRST1}$) and ($\ref{BRST2}$) will be discussed in the next section.

\end{subsection}

\begin{subsection}{Restoration of the Slavnov-Taylor identities for effective quantum gravity}  \label{STIRest}
As has been mentioned on various occasions, the introduction of the momentum space cutoff regularization violates the gauge (BRS) invariance of the (total) gravity action ($\ref{Sgrav}$) and leads to violated Slavnov-Taylor-identities (vSTI). We will now derive the vSTI and then discuss how they can be restored for effective quantum gravity with \textit{finite accuracy}.

To do so, we will have to consider the renormalization of vertex functions with \textit{one} insertion of a composite field. Some key elements of this subject (employing the method of flow equations) have been summarized in Appendix ($\ref{OI}$), which we will repeatedly refer to. 

Let us begin by introducing bare versions of the composite BRS fields ($\ref{BRST1}$) and ($\ref{BRST2}$) according to the rules discussed in Appendix ($\ref{OI}$): 
\begin{eqnarray}
\Psi^{\mu \nu}(x, \Lambda_0) &:=& R^0_1 \delta^{\mu \nu} \partial_\rho C^\rho - R^0_2 \big( \delta^{\rho \nu} \partial_\rho  C^\mu + \delta^{\mu \rho} \partial_\rho C^\nu \big) \nonumber \\ && \ \ + \ R^0_3  \partial_\rho \big(C^\rho  h^{\mu \nu} \big)   - R^0_4  \big( h^{\rho \nu} \partial_\rho  C^\mu + h^{\mu \rho} \partial_\rho C^\nu \big)  \label{BRST1_0} \\
\Omega^\mu(x, \Lambda_0)   &:=& R^0_5  C^\nu \partial_\nu C^\mu  \label{BRST2_0}
\end{eqnarray}
where the $R^0_{i}$ are some bare coupling constants\footnote{We will introduce the couplings $R_i$ more systematically later.} having canonical dimensions
\begin{eqnarray}
D_{R^0_i} &=& 0, \ \ \ i=1,2 \label{DR1_0} \\
D_{R^0_i} &=& -1, \ \ \ i=3...5 .   \label{DR2_0}
\end{eqnarray}
For the definition of ($\ref{BRST1_0}$) we have made use of eq. ($\ref{lg}$).  The BRS fields $\Psi^{\mu \nu}(x)$ and $\Omega^\mu(x) $ have ghost numbers $1$ and $2$ respectively, and their canonical dimensions are
\begin{equation}
D_{\Psi}=D_{\Omega}=2.  \label{DBRS}
\end{equation}
Note that there is again no (bare) insertion for the linear BRS variation ($\ref{BRST3}$) of the antighost field, because the STI ($\ref{STW}$) can be established by generating this variation through functional derivation with respect to the source field $t_{\mu \nu}$.    

By virtue of the bare BRS variations ($\ref{BRST1_0}$)  and ($\ref{BRST2_0}$), the extended effective potential ($\ref{LGravEG}$) can now be written as
\begin{eqnarray}
\tilde{L} \big(\Phi, \beta, \tau, \Lambda_0 \big)  = {L} \big(\Phi,  \Lambda_0 \big) +  \int d^4 x \left( \beta_{\mu \nu} \Psi^{\mu \nu}  + \tau_\mu  \Omega^\mu  \right)   \label{LGravEG2}
\end{eqnarray}
where the notation ${\Phi}= \left(h_{\mu \nu}, C^\mu, \overline{C}_\mu \right)$ for the graviton and ghost fields has been employed, and the potential ${L} \big(\Phi,  \Lambda_0 \big)$ has been defined in eq. ($\ref{LGravG}$). Moreover, introducing the quantity 
\begin{eqnarray}
Q(\Phi, \Lambda_0) :=  \int_x A  \frac{(2 \pi)^4}{2}  h   - \frac{1}{2} \langle h^{\mu \nu}, \Delta^{\Lambda_0 \ -1}_{\mu \nu \rho \sigma} h^{\rho \sigma} \rangle  -  \langle \overline{C}^\mu, \Delta_{GH \mu \nu}^{\Lambda_0 \ -1} {C}^\nu    \rangle  \label{Q}
\end{eqnarray}
the (extended) total bare gravity action ($\ref{StotdeRegG}$) can be compactly expressed in terms of $L$ and $Q$:
\begin{eqnarray}
{S}_{tot}({\Phi, \Lambda_0}) &=& Q(\Phi, \Lambda_0) +  {L} \big(\Phi,  \Lambda_0 \big) \label{slq1}\\
\tilde{S}_{tot}({\Phi, \beta, \tau, \Lambda_0}) &=& Q(\Phi, \Lambda_0) + \tilde{L} \big(\Phi, \beta, \tau, \Lambda_0 \big).
\end{eqnarray}
Finally, let us define \textit{regularized} bare BRS variations of the fields:
\begin{eqnarray}
\delta^{\Lambda_0}_\epsilon h^{\mu \nu} &:=& ( K^{\Lambda_0} \Psi^{\mu \nu} ) (x) \epsilon  \label{BRST1_R} \\
\delta^{\Lambda_0}_\epsilon  C^\mu   &:=& (K^{\Lambda_0} \Omega^\mu)(x)  \epsilon \label{BRST2_R} \\
\delta^{\Lambda_0}_\epsilon \overline{C}_\mu &:=& (K^{\Lambda_0} \xi^{-1} F_\mu)(x) \epsilon  \label{BRST3_R}
\end{eqnarray}
where $K^{\Lambda_0}:= K(\partial^2/\Lambda^2_0)$ is the (position-space) cutoff function introduced in eq. ($\ref{cutf}$), and $\xi^{-1} F_\mu= \xi^{-1} \partial^\sigma h_{\mu \sigma}$ refers to the gauge fixing, see eq. ($\ref{GFH}$). The motivation for employing regularized BRS variations will become clear later. 

Denoting the sources of the graviton and ghost fields by ${J} = \{ t_{\mu \nu},\sigma_\mu , \overline{\sigma}^\mu \}$, the violated Slavnov-Taylor identities now follow from the requirement that the regularized generating functional of quantum gravity,
\begin{eqnarray}
W({J}; {\Lambda_0}) = \int \mathcal{D} \Phi \ e^{ \left. S_{tot}(\Phi, {\Lambda_0}) + \langle \Phi, J \rangle \right.} ,
\end{eqnarray}
be invariant\footnote{This amounts to the statement that the Jacobian of the BRS transformation is equal to $1$.} under the BRS variations ($\ref{BRST1_R}$)-($\ref{BRST3_R}$) of the fields:
\begin{eqnarray}
0 \stackrel{!}{=} \int \mathcal{D} \Phi \ e^{ \left. S_{tot}(\Phi, {\Lambda_0}) + \langle \Phi, J \rangle \right.} \left( \delta^{\Lambda_0}_\epsilon \langle \Phi, J \rangle + \delta^{\Lambda_0}_\epsilon  S_{tot}(\Phi, {\Lambda_0})  \right) .   \label{vSTI1}
\end{eqnarray}
Note that the BRS variation of the source term $\delta^{\Lambda_0}_\epsilon \langle \Phi, J \rangle$ gives rise to ''conventional'' STI ($\ref{STW}$), whereas $\delta^{\Lambda_0}_\epsilon  S_{tot}(\Phi, {\Lambda_0})$ contains contributions stemming from the cutoff regularized inverse propagators appearing in ($\ref{Q}$) as well as from the BRS violating counterterms that had to be included in the bare potential ($\ref{LGravEG2}$). Let us analyze $\delta^{\Lambda_0}_\epsilon  S_{tot}(\Phi, {\Lambda_0})$ in some more detail, employing the ''split-up'' ($\ref{slq1}$). The BRS variation of the ''free'' part ($\ref{Q}$) is
\begin{eqnarray}
\hspace{-1cm} \delta^{\Lambda_0}_\epsilon  Q(\Phi, \Lambda_0) &=& \int_x \Big( A  \frac{(2 \pi)^4}{2} K^{\Lambda_0} \Psi^{\mu}_\mu   - \frac{1}{2} \langle h^{\mu \nu}, \Delta^{-1}_{\mu \nu \rho \sigma} \Psi^{\rho \sigma} \rangle  \nonumber \\ && \hspace{2cm} -  \langle \overline{C}^\mu, \Delta_{GH \mu \nu}^{ -1} \Omega^\nu    \rangle  -  \langle \xi^{-1} F_\mu, \Delta_{GH \mu \nu}^{ -1} {C}^\nu     \rangle \Big) \epsilon .   \label{Qvar}
\end{eqnarray}
Observe that in the terms bilinear in the fields, the cutoff function in the regularized BRS variations ($\ref{BRST1_R}$)-($\ref{BRST3_R}$) has cancelled its inverse appearing in the inverted propagators. On the other hand, the part linear in the field as well as the BRS variation of the bare effective potential $\delta^{\Lambda_0}_\epsilon {L} \big(\Phi,  \Lambda_0 \big)$ contain the cutoff function and will therefore be nonpolynomial in the field derivatives\footnote{Concerning the linear term, the nonpolynomial parts will go away because they are total derivatives.}. However, all nonpolynomial parts will be multiplied with powers of $\Lambda_0^{-2}$.

The vSTI ($\ref{vSTI1}$) can be rewritten as follows. We employ an extended regularized generating functional of quantum gravity,
\begin{eqnarray}
\tilde{W}({J}, \beta ,  \tau; \Lambda_0) = \int \mathcal{D} \Phi \ e^{ \left. \tilde{S}_{tot}(\Phi, \beta ,  \tau, {\Lambda_0}) +  \langle \Phi, J \rangle   \right.},
\end{eqnarray}
and define a regularized version of the BRS operator ($\ref{BRSO}$):
\begin{eqnarray}
\mathcal{D}^{\Lambda_0}:=  \langle t_{\mu \nu}, K^{\Lambda_0} \frac{\delta }{\delta \beta_{\mu \nu}} \rangle + \langle \overline{\sigma}_\mu, K^{\Lambda_0} \frac{\delta }{\delta \tau_{\mu}} \rangle + \langle  \sigma^\mu  ,  \xi^{-1} K^{\Lambda_0} F_{\mu \rho \sigma} \Big( \frac{\delta }{\delta t_{\rho \sigma}} \Big)\rangle . \label{BRSOR}
\end{eqnarray}
Here, $K^{\Lambda_0}$ is again the cutoff function. Moreover, we treat the BRS variation of the bare action as a space-time integrated operator insertion:
\begin{eqnarray}
L_{(1)}(\Phi, \Lambda_0) \epsilon &:=& \delta^{\Lambda_0}_\epsilon  S_{tot}(\Phi, {\Lambda_0}). \label{L1}
\end{eqnarray}
From the properties of the BRS fields ($\ref{BRST1_0}$)  and ($\ref{BRST2_0}$) and from looking at eq. ($\ref{Qvar}$) it follows that the operator insertion $L_{(1)}(\Phi, \Lambda_0)$ has ghost number $1$ and canonical dimension 
\begin{eqnarray}
D_{L_{(1)}}=5   \label{DL1G}
\end{eqnarray}
in the sense of space-time integrated operator insertions as it is discussed at the end of Appendix ($\ref{OI}$).

By introducing a modified functional
\begin{eqnarray}
W_\chi({J}; {\Lambda_0}) :=  \int \mathcal{D} \Phi \ e^{ \left. S_{tot}(\Phi, {\Lambda_0}) + \chi L_{(1)}(\Phi, \Lambda_0) + \langle \Phi, J \rangle \right.}  \label{Wmod}
\end{eqnarray}
with $\chi \in \mathbb{R}$, the violated Slavnov-Taylor identities ($\ref{vSTI1}$) now take the form
\begin{eqnarray}
\mathcal{D}^{\Lambda_0} \tilde{W}({J}, \beta ,  \tau)|_{\beta=\tau=0} = \frac{d}{d \chi} W_\chi({J}; {\Lambda_0})|_{\chi=0}  . \label{vSTI2}
\end{eqnarray}
From the viewpoint of the analysis with flow equations, the central object of interest is the effective potential $L$, rather than the generating functional of Greens functions $W$. We will therefore show how the vSTI  ($\ref{vSTI2}$) can be expressed in terms of $L$.

Let $\Lambda< \Lambda_0$ denote some scale. Consider the ''running'' extended effective potential 
\begin{equation}
\tilde{L} \big(\Phi, \beta, \tau, \Lambda, \Lambda_0 \big)  \label{LE1}
\end{equation}
that can be seen as a solution of Polchinski's equation for quantum gravity ($\ref{polgravps}$) with initial conditions ($\ref{LGravEG2}$).  In complete analogy, we may introduce a running potential 
\begin{equation}
{L}_\chi \big(\Phi, \Lambda, \Lambda_0 \big) \label{LC1}
\end{equation}
by employing initial conditions ${L} \big(\Phi, \Lambda_0 \big) + \chi L_{(1)}(\Phi, \Lambda_0)$, see eq. ($\ref{Wmod}$). In order to reduce the notational complexity of the upcoming equations, let us also define the following abbreviations:
\begin{eqnarray}
L (\Lambda)& \equiv & \tilde{L} \big(\Phi, \beta, \tau, \Lambda, \Lambda_0 \big)|_{\beta=\tau=0} \label{abbr1} \\
L_\beta^{\mu \nu} (x,\Lambda) & \equiv &  \frac{\delta }{\delta \beta_{\mu \nu}(x)}  \tilde{L} \big(\Phi, \beta, \tau, \Lambda, \Lambda_0 \big)|_{\beta=\tau=0} \\
L_\tau^\mu (x,\Lambda) & \equiv &  \frac{\delta }{\delta \tau_{\mu}(x)}  \tilde{L} \big(\Phi, \beta, \tau, \Lambda, \Lambda_0 \big)|_{\beta=\tau=0} \label{abbr3} \\
L_{(1)} (\Lambda) & \equiv & \frac{d}{d \chi} {L}_\chi \big(\Phi, \Lambda, \Lambda_0 \big)|_{\chi=0} . \label{abbr4}
\end{eqnarray}
We are now ready to give the vSTI in terms of the effective potential.

\begin{satz}[Violated Slavnov-Taylor Identities]  \label{vSTITh}
Employing the notations ($\ref{abbr1}$)-($\ref{abbr4}$), the violated Slavnov-Taylor identities for quantum gravity are
\begin{eqnarray}
 && \hspace{-1cm} \big\langle h^{\mu \nu}, \Delta^{ -1}_{\mu \nu \rho \sigma} L_\beta^{\rho \sigma}(0) \big\rangle  +  \big\langle \overline{C}^\mu, \Delta_{GH \mu \nu}^{ -1} L_\tau^\nu(0)  \big\rangle  +   \xi^{-1} \big\langle {C}^\mu, \Delta_{GH \mu \nu}^{ -1}  F^{\nu \rho \sigma} \big( h_{\rho \sigma}  \big)  \big\rangle  \nonumber \\ && \hspace{3.5cm} + \ \xi^{-1} \big\langle {C}^\mu, \Delta_{GH \mu \nu}^{ -1} F^{\nu \rho \sigma} \Big( \Delta^{\Lambda_0}_{\rho \sigma \alpha \beta} \frac{\delta L(0)}{\delta h_{\alpha \beta}} \Big)  \big\rangle \ = \ L_{(1)}(0) \ \  \nonumber \\ \label{vSTIL}
\end{eqnarray} 
\vspace{-1ex}
where 
\begin{equation}
\Big( \Delta^{\Lambda_0}_{\rho \sigma \alpha \beta} \frac{\delta L(0)}{\delta h_{\alpha \beta}} \Big) \big(x \big) \equiv \int_y   \Delta^{\Lambda_0}_{\rho \sigma \alpha \beta} (x-y) \frac{\delta L(0)}{\delta h_{\alpha \beta} (y)}.
\end{equation} 
\end{satz}

\setcounter{proof}{10}
\begin{proof} 
As follows from the discussion in Appendix ($\ref{LZ}$), see in particular eqns. ($\ref{WL}$) and ($\ref{phij}$), the effective potentials ($\ref{LE1}$)  and ($\ref{LC1}$) can be related to the functionals $\tilde{W}({J}, \beta ,  \tau; \Lambda_0)$  and $W_\chi({J}; {\Lambda_0}) $ by 
\begin{eqnarray}
\tilde{W}({J}, \beta ,  \tau; \Lambda_0) = e^{\frac{1}{2} \langle J, \Delta_{\Lambda_0} J \rangle} e^{\tilde{L} \big(\Phi, \beta, \tau, 0, \Lambda_0 \big)} \label{WLrel1} \\
{W}_\chi({J}; \Lambda_0) = e^{\frac{1}{2} \langle J, \Delta_{\Lambda_0} J \rangle} e^{\tilde{L}_\chi \big(\Phi, 0, \Lambda_0 \big)} \label{WLrel2}
\end{eqnarray}
where
\begin{eqnarray}
\langle J, \Delta_{\Lambda_0} J  \rangle := \langle t^{\mu \nu}, \Delta^{\Lambda_0 }_{\mu \nu \rho \sigma} t^{\rho \sigma}  \rangle + \langle \sigma^\mu, \Delta_{GH \mu \nu}^{\Lambda_0 } \overline{\sigma}^\nu    \rangle
\end{eqnarray}
and the fields $\Phi$ are understood to be given in terms of their respective sources $J$,
\begin{eqnarray}
h_{\mu \nu}(x) &=& \int_y  \Delta^{\Lambda_0}_{\mu \nu \rho \sigma}(x-y) t^{\rho \sigma}(y) \label{phij1} \\
C_\mu (x)&=& \int_y \Delta_{GH \mu \nu}^{\Lambda_0}(x-y) \sigma^\nu (y) \\
\overline{C}_\mu (x) &=& \int_y \Delta_{GH \mu \nu}^{\Lambda_0}(x-y) \overline{\sigma}^\nu (y).  \label{phij3}
\end{eqnarray}
The vSTI for the effective potential can now be derived by plugging eqns. ($\ref{WLrel1}$) and ($\ref{WLrel2}$) into ($\ref{vSTI2}$) and reexpressing the resulting equation for the sources $J$ in terms of the fields $\Phi$. The latter is done by inverting the relations ($\ref{phij1}$)-($\ref{phij3}$).
\begin{flushright}
$\Box$
\end{flushright}
\end{proof}
In the following, the violated Slavnov-Taylor identities will be analyzed. At first we note that the potentials ($\ref{abbr1}$)-($\ref{abbr4}$) appear in eq. ($\ref{vSTIL}$)  with ''floating'' cutoff $\Lambda=0$. This implies a relation between the the effective potentials $L$ and the generating functionals $W$ $\grave{a}$ la eqns. ($\ref{WLrel1}$) and ($\ref{WLrel2}$), which in particular allows for external momenta $k_i^2 > 0$ in the vertex functions $L_n(k_i,0, \Lambda_0)$. However, so far in this work we have employed the perception that we probe the physics only below $\Lambda$, corresponding to a relation ($\ref{WP}$) between $L$ and $W$ and  external momenta $k_i^2 \le \Lambda^2$. In addition, problems may arise when $\Lambda^2 < \Lambda_K$, see the discussion at the end of section ($\ref{propsec}$). We will therefore analyze the functionals appearing in the vSTI ($\ref{vSTIL}$) for nonvanishing floating cutoff $\Lambda$ and keep the latter above the external momenta $k_i$ as well as the (at this point arbitrary\footnote{After the restoration of the STI has been accomplished, $B_1$ will be related to $\Lambda_K$.}) ''mass'' squares $B_1, B_2$ appearing in the graviton and ghost propagators ($\ref{gravpropG}$) and ($\ref{ghpropG}$):
\begin{eqnarray}
\Lambda^2 &>& k_i^2, B_1, B_2 .
\end{eqnarray}
The $\Lambda$-dependent results of this analysis will then (at least qualitatively) be related to the functionals appearing in eq. ($\ref{vSTIL}$) by replacing $\Lambda^2$ with the mass squares\footnote{Modulo problems that may arise with ''wrong sign'' mass squares such as $B_1 \sim \Lambda_K$ for $\Lambda_K > 0$.} $B_1, B_2$ and the external momenta $k_i^2$ respectively. This is suggested by analogous results\footnote{See in particular eqns. $4.52$ and $4.117$ of \cite{Mull} in the limit $\Lambda \rightarrow 0$.} for Yang-Mills Theory \cite{Mull}.

We begin with the LHS of eq.  ($\ref{vSTIL}$) and notice that the three functionals $L(\Lambda)$, $L_\beta^{\mu \nu}(\Lambda)$ and $L_\tau^\mu(\Lambda)$ defined in eqns.  ($\ref{abbr1}$)-($\ref{abbr3}$) appear.

In Theorem ($\ref{BoundThG}$) of the last section, we have already established bounds for the (dimensionless) vertex functions $A_{n} (k_1,...,k_{n}, \Lambda, \Lambda_0) $ of the functional $L(\Lambda)$. In the discussion following Theorem ($\ref{invThG}$) we then have argued that the UV cutoff of effective quantum gravity should be the Planck scale, $\Lambda_0=M_P$. Finally, improvement conditions have been introduced in eq. ($\ref{rcG3}$) and the dependence of the vertex functions on the unknown initial conditions ($\ref{shapeANRG}$) has been analyzed in Theorem ($\ref{PreThG}$). The result, summarized in eq. ($\ref{PreFGnp}$), was that at some scale $\Lambda<M_P$ the vertex functions $A_{n}(\Lambda)$ (and therefore the functional $L(\Lambda)$) are known to an accuracy of $(\Lambda/M_P)^2$.

In order to establish these bounds for the entire LHS of  eq. ($\ref{vSTIL}$), a similar analysis has to be performed for the generating functionals $L_\beta^{\mu \nu}(x,\Lambda)$ and $L_\tau^\mu(x,\Lambda)$ of the vertex functions carrying one BRS field $\Psi^{\mu \nu}(x)$ or $\Omega^\mu(x)$ as  operator insertion. The necessary steps to do so have been developed in Appendix ($\ref{OI}$), leading to Theorems  ($\ref{BoundThOI}$) and  ($\ref{PreThI}$) concerning bounds for vertex functions with one operator insertion and their dependence on the unknown initial conditions.

Let us therefore apply these theorems to the functionals $L_\beta^{\mu \nu}(x,\Lambda)$ and $L_\tau^\mu(x,\Lambda)$. We introduce the Fourier transforms
\begin{eqnarray}
L_\beta^{\mu \nu}(q,\Lambda) := \int_x e^{iqx} L_\beta^{\mu \nu}(x, \Lambda)  \\
L_\tau^\mu(q, \Lambda) := \int_x e^{iqx} L_\tau^\mu(x, \Lambda)
\end{eqnarray}
and define momentum-space vertex functions
\begin{eqnarray}
\delta^4(q+ k_1 + ... + k_n ) L_{\beta n}(q, k_1, ..., k_n, \Lambda) \ = \ (2 \pi )^{4n} \delta^{(n)}_{\hat{\Phi}} L_\beta(q, \Phi, \Lambda) \big|_{\Phi=0} \  \label{LBetaExp} \\ \delta^4(q+ k_1 + ... + k_n ) L_{\tau n}(q, k_1, ..., k_n, \Lambda) \ = \ (2 \pi )^{4n} \delta^{(n)}_{\hat{\Phi}} L_\tau(q, \Phi, \Lambda) \big|_{\Phi=0} . \label{LTauExp}
\end{eqnarray}
Here, the schematic notation introduced in eq. ($\ref{LexpG}$) has been employed. As has been pointed out in eq. ($\ref{DBRS}$), the canonical dimension of the BRS composite fields $\Psi^{\mu \nu}(x)$ and $\Omega^\mu(x) $ equals $2$.  Hence, the canonical dimension of the momentum-space vertex functions $L_{\beta n}$ and $L_{\tau n}$ is 
\begin{eqnarray}
D_{L_{\beta n}} = D_{L_{\tau n}} \ = \ 2-n  \label{Dbetatau}
\end{eqnarray}
in accordance with eq. ($\ref{DLnOI}$) of Appendix  ($\ref{OI}$).

Note that the (suppressed) tensor structure of eqns. ($\ref{LBetaExp}$) and ($\ref{LTauExp}$), as well as the ghost numbers $+1$ and $+2$ associated with the functionals $L_\beta^{\mu \nu}(q,\Lambda)$ and $L_\tau^\mu(q, \Lambda)$, restrain the number of nonvanishing vertex functions $L_{\beta n}$ and $L_{\tau n}$. Introducing running coupling constants $R_i(\Lambda)$ as we have done in eqns. ($\ref{rcc0OI}$)- ($\ref{rcc3OI}$), it therefore turns out that all couplings having canonical dimension  $D_{R_i} \ge -1$ are given by the following (schematic) definitions:
\begin{eqnarray}
 R_i(\Lambda) &:=&  \partial_{1} \Big( \frac{\delta}{\delta C (k_1)} L_\beta(q,\Lambda) \Big)\big|_{C=k_1=q=0}, \ \ \ i=1,2 \nonumber\\  R_i(\Lambda) &:=&  \partial_{j} \Big( \frac{\delta}{\delta C(k_1)} \frac{\delta}{\delta h (k_2)} L_\beta(q,\Lambda) \Big)\big|_{h=C=k_j=q=0} , \ \ \ i=3,4  \nonumber  \\  R_i(\Lambda) &:=&   \partial_{j} \Big( \frac{\delta}{\delta C(k_1)} \frac{\delta}{\delta C (k_2)} L_\tau(q,\Lambda) \Big)\big|_{C=k_j=q=0}, \ \ \ i=5   .  \nonumber  \\
  \label{rcR}
\end{eqnarray}
The index $i$ stems from the different possibilities to contract the (suppressed) tensor structure on the RHS of eqns. ($\ref{rcR}$).  

The running couplings $R_i(\Lambda), \ i=1...5,$ are of course the pendants to the bare couplings $R_i^0, \ i=1...5,$ appearing in eqns. ($\ref{BRST1_0}$)  and ($\ref{BRST2_0}$). This is underlined by the fact that if we look at eqns. ($\ref{rcR}$) and ($\ref{Dbetatau}$), we see that the canonical dimensions of the running couplings $R_i(\Lambda)$ are 
\begin{eqnarray}
D_{R_i} &=& 0, \ \ \ i=1,2 \label{DR1} \\
D_{R_i} &=& -1, \ \ \ i=3...5   \label{DR2}
\end{eqnarray}
in agreement with the previous eqns. ($\ref{DR1_0}$) and ($\ref{DR2_0}$). Note also that eq. ($\ref{DR2}$)  means that the BRS transformations ($\ref{BRST1_R}$)-($\ref{BRST3_R}$) for gravity contain \textit{nonrenormalizable parts}. This is an important difference to the Yang-Mills case \cite{KM}.

Proceeding as in Appendix ($\ref{OI}$), we define renormalization conditions for the couplings having canonical dimensions $D_{R_i} = 0$,
\begin{eqnarray} 
R_i^R := R_i(\Lambda_R), \ \ \ i=1,2  \label{rccOIG}
\end{eqnarray}
where $\Lambda_R < M_P$ is some renormalization scale, and initial conditions for the couplings having $D_{R_i} = -1$:
\begin{eqnarray} 
R_i^0 := R_i(M_P), \ \ \ i=3...5 . \label{iniOIG}
\end{eqnarray}
Moreover, we introduce dimensionless couplings 
\begin{eqnarray} 
\mathcal{R}_i^R (\Lambda) &=& R_i^R, \ \ \ i=1,2   \label{ROIG1} \\
\mathcal{R}_i^0 (\Lambda) &=& \Lambda R_i^0,   \ \ \ i=3...5   \label{ROIG2}
\end{eqnarray}
and impose as an additional constraint to the renormalization and initial conditions that ($\ref{ROIG1}$), ($\ref{ROIG2}$) be small for $\Lambda_R< \Lambda < M_P$:
\begin{eqnarray} 
\mathcal{R}_i^R(\Lambda)  &\le& 1, \ \ \ i=1,2  \label{smallOIGa} \\
\mathcal{R}_i^0(\Lambda)  &\le& 1, \ \ \ i=3...5.  \label{smallOIGb}
\end{eqnarray}
Note that because of the definition ($\ref{ROIG2}$), this in particular implies ${R}_i^0 \le M_P^{-1}, \ i=3...5$, for $\Lambda \le M_P$ and thus
\begin{equation}
\mathcal{R}_i^0 (\Lambda)  \le {\Lambda}/{M_P}, \ \ \ \ i=3...5.  \label{smallOIGExp}
\end{equation}
As always, we will consider dimensionless vertex functions
\begin{eqnarray}
A_{\beta n}(\Lambda) := \Lambda^{n-2} L_{\beta n}, \ \ \ \ A_{\tau n}(\Lambda) := \Lambda^{n-2} L_{\tau n}
\end{eqnarray}
and expand them in the dimensionless renormalized renormalizable couplings $\mathcal{R}_1^R$, $\mathcal{R}_2^R$, $\lambda_4^R$ and $\lambda_5^R$ defined in eqns.  ($\ref{ROIG1}$), ($\ref{lgren}$) and the bare nonrenormalizable couplings $\mathcal{R}_i^0, \ i=3...5,$ and $\lambda_{\tilde{a}}^0, \  \tilde{a} =6...8$ of eqns. ($\ref{ROIG2}$), ($\ref{lgBare}$). For $A_{\beta n}(\Lambda)$ we thus obtain\footnote{Please refer to eq. ($\ref{Apert2I_0}$) of Appendix ($\ref{OI}$) for a definition of the coefficients $A_{\beta n}^{(i, r_1,...,r_5)} (q,k_1,...,k_{n}, \Lambda)$.}
\begin{eqnarray}
A_{\beta n} (q,k_1,...,k_{n}, \Lambda) &=& \sum_{r_1, ...,r_{5}=0}^{\infty}  \ \sum_{i=1}^4 \mathcal{R}_i^{R/0} \ (\lambda_4^R)^{r_1} (\lambda_5^R)^{r_2} \nonumber \\  &&  \hspace{2.5cm}  (\lambda_{6}^0)^{r_3} (\lambda_{7}^0)^{r_4} (\lambda_{8}^0)^{r_5}  A_{\beta n}^{(i, r_1,...,r_5)} (q,k_1,...,k_{n}, \Lambda) \nonumber \\  \label{GenPertIG}
\end{eqnarray}
where we have employed the notation
\begin{eqnarray}
\mathcal{R}_i^{R/0} := \left\{ \begin{array}{l} \mathcal{R}_i^R, \ \ \  i=1,2 \\ \mathcal{R}_i^0, \ \ \ \ i=3,4 \ . \end{array} \right. \label{R_R0}
\end{eqnarray}
Hence each graph contributing to ($\ref{GenPertIG}$) contains \textit{one} extra vertex associated with a renormalized coupling constant $\mathcal{R}_1^R$, $\mathcal{R}_2^R$ or a bare one $\mathcal{R}_3^0$, $\mathcal{R}_4^0$ because of the BRS operator insertion.  Remember also that the couplings $\lambda_{\tilde{a}}$ have to be understood as $k_{\tilde{a}}$-tuples, and recall the notation conventions of eqns. ($\ref{tupord}$) and ($\ref{tupexp}$). 

The procedure for $A_{\tau n}(\Lambda)$ goes in complete analogy, with the exception that the extra vertex is always associated with the bare nonrenormalizable coupling $\mathcal{R}_5^0$.

Finally, we introduce symbols $\Theta_i^R$ and  $\Theta_i^0$ as we have done it in Appendix ($\ref{OI}$):
\begin{eqnarray}
i=1,2: && \Theta_i^R:=1 \  \wedge \  \Theta_i^0:=0  \label{thet1} \\ 
i=3...5:&&   \Theta_i^R:=0  \  \wedge \  \Theta_i^0:=1.  \label{thet2} 
\end{eqnarray}
We are now ready to establish bounds for the vertex functions  $A_{\beta n}$ by applying Theorem ($\ref{BoundThOI}$) of Appendix ($\ref{OI}$).

\begin{satz}[Boundedness of Vertex Functions with BRS Insertion] \label{BoundThOIG}
Let \\ $\Lambda_R \le \Lambda \le M_P$. Given Theorem ($\ref{BoundThG}$), the renormalization conditions ($\ref{rccOIG}$), the initial conditions ($\ref{iniOIG}$) and assuming that
\begin{eqnarray}
||\partial^p A_{\beta n}^{(i, r_1,..., r_5)}(q, p_1,...,p_{n}, \Lambda_0)|| \le M_P^{-p}  \left( \frac{M_P}{\Lambda_R} \right)^{|r_1|} Pln \left( \frac{M_P}{\Lambda_R} \right) \label{iniNRIG}
\end{eqnarray} 
for $n+p \ge 4$, to order $r_1,..., r_5$ in perturbation theory in $\lambda_4^R$, $\lambda_5^R$, $\lambda_{6}^0$,  $\lambda_{7}^0$ and $\lambda_{8}^0$
\begin{eqnarray}
&& ||\partial^p A_{\beta n}^{(i,r_1,..., r_5)}(p_1,...,p_{n}, \Lambda)|| \nonumber \\ &&  \qquad \qquad   \le \Lambda^{-p} \left( \frac{\Lambda}{\Lambda_R} \right)^{r_1} \left( \frac{M_P}{\Lambda} \right)^{r_{NR} + \Theta_i^0}   \left( \Theta_i^R \ Pln\left( \frac{\Lambda}{\Lambda_R} \right) +  \frac{\Lambda}{M_P}  Pln \left( \frac{M_P}{\Lambda_R} \right) \right) \nonumber \\     \label{BoundNRIG}
\end{eqnarray}
where the index $i$ refers to an extra vertex associated with a renormalized coupling ${R}_i^R, \ i=1,2$ having canonical dimension $D_{R_i} = 0$ ($\Theta_i^R=1$, $\Theta_i^0=0$) or a bare one ${R}_i^0, \ i=3,4$ having canonical dimension $D_{R_i} = -1$ ($\Theta_i^0=1$, $\Theta_i^R=0$), and $r_{NR} = |r_3| + |r_4| + |r_5|$.
\end{satz}
The Theorem concerning boundedness of the vertex functions $A_{\tau n}$ is the same with the exception that the extra vertex is always associated with the bare nonrenormalizable coupling ${R}_5^0$. Thus, we do not need the renormalization conditions ($\ref{rccOIG}$) and we only have the case $\Theta_i^0=1$, $\Theta_i^R=0$.

Let us introduce \textit{improvement conditions} for the nonrenormalizable couplings $R_i, \ i=3...5$:
\begin{eqnarray} 
R_i^{NR} := R_i(\Lambda_R), \ \ i=3...5. \label{ImproOIG}
\end{eqnarray}
It is assumed that the improvement conditions ($\ref{ImproOIG}$) are taken such that they are compatible with small initial conditions ($\ref{iniOIG}$), where small is meant in the sense of eq. ($\ref{smallOIGb}$). By the arguments given in the last section, this will amount to the requirement that the improvement conditions should be (not larger than) of the order of the inverse Planck scale,
\begin{eqnarray} 
R_i^{NR} \sim M_P^{-1}.
\end{eqnarray}
Considering the renormalized BRS variations ($\ref{BRST1}$) and ($\ref{BRST2}$) (after some restoration of the STI), it seems reasonable that this will be so because we observe that effectively, the $R_i^{NR}$ become replaced with the gravitational constant $\lambda$.

In order to investigate the dependence of the vertex functions $A_{\beta n}(\Lambda)$ and $A_{\tau n}(\Lambda)$ on the unknown initial conditions, we introduce the parametrization
\begin{eqnarray}
 t \ \partial^p A_{\beta n}^{(i,r_1, ..., r_5)}( M_P), \ \ \ t \in [0,1], \ \ \ n+p \ge 4   \label{shapeANRGI}
\end{eqnarray}
and similarly for $A_{\tau n}^{(i,r_1, ..., r_5)}( M_P)$. This leads to $t$-dependend ''running'' vertex functions $A_{\beta n}(\Lambda, t)$ and $A_{\tau n}(\Lambda, t)$, and we may apply Theorem ($\ref{PreThI}$) of Appendix ($\ref{OI}$):

\begin{satz}[Predictivity of Quantum Gravity with BRS Insertion] \label{PreThIG} Let \\ there be renormalization conditions  ($\ref{rccOIG}$) and  improvement conditions ($\ref{ImproOIG}$). Assume that to order $r_1, ..., r_5$ in perturbation theory in $\lambda_4^R$, $\lambda_5^R$, $\lambda_{6}^0$,  $\lambda_{7}^0$ and $\lambda_{8}^0$
\begin{eqnarray}
||\partial^p A_{\beta n}^{(i,r_1,..., r_5)}(q, p_1,...,p_{n}, \Lambda)|| \le \Lambda^{-p}  \left( \frac{\Lambda}{\Lambda_R} \right)^{|r_1|}  \left( \frac{M_P}{\Lambda} \right)^{r_{NR}+ \Theta_i^0} Pln \left( \frac{M_P}{\Lambda_R} \right), \label{BoundWPIG}
\end{eqnarray}
and that for $n+p \ge 4$
\begin{eqnarray}
|| \frac{d}{d t} \partial^p A_{\beta n}^{(i,r_1,..., r_5)}(q, p_1,...,p_{n}, M_P)|| \le M_P^{-p}  \left( \frac{M_P}{\Lambda_R} \right)^{|r_1|} Pln \left( \frac{M_P}{\Lambda_R} \right) . \label{ini2PIG}
\end{eqnarray}
Given Theorems ($\ref{BoundThG}$),  ($\ref{BoundThOIG}$) and ($\ref{PreThG}$) we then have for $\Lambda_R \le \Lambda \le M_P$
\begin{eqnarray}
|| \frac{d}{d t} \partial^p A_{\beta n}^{(i, r_1,..., r_5)}(q, p_1,...,p_{n}, \Lambda)|| \le \Lambda^{-p}  \left( \frac{\Lambda}{\Lambda_R} \right)^{|r_1|}  \left( \frac{M_P}{\Lambda} \right)^{r_{NR}+\Theta_i^0} \left(\frac{\Lambda}{M_P}\right)^{2} Pln \left( \frac{M_P}{\Lambda_R} \right) \nonumber \\ \label{PreIG}
\end{eqnarray}
where the index $i$ refers to an extra vertex associated with a renormalized coupling ${R}_i^R, \ i=1,2$ having canonical dimension $D_{R_i} = 0$ ($\Theta_i^0=0$) or a bare one ${R}_i^0, \ i=3,4$ having canonical dimension $D_{R_i} = -1$ ($\Theta_i^0=1$), and  $r_{NR} = |r_3| + |r_4| + |r_5|$.
\end{satz}
For the vertex functions $A_{\tau n}(\Lambda, t)$ we may proceed similarly, with the exception that once again the extra vertex is always associated with the bare nonrenormalizable coupling ${R}_5^0$. Thus, we do not need the renormalization conditions ($\ref{rccOIG}$) and we only have the case $\Theta_i^0=1$.

By the arguments following Theorem  ($\ref{PreThG}$), eq.  ($\ref{PreIG}$) means that for small couplings in the sense of eqns. ($\ref{smallRG}$)-($\ref{smallNRGExp}$) and ($\ref{smallOIGa}$)-($\ref{smallOIGExp}$), the ignorance about the initial values ($\ref{shapeANRGI}$) amounts to an indetermination of the running vertex functions $A_{\beta n}(\Lambda)$ and $A_{\tau n}(\Lambda)$ of the order of $(\Lambda/M_P)^2$.

Thus, we have achieved our aim: Theorems ($\ref{PreThG}$) and ($\ref{PreThIG}$) allow us to control the entire LHS of the violated Slavnov-Taylor Identities ($\ref{vSTIL}$), since all three functionals $L(\Lambda)$, $L_\beta^{\mu \nu}(\Lambda)$ and $L_\tau^\mu(\Lambda)$ appearing are now known to an accuracy of $(\Lambda/M_P)^2$.

\medskip
Let us therefore move on and consider the functional $L_{(1)} (\Lambda)$ that forms the RHS of the vSTI ($\ref{vSTIL}$). It has been defined in eq. ($\ref{abbr4}$) as the generating functional of vertex functions carrying the BRS variation of the bare total action ($\ref{L1}$) as a (space-time integrated) operator insertion. Remember that $L_{(1)} (\Lambda)$  has ghost number $1$ and canonical dimension $5$ in the sense of space-time integrated operator insertions discussed at the end of Appendix ($\ref{OI}$). Thus, according to eq. ($\ref{DLnOI}$) the momentum-space vertex functions
\begin{eqnarray}
\delta^4(k_1 + ... + k_n ) L_{(1)n}(k_1, ..., k_n, \Lambda) = (2 \pi )^{4n} \delta^{(n)}_{\hat{\Phi}} L_{(1)}(\Phi, \Lambda) \big|_{\Phi=0}    \label{LexpGIG}
\end{eqnarray}
have canonical dimensions
\begin{eqnarray}
D_{ L_{(1) n} } = 5 -n.  \label{DL1_5}
\end{eqnarray}
We may use the vertex functions ($\ref{LexpGIG}$) to introduce running coupling constants $F_i(\Lambda)$, as we have done it in eqns. ($\ref{rcc0OI}$)- ($\ref{rcc3OI}$) of Appendix ($\ref{OI}$). If we attempt to control the dimensionless counterparts $A_{(1)n}(k_1, ..., k_n, \Lambda)$ of ($\ref{LexpGIG}$) to an accuracy of $(\Lambda/M_P)^2$, we will have to define renormalization and improvement conditions:
\begin{eqnarray}
F_i^R &:=& F_i(\Lambda_R), \ \ \ D_{F_i} \ge 0 \label{RenL1} \\
F_i^{NR} &:=& F_i(\Lambda_R), \ \ \ D_{F_i} = -1.  \label{ImproL1}
\end{eqnarray}
This follows again from the analogous procedure leading to Theorem ($\ref{PreThI}$).

The renormalization and improvement conditions ($\ref{RenL1}$) and ($\ref{ImproL1}$) are associated with position-space local operators $\mathcal{O}_i(h, C, \overline{C})$ having canonical dimensions $D_{\mathcal{O}_i} \le 6$, as follows from a derivative expansion of the functional $L_{(1)} (\Lambda_R)$ $\grave{a}$ la Appendix ($\ref{MDV}$). On the other hand, the \textit{bare} insertion $L_{(1)} (M_P)$ has been defined in eq. ($\ref{L1}$) as the (regularized) BRS variation of the bare total action:
\begin{eqnarray}
L_{(1)}(\Phi, M_P) \epsilon &=& \delta^{M_P}_\epsilon  S_{tot}(\Phi, {M_P}) \ \equiv \ \sum_i F_i^0 \int_x \mathcal{O}_i(h, C, \overline{C}) \epsilon . \label{L1_2}
\end{eqnarray}
If we look at Table ($\ref{tab}$) and remember the definitions of the BRS fields ($\ref{BRST1_0}$)  and ($\ref{BRST2_0}$), we see that the operators $\mathcal{O}_i(h, C, \overline{C})$ having canonical dimensions $D_{} \le 6$ are those that are asscociated with the bare coupling constants $\rho^0_{\tilde{a}}, \  \tilde{a}=1...8$, and $R^0_i, \ i=1...5$, defined in  Table ($\ref{tab}$) and  eqns.  ($\ref{rcR}$). Hence
\begin{eqnarray}
 F_i^0 = F_i^0(\rho^0_{\tilde{a}}, R^0_j), \ \ \ \tilde{a}=1...8, \ j=1...5, \ D_{F^0_i} \ge -1.   \label{FB1}
\end{eqnarray}
However, the bare couplings $\rho^0_{\tilde{a}}, \  \tilde{a}=1...8$, and $R^0_i, \ i=1...5$, are already \textit{implicitly defined} by their renormalized counterparts because of the renormalization and improvement conditions ($\ref{rcG1}$)-($\ref{rcG2}$), ($\ref{rccOIG}$), ($\ref{rcG3}$) and ($\ref{ImproOIG}$):
\begin{eqnarray}
\rho^0_{\tilde{a}} &=& \rho^0_{\tilde{a}} \big( \rho^R_{\tilde{b}}, \rho^{NR}_{\tilde{b}}, \Lambda_R, M_P \big)   \label{FB2}   \\
R^0_i &=& R^0_i \big(R^R_j, R^{NR}_j, \Lambda_R, M_P \big).   \label{FB3} 
\end{eqnarray}
We therefore conclude that while establishing Theorem ($\ref{PreThI}$) for the (dimensionless) vertex functions $A_{(1)n}(k_1, ..., k_n, \Lambda)$ of the functional $L_{(1)} (\Lambda)$, the renormalization and improvement conditions ($\ref{RenL1}$) and ($\ref{ImproL1}$) cannot be chosen freely. In fact, they are (in principle) determined by solving the Polchinski RGE for initial conditions ($\ref{FB1}$). Hence, they must be functions\footnote{In the analogous treatment for Yang-Mills Theory, it is shown that these functions can be explicitly worked out by analyzing the relevant part of an 1PI fuctional $\Gamma_{(1)}$ describing the violation of the STI \cite{Mull}. The functional $\Gamma_{(1)}$ is the 1PI Yang-Mills counterpart to our functional $L_{(1)}$. Note that for effective quatum gravity, one would have to analyze the relevant \textit{and} the least irrelevant parts of a functional $\Gamma_{(1)}$ being the 1PI counterpart to $L_{(1)}$.} $G_i$ of the renormalization and improvement conditions ($\ref{rcG1}$)-($\ref{rcG2}$), ($\ref{rccOIG}$), ($\ref{rcG3}$) and ($\ref{ImproOIG}$):
\begin{eqnarray}
F_i^R &=& G_i(\rho^R_{\tilde{a}}, \rho^{NR}_{\tilde{a}}, R^R_j, R^{NR}_j) \ + \ \mathcal{O}\big( (\Lambda_R/{M_P})^2 \big)  \label{FtR1}    \\
F_i^{NR} &=& G_i(\rho^R_{\tilde{a}}, \rho^{NR}_{\tilde{a}}, R^R_j, R^{NR}_j) \ + \ \mathcal{O}\big( (\Lambda_R/{M_P})^2 \big)  \label{FtR2} 
\end{eqnarray}
where the indetermination that is left in eqns. ($\ref{FtR1}$), ($\ref{FtR2}$) stems once more from the fact that we do not know about the initial conditions
\begin{eqnarray}
\partial^p A_{(1)n}(p_1,...,p_{n}, M_P), \ \ \ n+p \ge 7    \label{inivSTI}
\end{eqnarray}
associated with couplings $ F_i^0, \ \ D_{F^0_i} \le -2$. Note however that for establishing Theorems ($\ref{BoundThOI}$) and ($\ref{PreThI}$) for the vertex functions ($\ref{LexpGIG}$), it is necessary that ($\ref{inivSTI}$) are sufficiently small, see eq. ($\ref{iniNRI}$). This will be the case because in eq. ($\ref{L1_2}$), the bare couplings $\rho^0_{\tilde{n}}, \  \tilde{n} \ge 9$, are small by the reasoning of eq. ($\ref{PlanckNat}$), and the bare couplings $R^0_i, \ i=1...5$, are small according to eq. ($\ref{smallOIGb}$). Furthermore, the nonpolynomial parts in  $\delta^{M_P}_\epsilon {L} \big(\Phi,  M_P \big)$ introduced by the \textit{regularized} BRS variations (see the discussion regarding eq. ($\ref{Qvar}$)) will be also small since they are multiplied by powers of $\Lambda_0^{-2}$.

Let us therefore assume that Theorems ($\ref{BoundThOI}$) and ($\ref{PreThI}$) have been established for the vertex functions of the functional $L_{(1)} (\Lambda)$ while employing renormalization and improvement conditions $\grave{a}$ la ($\ref{FtR1}$) and ($\ref{FtR2}$).

\medskip
We are now ready to discuss the restoration of the violated Slavnov-Taylor identities ($\ref{vSTIL}$). The violation is expressed in terms of the functional $L_{(1)} (\Lambda)$ appearing on the RHS of ($\ref{vSTIL}$), since the LHS of ($\ref{vSTIL}$) effectively describes the regularized BRS operator ($\ref{BRSOR}$) acting on the extended effective potential $\tilde{L} \big(\Phi, \beta, \tau, 0, \Lambda_0 \big)$. Note that ($\ref{BRSOR}$) equals its unregularized counterpart ($\ref{BRSO}$) at scales $\Lambda< \Lambda_0$. This means that ultimatively, the LHS of ($\ref{vSTIL}$) is identical to the the LHS of the STI ($\ref{STW}$). Thus, we have restored the vSTI if we can make $L_{(1)} (\Lambda)$ vanish\footnote{Remember the discussion following Theorem ($\ref{vSTITh}$) concerning the value $\Lambda=0$ of the floating cutoff.}. 

However, since the three functionals $L(\Lambda)$, $L_\beta^{\mu \nu}(\Lambda)$ and $L_\tau^\mu(\Lambda)$ appearing at the LHS of ($\ref{vSTIL}$) are only known to an accuracy of $(\Lambda/M_P)^2$, it will suffice to show that
\begin{eqnarray}
|| L_{(1)} (\Lambda)  || &\le& \left(\frac{\Lambda}{M_P} \right)^2.  \label{L1Re}
\end{eqnarray}
This is the restoration of the vSTI ($\ref{vSTIL}$) to \textit{finite accuracy}. In order to achieve ($\ref{L1Re}$), two steps are necessary.
\begin{itemize}
\item Renormalization and improvement conditions for the \textit{physical} couplings $\Lambda_K$, $\lambda$ and the gravitational field $h_{\mu \nu}$ are specified at some renormalization scale $\Lambda_R < M_P$. To do so, we employ again a schematic notation suppressing all indices. Moreover, we would like to remind the reader of the notations and conventions employed in table ($\ref{tab}$) and eqns. ($\ref{r67}$) and ($\ref{gravpropG}$):
\begin{eqnarray}
\frac{\delta}{\delta h(k_1)} \frac{\delta}{\delta h(k_2)} L(\Phi, \Lambda_R)\big|_{h=k_i=0} &\stackrel{!}{=}&0, \ \ \ \  B_1 \stackrel{!}{=} \Lambda_K \ \nonumber \\
\partial_{i}^2 \left(\frac{\delta}{\delta h(k_1)} \frac{\delta}{\delta h(k_2)} L(\Phi, \Lambda_R)\right) \Big|_{h=k_i=0} &\stackrel{!}{=}& 0 \nonumber \\
\Big(\frac{1}{2} \partial_{i} \partial_{j} -\partial_{i}^2  \Big) \left(\frac{\delta}{\delta h(k_1)} \frac{\delta}{\delta h(k_2)} \frac{\delta}{\delta h(k_3)} L(\Phi, \Lambda_R)\right) \Big|_{h=k_{i}=0} &\stackrel{!}{=}& \lambda \ . \nonumber \\   \label{physrenGK}
\end{eqnarray}

\item \textit{One particular set} of ''arbitrary'' renormalization and improvement conditions ($\ref{rcG1}$)-($\ref{rcG2}$), ($\ref{rccOIG}$), ($\ref{rcG3}$) and ($\ref{ImproOIG}$) for the remaining couplings $\rho_{\tilde{a}}(\Lambda_R), \  \tilde{a}=1...8$, and $R_i(\Lambda_R), \ i=1...5$, has to be determined such that
\begin{eqnarray}
|| G_i(\rho^R_{\tilde{a}}, \rho^{NR}_{\tilde{a}}, R^R_j, R^{NR}_j) || &\le& \left(\frac{\Lambda_R}{M_P}\right)^2  \ \ \forall \ i   \label{Gismall}
\end{eqnarray}
where the functions $G_i$ have been defined in eqns. ($\ref{FtR1}$) and ($\ref{FtR2}$).
\end{itemize}
Once the second step has been achieved, also the renormalization and improvement conditions $F_i^R$ and $F_i^{NR}$ of the functional $L_{(1)} (\Lambda)$ will obey the bound ($\ref{Gismall}$)  because of eqns. ($\ref{FtR1}$) and ($\ref{FtR2}$). Since the initial conditions ($\ref{inivSTI}$) are also small, we then may conclude that ($\ref{L1Re}$) will be satisfied and the vSTI ($\ref{vSTIL}$) are restored with an accuracy of $(\Lambda/M_P)^2$. 

Note that this means in particular that at the scale $\Lambda< M_P$, our effective theory of quantum gravity is now determined by the cosmologogical constant $\Lambda_K$ and the gravitational constant $\lambda$ to an accuracy of $(\Lambda/M_P)^2$.

\medskip
Some concluding remarks are in order. In section ($\ref{CutRegG}$) we have discussed that for the restoration of the STI for Yang-Mills Theory \cite{KM}, it turns out that there are \textit{more} relevant parts of the $L_{(1)}$-insertion than ''arbitrary'' renormalization conditions $\grave{a}$ la  ($\ref{rcG1}$)-($\ref{rcG2}$), ($\ref{rccOIG}$). To be concrete, one finds that $53$ conditions analogous to ($\ref{Gismall}$) have to be achieved by adjusting $37+7$ renormalization conditions for the effective potential and the BRS variations. Thus, it has to be shown that there are some linear interdependences between the relevant parts of the $L_{(1)}$-insertion for Yang-Mills theory. This is a rather awkward task involving also the renormalization of 1PI vertex functions that we did not consider in this work. Please refer to \cite{KM}, \cite{Mull} for details.

However, it seems probable that the same difficulties will arise in the case of effective quantum gravity: it may turn out that there are \textit{more} inequalities ($\ref{Gismall}$) to fulfill than there are renormalization and improvement conditions ($\ref{rcG1}$)-($\ref{rcG2}$), ($\ref{rccOIG}$), ($\ref{rcG3}$) and ($\ref{ImproOIG}$). In this case, we would have to prove linear interdependences between the relevant and least irrelevant parts of $L_{(1)}$, i.e. eqns.  ($\ref{FtR1}$), ($\ref{FtR2}$), by employing the methods described in \cite{KM}.

\end{subsection}

\begin{section}[The no-cutoff limit of quantum qravity]{The no-cutoff limit of quantum gravity from the viewpoint of the Polchinski analysis} \label{GravNoCut}

It has been discussed in chapter ($\ref{OvM}$) that the running nonrenormalizable couplings $\rho_n(\Lambda)$ of some quantum field theory become \textit{determined} by the renormalizable ones $\rho_a(\Lambda)$ in the no-cutoff limit $\Lambda_0 \rightarrow \infty$:
\begin{equation}
\lim_{\Lambda_0 \rightarrow \infty}  \rho_n(\Lambda,\Lambda_0, \rho_a^0(\Lambda, \Lambda_0, \rho_a(\Lambda))) = \rho_n^{cont}( \Lambda, \rho_a(\Lambda)).  \label{snc}
\end{equation}
This is, however, only true if the bare initial values $\rho_n^0= \rho_n^0 (\rho_a^0)$ of the nonrenormalizable couplings are sufficiently small in the sense of eq. ($\ref{sini}$). See section ($\ref{RenFlowOver}$) for more details. In addition, we would like the reader to recall that in general, the no-cutoff limits $\rho_n^{cont}( \Lambda, \rho_a(\Lambda))$ will be \textit{nonzero}.

On the other hand, from the viewpoint of the flow equations a theory has been called ''nonrenormalizable'' if its field and symmetry content is such that \textit{no} renormalizable interactions, that is no couplings except for kinetic and mass terms, are permitted. However, there is nothing wrong with introducing renormalization conditions for the fields and masses of such a theory (i.e. for the only operators that are renormalizable), and considering the no-cutoff limit employing analogons to Theorems ($\ref{BoundTh}$)-($\ref{UniTh}$). In fact, we can produce a related situation for the scalar field theory considered in chapter ($\ref{RenFlow}$) by choosing vanishing renormalization conditions  ($\ref{c3}$) and ($\ref{c4}$) for the renormalizable $\phi^3$ and $\phi^4$ couplings:
\begin{eqnarray} 
\rho_4^R=\rho_5^R = 0.  \label{rencond0}
\end{eqnarray}
If the bare nonrenormalizable couplings are assumed to be small, see eq. ($\ref{ini}$),  the outcome of such an analysis, already discussed at the end of section ($\ref{EffFlowOver}$),  is a forgone conclusion: all interactions will go away as $\Lambda_0 \rightarrow \infty$ because there are no renormalizable interactions generating new contributions to the nonrenormalizable ones while integrating out field modes. Thus, the running effective potential vanishes:
\begin{eqnarray} 
L(\Lambda, \Lambda_0) \rightarrow 0 \ \ \ \text{for} \ \ \ \Lambda_0 \rightarrow \infty,  \label{LzNR}
\end{eqnarray}
and we are left with the ''free'' part of the effective action\footnote{See eq. ($\ref{SpolM}$) for definitions of the potential $L$ and the effective action.}. Equivalently, one can say that all running nonrenormalizable couplings will die out in the no-cutoff limit,
\begin{equation}
 \rho_n^{cont}( \Lambda, \rho_a(\Lambda)) = 0.   \label{nnc}
\end{equation}
Note that this is not in contradiction to eq. ($\ref{snc}$). One can rather say that in the special case of a nonrenormalizable theory, the running nonrenormalizable couplings  $\rho_n(\Lambda)$ are determined by the renormalizable ones  $\rho_a(\Lambda)$ (i.e. only field strength renormalization and masses) to be \textit{zero}.

This can be reformulated in more technical terms for the scalar field theory of chapter ($\ref{RenFlow}$) as follows. Since the vertex functions of the effective potential are evaluated in perturbation theory in the couplings $\rho_4^R$ and $\rho_5^R$, see eq. ($\ref{Pert}$), and vanish to $0th$ order in perturbation theory, see eq. ($\ref{zero}$), they will all be zero  for the choice ($\ref{rencond0}$). This in particular holds true for the bare initial conditions ($\ref{ini}$). However, in the no-cutoff limit the vertex functions  converge to limits that are independent of the initial conditions as long as these are sufficiently small. This is assured by Theorems ($\ref{ConvTh}$) and ($\ref{UniTh}$). In particular, we may choose vanishing initial conditions. Thus eq. ($\ref{LzNR}$) is confirmed.

\medskip
Let us now discuss the situation for quantum gravity. As has been argued in section ($\ref{GENR}$), the expansion of the gravity action leaves us with infinitely many nonrenormalizable operators\footnote{We always assume that the UV behaviour of the theory is governed by the $1/k^2$ propagators ($\ref{gravprop}$) and ($\ref{propgh}$). This once more amounts to the assumption of small bare couplings associated with the higher field invariants such as $R^2$. See the discussion in section ($\ref{GAFG}$).}. However, as we have just pointed out, from the viewpoint of the flow equations the maybe more interesting question is \textit{''are there any renormalizable interactions at all?''}. Looking at table ($\ref{tab}$), we see that this is indeed the case: the couplings $\rho_4$ and $\rho_5$ accociated with $\phi^3$ and $\phi^4$-like operators are renormalizable. Note that after the restoration of the Slavnov-Taylor identities, these couplings will be given in terms of the gravitational coupling $\lambda$ and the cosmological constant $\Lambda_K$ as $\lambda \Lambda_K$ and $\lambda^2 \Lambda_K$. Hence, they will only be there for nonvanishing (renormalized) cosmological constant.

At the end of this section, the above observation will lead to us the speculation whether quantum gravity \textit{with cosmological constant} has a no-cutoff limit with nonvanishing gravitational constant $\lambda$. Note that since the latter is a nonrenormalizable coupling, it would then be \textit{determined by the cosmological constant} $\Lambda_K$ in the sense of eq. ($\ref{snc}$). Moreover, we will discuss that the coupling of mass terms of some massive field $\varphi$ to gravity produces again $\phi^3$ and $\phi^4$-like operators. This will give rise to the speculation of a gravitational constant that is \textit{determined by the cosmological constant and the masses of the elementary particles} in the no-cutoff limit. If true, this might indicate that there is a deeper relation between the Higgs mechanism (that produces the mass terms) and the gravitational force.  

In the following, we will give the formal steps that have to be employed in order to consider the no-cutoff limit of quantum gravity along the lines of renormalization with flow equations. If we refer to equations or Theorems of sections ($\ref{ToyM}$) and ($\ref{STIRest}$), it is understood that the Planck scale $M_P$ is replaced by an arbitrary UV cutoff $\Lambda_0$. After the formal treatment, we will discuss the implications of our results. 

We begin with the vertex functions ($\ref{LexpG}$) of the gravity potential $L(h, C, \overline{C}, \Lambda)$ introduced in eq. ($\ref{SGraveffG}$) and expand their dimensionless counterparts $A_{n} (k_1,...,k_{n}, \Lambda)$ in the dimensionless renormalized renormalizable couplings $\lambda_4^R$ and $\lambda_5^R$ introduced in eqns. ($\ref{lgren}$):
\begin{eqnarray}
A_{n} (k_1,...,k_{n}, \Lambda) = \sum_{r_1, r_2=0}^{\infty} (\lambda_4^R)^{r_1} (\lambda_5^R)^{r_2}   A_{n}^{(r_1,r_2)} (k_1,...,k_{n}, \Lambda).  \label{Apert4}
\end{eqnarray}
Remember that the couplings $\lambda_{\tilde{a}}$ have to be understood as $k_{\tilde{a}}$-tuples, and recall the notation conventions of eqns. ($\ref{tupord}$) and ($\ref{tupexp}$). Comparing to the expansion ($\ref{Apert3}$) where also the bare nonrenormalizable couplings ($\ref{lgBare}$) are used as expansion parameters, we notice that we may identify our ''new'' expansion ($\ref{Apert4}$) with the $0th$ order of perturbation theory in the bare nonrenormalizable couplings of the former expansion ($\ref{Apert3}$). Thus, Theorem ($\ref{BoundThG}$) concerning the boundedness of gravity vertex functions can be applied to the vertex functions $A_{n}^{(r_1,r_2)} (k_1,...,k_{n}, \Lambda)$ of ($\ref{Apert4}$), and we may proceed by establishing their convergence to a no-cutoff limit in analogy to Theorem ($\ref{ConvTh}$) for the scalar field theory considered in chapter ($\ref{RenFlow}$).

\begin{satz}[Convergence of Gravity Vertex Functions] \label{ConvThG} 
Let there be renormalization conditions ($\ref{rcG1}$)-($\ref{rcG2}$). Assume that to order $r_1, r_2$ in perturbation theory in $\lambda_4^R$ and $\lambda_5^R $ 
\begin{eqnarray}
||\partial^p A_{n}^{(r_1, r_2)}(p_1,...,p_{n}, \Lambda)|| \le \Lambda^{-p}  \left( \frac{\Lambda}{\Lambda_R} \right)^{|r_1|} Pln \left( \frac{\Lambda_0}{\Lambda_R} \right), \label{BoundWGG}
\end{eqnarray}
and that for $n+p \ge 5$
\begin{eqnarray}
|| \Lambda_0 \frac{d}{d \Lambda_0} \partial^p A_{n}^{(r_1, r_2)}(p_1,...,p_{n}, \Lambda_0)|| \le \Lambda_0^{-p}  \left( \frac{\Lambda_0}{\Lambda_R} \right)^{|r_1|} Pln \left( \frac{\Lambda_0}{\Lambda_R} \right) . \label{ini2GG}
\end{eqnarray}
Then
\begin{eqnarray}
||\Lambda_0 \frac{d}{d \Lambda_0} \partial^p A_{n}^{(r_1, r_2)}(p_1,...,p_{n}, \Lambda)|| \le \Lambda^{-p}  \left( \frac{\Lambda}{\Lambda_R} \right)^{|r_1|} \frac{\Lambda}{\Lambda_0} Pln \left( \frac{\Lambda_0}{\Lambda_R} \right)  \label{ConvEGG}
\end{eqnarray}
where $\Lambda_R \le \Lambda \le \Lambda_0$.
\end{satz}
The conditions ($\ref{BoundWGG}$)  and ($\ref{ini2GG}$)  are guaranteed by Theorem ($\ref{BoundThG}$). The remaining proof runs in complete analogy to the one of Theorem ($\ref{ConvTh}$), and we will therefore skip it. 

Integrating ($\ref{ConvEGG}$) with respect to $\Lambda_0$ we may conclude with Cauchy's criterion that the gravity vertex functions converge to a no-cutoff limit
\begin{equation}
A_{n}^{cont \ (r_1, r_2)}(\Lambda) := \lim_{\Lambda_0 \rightarrow \infty} A_{n}^{(r_1, r_2)}(\Lambda, \Lambda_0) \label{contlimAG} 
\end{equation}
where the rate of convergence is given by
\begin{eqnarray}
|| A_{n}^{(r_1, r_2)}(p_1,...,p_{n}, \Lambda, \Lambda_0)- A_{n}^{cont \ (r_1, r_2)}(p_1,...,p_{n},\Lambda)  || \le \left( \frac{\Lambda}{\Lambda_R} \right)^{|r_1|} \frac{\Lambda}{\Lambda_0} Pln \left( \frac{\Lambda_0}{\Lambda_R} \right) . \nonumber \\   \label{conv2G}
\end{eqnarray}
We would like to stress that at this point of the analysis, with gauge invariance still violated by the momentum cutoff regularization, we have only established that there exists a family of finite theories (with field content ${\Phi} = \left(h_{\mu \nu}, C^\mu, \overline{C}_\mu \right)$) that is parametrized by the arbitrary renormalization conditions ($\ref{rcG1}$)-($\ref{rcG2}$). 

Let us therefore proceed with the discussion of the formal steps that are necessary in order to restore the violated Slavnov-Taylor identities ($\ref{vSTIL}$) in the no-cutoff limit. To do so, we have to consider the effective potentials $L(\Lambda)$, $L_\beta^{\mu \nu}(\Lambda)$, $L_\tau^\mu(\Lambda)$ and $L_{(1)} (\Lambda)$ that have been defined in eqns. ($\ref{abbr1}$)-($\ref{abbr4}$). We have just established the boundedness and convergence of the vertex functions of $L(\Lambda)$  in the no-cutoff limit, so let us move on and carry out the same program for the vertex functions carrying the nonlinear BRS variations  ($\ref{BRST1_0}$) and  ($\ref{BRST2_0}$) as an operator insertion.

However, here we are confronted with the following problem. In the last section, the dimensionless vertex functions $A_{\beta n}(\Lambda)$ and $A_{\tau n}(\Lambda)$ of the functionals $L_\beta^{\mu \nu}(\Lambda)$ and $L_\tau^\mu(\Lambda)$ have been evaluated in perturbation theory in the couplings defined in eqns. ($\ref{ROIG1}$), ($\ref{lgren}$), ($\ref{ROIG2}$) and ($\ref{lgBare}$). We summarize the respective expansion parameters again in the following table: 
\begin{table}[here]
\hspace{0.5cm} \begin{tabular}{|c|c|c|} \hline  
Vertex function & Renormalized renormalizable &  Bare nonrenormalizable  \\  & expansion parameters & expansion parameters \\ \hline \hline  $A_{\beta n}(\Lambda)$  & $\mathcal{R}_1^R$, $\mathcal{R}_2^R$, $\lambda_4^R$, $\lambda_5^R$ & $\mathcal{R}_3^0$, $\mathcal{R}_4^0$, $\lambda_{6}^0$, $\lambda_{7}^0$ , $\lambda_{8}^0$  \\ \hline $A_{\tau n}(\Lambda)$ & $\lambda_4^R$, $\lambda_5^R$ & $\mathcal{R}_5^0$, $\lambda_{6}^0$, $\lambda_{7}^0$ , $\lambda_{8}^0$  \\ \hline
\end{tabular} 
\caption{Expansion parameters of the vertex functions with BRS insertion}   \label{tab2}
\end{table}

Recall that the couplings $\mathcal{R}_i^R$, $\mathcal{R}_i^0$ are associated with one extra vertex of each graph contributing to the vertex functions $A_{\beta n}(\Lambda)$ and $A_{\tau n}(\Lambda)$ because of the BRS operator insertion.

Proceeding as we have done in eq. ($\ref{Apert4}$) for the vertex functions without operator insertion would suggest to expand $A_{\beta n}(\Lambda)$ solely in the renormalizable couplings $\mathcal{R}_1^R$, $\mathcal{R}_2^R$, $\lambda_4^R$ and $\lambda_5^R$, and accordingly $A_{\tau n}(\Lambda)$ solely in the renormalizable couplings $\lambda_4^R$ and $\lambda_5^R$. This would again correspond to the $0$th order of perturbation theory in the nonrenormalizable couplings, or equivalently mean that we employ the special case of vanishing values for them:
\begin{eqnarray}
\mathcal{R}_i^0(\Lambda)&=&0, \ \ \ i=3...5   \label{zeroR1} \\
\lambda_{\tilde{a}}^0 (\Lambda)&=&0, \ \ \ \tilde{a}=6...8.  \label{zeroR2}
\end{eqnarray}
However, since each graph of the vertex functions $A_{\beta n}(\Lambda)$ and $A_{\tau n}(\Lambda)$ must have one extra vertex associated with one of the couplings $\mathcal{R}_i^R$ or $\mathcal{R}_i^0$, we are forced to conclude that for the choice ($\ref{zeroR1}$) we have  
\begin{eqnarray}
A_{\tau n}(\Lambda) = 0  \label{noat}
\end{eqnarray}
because all extra vertices of the graphs of $A_{\tau n}(\Lambda)$ are due to the nonrenormalizable bare coupling $\mathcal{R}_5^0$.

In order to avoid this, let us therefore establish the boundedness and convergence of the vertex functions $A_{\beta n}(\Lambda)$ and $A_{\tau n}(\Lambda)$ by expanding them in the renormalized renormalizable couplings $\mathcal{R}_1^R$, $\mathcal{R}_2^R$, $\lambda_4^R$ and $\lambda_5^R$ and also in the bare nonrenormalizable couplings $\mathcal{R}_i^0, \ i=3...5$. We obtain
\begin{eqnarray}
A_{\beta n} (q,k_1,...,k_{n}, \Lambda) &=& \sum_{r_1,r_{2}=0}^{\infty}  \ \sum_{i=1}^4 \mathcal{R}_i^{R/0} \ (\lambda_4^R)^{r_1} (\lambda_5^R)^{r_2}   A_{\beta n}^{(i, r_1,r_2)} (q,k_1,...,k_{n}, \Lambda) \nonumber \\  \label{GenPertIGNC}
\end{eqnarray}
where $\mathcal{R}_i^{R/0}$ is defined as in eq. ($\ref{R_R0}$). The expansion of $A_{\tau n}(\Lambda)$ is similar, with the exception that the extra vertex is always associated with the bare nonrenormalizable coupling $\mathcal{R}_5^0$.

If we want to take a no-cutoff limit $\Lambda_0 \rightarrow \infty$ while retaining nonvanishing couplings  $\mathcal{R}_i^0(\Lambda), \ i=3...5$, we are forced to give up the requirement ($\ref{smallOIGb}$) that these couplings remain small on scales $\Lambda \le \Lambda_0$. To see this, consider some scale $\Lambda_D \le \Lambda_0$. From the definition ($\ref{ROIG2}$) follows that for bare couplings
\begin{eqnarray}
R_i^0 \sim \Lambda_D^{-1}, \ \ \ i=3...5,   \label{divR0}
\end{eqnarray}
their dimensionless counterparts $\mathcal{R}_i^0(\Lambda)= \Lambda R_i^0$ will grow large\footnote{In particular, at the bare scale the $\mathcal{R}_i^0(\Lambda_0)$ will diverge as ${\Lambda_0}/{\Lambda_D}$ for  $\Lambda_0 \rightarrow \infty$.} as 
\begin{equation}
\mathcal{R}_i^0(\Lambda) \sim  {\Lambda}/{\Lambda_D}, \ \ \ i=3...5,  \label{divR}
\end{equation}
for $\Lambda \ge \Lambda_D$. We will repeatedly come back to the consequences of this behaviour.

Comparing ($\ref{GenPertIGNC}$) to the expansion ($\ref{GenPertIG}$) where also the bare nonrenormalizable couplings $\lambda_{\tilde{a}}^0, \  \tilde{a} =6...8$, are used as expansion parameters, we notice that we may again identify our ''new'' expansion ($\ref{GenPertIGNC}$) with the $0th$ order of perturbation theory in the bare nonrenormalizable couplings $\lambda_{\tilde{a}}^0, \  \tilde{a} =6...8$, of the former expansion ($\ref{GenPertIG}$). Thus, Theorem ($\ref{BoundThOIG}$) concerning the boundedness of gravity vertex functions with BRS insertion can be applied\footnote{It is understood that we substitute  $M_P \rightarrow \Lambda_0$ in Theorem ($\ref{BoundThOIG}$).} to the vertex functions $A_{\beta n}^{(i, r_1,r_2)} (q,k_1,...,k_{n}, \Lambda)$ of ($\ref{GenPertIGNC}$) The resulting  bounds are
\begin{eqnarray}
 && \hspace{-2em} ||\partial^p A_{\beta n}^{(i,r_1, r_2)}(q, p_1,...,p_{n}, \Lambda)|| \nonumber \\ &&  \qquad \qquad   \le \Lambda^{-p} \left( \frac{\Lambda}{\Lambda_R} \right)^{|r_1|} \left( \frac{\Lambda_0}{\Lambda} \right)^{ \Theta_i^0}   \left( \Theta_i^R \ Pln\left( \frac{\Lambda}{\Lambda_R} \right) +  \frac{\Lambda}{\Lambda_0}  Pln \left( \frac{\Lambda_0}{\Lambda_R} \right) \right) \nonumber \\     \label{BoundNRIGNC}
\end{eqnarray}
where the symbols $\Theta_i^R$ and  $\Theta_i^0$ have been defined in eqns. ($\ref{thet1}$) and ($\ref{thet2}$). Similar bounds for the expansion coefficients $A_{\tau n}^{(5, r_1,r_2)} (q,k_1,...,k_{n}, \Lambda)$ of the vertex functions $A_{\tau n}(\Lambda)$ can be established, where in this case we have $\Theta_i^R=0, \ \Theta_i^0=1$. 

Before we discuss the implications of the bounds ($\ref{BoundNRIGNC}$), let us move on and establish the convergence of the vertex functions $A_{\beta n}(\Lambda)$ and $A_{\tau n}(\Lambda)$ to a no-cutoff limit. To do so, some preliminary remarks are needed. In the last two sections, we have imposed improvement conditions for some of the nonrenormalizable couplings in an effective field theory context. However, as it has been discussed in section ($\ref{EffFlowOver}$) and again at the end of section ($\ref{Presec}$), it has been shown by C. Wieczerkowski \cite{Wiec} that improvement conditions may also be used to enhance the rate of convergence of the effective action to its no-cutoff limit. To do so, some of the nonrenormalizable couplings have to be fixed at the renormalization scale $\Lambda_R$  \textit{at their no-cutoff values}. The latter are determined by the renormalized renormalizable couplings and can in principle be calculated in perturbation theory in them.

For the gravity vertex functions without operator insertion ($\ref{Apert4}$), such a strategy would mean to impose improvement conditions
\begin{eqnarray}
\rho_{\tilde{a}}(\Lambda_R) &=& \rho^{cont}_{\tilde{a}}(\Lambda_R, \rho^R_{\tilde{b}}) , \ \  {\tilde{a}}=6,...,8, \ \  {\tilde{b}}=1,...,5,  \label{rcG3NC}  
\end{eqnarray}
instead of ($\ref{rcG3}$). By the same arguments that led to eq. ($\ref{conv2impro}$) we would then obtain an improved convergence as compared to Theorem ($\ref{ConvThG}$), eq. ($\ref{ConvEGG}$):
\begin{eqnarray}
||\Lambda_0 \frac{d}{d \Lambda_0} \partial^p A_{n}^{(r_1, r_2)}(p_1,...,p_{n}, \Lambda)|| \le \Lambda^{-p}  \left( \frac{\Lambda}{\Lambda_R} \right)^{|r_1|} \left( \frac{\Lambda}{\Lambda_0} \right)^2 Pln \left( \frac{\Lambda_0}{\Lambda_R} \right) . \label{ConvEGGNC}
\end{eqnarray}
In the following, we assume that  ($\ref{ConvEGGNC}$) has been established.

We now come back to our discussion concerning the convergence of the vertex functions $A_{\beta n}(\Lambda)$ and $A_{\tau n}(\Lambda)$ to no-cutoff limits. Corresponding to the expansion ($\ref{GenPertIGNC}$) where also the nonrenormalizable expansion parameters  $\mathcal{R}_i^0, \ i=3...5,$ have been employed, let us impose improvement conditions for the nonrenormalizable couplings ${R}_i(\Lambda_R), \ i=3...5$:
\begin{eqnarray} 
R_i^{NR} := R_i(\Lambda_R), \ \ i=3...5. \label{ImproOIGNC}
\end{eqnarray}
In order to avoid possible confusion, we would like to stress that the improvement conditions ($\ref{ImproOIGNC}$) are \textit{not} understood to be (necessarily) given in terms of renormalizable couplings in the no-cutoff limit (for small bare couplings ${R}_i^0$, this would in particular imply $R_5^{NR}=0$ by the reasoning leading to eq.  ($\ref{noat}$)), but are free input parameters. This is in contrast to ($\ref{rcG3NC}$). The conjecture (that we will not prove) is that they nevertheless can be met in the no-cutoff limit if we give up the requirement that the bare couplings ${R}_i^0$ be small, see eqns. ($\ref{divR0}$), ($\ref{divR}$).

The necessary ingredients for establishing the convergence of the vertex functions $A_{\beta n}(\Lambda)$ and $A_{\tau n}(\Lambda)$ to a no-cutoff limit have now been collected. The following Theorem can be deduced inductively by integrating the RG inequality ($\ref{RGI1_NRI}$) along the lines of eqns. ($\ref{RGI3}$) and ($\ref{RGI5}$) and applying the bounds of Theorems ($\ref{BoundThG}$), ($\ref{PreThG}$) and eq. ($\ref{ConvEGGNC}$).

\begin{satz}[Convergence of Vertex Functions with BRS Insertion] \label{ConvThIGNC} Let \\ there be renormalization conditions  ($\ref{rccOIG}$) and  improvement conditions ($\ref{ImproOIGNC}$). Assume that to order $r_1, r_2$ in perturbation theory in $\lambda_4^R$ and $\lambda_5^R$
\begin{eqnarray}
||\partial^p A_{\beta n}^{(i,r_1, r_2)}(q, p_1,...,p_{n}, \Lambda)|| \le \Lambda^{-p}  \left( \frac{\Lambda}{\Lambda_R} \right)^{|r_1|}  \left( \frac{\Lambda_0}{\Lambda} \right)^{\Theta_i^0} Pln \left( \frac{\Lambda_0}{\Lambda_R} \right), \label{BoundWPIGNC}
\end{eqnarray}
and that for $n+p \ge 4$
\begin{eqnarray}
||\Lambda_0 \frac{d}{d \Lambda_0} \partial^p A_{\beta n}^{(i,r_1,r_2)}(q, p_1,...,p_{n}, \Lambda_0)|| \le \Lambda_0^{-p}  \left( \frac{\Lambda_0}{\Lambda_R} \right)^{|r_1|} Pln \left( \frac{\Lambda_0}{\Lambda_R} \right) . \label{ini2PIGNC}
\end{eqnarray}
Given Theorems ($\ref{BoundThG}$), ($\ref{PreThG}$) and eq. ($\ref{ConvEGGNC}$) we then have for $\Lambda_R \le \Lambda \le \Lambda_0$
\begin{eqnarray}
|| \Lambda_0 \frac{d}{d \Lambda_0} \partial^p A_{\beta n}^{(i, r_1, r_2)}(q, p_1,...,p_{n}, \Lambda)|| \le \Lambda^{-p}  \left( \frac{\Lambda}{\Lambda_R} \right)^{|r_1|}  \left( \frac{\Lambda_0}{\Lambda} \right)^{\Theta_i^0} \left(\frac{\Lambda}{\Lambda_0}\right)^{2} Pln \left( \frac{\Lambda_0}{\Lambda_R} \right) \nonumber \\ \label{PreIGNC}
\end{eqnarray}
where the index $i$ refers to an extra vertex associated with a renormalized coupling ${R}_i^R, \ i=1,2$ having canonical dimension $D_{R_i} = 0$ ($\Theta_i^0=0$) or a bare one ${R}_i^0, \ i=3,4$ having canonical dimension $D_{R_i} = -1$ ($\Theta_i^0=1$), and  $r_{NR} = |r_3| + |r_4| + |r_5|$.
\end{satz}
For the vertex functions $A_{\tau n}(\Lambda)$ we may proceed similarly, with the exception that the extra vertex is always associated with the bare nonrenormalizable coupling ${R}_5^0$. Thus, we do not need the renormalization conditions ($\ref{rccOIG}$) and we only have the case $\Theta_i^0=1$.

The inequalities ($\ref{BoundNRIGNC}$) and ($\ref{PreIGNC}$) can be converted into the following upper bounds which are valid for $i=1...4$:
\begin{eqnarray}
||\partial^p A_{\beta n}^{(i,r_1, r_2)}(q, p_1,...,p_{n}, \Lambda)|| & \le & \Lambda^{-p} \left( \frac{\Lambda}{\Lambda_R} \right)^{|r_1|} Pln \left( \frac{\Lambda_0}{\Lambda_R} \right)  \\
||\Lambda_0 \frac{d}{d \Lambda_0} \partial^p A_{\beta n}^{(i, r_1, r_2)}(q, p_1,...,p_{n}, \Lambda)|| &\le& \Lambda^{-p}   \left( \frac{\Lambda}{\Lambda_R} \right)^{|r_1|}  \frac{\Lambda}{\Lambda_0} Pln \left( \frac{\Lambda_0}{\Lambda_R} \right) . \quad \quad \  \label{convNCID}
\end{eqnarray}
These also apply to the vertex functions  $A_{\tau n}^{(5, r_1,r_2)} (\Lambda)$.

The point is that eq. ($\ref{convNCID}$) is already sufficient for proving the convergence of the vertex functions $A_{\beta n}^{(i, r_1,r_2)} (\Lambda)$ and $A_{\tau n}^{(5, r_1,r_2)} (\Lambda)$ to no-cutoff limits 
\begin{eqnarray}
A_{\beta n}^{cont \ (i, r_1,r_2)} (\Lambda) := \lim_{\Lambda_0 \rightarrow \infty} A_{\beta n}^{(i, r_1,r_2)} (\Lambda, \Lambda_0) \label{NCAb} \\ A_{\tau n}^{cont \ (i, r_1,r_2)} (\Lambda) := \lim_{\Lambda_0 \rightarrow \infty} A_{\tau n}^{(i, r_1,r_2)} (\Lambda, \Lambda_0)  \label{NCAt}
\end{eqnarray}
where the arguments are the same as those following Theorem ($\ref{ConvThG}$). We therefore do not need the requirement of small couplings $\mathcal{R}_i^0(\Lambda) \sim \Lambda/\Lambda_0$, see eqns. ($\ref{smallOIGb}$), ($\ref{smallOIGExp}$), in order to cancel out the factor $\left( {\Lambda_0}/{\Lambda} \right)^{\Theta_i^0}$ appearing in eqns. ($\ref{BoundNRIGNC}$) and ($\ref{PreIGNC}$). This is because there is at most \textit{one} vertex associated with a nonrenormalizable coupling $\mathcal{R}_i^0, \ i=3...5$, in each graph contributing to the vertex functions $A_{\beta n}(\Lambda)$  and $A_{\tau n}(\Lambda)$. 

One should, however, keep in mind that couplings $\mathcal{R}_i^0(\Lambda) \sim  {\Lambda}/{\Lambda_D}, \ i=3...5$, $\grave{a}$ la  eqns. ($\ref{divR0}$), ($\ref{divR}$) will grow large at scales $\Lambda > \Lambda_D$.

By virtue of Theorems ($\ref{ConvThG}$) and  ($\ref{ConvThIGNC}$) we have now established that the dimensionless vertex functions of the functionals $L(\Lambda)$, $L_\beta^{\mu \nu}(\Lambda)$ and $L_\tau^\mu(\Lambda)$ converge to no-cutoff limits as $\Lambda_0 \rightarrow \infty$. Hence, the LHS of the violated Slavnov-Taylor identities ($\ref{vSTIL}$) of quantum gravity is under control and we may proceed with the restoration of the STI  in the no-cutoff limit.

\medskip
To do so, we will again have to consider the functional $L_{(1)} (\Lambda)$ appearing on the RHS of the vSTI ($\ref{vSTIL}$). It has been defined in eq. ($\ref{abbr4}$) as the generating functional of vertex functions carrying the BRS variation of the bare total gravity action ($\ref{L1}$) as a (space-time integrated) operator insertion. Recall that $L_{(1)} (\Lambda)$  has ghost number $1$ and canonical dimension $5$ in the sense of space-time integrated operator insertions discussed at the end of Appendix ($\ref{OI}$). Furthermore, momentum-space vertex functions $L_{(1)n}(k_1, ..., k_n, \Lambda)$ of $L_{(1)} (\Lambda)$ have been introduced in eq. ($\ref{LexpGIG}$) of the last section, and it has been argued that these may be used to define running coupling constants $F_i(\Lambda)$ as it has been done it in eqns. ($\ref{rcc0OI}$)- ($\ref{rcc3OI}$) of Appendix ($\ref{OI}$).

If we attempt to control the dimensionless counterparts  $A_{(1)n}(k_1, ..., k_n, \Lambda)$ of the vertex functions ($\ref{LexpGIG}$)  in the no-cutoff limit, we will have to impose the renormalization conditions ($\ref{RenL1}$), that is
\begin{eqnarray}
F_i^R &:=& F_i(\Lambda_R), \ \ \ D_{F_i} \ge 0 .  \nonumber
\end{eqnarray}
Establishing analogons to Theorems ($\ref{BoundThOI}$) and ($\ref{ConvThI}$) concerning the boundedness and convergence of vertex functions with operator insertion for the $A_{(1)n}(k_1, ..., k_n, \Lambda)$ requires smallness of the bare initial conditions in the sense of eq. ($\ref{iniNRI}$):
\begin{eqnarray}
||\partial^p A_{(1)n}(p_1,...,p_{n}, \Lambda_0)|| &\le& \Lambda_0^{-p} \ Pln \left( \frac{\Lambda_0}{\Lambda_R} \right), \ \ \ n+p \ge 6.   \label{iniNRINC}
\end{eqnarray} 
However, it seems probable that the condition ($\ref{iniNRINC}$) will be spoiled when we drop the requirement that the dimensionless nonrenormalizable BRS couplings $\mathcal{R}^0_i(\Lambda)$, $i=3...5$, be small on scales $\Lambda \le \Lambda_0$. This follows from the fact that ($\ref{iniNRINC}$) is an inequality for vertex functions of the \textit{bare} insertion $L_{(1)} (\Lambda_0)$. Since the latter has been defined in eq. ($\ref{L1}$) as the (regularized) BRS variation of the bare total gravity action, it will contain the couplings $\mathcal{R}^0_i(\Lambda_0), \ i=3...5$. If these are chosen as in eqns. ($\ref{divR0}$), ($\ref{divR}$), they will diverge as $\Lambda_0/\Lambda_D$ for $\Lambda_0\rightarrow \infty$. Thus, it is suggested that we are forced to employ small values
\begin{equation}
\mathcal{R}^0_i(\Lambda) \sim \Lambda/\Lambda_0, \ \ \ i=3...5,  \label{smallRNR}
\end{equation} 
in order to satisfy ($\ref{iniNRINC}$). Note that in the no-cutoff limit $\Lambda_0 \rightarrow \infty$ this amounts to vanishing couplings  $\mathcal{R}^0_i(\Lambda)=0, \  i=3...5$, with the unpleasant consequences that have been discussed following eq. ($\ref{zeroR1}$). For the time being, let us nevertheless ensure that the condition ($\ref{iniNRINC}$) is satisfied by setting the nonrenormalizable BRS couplings to zero, eq. ($\ref{zeroR1}$) .

As it has been explained in the last section in eqns. ($\ref{L1_2}$) and ($\ref{FB1}$), the renormalization conditions ($\ref{RenL1}$) cannot be chosen freely. Since the bare insertion $L_{(1)} (\Lambda_0)$ has been defined in eq. ($\ref{L1}$) as the (regularized) BRS variation of the bare total gravity action, the bare renormalizable couplings $F_i^0=F_i(\Lambda_0), \ D_{F_i} \ge 0$, will be given in terms of the bare renormalizable coupling constants $\rho^0_{\tilde{a}}, \  \tilde{a}=1...5$, and the bare nonzero BRS couplings $R^0_i, \ i=1,2$, defined in  Table ($\ref{tab}$) and  eqns. ($\ref{rcR}$):
\begin{eqnarray}
 F_i^0 = F_i^0(\rho^0_{\tilde{a}}, R^0_j), \ \ \ \tilde{a}=1...5, \ j=1,2, \ D_{F^0_i} \ge 0.   \label{FB1NC}
\end{eqnarray}
This follows from the dimension $5$ of the  $L_{(1)}$ insertion and the canonical dimensions of the couplings $\rho^0_{\tilde{a}}, \  \tilde{a}=1...5$, and $R^0_i, \ i=1,2$. See Table ($\ref{tab}$) and eqns. ($\ref{rcR}$)-($\ref{DR2}$) for definitions of the operators, couplings and their respective canonical dimensions.  

By the reasoning of the last section, we therefore conclude that the renormalization conditions ($\ref{RenL1}$) must be functions $G_i$ of the renormalization conditions ($\ref{rcG1}$), ($\ref{rcG2}$) and ($\ref{rccOIG}$):
\begin{eqnarray}
F_i^R &=& G_i(\rho^R_{\tilde{a}}, R^R_j) \ + \ \mathcal{O}  \big( \Lambda_R/\Lambda_0 \big)  \label{FtRNC} 
\end{eqnarray}
where the indetermination that is left in eq. ($\ref{FtR1NC}$) stems from the ignorance about the initial conditions
\begin{eqnarray}
\partial^p A_{(1)n}(p_1,...,p_{n}, \Lambda_0), \ \ \ n+p \ge 6    \label{inivSTINC}
\end{eqnarray}
associated with couplings $ F_i^0, \ \ D_{F^0_i} \le -1$. Since we have assured that these are small $\grave{a}$ la eq. ($\ref{iniNRINC}$) by employing vanishing nonrenormalizable BRS couplings ($\ref{zeroR1}$), there should be no further problems with establishing the convergence of the vertex functions $A_{(1)n}(\Lambda, \Lambda_0)$ to no-cutoff limits along the lines of Theorems ($\ref{BoundThOI}$) and ($\ref{ConvThI}$) while employing the renormalization conditions ($\ref{FtRNC}$): 
\begin{eqnarray}
A_{(1)n}^{cont}(\Lambda) := \lim_{\Lambda_0 \rightarrow \infty} A_{(1)n}(\Lambda, \Lambda_0). \label{NCL1} 
\end{eqnarray}

Let us now return to the case of \textit{nonvanishing} nonrenormalizable BRS couplings $\mathcal{R}^0_i(\Lambda) \ne 0, \ i=3...5$. It may actually be possible to prove the boundedness and convergence of the vertex functions $A_{(1)n}(p_1,...,p_{n}, \Lambda)$ of the functional  $L_{(1)} (\Lambda)$ for this choice. The idea is to employ a strategy similar to the one leading to  eq. ($\ref{ImproOIGNC}$) and Theorem ($\ref{ConvThIGNC}$) for the vertex functions carrying the gravity BRS fields. This would mean to impose the improvement conditions ($\ref{ImproL1}$), that is 
\begin{eqnarray}
F_i^{NR} &:=& F_i(\Lambda_R), \ \ \ D_{F_i} = -1, \nonumber
\end{eqnarray}
in addition to the renormalization conditions ($\ref{RenL1}$). Moreover, we would expand the graphs contributing to the vertex functions $A_{(1)n}(p_1,...,p_{n}, \Lambda)$ in perturbation theory in the dimensionless renormalizable couplings $\lambda_4^R$, $\lambda_5^R$ and the dimensionless counterparts of the couplings  $F_i^{R}, \ D_{F_i} \ge 0$, but also in the dimensionless versions of the bare nonrenormalizable $F_i^0, \ D_{F_i} = -1$. See the analogous expansion ($\ref{GenPertIGNC}$). The conjecture is that the effect of the diverging bare BRS couplings $\mathcal{R}^0_i(\Lambda_0), \ i=3...5,$ appearing in the initial conditions  ($\ref{inivSTINC}$) can be expressed in terms of the dimensionless counterparts of the $F_i^0, \ D_{F_i} = -1$. By the arguments following eqns. ($\ref{NCAb}$), ($\ref{NCAt}$) the improved rate of convergence will then be sufficient to come up for factors $\Lambda_0/\Lambda$ that arise due to the couplings $F_i^0, \ D_{F_i} = -1$. Note that the described procedure is probably equivalent\footnote{We could multiply the BRS variations ($\ref{BRST1}$)-($\ref{BRST3}$) with one inverse gravitational constant $\lambda^{-1}$ and treat them as dimension $3$ operator insertions having no  nonrenormalizable parts. However, ($\ref{BRST1}$) will then contain a part associated with a relevant coupling $\lambda^{-1} \sim M_P$ which grows large at scales $\Lambda < M_P$. In order to avoid this, we have treated the BRS fields as dimension $2$ operator insertions in the last section. In the present context, the $\lambda^{-1}$ coupling would be problematic because it deters us from considering the limit $\lambda \rightarrow 0$. However, for dimension $3$ of the BRS fields the functional $L_{(1)}$ would have dimension $6$, and we would need more renormalization conditions. This might be equivalent to intruducing renormalization \textit{and} improvement conditions for a dimension $5$ functional $L_{(1)}$ as it is suggested here.} to treating the BRS fields  ($\ref{BRST1_0}$) and  ($\ref{BRST2_0}$) as operator insertions with canonical dimension $3$. 

In analogy to the arguments leading to eqns. ($\ref{FtR1}$), ($\ref{FtR2}$) and ($\ref{FtR1NC}$) we furthermore conjecture that in the above scenario, the renormalization and improvement conditions ($\ref{RenL1}$) and ($\ref{ImproL1}$) must be functions $G_i$ of the renormalization conditions ($\ref{rcG1}$), ($\ref{rcG2}$) and ($\ref{rccOIG}$) and the improvement conditions ($\ref{ImproOIGNC}$): 
\begin{eqnarray}
F_i^R &=& G_i(\rho^R_{\tilde{b}}, \rho^{cont}_{\tilde{a}}( \rho^R_{\tilde{b}}, \Lambda_R, \Lambda_0), R^R_j, R^{NR}_k) \ + \ \mathcal{O}\big( \Lambda_R/{\Lambda_0} \big)  \label{FtR1NC}    \\
F_i^{NR} &=& G_i(\rho^R_{\tilde{b}}, \rho^{cont}_{\tilde{a}}( \rho^R_{\tilde{b}}, \Lambda_R, \Lambda_0), R^R_j, R^{NR}_k) \ + \ \mathcal{O}\big( \Lambda_R/{\Lambda_0} \big)  \label{FtR2NC} 
\end{eqnarray}
with $\tilde{b}=1...5$, $\tilde{a}=6...8$, $j=1,2$, $k=3...5$. Here it has been assumed that the bare nonrenormalizable couplings $\rho^{0}_{\tilde{a}}, \ \tilde{a}=6...8,$ are \textit{implicitly defined} by their renormalized no-cutoff values ($\ref{rcG3NC}$), meaning that they are {determined by the renormalized renormalizable couplings} $ \rho^{R}_{\tilde{a}}, \ \tilde{a}=1...5$. Compare this to the analogous eqns. ($\ref{FtR1}$) and ($\ref{FtR2}$) in the effective field theory context\footnote{Note also that in eqns. ($\ref{FtR1NC}$) and ($\ref{FtR2NC}$) the remaining indetermination is larger as compared to ($\ref{FtR1}$) and ($\ref{FtR2}$) because the couplings $F_i^0, \ D_{F_i} = -1$, will grow large.} where the improvement conditions for the couplings $\rho_{\tilde{a}}, \ \tilde{a}=6...8,$ have been specified \textit{freely}. On the other hand, we would like to stress that the improvement conditions $R^{NR}_k, \ k=3...5,$ for the nonrenormalizable BRS couplings appearing in ($\ref{FtR1NC}$) and ($\ref{FtR1NC}$) still have to be understood as free input parameters. They can be met because we have allowed for nonvanishing $\mathcal{R}^0_i(\Lambda) \ne 0, \ i=3...5$.

It is only fair to admit that this second scenario is more speculative than the first one leading to eq. ($\ref{NCL1}$). It is at this point not entirely clear\footnote{We are, however, pretty confident because of the conjectured analogy to an analysis employing dimension $3$ BRS fields,  leading to a dimension $6$ functional $L_{(1)} (\Lambda)$. Then only renormalization conditions have to be imposed for the functionals with operator insertion.} whether boundedness and convergence of the vertex functions $A_{(1)n}(p_1,...,p_{n}, \Lambda)$ of the functional  $L_{(1)} (\Lambda)$ can be established along the lines of eqns. ($\ref{BoundNRIGNC}$) and  ($\ref{PreIGNC}$) employing the renormalization and improvement conditions ($\ref{FtR1NC}$) and ($\ref{FtR2NC}$). Let us all the same assume that it can be done and therefore nonvanishing BRS couplings $\mathcal{R}^0_i(\Lambda), \ i=3...5$, may be kept in the no-cutoff limit with consequences that we will evaluate shortly.

\medskip
The restoration of the violated Slavnov-Taylor identities ($\ref{vSTIL}$) in the no-cutoff limit can now be disussed. Since the violation is expressed in terms of the functional $L_{(1)} (\Lambda)$ appearing on the RHS of ($\ref{vSTIL}$), the STI are restored if we can make $L_{(1)} (\Lambda)$ vanish\footnote{Remember the discussion following Theorem ($\ref{vSTITh}$) concerning the value $\Lambda=0$ of the floating cutoff.} in the limit $\Lambda_0 \rightarrow \infty$:
\begin{eqnarray}
|| L_{(1)} (\Lambda)  || &\le& \frac{\Lambda}{\Lambda_0} \ .  \label{L1ReNC}
\end{eqnarray}
We now follow the proceedings that have been proposed at the end of the last section while adapting them to the present case. We begin with the ''conservative'' assumption of  vanishing nonrenormalizable BRS couplings $\mathcal{R}^0_i(\Lambda)=0, \ i=3...5$:
\begin{itemize}
\item Renormalization conditions for the physical \textit{renormalizable} coupling $\Lambda_K$ and the gravitational field $h_{\mu \nu}$ are specified at some renormalization scale $\Lambda_R$:
\begin{eqnarray}
\frac{\delta}{\delta h(k_1)} \frac{\delta}{\delta h(k_2)} L(\Phi, \Lambda_R)\big|_{h=k_i=0} &\stackrel{!}{=}&0, \ \ \ \  B_1 \stackrel{!}{=} \Lambda_K \ \nonumber \\
\partial_{i}^2 \left(\frac{\delta}{\delta h(k_1)} \frac{\delta}{\delta h(k_2)} L(\Phi, \Lambda_R)\right) \Big|_{h=k_i=0} &\stackrel{!}{=}& 0  .  \label{PRCNC}
\end{eqnarray}

\item \textit{One particular set} of ''arbitrary'' renormalization conditions ($\ref{rcG1}$), ($\ref{rcG2}$) and ($\ref{rccOIG}$) for the remaining couplings $\rho_{\tilde{a}}(\Lambda_R), \  \tilde{a}=1...5$, and $R_i(\Lambda_R), \ i=1,2$, has to be determined such that
\begin{eqnarray}
|| G_i(\rho^R_{\tilde{a}}, R^R_j) || &\le& \frac{\Lambda_R}{\Lambda_0}   \ \ \forall \ i   \label{Gismall1NC}
\end{eqnarray}
where the functions $G_i$ have been defined in eqns. ($\ref{FtRNC}$).
\end{itemize}
Once the bound ($\ref{Gismall1NC}$)  has been achieved, also the renormalization conditions $F_i^R$ of the functional $L_{(1)} (\Lambda)$ will obey ($\ref{Gismall1NC}$)  because of eq. ($\ref{FtRNC}$). Since the initial conditions ($\ref{iniNRINC}$) are also small, we then may conclude that ($\ref{L1ReNC}$) will be satisfied and the vSTI ($\ref{vSTIL}$) are restored in the no-cutoff limit.

If we  employ the more speculative scenario that allows for nonvanishing BRS couplings $\mathcal{R}^0_i(\Lambda) \ne 0, \ i=3...5$, the second point has to be modified as follows:
\begin{itemize}
\item \textit{One particular set} of ''arbitrary'' renormalization conditions ($\ref{rcG1}$), ($\ref{rcG2}$), ($\ref{rccOIG}$) \textit{and} improvement conditions ($\ref{ImproOIGNC}$) for the remaining couplings $\rho_{\tilde{b}}(\Lambda_R)$, $\tilde{b}=1...5$, and $R_j(\Lambda_R), \ j=1...5$, has to be determined such that
\begin{eqnarray}
|| G_i(\rho^R_{\tilde{b}}, \rho^{cont}_{\tilde{a}}( \rho^R_{\tilde{b}}, \Lambda_R, \Lambda_0), R^R_j, R^{NR}_k) || &\le&   \frac{\Lambda_R}{\Lambda_0} \ \ \forall \ i   \label{Gismall2NC}
\end{eqnarray}
where $\tilde{a}=6...8$, and the functions $G_i$ have been defined in eqns. ($\ref{FtR1NC}$) and ($\ref{FtR2NC}$).
\end{itemize}
If we manage to establish the bounds ($\ref{Gismall2NC}$), also the renormalization and improvement conditions $F_i^R$ and $F_i^{NR}$ of the functional $L_{(1)} (\Lambda)$ will obey ($\ref{Gismall2NC}$)  because of eqns. ($\ref{FtR1NC}$) and ($\ref{FtR2NC}$). If the remaining initial conditions (corresponding to the bare couplings $F_i^0, \ \ D_{F^0_i} \le -2$) are small, one may again conclude that ($\ref{L1ReNC}$) will be satisfied and the vSTI ($\ref{vSTIL}$) are restored in the no-cutoff limit. Note that in the present case it will \textit{not} suffice to prove that the renormalization conditions $F_i^R$ obey ($\ref{Gismall2NC}$) because the dimensionless counterparts of the bare couplings $F_i^0, \ D_{F_i} = -1$, are supposed to diverge as $\Lambda_0/\Lambda_D$ for $\Lambda_0 \rightarrow \infty$. Hence, it \textit{does} have to be established that also the improvement conditions $F_i^{NR}$ satisfy ($\ref{Gismall2NC}$).

\medskip
Let us finally discuss the physical implications of what we have done so far. Therefore, it will turn out helpful to explicitly give the renormalized versions of the bare BRS fields ($\ref{BRST1_0}$) and ($\ref{BRST2_0}$) \textit{after} the restoration of the STI. These should just be the standard gravity BRS variations ($\ref{BRST1}$) and ($\ref{BRST2}$):
\begin{eqnarray}
\Psi^{\mu \nu}(x, \Lambda_R) &=& \delta^{\mu \nu} \partial_\rho C^\rho -  \big( \delta^{\rho \nu} \partial_\rho  C^\mu + \delta^{\mu \rho} \partial_\rho C^\nu \big) \nonumber \\ && \ \ + \ \lambda  \partial_\rho \big(C^\rho  h^{\mu \nu} \big)   - \lambda  \big( h^{\rho \nu} \partial_\rho  C^\mu + h^{\mu \rho} \partial_\rho C^\nu \big)  \label{BRST1_Rb} \\
\Omega^\mu(x, \Lambda_R)   &=& \lambda  C^\nu \partial_\nu C^\mu  \label{BRST2_Rb}
\end{eqnarray}
where $\lambda$ is the gravitational constant. Comparing to ($\ref{BRST1_0}$) and ($\ref{BRST2_0}$), we see that for the restoration of the STI the ''arbitrary'' renormalization and improvement conditions ($\ref{rccOIG}$) and ($\ref{ImproOIGNC}$) will have to be chosen such that
\begin{eqnarray}
R_i^R &=& 1, \ \ \ i=1,2 \label{}    \\
R_i^{NR} &=& \lambda, \ \ \ i=3...5.   \label{Rphys}
\end{eqnarray}
After this preliminary remark, let us begin by considering the most unambigious case. It consists of choosing vanishing renormalization conditions for the one and only physical renormalizable coupling:
\begin{eqnarray}
\Lambda_K &=& 0,   \label{zeroKK}
\end{eqnarray}
meaning that we consider the case of zero renormalized cosmological constant. We will now try to figure out how our quantum theory of gravity will look like for the choice  ($\ref{zeroKK}$) in the no-cutoff limit \textit{after} the restoration of the STI. 

Once the latter have been accomplished, the renormalized renormalizable couplings  $\rho^R_4$ and $\rho^R_5$ will be associated with the cosmological constant $\Lambda_K$ and the gravitational constant $\lambda$:
\begin{eqnarray}
\rho^R_4 &\sim& \lambda \Lambda_K \label{r4NC}  \\
\rho^R_5 &\sim& \lambda^2 \Lambda_K  \label{r5NC}
\end{eqnarray}
as follows from Table ($\ref{tab}$). However, the dimensionless counterparts  $\lambda_4^R$ and $\lambda_5^R$ of the couplings ($\ref{r4NC}$) and ($\ref{r5NC}$) are our expansion parameters of the vertex functions $A_{n} (k_1,...,k_{n}, \Lambda)$  of the gravity potential  $L(h, C, \overline{C}, \Lambda)$ in eq. ($\ref{Apert4}$). Since eqns. ($\ref{r4NC}$) and ($\ref{r5NC}$) imply that $\lambda_4^R=\lambda_5^R=0$ for $\Lambda_K=0$, we conclude by the reasoning leading to eq. ($\ref{LzNR}$) that the no-cutoff limit ($\ref{contlimAG}$) of the gravity vertex functions must vanish\footnote{Note that this is based on the assumption of \textit{small} bare nonrenormalizable couplings as they are implied by eq. ($\ref{iniNRG}$).}:
\begin{equation}
A_{n}^{cont \ (r_1, r_2)}(\Lambda)  \ = \ 0    \label{ngvf}
\end{equation}
for $\Lambda_K=0$. As we have explained in eq. ($\ref{nnc}$), this amounts to the statement that all nonrenormalizable running couplings $\rho_n(\Lambda, \Lambda_0), \ n \ge 6$, will die out as $\Lambda_0 \rightarrow \infty$ because there are no renormalizable interactions associated with ($\ref{r4NC}$) and ($\ref{r5NC}$) generating new contributions for them. In particular, the renormalized gravitational constant will vanish in the no-cutoff limit:
\begin{equation}
\lambda(\Lambda, \Lambda_0) \rightarrow 0 \ \ \ \text{for} \ \ \ \Lambda_0 \rightarrow \infty.
\end{equation}
Thus, our theory will become free as $\Lambda_0 \rightarrow \infty$. Let us check how this complies with gauge invariance, which on the quantum level means that the vertex functions must satisfy the Slavnov-Taylor identities. 

As follows from eq. ($\ref{Rphys}$), a vanishing gravitational constant $\lambda=0$ implies that after the restoration of the STI, the renormalized nonrenormalizable BRS couplings must also vanish:
\begin{eqnarray}
\lim_{\Lambda_0 \rightarrow \infty} R_i(\Lambda_R, \Lambda_0) &=& 0, \ \ \ i=3...5.   \label{RNCzero}
\end{eqnarray}
This will be automatically the case if we employ \textit{small} bare nonrenormalizable BRS couplings in the sense of eq. ($\ref{smallRNR}$), because they will then vanish in the limit $\Lambda_0 \rightarrow \infty$. Furthermore, new contributions to the $R_i(\Lambda), \ i=3,4$, will not be generated since  $\lambda_4^R=\lambda_5^R=0$. Note that small BRS couplings correspond to the more conservative scenario leading to eq. ($\ref{noat}$) and a restoration of the STI $\grave{a}$ la eq. ($\ref{Gismall1NC}$). We therefore conclude that
\begin{eqnarray}
L^{cont}(\Lambda) \ = \ L_\tau^{\mu cont}(\Lambda) \ = \ L_{(1)}^{cont} (\Lambda) \ = \ 0  \label{allzero}
\end{eqnarray}
for $\Lambda_K=0$, where we have employed
\begin{eqnarray}
L^{cont}(\Lambda):= \lim_{\Lambda_0 \rightarrow \infty} L(\Lambda, \Lambda_0)
\end{eqnarray}
and similarly for the functionals $L_\tau^\mu(\Lambda, \Lambda_0)$ and $L_{(1)} (\Lambda, \Lambda_0)$. Recall that these functionals have been defined in eqns. ($\ref{abbr1}$)-($\ref{abbr4}$), and that their respective no-cutoff limits have been established by virtue of eqns. ($\ref{contlimAG}$), ($\ref{NCAb}$), ($\ref{NCAt}$) and  ($\ref{NCL1}$). Eqns. ($\ref{allzero}$) should be valid in particular for $\Lambda=0$.

In the no-cutoff limit, the restored STI for quantum gravity \textit{without} a cosmological constant follow therefore from ($\ref{vSTIL}$) as
\begin{eqnarray}
\big\langle h^{\mu \nu}, \Delta^{ -1}_{\mu \nu \rho \sigma} L_\beta^{\rho \sigma}(0) \big\rangle  +  \xi^{-1} \big\langle {C}^\mu, \Delta_{GH \mu \nu}^{ -1}  F^{\nu \rho \sigma} \big( h_{\rho \sigma}  \big)  \big\rangle  &=& 0  \label{vSTIRNCn}
\end{eqnarray} 
where $F^{\nu \rho \sigma} \big( h_{\rho \sigma}  \big)= \partial_\rho h^{\rho \nu}$ and $ \Delta^{ -1}_{\mu \nu \rho \sigma}$, $\Delta_{GH \mu \nu}^{ -1}$ are the inverse graviton and ghost propagators, respectively.  In position space, the latter are given (for $\Lambda_K=0$ and $ \xi=\frac{1}{2}$) by eqns. ($\ref{invprop}$)  and ($\ref{propghPS0}$) as
\begin{eqnarray}
\Delta^{ -1}_{\mu \nu \rho \sigma} &=&-\frac{1}{2} \left( \delta_{\mu \rho} \delta_{\nu \sigma}+ \delta_{\mu \sigma} \delta_{\nu \rho} - \delta_{\mu \nu} \delta_{\rho \sigma}       \right) \partial_\alpha \partial^\alpha  \\
\Delta_{GH \mu \nu}^{ -1}  &=& \delta_{\mu \nu} \partial_\alpha \partial^\alpha .
\end{eqnarray} 
We furthermore note that because of $\lambda_4^R=\lambda_5^R=0$  and eqns. ($\ref{RNCzero}$), ($\ref{BRST1_Rb}$) the no-cutoff limit of the functional $L_\beta^{\mu \nu}(0, \Lambda_0)$ must be 
\begin{eqnarray}
L_\beta^{\mu \nu cont}(0) =  \delta^{\mu \nu} \partial_\rho C^\rho - \big( \delta^{\rho \nu} \partial_\rho  C^\mu + \delta^{\mu \rho} \partial_\rho C^\nu \big)
\end{eqnarray}
for $\Lambda_K=0$. One can explicitly check  that the no-cutoff STI ($\ref{vSTIRNCn}$) are indeed satisfied and gauge invariance is in effect in the limit $\lambda \rightarrow 0$. Looking at the (renormalized) BRS fields ($\ref{BRST1_Rb}$) and ($\ref{BRST2_Rb}$), this can also be interpreted in the way that the gauge/BRS transformations get deformed in this limit because the parts that are associated with the gravitational constant will go away.

We would like to point out that the vanishing of the gravitational constant $\lambda$ in the no-cutoff limit of quantum gravity while employing an analysis $\grave{a}$ la Polchinski has already been conjectured\footnote{S. Weinberg QFT I, Page 526} by S. Weinberg in his book about quantum field theory \cite{Wein1}. However, he presumably refers to the case $\Lambda_K=0$.

In the discussion taking place at the end of section ($\ref{PolGravSec}$), we have given four conditions that a theory must satisfy in order to be a valid candidate for a (fundamental)  quantum theory of gravitation. The no-cutoff limit of quantum gravity \textit{without} a cosmological constant obtained from our analysis with flow equations satisfies three of the four conditions: we have imposed only a finite number of renormalization conditions in eqns. ($\ref{PRCNC}$), the Slavnov-Taylor identities  ($\ref{vSTIRNCn}$) are fulfilled, and the remaining theory will be unitary. However, it fails in the last point, because it predicts a vanishing gravitational constant which is in conflict with the experimental facts. Note that this prediction is based on the assumption of \textit{small} bare nonrenormalizable couplings as it has been stressed on various occasions.

Let us now consider the case of gravity with a \textit{nonzero} renormalized cosmological constant,
\begin{eqnarray}
\Lambda_K & \ne & 0.   \label{NzeroKK}
\end{eqnarray}
Here the situation is less clear, because it is not obvious that the renormalized renormalizable couplings  ($\ref{r4NC}$) and ($\ref{r5NC}$) associated with the $\phi^3$ and $\phi^4$-like operators must be zero after the restoration of the STI. However, nonzero renormalizable interactions generate new contributions to the nonrenormalizable ones while integrating out field modes. This fact leads to the speculation whether  the cosmological constant may prevent the gravitational constant from dying off in the no-cutoff limit. 

To be more precise, the question is \textit{do we have a nonzero no-cutoff limit of the gravitational constant for $\Lambda_K \ne 0$},
\begin{equation}
\lim_{\Lambda_0 \rightarrow \infty} \lambda(\Lambda, \Lambda_0, \Lambda_K) :=  \lambda^{cont}(\Lambda, \Lambda_K) \ \ne \ 0 \ ?  \label{nonzg}
\end{equation}
Since the running nonrenormalizable couplings are determined by the renormalizable ones in the no-cutoff limit, this would in particular mean that the value of the renormalized gravitational constant is determined by the renormalized cosmological constant for $\Lambda_0 \rightarrow \infty$.  

To be very clear, the idea is \textit{not} to impose a renormalization (or, in our terminology, improvement) condition for the gravitational coupling, as we have done it in the effective field theory context in section ($\ref{STIRest}$), eq. ($\ref{physrenGK}$). The conjecture is rather that the gravitational constant should come out of the theory as a prediction. Note, however, that the Polchinski analysis works only for \textit{small} couplings as we have discussed in section ($\ref{RenFlowOver}$) and emphasized on various occasions. In particular, our dimensionless expansion parameters $\lambda_4^R$ and $\lambda_5^R$ must be small in the sense of eqns.  ($\ref{smallRG}$) and ($\ref{smallRGExp}$), and the bare nonrenormalizable couplings $\lambda_{\tilde{a}}^0, \ \tilde{a} \ge 6$, must be sufficiently small as it is implied by eq. ($\ref{iniNRG}$). Moreover, it has been discussed in eqns. ($\ref{smallm}$) and ($\ref{smallg}$) of section ($\ref{RGIs}$) that also the dimensionless versions of the couplings\footnote{Note that after the restoration of the STI we will have $A=\Lambda_K/\lambda$, $B_1=\Lambda_K$, $B_2=0$ as follows from eq. ($\ref{Stoy}$) and the graviton and ghost propagators ($\ref{gravpropG}$) and ($\ref{ghpropG}$).  } $A$, $B_1$ and $B_2$ appearing in the ''free'' part\footnote{$B_1$ and $B_2$ are the mass squares of the graviton and ghost propagators ($\ref{gravpropG}$) and ($\ref{ghpropG}$).}  of the gravity effective action ($\ref{SGraveffG}$) have to stay small on scales $\Lambda_R < \Lambda < \Lambda_0$,
\begin{eqnarray}
|\Lambda^{-2} B_i | &\le& 1, \ \ \ i=1,2 \label{smallmNC} \\
|\Lambda^{-3} A |   &\le&  \label{smallgNC} 1.
\end{eqnarray}
The point is that the requirement of small couplings amounts to additional constraints on the theory.

The no-cutoff value $\lambda^{cont}(\Lambda, \Lambda_K)$ of the gravitational constant can then in principle be obtained in the process of restoring the Slavnov-Taylor identities for quantum gravity. Once the renormalization conditions for the cosmological constant and the gravitational field have been specified, eq. ($\ref{PRCNC}$), one has to determine one particular set of ''arbitrary'' renormalization conditions ($\ref{rcG1}$), ($\ref{rcG2}$) and ($\ref{rccOIG}$) for the remaining couplings $\rho_{\tilde{a}}(\Lambda_R), \  \tilde{a}=1...5$, and $R_i(\Lambda_R), \ i=1,2$, such that eqns. ($\ref{Gismall1NC}$) are fulfilled \textit{and} the constraint of small couplings can be met. Note that in the more speculative scenario concerning the restoration of the STI that allows for nonvanishing BRS couplings $\mathcal{R}^0_i(\Lambda) \ne 0, \ i=3...5$, one has also to adjust the improvement conditions ($\ref{ImproOIGNC}$) such that that eqns. ($\ref{Gismall2NC}$) are satisfied.

If such a set of ''arbitrary'' renormalization (and improvement) conditions would exist, it should be uniquely defined in the limit $\Lambda_0 \rightarrow \infty$ because of the uniqueness of the no-cutoff limit, Theorem ($\ref{UniTh}$). Moreover, the no-cutoff value $\lambda^{cont}(\Lambda, \Lambda_K)$ of the gravitational constant could then be calculated in perturbation theory in the couplings $\lambda_4^R$ and $\lambda_5^R$ by means of the definition
\begin{eqnarray}
\lambda^{cont}(\Lambda, \Lambda_K) := \Big(\frac{1}{2} \partial_{i} \partial_{j} -\partial_{i}^2  \Big) \left(\frac{\delta}{\delta h(k_1)} \frac{\delta}{\delta h(k_2)} \frac{\delta}{\delta h(k_3)} L^{cont}(\Phi, \Lambda)\right) \Big|_{h=k_{i}=0}, \ \ \ \label{physrenGKNC}
\end{eqnarray}
that is by solving the Polchinski equation.

It follows from eq. ($\ref{Rphys}$) that a nonvanishing no-cutoff value for the gravitational constant would imply that after the restoration of the STI, the renormalized nonrenormalizable BRS couplings must be nonzero, too:
\begin{eqnarray}
\lim_{\Lambda_0 \rightarrow \infty} R_i(\Lambda_R, \Lambda_0) & \ne & 0, \ \ \ i=3...5.   \label{RNCzeroN}
\end{eqnarray}
As has been pointed out before, this can only be achieved if we give up the requirement of small bare nonrenormalizable BRS couplings ${R}^0_i, \ i=3...5$, as we have done in eqns.  ($\ref{divR0}$), ($\ref{divR}$), and impose improvement conditions for them, eq.  ($\ref{ImproOIGNC}$). In light of eq. ($\ref{Rphys}$) this might of course seem suspicious\footnote{A related issue might be that it is at the present point not clear whether it can be prevented that the graviton-ghost couplings go away in the no-cutoff limit. This would probably require renormalizable operators of the type $\rho_4 \int \overline{C} h C$ and $\rho_5 \int \overline{C} h^2 C$. However, we do dot see them to be BRS invariant.} because it can be interpreted in the way that we actually \textit{do} have to impose renomalization conditions for the gravitational coupling at some point. Otherwise, we would be forced in the conclusion that it must be zero in the no-cutoff limit.      

However, as we have shown in eqns. ($\ref{NCAb}$) and  ($\ref{NCAt}$), at least from a technical point of view it should be acceptable because the convergence of the vertex functions $A_{\beta n}(\Lambda)$ and $A_{\tau n}(\Lambda)$ carrying the BRS fields as operator insertions can be established in the no-cutoff limit for nonvanishing couplings $\mathcal{R}^0_i(\Lambda)\ne 0, \ i=3...5$. This is in contrast to imposing arbitrary improvement conditions for the nonrenormalizable gravity couplings $\lambda_{\tilde{a}}, \ \tilde{a}=6...8$, which would spoil convergence. Moreover, note that the restoration of the STI will \textit{not} be possible for arbitrary values of the gravitational constant because of the gravity potential $L(\Lambda)$ appearing in the vSTI ($\ref{vSTIL}$). For establishing the convergence of its vertex functions in ($\ref{contlimAG}$) we have imposed only the renormalization conditions ($\ref{rcG1}$)-($\ref{rcG2}$) and \textit{no} (arbitrary) improvement conditions. In particular, we have pointed out in eq. ($\ref{rcG3NC}$) that the couplings $\rho_{\tilde{a}}(\Lambda_R), \ \tilde{a}=6,7$ (which will be given by the gravitational coupling after the restoration of the STI, see Table ($\ref{tab}$)) are understood to be determined by the renormalizable ones $\rho_{\tilde{a}}(\Lambda_R), \tilde{a}=1...5$.

We have furthermore disussed that nonvanishing BRS couplings $\mathcal{R}^0_i(\Lambda)\ne 0, \ i=3...5$, force us to employ a more speculative scenario for the restoration of the STI, leading to eq. ($\ref{Gismall2NC}$). This scenario ultimatively amounts to proving that the renormalization \textit{and} improvement conditions ($\ref{FtR1NC}$) and ($\ref{FtR2NC}$)  of the functional $L_{(1)} (\Lambda)$ describing the violation of the STI vanish according to eq.  ($\ref{Gismall2NC}$). Hence, the same number of conditions is to be fulfilled as in the effective field theory context, see eqns. ($\ref{FtR1}$) and ($\ref{FtR2}$). However, in the present case we have less free parameters at hand that we may adjust, because the improvement conditions for the gravitational couplings $\rho_{\tilde{a}}, \ \tilde{a}=6...8,$ appearing in eq. ($\ref{Gismall2NC}$) are understood to be given in terms of no-cutoff limits ($\ref{rcG3NC}$) that are already determined by the renormalized renormalizable couplings $ \rho^{R}_{\tilde{a}}, \ \tilde{a}=1...5$.  One might, however, still hope for linear interdependencies in the conditions ($\ref{FtR1NC}$) and ($\ref{FtR2NC}$).

\medskip
Let us finally remark that the renormalizable $\phi^3$ and $\phi^4$-like operators that are introduced because of the cosmological constant can also be obtained in a different way. Consider some massive fields $\varphi_i$ having masses $m_i$ that couple to the gravitational field $h$. Employing the expansion ($\ref{KKexp}$) we obtain schematically
\begin{equation}
\int_x \sqrt{g} m_i^2 \varphi_i^2  \sim \int_x \left( m_i^2 \varphi_i^2  + \lambda m_i^2 \varphi_i^2 h + \lambda^2 m_i^2 \varphi_i^2 h^2 + ... \right) \label{higgsgrav}
\end{equation}
Thus, we again observe renormalizable interactions $\varphi_i^2 h$ and $\varphi_i^2 h^2$ with ''mixed'' coupling constants $\lambda m_i^2$ and $\lambda^2 m_i^2$. The latter have canonical dimensions $D_{\lambda m_i^2}=1$ and $D_{\lambda^2 m_i^2}=0$ in analogy to the couplings $\rho_4$ and $\rho_5$, see Table ($\ref{tab}$). The speculations concerning a nonvanishing no-cutoff limit of the gravitational constant for a nonzero renormalized cosmological constant may therefore be extended to the case of massive fields coupling to the gravitational field. The gravitational constant might then be determined by the cosmological constant and the masses of the elementary particles in the no-cutoff limit. We are not sure whether this might be related to induced gravity \cite{Adler} \cite{Sak}.

Note that what has been said above could have yet another implication. Since all Standard Model operators have canonical dimension $D_{\mathcal{O}}=4$, coupling them to the gravitational  field would yield operators of canonical dimension $D_{\mathcal{O}} \ge 5$. Hence, these would be all nonrenormalizable. The only exception forms the Higgs field, which couples via its tachyonic mass term to the gravitational field in the manner of eq. ($\ref{higgsgrav}$). One might therefore speculate whether there is a deeper connection between the mechansim that gives mass to all elementary particles, and the force that acts on it.

\end{section}

\end{section}
\end{chapter}

\chapter{Summary and Outlook} \label{Outlook}

In this thesis Euclidean quantum gravity was analyzed from the viewpoint of the renormalization group flow equations. The analysis is based on methods introduced by J. Polchinski \cite{Pol} concerning the perturbative renormalization of field theories via flow equations.

As a first step, we gave the proof of perturbative renormalizability of scalar $\phi^3+\phi^4$ field theory in the framework of renormalization with flow equations. Some generalizations as compared to earlier works \cite{KKS}, \cite{Pol} have been included. 

We then proceeded by extending the concepts of renormalization via flow equations to effective field theories that have a finite cutoff. This was again done for a scalar field theory by imposing additional renormalization conditions for some of the nonrenormalizable couplings. The additional renormalization conditions have been named ''improvement conditions'', and it has been pointed out that our treatment in particular applies to nonrenormalizable theories that do not allow for \textit{any} renormalizable interactions. As a result of our analysis we have established that effective field theories are predictive at scales $\Lambda$ far below the cutoff $\Lambda_0$ with finite accuracy $(\Lambda/\Lambda_0)^{|D_{\rho_l}|+s}$. Here, $D_{\rho_l}$ denotes the canonical dimension of the least irrelevant coupling of the QFT, and $s$ is the ''improvement index'' that refers to the number of improvement conditions that have been imposed.

Turning to quantum gravity, the standard  covariant BRS quantization procedure for Euclidean Einstein gravity with a cosmological constant was reviewed. As a dynamical variable we used a perturbation of the metric density $\sqrt{g} \ g^{\mu \nu}$ around flat space. We introduced a momentum cutoff regularization for the generating functional of quantum gravity and discussed the resulting violation of the gauge invariance and hence the Slavnov-Taylor-identities (STI). Polchinski's renormalization group equation for Euclidean quantum gravity has been derived. 

Applying the methods that we developed for analyzing effective field theories with flow equations to quantum gravity, we disregarded in a first step of our analysis the violation of the STI. A set of arbitrary renormalization and improvement conditions has been imposed, and by inverting the renormalization group trajectory it has been shown that the improvement conditions force the UV cutoff of effective quantum gravity to be the Planck scale $M_P$. We then have established that for generic bare gravity actions the family of theories described by the arbitrary renormalization and improvement conditions is predictive at scales $\Lambda$ far below the Planck scale $M_P$ with finite accuracy $(\Lambda/M_P)^2$. 

Proceeding with the restoration of the STI, bare regularized BRS variations of the graviton and ghost fields have been introduced and the violated Slavnov-Taylor identities (vSTI) of quantum gravity have been worked out. Extending our concepts to composite operator renormalization, it has been found that the gravity BRS fields contain {nonrenormalizable parts}. By introducing renormalization and improvement conditions for them, it has then been established that  vertex functions carrying the BRS fields as operator insertions are known at scales $\Lambda$ far below the Planck scale $M_P$ with finite accuracy $(\Lambda/M_P)^2$. We furthermore have found that the violation of the STI can be described in terms of vertex functions carrying the BRS variation of the bare gravity action as a space-time integrated operator insertion. It has therefore been argued that the STI can be restored to {finite accuracy} if {one particular} set of arbitrary renormalization and improvement conditions for the couplings and BRS fields can be found such that the relevant and leading irrelevant parts of the vertex functions describing the violation of the STI are driven small as $(\Lambda/M_P)^2$ at scales $\Lambda$ far below the Planck scale $M_P$. 

Finally, we considered the no-cutoff limit $\Lambda_0 \rightarrow \infty$ of quantum gravity in the framework of renormalization with flow equations. The vertex functions of the gravity effective potential were expanded solely in the renormalizable couplings, and their boundedness and convergence  has been established in the limit $\Lambda_0 \rightarrow \infty$. Attempting to apply the same program to the vertex functions carrying the  BRS fields as operator insertions, we observed that the nonrenormalizable parts of the gravity BRS fields vanish in the no-cutoff limit if smallness of the bare BRS couplings is imposed. It could however be shown that if the latter constraint is dropped, convergence of the BRS vertex functions may still be proven. Proceeding with the restoration of the STI, we argued that for zero renormalized cosmological constant $\Lambda_K=0$ the theory will become free as $\Lambda_0 \rightarrow \infty$, and that the latter statement is compatible with gauge invariance. It was speculated whether a \textit{nonzero} cosmological constant $\Lambda_K \ne 0$ leads to a \textit{nonvanishing} value of the gravitational constant in the no-cutoff limit, and we pointed out that the gravitational coupling should then be determined by the cosmological constant. Finally, we conjectured that a similar situation might arise if  massive fields are coupled to gravity. This led to the speculation whether for $\Lambda_0 \rightarrow \infty$ the gravitational constant is given in terms of the cosmological constant and the masses of the elementary particles.

\medskip
Let us now briefly discuss some relations of our results to the work of others. We begin with a comparison to treatments of quantum gravity as an effective field theory in the ''conventional''  perturbative framework. Such an analysis has been carried out by J. Donoghue, N.E.J. Bjerrum-Bohr and others \cite{Dono1} \cite{BB2} \cite{BB3}. Generally speaking, the main advantage of their approach as compared to ours is that they may employ dimensional regularization, which is a symmetry respecting regulator. Hence, they do not get involved in the complicated discussion of restoring the violated Slavnov-Taylor identities. On the other hand, our treatment involving a specific number of improvement conditions which then lead to a defined predictivity of the effective theory is maybe more systematic and transparent from a conceptual point of view. Moreover, we employ the modern language of the renormalization group. 

Let us furthermore stress that we resctricted our considerations to \textit{pure} Euclidean gravity, whereas the focus of the above authors lies in the derivation of quantum corrections to Newton's potential. Thus, they mainly deal with the gravitational interaction of matter. 

As we have discussed, our results for pure effective quantum gravity are such that the vertex functions of the theory are known at the scale $\Lambda$ to an accuracy of $(\Lambda/M_P)^2$ where $M_P$ is the Plack scale. On the other hand, this means that already to second order in perturbation theory in the gravitational coupling $\lambda \sim M_P^{-1}$, the perturbative contributions will be of the same order as the indetermination of the vertex functions\footnote{Perturbative contributions that stem from the ''mixed'' couplings $\lambda \Lambda_K$,  $\lambda^2 \Lambda_K$ etc. appearing for nonvanishing cosmological constant $\Lambda_K$ will be neglectible because of the smallness of $\Lambda_K$.}. In particular, we will not be able to calculate quantum corrections because already the 1-loop correction to the graviton propagator (the ''vacuum polarization'') involves two $3$-graviton vertices.   

However, when deriving the quantum corrections to Newton's potential in \cite{Dono1} \cite{BB2}, the authors actually \textit{do} extract quantum corrections out of the graviton vacuum polarization diagram. This is not a contradiction to our results by the following reasoning. In our considerations we have always assumed that we may perform a derivative expansion of the gravity effective potential into local composite field operators. Equivalently, the momentum space vertex functions of the effective potential have been Taylor expanded around vanishing momenta $k_i=0$. This procedure was justified because we kept a nonvanishing effective IR cutoff. However, the low energy propagation of massless particles (such as the graviton\footnote{Modulo ''mass'' squares that appear for nonvanishing consmological constant $\Lambda_K$ in the graviton propagator. However, we have always kept the IR cutoff $\Lambda$ above $\Lambda_K$, $\Lambda^2 > |\Lambda_K|$.}) leads to nonanalytic contributions that cannot be expanded in a Taylor series. It can be shown that these contributions lead to quantum corrections that are dominant over the analytic ones, and it has been such quantum effects\footnote{The (Minkowski space) contribution is $\sim \ln(-k^2)$.} that have been extracted out of the 1-loop correction to the graviton propagator by the authors mentioned. 

We will now come to the discussion of the relation of our work to a \textit{nonperturbative} analysis of Euclidean quantum gravity with flow equations that has been performed by M. Reuter and others \cite{MR1}  \cite{MR3}  \cite{MR4}. As has been mentioned on various occasions, the Polchinski analysis works only for sufficiently small couplings. Hence, it implictly relies on the Gaussian fixed point in the space of couplings as does perturbation theory. However, the {nonperturbative} investigation of the renormalization group flow of gravity employing truncations of the space of actions seems to indicate that the UV behaviour of quantum gravity is governed by a nontrivial (i.e. non-Gaussian) fixed point \cite{MR2}. As long as we impose our effective field theory scenario for quantum gravity, meaning that we retain a finite UV cutoff $\Lambda_0=M_P$, this is again not in contradiction to our analysis because the regime of the conjectured non-Gaussian fixed point sets in at scales $\Lambda$ \textit{above} the Planck scale $M_P$, $\Lambda > M_P$ \cite{MR6}. At scales far \textit{below} the Planck scale, the theory is governed by the Gaussian fixed point and our results may be applied. These in particular imply that for generic bare gravity actions defined at the Planck scale $M_P$, the effective Lagrangian of quantum gravity is attracted towards a finite dimensional submanifold in the space of possible Lagrangians at scales $\Lambda << M_P$.  To be more precise, the analysis of section ($\ref{ToyM}$) has shown that for tiny renormalized values of the cosmological constant $\Lambda_K \sim 10^{-120} M_P^2$ and a renormalized gravitational constant $\lambda \sim M_P^{-1}$ the effective Lagrangian of quantum gravity (and hence all couplings associated with the higher field invariants) are effectively determined by the gravitational coupling\footnote{For larger values of the cosmological constant $\Lambda_K$, the effective Lagrangian is determined by $\lambda$ \textit{and} $\Lambda_K$ to an accuracy of $(\Lambda/M_P)^2$.} to an accuracy of $(\Lambda/M_P)^2$. This behaviour should  be also visible in a study of the RG trajectories of the running gravity couplings that is carried out by truncating the space of actions. However, due to the extreme algebraic complexity of the respective $\beta$-functions the authors mentioned have mostly considered an Einstein-Hilbert truncation that involves only the running cosmological constant and the running gravitational coupling. This makes it difficult to recover the convergence of the gravity couplings to a finite dimensional submanifold at scales $\Lambda << M_P$ from their results.

Let us also mention that our speculations on the no-cutoff limit of quantum gravity in section ($\ref{GravNoCut}$) still rely on the Gaussian fixed point. If such a no-cutoff limit (with nonvanishing gravitational constant $\lambda$) would indeed exist, it might therefore be in conflict with the results concerning the nontrivial fixed point that have been found by the above authors.

Finally, we would like to point out that the vertex functions of the effective potential for which the RG analysis has been performed in this work are more or less equivalent to the connected Greens functions. This has been shown in Appendix ($\ref{LZ}$). On the other hand, the nonperturbative RG analysis of M. Reuter et al. has been carried out for one particle irreducible (1PI) vertex functions. This is another important difference between their approach and ours.  

\medskip
To conclude with an outlook, we would like to stress once more that the methods developed in this work for investigating effective field theories with flow equations should be applicable to \textit{any} effective field theory, not just gravity.

Concering our analysis of effective quantum gravity with flow equations, it remains to show that there actually exists a set of arbitrary renormalization and improvement conditions for the couplings and the BRS fields such that the restoration of the Slavnov-Taylor identities with finite accuracy can be accomplished. The analogous procedure for Yang-Mills theory \cite{KM} \cite{Mull} suggests that this will be a rather tedious task.

Finally, it should be investigated whether for nonzero cosmological constant there does indeed exist a choice of arbitrary renormalization  conditions such that in the no-cutoff limit the Slavnov-Taylor identities can be restored, the requirement of small couplings can be met and, finally, a \textit{nonvanishing} value of the renormalized gravitational constant is obtained. One should carry out the same program for the case of a massive field coupled to the gravitational field.




\begin{appendix}

\chapter{}

\begin{section}{Canonical dimensions of fields}   \label{D}
The canoncial dimension $D_\phi$ of a field $\phi(x)$ of some QFT is determined by the kinetic term (or ''free'' part) of its action,
\begin{eqnarray}
S_f=- \frac{1}{2} \int_{xy} \phi(x) \Delta^{-1}(x-y) \phi(y)
\end{eqnarray}
where $\Delta^{-1}(x-y) $ denotes the inverse propagator. The action must be dimensionless\footnote{We employ $c=\hbar =1$.}, and hence $D_\phi$ follows from the number of space-time dimensions $d$ and the properties of $\Delta^{-1}(x-y)$. To be more precise, it turns out that if the Fourier transformed propagator $\tilde{\Delta}(k)$ has a large momentum behaviour
\begin{equation}
| \tilde{\Delta}(\alpha k) | \sim C \alpha^{-\sigma}, \  \ k\ne 0,  \ \ \alpha \rightarrow \infty,
\end{equation}
where $C$ is some positive constant and $\alpha$ is a parameter, then \cite{Z-J}
\begin{equation}
 D_\phi  = {\frac{1}{2}(d-\sigma)} .  \label{Dphix}
\end{equation}
Consequently, the Fourier transformed field $\tilde{\phi}(k)$ has canonical dimension
\begin{equation}
 D_{\tilde{\phi}}  = -{\frac{1}{2}(d+\sigma)}.  \label{Dphip}
\end{equation}
In this work, we will consider $d=4$ space-time dimensions and deal with  $1/k^2$ (momentum-space) propagators (meaning that $\sigma=2$). Thus we have
\begin{equation}
 D_\phi  = 1, \ \ \ \ D_{\tilde{\phi}} =-3.
\end{equation}

\end{section}

\begin{section}{Derivative expansion of the effective potential}  \label{MDV}
The effective potential $L(\phi, \Lambda)$ introduced in eq. ($\ref{Spol}$) can be expanded in powers of the fields $\phi(x)$:
\begin{eqnarray}
L(\phi, \Lambda) &=& \sum_{n=1}^\infty \frac{1}{n!} \int d^4x_1 ... d^4x_n L_n(x_1...x_n, \Lambda) \phi(x_1)...\phi(x_n) . \label{Lexps}
\end{eqnarray}
We call the expansion coefficients $L_n(x_1...x_n, \Lambda)$ position space vertex functions. They are obviously nonlocal objects. Formally, ($\ref{Lexps}$) can always be written as a local expression plus a nonlocal remainder term by expanding all fields $\phi(x_i)$ into power series around a common point $x$. Such a procedure is called derivative expansion, and we will now do it explicitly for the example
\begin{eqnarray}
\frac{1}{3!} \int d^4x_1 d^4x_2 d^4x_3 L_3(x_1, x_2, x_3, \Lambda) \phi(x_1) \phi(x_2) \phi(x_3) . \label{Exexp}
\end{eqnarray}
We choose $x_1$ as an expansion point for the fields $\phi(x_i)$:
\begin{eqnarray}
\phi(x_i) = \phi(x_1) + \partial_\mu \phi(x_1) (x_i - x_1)^\mu + \frac{1}{2} \partial_\mu \partial_\nu \phi(x_1) (x_i - x_1)^\mu (x_i - x_1)^\nu + ...
\end{eqnarray}
Furthermore, because of the homogeneity of space-time we have
\begin{eqnarray}
L_3(x_1, x_2, x_3, \Lambda) = L_3(\overline{x}_2, \overline{x}_3, \Lambda)
\end{eqnarray}
where $\overline{x}_2:=x_2-x_1$, $\overline{x}_3:=x_3-x_1$. Thus we may write
\begin{eqnarray}
 (\ref{Exexp})&=& \int_{x_1}  \phi(x_1)^3 \int_{\overline{x}_2 \overline{x}_3} L_3(\overline{x}_2, \overline{x}_3, \Lambda)  +  \int_{x_1} \phi(x_1)^2 \partial_\mu \phi(x_1) \sum_{i=2}^3  \int_{ \overline{x}_2 \overline{x}_3} \overline{x}_i^\mu L_3(\overline{x}_2, \overline{x}_3, \Lambda) \nonumber \\ && + \ \int_{x_1} \phi(x_1)^2 \partial_\mu \partial_\nu \phi(x_1) \sum_{i=2}^3 \frac{1}{2} \int_{\overline{x}_2 \overline{x}_3} \overline{x}_i^\mu \overline{x}_i^\nu L_3(\overline{x}_2, \overline{x}_3, \Lambda)  \nonumber \\ &&  + \ \int_{x_1} \phi(x_1) \partial_\mu \phi(x_1) \partial_\nu \phi(x_1)  \int_{\overline{x}_2 \overline{x}_3} \overline{x}_2^\mu \overline{x}_3^\nu L_3(\overline{x}_2, \overline{x}_3, \Lambda) + ... \ . \label{DExp1}
\end{eqnarray}
Let us analyze the different terms of ($\ref{DExp1}$). To do so, consider $G^{\mu \nu} \in SO(4)$. Then
\begin{eqnarray}
G^{\mu \nu} \int_{\overline{x}_2 \overline{x}_3}   \overline{x}^{\nu}_{i} L_3(\overline{x}_2, \overline{x}_3, \Lambda) &=&  \int_{\overline{x}_2 \overline{x}_3}    G^{\mu \nu} \overline{x}^{\nu}_{i} L_3(G \overline{x}_2, G \overline{x}_3, \Lambda) \nonumber \\ &=& \int_{\overline{x}_2 \overline{x}_3}    \overline{x}^{\nu}_{i} L_3(\overline{x}_2, \overline{x}_3, \Lambda).  \label{tenint}
\end{eqnarray}
Hence, the invariance of $L_3$ under the orthogonal group and the tensor structure of the integral in ($\ref{tenint}$) force
\begin{eqnarray}
\int_{\overline{x}_2 \overline{x}_3}  \overline{x}^{\nu}_{i} L_3(\overline{x}_2, \overline{x}_3, \Lambda)  =0.  \label{oddvan}
\end{eqnarray}
With a similar argument and the symmetry $L_3(\overline{x}_2, \overline{x}_3, \Lambda)=L_3(\overline{x}_3, \overline{x}_2, \Lambda)$ we conclude that
\begin{eqnarray}
-\int_{\overline{x_2} \overline{x_3}}  \ \overline{x}^{\mu}_{i} \overline{x}^{\nu}_{i} L_3(\overline{x}_2, \overline{x}_3, \Lambda) &=& \delta^{\mu \nu} \rho_{6}(\Lambda), \ \ \ i=2,3 \label {rcc6ps} \\
-\int_{\overline{x_2} \overline{x_3}}  \ \overline{x}^{\mu}_{i} \overline{x}^{\nu}_{j} L_3(\overline{x}_2, \overline{x}_3, \Lambda) &=& \delta^{\mu \nu} \rho_{7}(\Lambda), \ \ \  i \ne j = 2,3
\end{eqnarray}
where we introduced the running coupling constants $\rho_{6}(\Lambda)$ and $\rho_{7}(\Lambda)$. If, in addition, we define
\begin{eqnarray}
\rho_4(\Lambda):=  \int_{\overline{x}_2 \overline{x}_3} L_3(\overline{x}_2, \overline{x}_3, \Lambda),   \label{rcc4ps}
\end{eqnarray}
the expansion ($\ref{DExp1}$) becomes
\begin{eqnarray}
 (\ref{Exexp})&=& \rho_4(\Lambda) \int_{x_1}  \phi(x_1)^3 - \rho_{6}(\Lambda) \int_{x_1} \phi(x_1)^2 \partial_\mu \partial^\mu \phi(x_1) \nonumber \\ && \hspace{3.5cm} - \ \rho_{7}(\Lambda) \int_{x_1} \phi(x_1) \partial_\mu \phi(x_1) \partial^\mu \phi(x_1) + ... \ . \label{DExp2}
\end{eqnarray}
Note that the last two operators of ($\ref{DExp2}$) are not linearly independent\footnote{Contrary to what is stated in \cite{Wiec}.}, because integration by parts yields
\begin{eqnarray}
\int_{x_1} \phi(x_1)^2 \partial_\mu \partial^\mu \phi(x_1) = -2 \int_{x_1} \phi(x_1) \partial_\mu \phi(x_1) \partial^\mu \phi(x_1).   \label{lindep}
\end{eqnarray}
Hence, we arrive at
\begin{eqnarray}
 (\ref{Exexp})= \rho_4(\Lambda) \int_{x_1}  \phi(x_1)^3 + \rho_{6/7} \int_{x_1} \phi(x_1)^2 \partial_\mu \partial^\mu \phi(x_1)  \label{DExp3}
\end{eqnarray}
where
\begin{eqnarray}
\rho_{6/7} :=-\rho_{6}(\Lambda) + \frac{1}{2} \rho_{7}(\Lambda).
\end{eqnarray}
In analogy to what we have shown for the example ($\ref{Exexp}$), we may proceed for other vertex functions. Introducing the running coupling constants
\begin{eqnarray}
\rho_1(\Lambda) &:=& L_1(\Lambda) \label{rcc1ps} \\
\rho_2(\Lambda) &:=& \int_{\overline{x}_2} L_2(\overline{x}_2,\Lambda) \\
\rho_3(\Lambda) \ \delta^{\mu \nu} &:=& - \int_{\overline{x}_2}  \ \overline{x}_2^{\mu} \overline{x}_2^{\nu} L_2(\overline{x}_2, \Lambda) \\
\rho_5(\Lambda) &:=& \int_{\overline{x}_2 \overline{x}_3 \overline{x}_4} L_4(\overline{x}_2,\overline{x}_3,\overline{x}_4,\Lambda) \\
\rho_8(\Lambda) &:=& \int_{\overline{x}_2 \overline{x}_3 \overline{x}_4 \overline{x}_5} L_5(\overline{x}_2,\overline{x}_3,\overline{x}_4, \overline{x}_5, \Lambda) \label{rcc8ps}
\end{eqnarray}
then allows us to write
\begin{eqnarray}
L(\phi, \Lambda) &=& \int d^4x \Big( \rho_1(\Lambda)  \phi(x) +  \rho_2(\Lambda) \phi(x)^2 - \rho_3(\Lambda) \phi(x) \partial_\mu \partial^\mu \phi(x)   \nonumber \\ &&  +  \ \rho_4(\Lambda) \phi(x)^3 + \rho_5(\Lambda) \phi^4 + \rho_{6/7} \phi(x)^2 \partial_\mu \partial^\mu \phi(x)  + \rho_8(\Lambda) \phi^5 \Big)  + R^{(1)}(\phi) \nonumber \\   \label{Dexp3}
\end{eqnarray}
where $ R^{(1)}(\phi)$ is a nonlocal remainder term.  Eq. ($\ref{Dexp3}$) is a derivative expansion of the effective potential $L(\phi, \Lambda)$ into local composite field operators $\mathcal{O}_i(x, \phi)$ of canonical dimension\footnote{The field $\phi(x)$ has canonical dimension $D_\phi=1$ in $d=4$ space-time dimensions. See section ($\ref{GENR}$) for the determination of the canonical dimension of a field.} $D_{\mathcal{O}_i} \le 5$. The latter are associated with running coupling constants $\rho_{i}(\Lambda) $ of canonical dimension
\begin{eqnarray}
D_{\rho_i} = 4 - D_{\mathcal{O}_i} \ge -1.
\end{eqnarray}
Thus, ($\ref{Dexp3}$) can be written as
\begin{eqnarray}
L(\phi, \Lambda) &=& \sum_{D_{\rho_i } \ge -1} \rho_i(\Lambda) \int_x \mathcal{O}_i(x, \phi) +  R^{(1)}(\phi) . \label{effpotexp}
\end{eqnarray} 
Eq. ($\ref{effpotexp}$) corresponds to the case $s=1$ in the general expression ($\ref{effpot}$) for a derivative expansion of $L(\phi, \Lambda)$  given in section ($\ref{RGInt}$).

So far in this appendix, we have worked in position space. However, in chapters ($\ref{RenFlow}$) and ($\ref{EffFlow}$) the objects of interest are momentum space vertex functions $L_{n} (k_1,...,k_{n}, \Lambda)$. In eqns. ($\ref{rcc1}$)-($\ref{rcc5}$) and ($\ref{rcc6}$)-($\ref{rcc8}$) these are even used to define the running coupling constants $\rho_i(\Lambda)$ via Taylor expansions around $k_i=0$. In the following, we will show that the coupling constants of eqns. ($\ref{rcc1}$)-($\ref{rcc5}$) and ($\ref{rcc6}$)-($\ref{rcc8}$) are identical to those defined in eqns. ($\ref{rcc6ps}$)-($\ref{rcc4ps}$) and ($\ref{rcc1ps}$)-($\ref{rcc8ps}$).

To do so, we will again consider the example of $L_3(x_1,x_2,x_3)$. Its Fourier transform is given by
\begin{eqnarray}
\delta(k_1+k_2+k_3) \ L_3(k_1,k_2,k_3,\Lambda) &=& \int_{x_1 x_2 x_3} e^{i k_1 x_1} e^{i k_2 x_2} e^{i k_3 x_3} L_3(\overline{x}_2, \overline{x}_3, \Lambda) \nonumber \\ &=& \delta(k_1+k_2+k_3) \ \int_{\overline{x}_2, \overline{x}_3} e^{i k_2 \overline{x}_2} e^{i k_3 \overline{x}_3} L_3(\overline{x}_2, \overline{x}_3, \Lambda) \nonumber . \\
\end{eqnarray}
Integrating over $k_1$ yields
\begin{eqnarray}
 L_3(-k_2 -k_3, k_2, k_3,\Lambda) =  \int_{\overline{x}_2, \overline{x}_3} e^{i k_2 \overline{x}_2} e^{i k_3 \overline{x}_3} L_3(\overline{x}_2, \overline{x}_3, \Lambda) .
\end{eqnarray}
Therefore, we obtain the following identities:\footnote{Remember that $\partial^{\mu_1}_{i,1} L_3(\tilde{k}_1,\tilde{k}_2,\tilde{k}_3,\Lambda)|_{\tilde{k}_i=0} = \partial^{\mu_1}_{i} L_3(-\tilde{k}_2-\tilde{k}_3,\tilde{k}_2,\tilde{k}_3,\Lambda)|_{\tilde{k}_i=0}$.}
\begin{eqnarray}
L_3(k_1, k_2, k_3,\Lambda)|_{{k}_i=0}  &=&  \int_{\overline{x}_2, \overline{x}_3} L_3(\overline{x}_2, \overline{x}_3, \Lambda) \label{md0} \\
\partial^{\mu}_{i,1} L_3(\tilde{k}_1,\tilde{k}_2,\tilde{k}_3,\Lambda)|_{\tilde{k}_i=0} &=&  i \int_{\overline{x}_2 \overline{x}_3}   \overline{x}^{\mu}_{i} L_3(\overline{x}_2, \overline{x}_3, \Lambda) \label{md1}  \\
\partial^{\mu_1}_{i_1,1} \partial^{\mu_2}_{i_2,1} L_3(\tilde{k}_1,\tilde{k}_2,\tilde{k}_3,\Lambda)|_{\tilde{k}_i=0} &=& -  \int_{\overline{x_2} \overline{x_3}}  \ \overline{x}^{\mu_1}_{i_1} \overline{x}^{\mu_2}_{i_2} L_3(\overline{x}_2, \overline{x}_3, \Lambda) \label{md2} .
\end{eqnarray}
From ($\ref{md1}$) and ($\ref{oddvan}$) follows that odd numbers of momentum derivatives of vertex functions vanish. Furthermore, eqns. ($\ref{md0}$) and ($\ref{md2}$) show that the definitions of $\rho_4(\Lambda)$, $\rho_6(\Lambda)$ and $\rho_7(\Lambda)$ given in eqns. ($\ref{rcc4}$), ($\ref{rcc6}$) and ($\ref{rcc7}$) are indeed identical to those of eqns. ($\ref{rcc6ps}$)-($\ref{rcc4ps}$). The same can be shown for the remaining coupling constants $\rho_i(\Lambda)$. Hence, if we insert the Taylor expansions ($\ref{TE0}$), ($\ref{TE1}$) and ($\ref{TE2}$) respectively of the momentum space vertex functions $L_{n} (k_1,...,k_{n}, \Lambda)$ into the expansion ($\ref{Lexp}$) of the effective potential, we just obtain the momentum space version of the derivative expansion ($\ref{effpot}$).

To conclude this section, we demonstrate how the argument ($\ref{lindep}$) for the linear dependence of the operators associated with the couplings $\rho_6(\Lambda)$ and $\rho_7(\Lambda)$ translates into momentum space. If we insert the parts of the Taylor expansion ($\ref{TE2}$) associated with the couplings $\rho_6(\Lambda)$ and $\rho_7(\Lambda)$ into ($\ref{Lexp}$)\footnote{It is understood that the integration over $k_1$ has been carried out in ($\ref{Lexp}$).}, we obtain
\begin{eqnarray}
\int_{k_2 k_3} \Big( \rho_{6} \sum_{i=2}^3 k_i^2 + \rho_{7} \sum_{i \ne j=2}^3 k_i k_j \Big) \ \phi(k_2) \phi(k_3) \phi(-k_2 -k_3) .  \label{r7r8m}
\end{eqnarray}
Substituting  $k_2$ by $\tilde{k}_2= -k_2 -k_3$ we have 
\begin{eqnarray}
 \int_{k_3} \phi(k_3) \int_{k_2} k_2^2 \ \phi(k_2) \phi(-k_2 -k_3) = \int_{k_3} \phi(k_3) \int_{\tilde{k}_2} (-\tilde{k}_2-k_3)^2 \phi(\tilde{k}_2) \phi(- \tilde{k}_2 -k_3) \nonumber 
\end{eqnarray}
and therefore
\begin{eqnarray}
-\frac{1}{2}  \int_{k_2 k_3} k_2^2 \ \phi(k_2) \phi(k_3) \phi(-k_2 -k_3) =    \int_{k_2 k_3} k_2 k_3 \ \phi(k_2) \phi(k_3) \phi(-k_2 -k_3) . \nonumber \\
\end{eqnarray}
Thus eq. ($\ref{r7r8m}$) becomes
\begin{eqnarray}
(\ref{r7r8m}) &=& -2 \rho_{6/7}(\Lambda) \int_{k_2 k_3} k_2^2 \ \phi(k_2) \phi(k_3) \phi(-k_2 -k_3)    \label{gravm}
\end{eqnarray}
in full correspondence with the postition space equation ($\ref{DExp3}$). 

\end{section}

\begin{section}{Linearized renormalization group theory}  \label{LRG}
Let $S_e(\Lambda)$ be an effective action. Its dependence on the scale $\Lambda$ is given by a renormalization group equation (RGE)
\begin{equation}
- \Lambda \frac{d}{d \Lambda} S_e(\Lambda) = \mathcal{F}(S_e(\Lambda)) . \label{RG_1}
\end{equation}
At a \textit{fixed point} $S_e^{*}$ we have
\begin{equation}
- \Lambda \frac{d}{d \Lambda} S^{*}_e(\Lambda)= 0. 
\end{equation}
Now consider small deviations 
\begin{equation}
\delta S_e (\Lambda):= S_e(\Lambda) - \overline{S}_e(\Lambda)
\end{equation}
from a solution $\overline{S}_e(\Lambda)$ of eq. ($\ref{RG_1}$). They satisfy a linearized RGE
\begin{equation}
- \Lambda \frac{d}{d \Lambda} \delta S_e(\Lambda) = M( \overline{S}_e(\Lambda)) \delta S_e(\Lambda) .
\end{equation}
For the rest of this section, we choose $\overline{S}_e(\Lambda)= S_e^{*}$. Then $M$ becomes independent of the scale $\Lambda$, and the eigenvalue equation
\begin{equation}
M \mathcal{O}_i = \xi_i \mathcal{O}_i 
\end{equation}
defines scaling exponents $\xi_i$ and a set of eigenoperators $\mathcal{O}_i$ (which is assumed to be complete). We may expand $\delta S $ with respect to the $\mathcal{O}_i$:
\begin{equation}
\delta S (\Lambda)= \sum_i \mu_i (\Lambda) \mathcal{O}_i .
\end{equation}
The ''scaling fields'' $\mu_i(\Lambda)$ obey
\begin{equation}
- \Lambda \frac{d}{d \Lambda} \mu_i (\Lambda)  = \xi_i \mu_i (\Lambda)  \label{RG_EW}
\end{equation}
which has the solution
\begin{equation}
\mu_i (\Lambda)= \mu_i(\Lambda_0) \left( \frac{\Lambda}{\Lambda_0 } \right)^{-\xi_i} .
\end{equation}
Near the fixed point $S^{*}_e$, we have
\begin{eqnarray}
S_e (\Lambda) &=& S_e^{*} + \sum_i \mu_i (\Lambda) \mathcal{O}_i \\ &=& S_e^{*} + \sum_i \mu_i(\Lambda_0) \left( \frac{\Lambda}{\Lambda_0 } \right)^{-\xi_i} \mathcal{O}_i .
\end{eqnarray}
There are three different kinds of operators $\mathcal{O}_i $ associated with the eigenvalues $\xi_i $. They are called ''scaling operators'' and are classified as follows:
\begin{itemize}
\item $\xi_i > 0$: The associated scaling field $\mu_i$ is \textit{relevant} because it brings the action $S_e(\Lambda)$ away from the fixed point $S_e^{*}$ when $\frac{\Lambda}{\Lambda_0} \rightarrow 0$. 
\item $\xi_i < 0$: The associated scaling field $\mu_i$ is \textit{irrelevant} because it decays to zero when $\frac{\Lambda}{\Lambda_0} \rightarrow 0$. $S_e(\Lambda)$ converges towards $S_e^{*}$.
\item $\xi_i = 0$: The associated scaling field $\mu_i$ is \textit{marginal}. Then $S_e^{*} + \mu_i \mathcal{O}_i$ is again a fixed point for any $\mu_i$. The latter property may be destroyed beyond the linear order.
\end{itemize}  
For the Gaussian fixed point, that is in the limit of vanishing coupling, the corresponding eigenvalues of the RG transformation are given precisely by the canonical dimension of the couplings. In order to evaluate this somewhat  further, consider couplings $\rho_i(\Lambda)$ which have canonical dimensions of
\begin{equation}
[\rho_i]= D_{\rho_i} .
\end{equation}
We may define dimensionless couplings by
\begin{equation}
\lambda_i(\Lambda):= \Lambda^{-D_{\rho_i}} \rho_i(\Lambda).
\end{equation}
These couplings $\lambda_i$ depend on the scale $\Lambda$ according to RG equations
\begin{equation}
- \Lambda \frac{d}{d \Lambda} \lambda_i(\Lambda) = \beta_i(\lambda(\Lambda)). \label{RG_dl}
\end{equation}
Eq. ($\ref{RG_dl}$) is the analogue to ($\ref{RG_1}$). Again, we consider small deviations
\begin{equation}
\delta \lambda_i (\Lambda):= \lambda_i (\Lambda) - \overline{\lambda}_i(\Lambda)
\end{equation}
which obey a linearized RG equation
\begin{eqnarray}
- \Lambda \frac{d}{d \Lambda} \delta \lambda_i(\Lambda) &=& M_{ij} (\overline{\lambda}(\Lambda)) \delta \lambda_j(\Lambda) \label{RG_dlL} \\ M_{ij} (\overline{\lambda}(\Lambda)) &=& \frac{\partial}{\partial \lambda_j} \xi_i(\overline{\lambda}(\Lambda)) . \label{RG_M}
\end{eqnarray}
In the limit of vanishing coupling, that is $\overline{\lambda}_i(\Lambda) \rightarrow 0 $, the $\rho_i$ become independent of the scale\footnote{In perturbation theory, this can be understood by noting that the $\Lambda$-dependence of $\rho_i$ arises through contributions from diagrams. These vanish if the couplings go to zero.} $\Lambda$ and we have
\begin{eqnarray}
\delta \lambda_i(\Lambda) & \equiv & \mu_i(\Lambda) \\ M_{ij} (\overline{\lambda}(\Lambda)) & \equiv & \delta_{ij} D_{\rho_i} .
\end{eqnarray}
The associated scaling fields (or couplings) are then often called superrenormalizable, nonrenormalizable and renormalizable, respectively.

\end{section}

\begin{section}{The relation between the effective potential $L(\phi)$ and the generating functional $Z(J)$ of connected Greens functions}  \label{LZ}
Let $L(\phi, \Lambda, \Lambda_0)$ be a solution of the Polchinski RGE ($\ref{pol}$) where initial conditions $L(\phi, \Lambda_0, \Lambda_0)$ have been specified. As is explained in section ($\ref{RGInt}$), this solution leads to a generating functional 
\begin{equation}
W(J_\Lambda; \Lambda_0)= \int \mathcal{D} \phi e^{- \frac{1}{2} ( \phi, \Delta^{-1}_\Lambda \phi ) + L(\phi, \Lambda, \Lambda_0) + (J_\Lambda, \phi) }  \label{W2}
\end{equation}
that does not depend on the scale $\Lambda$. However, in order to avoid that $L(\phi, \Lambda, \Lambda_0)$ becomes a complicated functional of the source $J_\Lambda$, we had to impose the additional assumption
\begin{equation}
J_\Lambda(k)=0 \ \ \text{for} \ \ k^2 > \Lambda^2.   \label{Jvan}
\end{equation}
We will now briefly discuss how this assumption, which is problematic in practical calculations \cite{Z-J} and might also cause troubles when considering the Slavnov-Taylor-Identities for the generating functional $W(J)$ of a gauge theory, can be avoided. This will be done by relating the effective potential $L(\phi, \Lambda, \Lambda_0)$ to a generating functional for connected Greens functions $\overline{Z}(J;\Lambda, \Lambda_0 )$.

Note that while obtaining the solution $L(\phi, \Lambda, \Lambda_0)$, only field modes corresponding to momenta $\Lambda^2 < k^2 < \Lambda^2_0$ have been integrated out and thus contributed to $L(\phi, \Lambda, \Lambda_0)$. Hence, $\Lambda_0$ and $\Lambda$ can be viewn as effective UV and IR momentum cutoffs for the functional $L(\phi, \Lambda, \Lambda_0)$ respectively.

It is therefore possible to write $L(\phi, \Lambda, \Lambda_0)$ in the following way:
\begin{eqnarray}
e^{ L(\phi, \Lambda, \Lambda_0)} = \int \mathcal{D} \varphi \ e^{-\frac{1}{2} ( \varphi, \Delta^{-1}_{\Lambda, \Lambda_0} \varphi )} e^{L(\phi+ \varphi, \Lambda_0, \Lambda_0)}  \label{AltL}
\end{eqnarray}
where the regularized propagator $\Delta_{\Lambda, \Lambda_0}$ now has a built-in IR momentum cutoff $\Lambda$ (that supresses momenta $k^2< \Lambda^2$) in addition to the UV cutoff $\Lambda_0$. Indeed it can be shown \cite{Mull} that, if we take eq. ($\ref{pol}$)  as a defining relation for $L(\phi, \Lambda, \Lambda_0)$, the dependence of the latter on the IR cutoff $\Lambda$ is again given by Polchinski's equation. 

Let us now introduce the generating functional 
\begin{equation}
\overline{W}({J}; \Lambda, \Lambda_0)= \int \mathcal{D} \varphi e^{- \frac{1}{2} ( \varphi, \Delta^{-1}_{\Lambda, \Lambda_0} \varphi ) + L(\varphi, \Lambda_0, \Lambda_0) + ({J}, \varphi) }  \label{W3}
\end{equation}
which, contrary to ($\ref{W2}$), depends on \textit{both} scales $\Lambda$ and $\Lambda_0$, but whose source term ${J}$ is not supposed to obtain the condition ($\ref{Jvan}$). Obviously, we have
\begin{equation}
\overline{W}(J; 0, \Lambda_0) = W(J_{\Lambda_0}; \Lambda_0).
\end{equation}
By substituting $\varphi \rightarrow \phi-\varphi$ in eq. ($\ref{AltL}$) it follows that $\overline{W}({J}; \Lambda, \Lambda_0)$ and $L(\phi, \Lambda, \Lambda_0)$ are related by the following equation:
\begin{eqnarray}
e^{ L(\phi, \Lambda, \Lambda_0)} = e^{-\frac{1}{2} ( \phi, \Delta^{-1}_{\Lambda, \Lambda_0} \phi )} \ \overline{W}(\Delta^{-1}_{\Lambda, \Lambda_0} \phi; \Lambda, \Lambda_0) .   \label{WL}
\end{eqnarray}
In terms of the generating functional for connected Greens functions $\overline{Z}(J;\Lambda, \Lambda_0 ) = - \ln \big(\overline{W}(\overline{J}; \Lambda, \Lambda_0)\big)$ this becomes
\begin{eqnarray}
{ L(\phi, \Lambda, \Lambda_0)} = {-\frac{1}{2} ( \phi, \Delta^{-1}_{\Lambda, \Lambda_0} \phi )} - \overline{Z}(\Delta^{-1}_{\Lambda, \Lambda_0} \phi; \Lambda, \Lambda_0).  \label{ZJe}
\end{eqnarray}
Hence, we have found an easy relation beween the effective potential $L(\phi, \Lambda, \Lambda_0)$ and the generating functional for connected Greens functions $\overline{Z}(J;\Lambda, \Lambda_0 )$.  Moreover, it is now clear how to overcome the potential troublesome condition ($\ref{Jvan}$) for the source  $J_\Lambda$: the important quantity we have to calculate is $L(\phi, 0, \Lambda_0)$. Once we have done this, we easliy obtain the physical (regularized) generating functional $\overline{Z}(J;0, \Lambda_0 )$ via eq. ($\ref{ZJe}$) and the substitution
\begin{eqnarray}
\phi \rightarrow \Delta_{0, \Lambda_0} J .   \label{phij}
\end{eqnarray}

\end{section}

\begin{section}{Renormalization of operator insertions}   \label{OI}
In this section, we will briefly discuss the renormalization of composite fields, in particular from the viewpoint of the flow equations. We will follow the treatment of Refs \cite{Mull} and \cite{Z-J} and adapt it to our work.

Let $\phi(x)$ be a scalar field in $d=4$ space-time dimensions. A composite field $\mathcal{O}(x)$ is a polynomial formed of the field  $\phi(x)$ and of its space-time derivatives. Typical examples are
\begin{eqnarray}
\mathcal{O}(x) = \phi^2(x), \  \phi(x) \partial^2 \phi(x), \ \phi^4(x) ... 
\end{eqnarray}
We denote the canonical dimension of the composite field by $D_{\mathcal{O}}$. 

In principle, correlation functions involving operator insertions (OI), like
\begin{eqnarray}
\langle \mathcal{O}(x), \phi(x_1), ..., \phi(x_n)  \rangle  , \label{Oex}
\end{eqnarray}
can be obtained from the field correlation functions by letting various points $x_i$ coincide. However, in momentum space this procedure amounts to additional integrations, and hence new divergences appear. For the case of \textit{one} operator insertion as in the example ($\ref{Oex}$), it turns out that in order to cancel the upcoming divergences, counterterms for all local operators (permitted by the field and symmetry content of the QFT) that have canonical dimensions $D \le D_{\mathcal{O}}$ must be added. 

Thus, for an insertion of one $\phi^2(x)$ we would have to add counterterms\footnote{If the QFT is symmetric under $\phi \rightarrow - \phi$, we do not need the counterterm for the operator $\phi(x)$.} for the operators $\phi^2(x)$, $\phi(x)$ and a field independent term, arriving at some bare operator insertion
\begin{eqnarray}
\mathcal{O}^0(x) = R^0_2 \phi^2(x) + R^0_1 \phi(x) + R^0_0 .   \label{exBOI}
\end{eqnarray}
We will restrict our considerations to the case of {one} operator insertion. To discuss the renormalization of a composite field with flow equations, let us introduce a bare extended effective potential
\begin{eqnarray}
\tilde{L}(\phi, \gamma, \Lambda_0) := L(\phi, \Lambda_0) + \int_x \gamma(x) \mathcal{O}(x, \Lambda_0) \label{barepotI}
\end{eqnarray}
where $L(\phi, \Lambda_0)$ is some bare effective potatial $\grave{a}$ la eq. ($\ref{barepot}$), and $\gamma(x)$ is a source that couples to the bare operator insertion $\mathcal{O}(x, \Lambda_0)$. Solving the Polchinski RGE ($\ref{polm}$) employing ($\ref{barepotI}$) as initial condition then yields the running potential
\begin{eqnarray}
\tilde{L}(\phi, \gamma, \Lambda, \Lambda_0) .
\end{eqnarray}
Recall that the derivation of the Polchinski equation does {not} depend on translational invariance of the effective potential \cite{Gerhard}.

We now may define the generating functional for vertex functions with one operator insertion:
\begin{equation}
L_{(1)}(\phi, x, \Lambda, \Lambda_0) :=  \frac{\delta }{\delta \gamma(x)}  \tilde{L}(\phi, \gamma, \Lambda, \Lambda_0)|_{\gamma=0} .   \label{L1ps}
\end{equation}
In the following, we will work in momentum space. The Fourier transform of ($\ref{L1ps}$) is defined as
\begin{eqnarray}
L_{(1)}(\phi, q, \Lambda, \Lambda_0) := \int_x e^{i q x} L_{(1)}(\phi, x, \Lambda, \Lambda_0)  \label{L1ms}.
\end{eqnarray}
The functionals ($\ref{L1ps}$), ($\ref{L1ms}$) obey a RGE linear in $L_{(1)}$, as can be seen by functional differentiation of the  RGE ($\ref{polm}$) for $\tilde{L}(\phi, \gamma, \Lambda, \Lambda_0) $ with respect to $\gamma$:
\begin{eqnarray}
- \Lambda \frac{d}{d \Lambda} L  = \frac{1}{2} \int d^4k   (2 \pi)^4  \Lambda \frac{d}{d \Lambda} \Delta_\Lambda  \left(\frac{\delta L}{\delta \phi(k)} \frac{\delta L_{(1)}}{\delta \phi(-k)}   + \frac{\delta^2 L_{(1)}}{\delta \phi(k) \delta\phi(-k)} + A \delta(k) \frac{\delta L_{(1)}}{\delta \phi(k)}    \right) \nonumber  . \\ \label{polL1}
\end{eqnarray}
The (momentum-space) vertex functions with one operator insertion are the coefficients $L_{(1) n} $ of an expansion of ($\ref{L1ms}$) in powers of the fields $\phi$,
\begin{eqnarray}
L_{(1)}(\phi,q, \Lambda,\Lambda_0) = \sum_{n=1}^{\infty} \frac{1}{n!} \int \frac{d^4 k_1 ... d^4 k_{n}}{(2 \pi )^{4n-4}} L_{(1) n} (q, k_1,...,k_{n}, \Lambda) \delta^4 \big( q+ \sum_i k_i \big) \phi(k_1)... \phi(k_{n}) . \nonumber \\   \label{LexpI}
\end{eqnarray}
Let us now do some dimensional analysis. The canonical dimension of the bare position-space operator insertion $\mathcal{O}(x, \Lambda_0)$ has been denoted by $D_{\mathcal{O}}$, and from eq. ($\ref{L1ps}$) follows that we also have 
\begin{eqnarray}
[L_{(1)}(\phi, x, \Lambda, \Lambda_0) ] = D_{\mathcal{O}}.
\end{eqnarray}
The Fourier transformed functional ($\ref{L1ms}$) then has canonical dimension
\begin{eqnarray}
[L_{(1)}(\phi, q, \Lambda, \Lambda_0) ] = D_{\mathcal{O}}-4.
\end{eqnarray}
As is  discussed in Appendix ($\ref{L1ms}$), a propagator that goes like $1/k^2$ implies canonical dimension $D_{\phi(k)}=-3$ of the (momentum-space) field $\phi(k)$. Thus eq. ($\ref{LexpI}$) tells us that the canonical dimension of the momentum-space vertex functions with one operator insertion equals
\vspace{-0.5ex}
\begin{equation}
D_{ L_{(1) n} } = D_{\mathcal{O}} -n.     \label{DLnOI}
\end{equation}
Compare this to the dimension $D_{L_n}=4-n$ of the vertex functions without operator insertions introduced in section ($\ref{RGIs}$). 

By Taylor expanding the vertex functions $L_{(1) n} (q, k_1,...,k_{n}, \Lambda)$ around $q=k_i=0$ (see eqns. ($\ref{TE1}$), ($\ref{TE2}$)) we may define running coupling constants $R_i$:
\begin{eqnarray} 
R_0(\Lambda) &:=& L_{(1) 0}(0,\Lambda) \label{rcc0OI} \\
R_1(\Lambda) &:=& L_{(1) 1}(0,0, \Lambda) \label{rcc1OI} \\
R_2(\Lambda) &:=& L_{(1) 2}(0,0,0,\Lambda) \\
R_3(\Lambda) \ \delta^{\mu \nu} &:=& \partial^\mu_{q,1} \partial^\nu_{q,1} L_{(1)1}(q,k_1,\Lambda)|_{q=k_1=0}  \label{rcc3OI} \\
& \vdots &   \nonumber
\end{eqnarray}
In the following, we will distinguish between renormalizable couplings $R_a$ and nonrenormalizable couplings $R_n$ by their canonical dimensions:
\vspace{-0.4ex}
\begin{eqnarray} 
{R}_a: && D_{R_a} \ge 0 \\
{R}_n: && D_{R_l} < 0. 
\end{eqnarray}
Note that the canonical dimensions $D_{R_i}$ follow from the definitions ($\ref{rcc0OI}$)- ($\ref{rcc3OI}$). We introduce renormalization conditions for the renormalizable couplings $R_a$
\begin{equation} 
R_a^R := R_a(\Lambda_R), \ \ \ D_{R_a} \ge 0 \label{rccOI}
\end{equation}
where $\Lambda_R < \Lambda_0$ is some renormalization scale. This is of course related to our discussion of counterterms for operator insertions at the beginning of this section. 

Finally, initial conditions for the \textit{least} irrelevant couplings $R_l^0$ (having canonical dimensions\footnote{For simplicity, we assume that the field and symmetry content of the theory is such that the least irrelevant couplings indeed have $D_{R_l} = -1$.} $D_{R_l} = -1$) are specified at the bare scale:
\begin{equation} 
R_l^0 := R_l(\Lambda_0), \ \ \ D_{R_l} = -1.  \label{iniOI}
\end{equation}
Introducing the dimensionless coupling constants
\begin{eqnarray} 
\mathcal{R}_a^R (\Lambda) &=& \Lambda^{-D_{R_a}} R_a^R  \label{ROI1} \\
\mathcal{R}_l^0 (\Lambda) &=& \Lambda \ R_l^0  \label{ROI2}
\end{eqnarray}
we impose as an additional constraint to the renormalization and initial conditions that the dimensionless couplings be small on scales $\Lambda_R \le \Lambda \le \Lambda_0$:
\begin{eqnarray} 
\mathcal{R}_a^R (\Lambda) &\le& 1 \\ 
\mathcal{R}_l^0 (\Lambda)   &\le& 1  . \label{smallOI}
\end{eqnarray}
In the following, we will employ the notations of chapters ($\ref{RenFlow}$) and ($\ref{EffFlow}$). Expanding the vertex functions $L_{(1) n}$
in the renormalized renormalizable couplings $R_a^R$, $\rho_4^R$ and $\rho_5^R$ introduced in eqns. ($\ref{rccOI}$), ($\ref{c3}$) and ($\ref{c4}$) and in the bare nonrenormalizable couplings $R_l^0$ and $\rho_{\tilde{a}}^0, \  \tilde{a} =6...8$, of eqns.  ($\ref{iniOI}$)  and ($\ref{iniCNR}$) yields
\begin{eqnarray}
L_{(1)n} (q,k_1,...,k_{n}, \Lambda) & = & \sum_{r_1,...,r_{5}=0}^{\infty} \ \sum_{i}  R_i^{R/0} \ (\rho_4^R)^{r_1} (\rho_5^R)^{r_2} \nonumber \\ &&  \hspace{2.5cm} (\rho_{6}^0)^{r_3}  (\rho_{7}^0)^{r_4}  (\rho_{8}^0)^{r_5}  L_{(1)n}^{(i, r_1,...,r_5)} (q,k_1,...,k_{n}, \Lambda) \nonumber  \\  \label{GenPertI}
\end{eqnarray}
where $R_i^{R/0}$ means $R_a^R$ or $R_l^0$. Hence, because of the operator insertion each graph contributing to ($\ref{GenPertI}$) contains \textit{one} extra vertex associated with a renormalized coupling constant ($\ref{rccOI}$) \textit{or} a bare one ($\ref{iniOI}$). To distiguish between the  two possibilities, we introduce the symbols $\Theta_i^R$ and  $\Theta_i^0$:
\begin{eqnarray}
D_{R_i} \ge 0: && \Theta_i^R:=1 \  \wedge \  \Theta_i^0:=0  \\ D_{R_i} =-1:&&   \Theta_i^R:=0  \  \wedge \  \Theta_i^0:=1
\end{eqnarray}
for a graph $L_{(1)n}^{(i, r_1,...,r_5)}$ associated with a coupling $R_i^{R/0}$. Employing dimensionless vertex functions
\begin{eqnarray}
A_{(1)n}^{(i,r_1, ..., r_5)} (q,k_1,...,k_{n}, \Lambda) := \Lambda^{n-D_{\mathcal{O}} +r_1 + \Theta_i^R D_{R_i}- (r_3 +r_4 +r_5 + \Theta_i^0 )}  L^{(i, r_1, ..., r_5)}_{(1)n} (q,k_1,...,k_{n}, \Lambda) \nonumber \\  \label{Apert2I_0}.
\end{eqnarray}
and the dimensionless couplings ($\ref{ROI1}$), ($\ref{dCR}$), ($\ref{ROI2}$) and ($\ref{dCNR}$), eq. ($\ref{GenPertI}$) becomes 
\begin{eqnarray}
A_{(1)n} (q,k_1,...,k_{n}, \Lambda) &=& \sum_{r_1, ...,r_{5}=0}^{\infty} \ \sum_{i} \mathcal{R}_i^{R/0} \ (\lambda_4^R)^{r_1} (\lambda_5^R)^{r_2} \nonumber \\  &&  \hspace{2.5cm}  (\lambda_{6}^0)^{r_3} (\lambda_{7}^0)^{r_4} (\lambda_{8}^0)^{r_5}  A_{(1)n}^{(i, r_1,...,r_5)} (q,k_1,...,k_{n}, \Lambda). \nonumber \\ \label{Apert2I}
\end{eqnarray}
We may now reformulate the RGE ($\ref{polL1}$) in terms of the dimensionless vertex functions ($\ref{Apert2I}$). Applying the bounds ($\ref{q1}$) and ($\ref{q2}$) as well as the condition ($\ref{smallg}$) for the renormalization constant $A$ we arrive at the RG inequality (RGI)
\begin{eqnarray} 
&& \hspace{-0.7cm} \lVert \frac{d}{d \Lambda} \Lambda^{D_{\mathcal{O}}-n-r_1 -\Theta_i^R D_{R_i} + r_{NR} + \Theta_i^0} \ \partial^p A^{(i, r_1,..., r_5)}_{(1)n}(\Lambda)  \rVert \nonumber \\   && \qquad \qquad \le   c_{n,p} \ \Lambda^{D_{\mathcal{O}}-1-n-r_1-\Theta_i^R D_{R_i} + r_{NR} + \Theta_i^0 } \Bigg( \lVert \partial^p A^{(i,r_1,..., r_5)}_{(1)n+2}(\Lambda) \rVert + \lVert \partial^p A^{(i,r_1,..., r_5)}_{(1)n+1}(\Lambda) \rVert \nonumber \\ && \qquad \qquad \qquad  \qquad \qquad +  2 \sum_{...} \Lambda^{-p_1} \lVert \partial^{p_2} A_{l}^{(s_1,..., s_5)}(\Lambda) \rVert \lVert \partial^{p_3} A^{(i,r_1-s_1,..., r_5-s_5)}_{(1)n+2-l}(\Lambda) \rVert \Bigg) \nonumber \\ \label{RGI1_NRI}
\end{eqnarray} 
where $\sum_{...}$ is defined as in eq. ($\ref{sumabb}$). Observe that the RGI ($\ref{RGI1_NRI}$) is again linear in $A_{(1) n}$.

Eq. ($\ref{RGI1_NRI}$) is the analogon to the key RG inequality ($\ref{RGI1_NR}$). We may therefore proceed as in chapter ($\ref{EffFlow}$) and establish the boundedness of vertex functions with one operator insertion. Since the RGI ($\ref{RGI1_NRI}$) involves also vertex functions without operator insertion, the bounds ($\ref{BoundNR}$) that have been established in Theorem ($\ref{BoundThII}$) will have to be employed. We will just state the Theorem without proving it:

\begin{satz}[Boundedness of Vertex Functions with OI] \label{BoundThOI}
Given Theorem ($\ref{BoundThII}$), the renormalization conditions ($\ref{rccOI}$), the initial conditions ($\ref{iniOI}$) and assuming that
\begin{eqnarray}
||\partial^p A_{(1)n}^{(i, r_1,..., r_5)}(q, p_1,...,p_{n}, \Lambda_0)|| \le \Lambda_0^{-p}  \left( \frac{\Lambda_0}{\Lambda_R} \right)^{r_1+ \Theta_i^R D_{R_i}} Pln \left( \frac{\Lambda_0}{\Lambda_R} \right) \label{iniNRI}
\end{eqnarray} 
for $n+p \ge D_{\mathcal{O}} + 2$, to order $r_1,..., r_5$ in perturbation theory in the couplings $\lambda_4^R$, $\lambda_5^R$, $\lambda_{6}^0$,  $\lambda_{7}^0$ and $\lambda_{8}^0$
\begin{eqnarray}
&& ||\partial^p A_{(1)n}^{(i,r_1,..., r_5)}(p_1,...,p_{n}, \Lambda)|| \nonumber \\ &&  \qquad \qquad   \le \Lambda^{-p} \left( \frac{\Lambda}{\Lambda_R} \right)^{r_1+\Theta_i^R D_{R_i}} \left( \frac{\Lambda_0}{\Lambda} \right)^{r_{NR} + \Theta_i^0}   \left( \Theta_i^R \ Pln\left( \frac{\Lambda}{\Lambda_R} \right) +  \frac{\Lambda}{\Lambda_0}  Pln \left( \frac{\Lambda_0}{\Lambda_R} \right) \right) \nonumber \\ \label{BoundNRI}
\end{eqnarray}
where the index $i$ refers to an extra vertex associated with a renormalized coupling ${R}_i^R$ having canonical dimension $D_{R_i} \ge 0$ ($\Theta_i^R=1$, $\Theta_i^0=0$) or a bare one ${R}_i^0$ having canonical dimension $D_{R_i} = -1$ ($\Theta_i^0=1$, $\Theta_i^R=0$), and $\Lambda_R \le \Lambda \le \Lambda_0$.
\end{satz}
Let us introduce \textit{improvement conditions} for the least irrelevant couplings $R_l$:
\begin{eqnarray} 
R_l^{NR} := R_l(\Lambda_R), \ \ \ D_{R_l} = -1.  \label{ImproOI}
\end{eqnarray}
It is assumed that the improvement conditions ($\ref{ImproOI}$) are taken such that they are compatible\footnote{In analogy to the discussion of section ($\ref{InvSec}$), this will amount to the requirement that the deviations of the improvement conditions ($\ref{ImproOI}$) from the values the $R_l(\Lambda_R), \ D_{R_l} = -1$, take for vanishing initial conditions ($\ref{iniOI}$) be small. See Theorem ($\ref{invTh}$) for details.} with small initial conditions ($\ref{iniOI}$), where small is meant in the sense of eq. ($\ref{smallOI}$).

In order to investigate the predictivity of an effective field theory with one operator insertion, we proceed as in section ($\ref{Presec}$). Hence, the bare vertex functions are parametrized:
\begin{eqnarray}
 t \ \partial^p A_{(1)n}^{(i,r_1, ..., r_5)}( \Lambda_0), \ \ \ t \in [0,1], \ \ \ n+p \ge D_{\mathcal{O}} + 2    \label{shapeANRI}
\end{eqnarray}
where $\partial^p A_{(1)n}^{(i,r_1, ..., r_5)}(\Lambda_0)$ are small $\grave{a}$ la eq. ($\ref{iniNRI}$). This leads to running vertex functions depending on the paramenter $t$,
\begin{eqnarray}
 \partial^p A_{(1)n}^{(i,r_1, ..., r_5)}(q,p_1,...,p_{n}, \Lambda, \Lambda_0, t).
\end{eqnarray}
We now state the Theorem concerning the predictivity of effective field theories with one operator insertion. It is the analogon to Theorem ($\ref{PreTh}$).

\begin{satz}[Predictivity of Effective Field Theories with OI] \label{PreThI} Let there be renormalization conditions  ($\ref{rccOI}$) and improvement conditions ($\ref{ImproOI}$). Assume that to order $r_1, ..., r_5$ in perturbation theory in $\lambda_4^R$, $\lambda_5^R$, $\lambda_{6}^0$,  $\lambda_{7}^0$ and $\lambda_{8}^0$
\begin{equation}
||\partial^p A_{(1)n}^{(i, r_1,..., r_5)}(q, p_1,...,p_{n}, \Lambda)|| \le \Lambda^{-p}  \left( \frac{\Lambda}{\Lambda_R} \right)^{r_1+\Theta_i^R D_{R_i}}  \left( \frac{\Lambda_0}{\Lambda} \right)^{r_{NR}+ \Theta_i^0} Pln \left( \frac{\Lambda_0}{\Lambda_R} \right), \label{BoundWPI}
\end{equation}
and that for $n+p \ge D_{\mathcal{O}} + 2$
\begin{eqnarray}
|| \frac{d}{d t} \partial^p A_{(1)n}^{(i,r_1,..., r_5)}(q, p_1,...,p_{n}, \Lambda_0)|| \le \Lambda_0^{-p}  \left( \frac{\Lambda_0}{\Lambda_R} \right)^{r_1+\Theta_i^R D_{R_i}} Pln \left( \frac{\Lambda_0}{\Lambda_R} \right) . \label{ini2PI}
\end{eqnarray}
Given Theorems ($\ref{BoundThII}$), ($\ref{BoundThOI}$) and ($\ref{PreTh}$) we then have 
\begin{eqnarray}
|| \frac{d}{d t} \partial^p A_{(1)n}^{(i, r_1,..., r_5)}(q, p_1,...,p_{n}, \Lambda)|| \le \Lambda^{-p}  \left( \frac{\Lambda}{\Lambda_R} \right)^{r_1+\Theta_i^R D_{R_i}}  \left( \frac{\Lambda_0}{\Lambda} \right)^{r_{NR}+ \Theta_i^0} \left(\frac{\Lambda}{\Lambda_0}\right)^{2} Pln \left( \frac{\Lambda_0}{\Lambda_R} \right) \nonumber \\ \label{PreI}
\end{eqnarray}
where the index $i$ refers to an extra vertex associated with a renormalized coupling ${R}_i^R$ having canonical dimension $D_{R_i} \ge 0$ ($\Theta_i^R=1$, $\Theta_i^0=0$) or a bare one ${R}_i^0$ having canonical dimension $D_{R_i} = -1$ ($\Theta_i^0=1$, $\Theta_i^R=0$), and $\Lambda_R \le \Lambda \le \Lambda_0$.
\end{satz}
Eq. ($\ref{PreI}$) can be converted into an equation similar to eq. ($\ref{PreF}$) and thus states the indetermination of the running vertex functions $A_{(1)n}(\Lambda)$ as a result of the ignorance about the exact initial values ($\ref{shapeANRI}$). Please refer to the discussion of section ($\ref{Presec}$) for more details.

Finally, we would also like to discuss the convergence of the vertex functions with operator insertion to a no-cutoff limit. To do so, we impose the renormalization conditions ($\ref{rccOI}$) and no improvement conditions. Furthermore, we expand the vertex functions $A_{(1)n}(\Lambda)$ in perturbation theory solely in the renormalizable couplings $\lambda_4^R$ and $\lambda_5^R$, which can also be seen as the $0th$ order in perturbation theory in the bare nonrenormalizable couplings $\lambda_{6}^0$,  $\lambda_{7}^0$ and $\lambda_{8}^0$. Finally, each contributing graph contains one extra vertex associated with a renormalized coupling constant ($\ref{rccOI}$), but \textit{not} with a bare nonrenormalizable one ($\ref{iniOI}$). We then obtain the following Theorem:

\begin{satz}[Convergence of Vertex Functions with OI] \label{ConvThI} Let there be renormalization conditions ($\ref{rccOI}$). Assume that to order $r_1, r_2$ in perturbation theory in $\lambda_4^R$ and $\lambda_5^R $ 
\begin{eqnarray}
||\partial^p A_{(1)n}^{(i, r_1, r_2)}(q, p_1,...,p_{n}, \Lambda)|| \le \Lambda^{-p}  \left( \frac{\Lambda}{\Lambda_R} \right)^{r_1+  D_{R_i}} Pln \left( \frac{\Lambda_0}{\Lambda_R} \right), \label{BoundWI}
\end{eqnarray}
and that for $n+p \ge D_{\mathcal{O}} + 1$
\begin{eqnarray}
|| \Lambda_0 \frac{d}{d \Lambda_0} \partial^p A_{(1)n}^{(i,r_1, r_2)}(q, p_1,...,p_{n}, \Lambda_0)|| \le \Lambda_0^{-p}  \left( \frac{\Lambda_0}{\Lambda_R} \right)^{r_1+ D_{R_i}} Pln \left( \frac{\Lambda_0}{\Lambda_R} \right) . \label{ini2I}
\end{eqnarray}
Given Theorems ($\ref{BoundTh}$), ($\ref{BoundThOI}$) and ($\ref{ConvTh}$) we then have
\begin{eqnarray}
||\Lambda_0 \frac{d}{d \Lambda_0} \partial^p A_{(1)n}^{(i,r_1, r_2)}(q, p_1,...,p_{n}, \Lambda)|| \le \Lambda^{-p}  \left( \frac{\Lambda}{\Lambda_R} \right)^{r_1+ D_{R_i}} \frac{\Lambda}{\Lambda_0} Pln \left( \frac{\Lambda_0}{\Lambda_R} \right)  \label{ConvI}
\end{eqnarray}
where the index $i$ refers to an extra vertex associated with a renormalized coupling ${R}_i^R$ having canonical dimension $D_{R_i} \ge 0$, and  $\Lambda_R \le \Lambda \le \Lambda_0$.
\end{satz}
It will turn out that we will also have to consider space-time integrated operator insertions
\begin{eqnarray}
L_{(1)}(\phi, \Lambda_0):= \int_x \mathcal{O}(x, \Lambda_0),   \label{SIOI}
\end{eqnarray}
leading to a bare extended effective potential
\begin{eqnarray}
{L}_\chi(\phi, \Lambda_0) := L(\phi, \Lambda_0) + \chi L_{(1)}(\phi, \Lambda_0) \label{barepotI2}
\end{eqnarray}
with $\chi \in \mathbb{R}$. For the corresponding running potential ${L}_ \chi(\phi,\Lambda, \Lambda_0)$ we find
\begin{eqnarray}
L_{(1)}(\phi, \Lambda, \Lambda_0) := \frac{d}{d \chi} {L}_\chi(\phi,\Lambda, \Lambda_0)|_{\chi=0} &=& \int_x L_{(1)}(x, \phi, \Lambda, \Lambda_0) \nonumber \\ &=& L_{(1)}(q, \phi, \Lambda, \Lambda_0)|_{q=0} 
\end{eqnarray}
where $ L_{(1)}(x, \phi, \Lambda, \Lambda_0)$ and $L_{(1)}(q, \phi, \Lambda, \Lambda_0)$ have been defined in eqns. ($\ref{L1ps}$) and ($\ref{L1ms}$) respectively. It is therefore clear that the generating functional $L_{(1)}(\phi, \Lambda, \Lambda_0)$ obeys a space-time integrated analogon of the RGE ($\ref{polL1}$), and that the vertex functions
\begin{eqnarray}
\delta^4(k_1 + ... + k_n ) L_{(1)n}(k_1, ..., k_n, \Lambda) = (2 \pi )^{4n} \frac{\delta}{\delta \phi(k_1)} ... \frac{\delta}{\delta \phi(k_n)}  L_{(1)}(\phi, \Lambda) \big|_{\phi=0}    \label{LexpGI}
\end{eqnarray}
again have canonical dimension
\begin{eqnarray}
D_{ L_{(1) n} } = D_{\mathcal{O}} -n.
\end{eqnarray}
We may therefore proceed for space-time integrated operator insertions as we have done in Theorems ($\ref{BoundThOI}$), ($\ref{PreThI}$) and ($\ref{ConvThI}$). In particular, we will somewhat sloppily talk about the canonical dimension of the insertion $L_{(1)}$ as given by $D_{\mathcal{O}}$.

\end{section}

\begin{section}[Conventions for the gravity action]{Conventions for the gravity action and the sign of the cosmological constant}  \label{Conv}
For the metric signature $(-1,+1,+1,+1)$ and the definitions
\begin{eqnarray}
R_{\mu \nu} &=& R^{\alpha}_{\ \mu \alpha \nu} \\
R &=& R^\mu_{\ \mu} 
\end{eqnarray}
of the Ricci tensor and the curvature scalar, the  action for gravity with a cosmological constant coupled to matter is \cite{Carroll} \cite{MTW} 
\begin{eqnarray}
S &=& S_{EH} + S_{M} \\ &=& \frac{1}{{\lambda^2}} \int d^4 x \sqrt{-g} \left(- {4\Lambda_K} + {2} R \right) + S_M
\end{eqnarray}
where $\lambda^2 := 32 \pi G$ with $G$ being Newton's constant. Employing the variational principle
\begin{eqnarray}
\delta_{g_{\mu \nu}} S =0
\end{eqnarray}
one obtains the field equations 
\begin{eqnarray}
R^{\mu \nu} - \frac{1}{2} g^{\mu \nu} R + \Lambda_K g^{\mu \nu} = \frac{1}{4} \lambda^2 T^{\mu \nu}  \label{EEQ}
\end{eqnarray}
where $T^{\mu \nu}$ is the energy-momentum tensor defined by
\begin{eqnarray}
\delta_{g_{\mu \nu}} S_{M} =  - \frac{1}{2} \int d^4x \sqrt{-g} T^{\mu \nu}  \delta {g_{\mu \nu}} .   \label{EMT}
\end{eqnarray}
$T^{\mu \nu}$ is symmetric in ${\mu, \nu}$ by definition. Note that with the Bianchi identity 
\begin{eqnarray}
D_\lambda R^\rho_{\ \sigma \mu \nu} + D_\nu R^\rho_{\ \sigma \lambda \mu } + D_\mu R^\rho_{\ \sigma \nu \lambda} =0
\end{eqnarray}
and the condition for a metric connection $D_\rho g_{\mu \nu}=0$ the field equation ($\ref{EEQ}$) implies that
\begin{eqnarray}
D_\nu T^{\mu \nu} =0.  \label{MEQ}
\end{eqnarray}
Here, $D_\mu$ denotes the covariant derivative. Consider the expansion
\begin{eqnarray}
g_{\mu \nu}= \eta_{\mu \nu} + \lambda h_{\mu \nu}
\end{eqnarray}
where $\eta_{\mu \nu}$ is the Minkowski metric. In the weak field approximation, eq. ($\ref{EMT}$) implies that gravity can be coupled to matter via the term
\begin{eqnarray}
- \frac{1}{2} \int d^4x \ \tilde{T}^{\mu \nu}  h_{\mu \nu}
\end{eqnarray}
where $\tilde{T}^{\mu \nu}$ is the flat spacetime energy-momentum tensor of all matter fields. Moreover, eq.  ($\ref{MEQ}$)  reduces to $0th$ order in $h_{\mu \nu}$ to 
\begin{eqnarray}
\partial_\nu \tilde{T}^{\mu \nu} =0.  \label{MEQW}
\end{eqnarray}
Thus, $\tilde{T}^{\mu \nu}$ can be viewn as a physical ''source'' that is coupled to the gravitational field $h_{\mu \nu}$.

In the following, we will investigate somewhat further the implications of the cosmological constant $\Lambda_K$. For $T^{\mu \nu}=0$ eq. ($\ref{EEQ}$) can be written as
\begin{eqnarray}
R^{\mu \nu} - \frac{1}{2} g^{\mu \nu} R  = - \Lambda_K g^{\mu \nu} .   \label{EEQV}
\end{eqnarray}
Thus the cosmological constant can be viewn as a kind of energy-momentum tensor associated with the vacuum. Since the $T^{00}$ component of the energy-momentum tensor is the energy density, we see that for $\Lambda_K>0$ the cosmological constant corresponds to a \textit{positive} energy density of the vacuum. 

Taking the trace of eq. ($\ref{EEQV}$) leads to
\begin{eqnarray}
 R  = 4 \Lambda_K .   \label{MS}
\end{eqnarray}
This means that there are solutions of ($\ref{EEQV}$) that are maximally symmetric, i.e. that have 10 Killing vectors \cite{Bronstein2}. The two different spaces associated with these solutions for $\Lambda_K>0$  and $\Lambda_K<0$ are known as de Sitter and anti de Sitter space, respectively.

An unregularized generating functional of pure quantum Einstein gravity as given by the action $S_{EH}$ can be formally defined as 
\begin{eqnarray}
W(\mathcal{J}) = \int \mathcal{D} h_{\mu \nu} \mathcal{D} C^\mu \mathcal{D} \overline{C}_\mu \ e^{ i \left( S_{tot} + S_{\mathcal{J}} \right)} 
\end{eqnarray}
where
\begin{equation}
S_{tot}(h, C, \overline{C})= S_{EH}(h) + S_{GF}(h )+ S_{GH}(h, C, \overline{C})  .
\end{equation}
$S_{GF}$ and $ S_{GH}$ are gauge fixing and ghost terms, whereas $S_{\mathcal{J}}$ is a general source term (see chapter ($\ref{PolGR}$) for details).

In order to employ a cutoff regularization and to consider the renormalization group flow of quantum gravity, we have to perform a Wick rotation to Euclidean space. This amounts to the substitution \cite{Ramond}
\begin{eqnarray}
x_0 \rightarrow -i x_0 . 
\end{eqnarray}
For our metric signature we then have
\begin{eqnarray}
d^4x &\rightarrow& -i d^4x^E \\
\partial_\mu \partial^\mu & \rightarrow & \partial^E_\mu \partial^E_\mu.
\end{eqnarray}
The Euclidean generating functional is therefore, in accordance with Ref \cite{MR1},
\begin{eqnarray}
W^E(\mathcal{J}) = \int \mathcal{D} h_{\mu \nu} \mathcal{D} C^\mu \mathcal{D} \overline{C}_\mu \ e^{ \left( S^E_{tot} + S^E_{\mathcal{J}} \right)} .
\end{eqnarray}
Since we will always work in Euclidean space, we will leave out the index $E$ from now on and we will also reemploy the summation convention with up and down indices.

\end{section}

\begin{section}{The contravariant metric density}   \label{cgdens}
For convenience of the reader, the properties of the contravariant Euclidean metric density 
\begin{eqnarray}
\tilde{g}^{\mu \nu}:= \sqrt{g} \ g^{\mu \nu} \label{defgdens}
\end{eqnarray}
with $g = \det g_{\mu \nu} $ are reviewed. We consider a coordinate transformation
\begin{eqnarray}
x^\alpha{'}(x^\mu) . \label{CT}
\end{eqnarray}
The (inverse) metric tensor transfoms under ($\ref{CT}$) as
\begin{eqnarray}
g^{\alpha \beta}{'} &=& \frac{\partial x^\alpha{'}}{\partial x^\mu} \frac{\partial x^\beta{'} }{\partial x^\nu} g^{\mu \nu}  \\
g{'}_{\alpha \beta} &=& \frac{\partial x^\mu}{\partial x^\alpha{'}} \frac{\partial x^\nu }{ \partial x^\beta{'}} g_{\mu \nu} .
\end{eqnarray}
With the properties of determinants follows the transformation law for $\tilde{g}^{\mu \nu}$:
\begin{eqnarray}
\tilde{g}^{\alpha \beta}{'}= \left| \det \frac{\partial x^\alpha{'}}{\partial x^\mu} \right|^{-1} \frac{\partial x^\alpha{'}}{\partial x^\mu} \frac{\partial x^\beta{'} }{\partial x^\nu} g^{\mu \nu} .
\end{eqnarray}
This is the transformation law \cite{Bronstein2} for a rank 2 contravariant tensor density of weight $+1$.

We will now derive the Lie derivative of $\tilde{g}^{\mu \nu}$. We have
\begin{eqnarray}
\mathcal{L}_X \tilde{g}^{\mu \nu} = \big( \mathcal{L}_X \sqrt{g} \big) g^{\mu \nu} + \sqrt{g} \mathcal{L}_X  g^{\mu \nu} 
\end{eqnarray}
where \cite{Nakahara}
\begin{eqnarray}
\mathcal{L}_X  g^{\mu \nu} = X^\rho \partial_\rho g^{\mu \nu} - g^{\nu \rho}  \partial_\rho X^\mu - g^{\mu \rho}  \partial_\rho X^\nu .
\end{eqnarray}
Therefore it remains to derive $\mathcal{L}_X \sqrt{g}$. With $\det g_{\mu \nu} = e^{ tr \ln g_{\mu \nu}}$ and the chain rule we find
\begin{eqnarray}
\mathcal{L}_X \sqrt{g} = \frac{1}{2} \sqrt{g} \ g^{\mu \nu} \mathcal{L}_X  g_{\mu \nu} .  \label{lieg}
\end{eqnarray}
Using
\begin{eqnarray}
 \mathcal{L}_X  g_{\mu \nu} = X^\rho \partial_\rho g_{\mu \nu} + g_{\nu \rho}  \partial_\mu X^\rho + g^{\mu \rho}  \partial_\nu X^\rho
\end{eqnarray}	
and, analogous to ($\ref{lieg}$), 
\begin{eqnarray}
 \frac{1}{2} \sqrt{g} \ g^{\mu \nu} X^\rho \partial_\rho  g_{\mu \nu} =X^\rho \partial_\rho  \sqrt{g} ,
\end{eqnarray}
 this can be written as
\begin{eqnarray}
\mathcal{L}_X \sqrt{g} = X^\rho \partial_\rho \sqrt{g} + \sqrt{g} \partial_\rho X^\rho .
\end{eqnarray}
Putting it all together we finally arrive at
\begin{eqnarray}
\mathcal{L}_X \tilde{g}^{\mu \nu} = X^\rho \partial_\rho \tilde{g}^{\mu \nu} + \tilde{g}^{\mu \nu} \partial_\rho X^\rho - \tilde{g}^{\rho \nu} \partial_\rho  X^\mu - \tilde{g}^{\mu \rho} \partial_\rho X^\nu. 
\end{eqnarray}
To conclude, we will show that the covariant derivative of $\tilde{g}^{\mu \nu}$ vanishes. For a metric connection, we have $D_\rho g_{\mu \nu}=D_\rho g^{\mu \nu}=0$. Thus, with the Leibnitz rule and the definition ($\ref{defgdens}$) 
\begin{eqnarray}
D_\rho \tilde{g}^{\mu \nu} = \big( D_\rho \sqrt{g} \big) g^{\mu \nu} .
\end{eqnarray}
In complete analogy to ($\ref{lieg}$) it can be shown that
\begin{eqnarray}
D_\rho \sqrt{g} = \frac{1}{2} \sqrt{g} \ g^{\mu \nu} D_\rho   g_{\mu \nu} .
\end{eqnarray}
Therefore we indeed obtain
\begin{eqnarray}
D_\rho \tilde{g}^{\mu \nu} = 0 .
\end{eqnarray}

\end{section}

\end{appendix}

\nocite{BB}
\nocite{Sun}
\nocite{Morris}
\nocite{Bag}
\nocite{Donoint}

\addtocontents{toc}{\bigskip \textbf{Bibliography}} 
\bibliography{LiteraturVZD}

\newpage
\thispagestyle{empty}
\vspace*{5ex}
\LARGE
\noindent
\textbf{Danksagung}

\vspace{2ex}
\normalsize
\noindent
F\"ur die Vergabe des interessanten Themas sowie die gute und freundliche Betreuung, insbesondere w\"ahrend der Endphase meiner Arbeit, m\"ochte ich mich herzlich bei meinem Betreuer Prof. Dr. Gerhard Mack bedanken. 

\smallskip
\noindent
Weiterhin danke ich Herrn Prof. Dr. Martin Reuter herzlich f\"ur hilfreiche Diskussionen und die Gelegenheit, meine Arbeit an der Universit\"at Mainz vorzustellen.

\smallskip
\noindent
Dank f\"ur Hilfe, Gespr\"ache vielerlei Art sowie konstruktive Kritik geb\"uhrt dar\"uber hinaus Thorsten Pr\"ustel, Florian Schwennsen, Olaf Hohm, Michael Olschewsky, Michael R\"ohrs, Martin Hentschinski, Iman Benmachiche und Niels Emil Jannik Bjerrum-Bohr.

\smallskip
\noindent
Ich danke der Deutschen Forschungsgemeinschaft (DFG) f\"ur die finanzielle F\"orderung meiner Arbeit im Rahmen des Graduiertenkollegs ''Zuk\"unftige Entwicklungen in der Teilchenphysik''.

\smallskip
\noindent
Zuletzt danke ich meinen Eltern sehr herzlich f\"ur ihre vielf\"altige Unterst\"utzung.

\end{document}